\documentclass[11pt]{amsart}
\usepackage{geometry}        
\usepackage{latexsym,amsfonts,amsmath,amssymb,amsthm}
\usepackage[normalem]{ulem}
\usepackage[norelsize]{algorithm2e}
\usepackage[usenames,dvipsnames]{xcolor}
\usepackage{graphicx,tikz, tkz-euclide, float, pgfplots, overpic, cancel}
\usepackage{subfigure}
\usetikzlibrary{arrows,calc, decorations.markings, positioning, fadings}
\usetkzobj{all}
\usepackage{mathtools} 
\usepackage{newfloat}
\DeclareFloatingEnvironment[
    fileext=los,
    listname=List of Schemes,
    name=Notes on Algorithm,
    placement=tbhp,
]{scheme}

\newcommand{\RHP}{RH problem}
\newcommand{\bigo}{\mathcal O}

\newcommand{\defeq}{\vcentcolon=}

\newtheorem{theorem}{Theorem}
\newtheorem{lemma}{Lemma}[section]
\newtheorem{prop}[lemma]{Proposition}

\theoremstyle{definition}
\newtheorem{definition}[lemma]{Definition}
\newtheorem{remark}[lemma]{Remark}
\newtheorem{rhp}{RH Problem}
\newcommand{\rhref}[1]{RH Problem~\ref{#1}}

\let\Re=\undefined\DeclareMathOperator{\Re}{Re}
\let\Im=\undefined\DeclareMathOperator{\Im}{Im}
\DeclareMathOperator{\diag}{diag}
\DeclareMathOperator{\sign}{sign}
\DeclareMathOperator{\sech}{sech}
\newcommand{\half}{\frac{1}{2}}
\newcommand{\goto}{\rightarrow}

\let\llldots=\ldots
\def\ldots{\llldots{}}
\DeclareSymbolFont{yhlargesymbols}{OMX}{yhex}{m}{n}
\DeclareMathAccent{\wideparen}{\mathord}{yhlargesymbols}{"F3}	

\newcommand{\oarc}[1]{\left(#1 \right)_{\text{arc}}}	
\newcommand{\carc}[1]{\left[#1 \right]_{\text{arc}}}

\title[Numerical inverse scattering for the Toda lattice]{Numerical inverse scattering for the Toda lattice}
\author[D.~Bilman]{Deniz Bilman}
\address{Deniz Bilman\\
     Department of Mathematics\\
     University of Michigan\\
     530 Church Street\\
     Ann Arbor, MI 48109
     }
\email{bilman@umich.edu}
\author[T.~Trogdon]{Thomas Trogdon}
\address{Thomas Trogdon\\
     University of California, Irvine\\
     Rowland Hall\\
     Irvine, CA 92697
   }
\email{ttrogdon@uci.edu}

\thanks{The authors wish to thank Percy Deift, Peter Miller, and Irina Nenciu for useful discussions and suggestions. The authors also thank the anonymous referees for their suggestions that greatly improved the readability of our paper. DB gratefully acknowledges the hospitality of Courant Institute of Mathematical Sciences, where the majority of this work was done.  The authors acknowledge the partial support of the National Science Foundation through the NSF grants DMS-1150427 (DB) and DMS-1303018 (TT). Any opinions, findings, and conclusions or recommendations expressed in this material are those of the authors and do not necessarily reflect the views of the funding sources.}

\begin{document}
\begin{abstract}
We present a method to compute the inverse scattering transform (IST) for the famed Toda lattice by solving the associated Riemann--Hilbert (RH) problem numerically. Deformations for the RH problem are incorporated so that the IST can be evaluated in $\mathcal O(1)$ operations for arbitrary points in the $(n,t)$-domain, including short- and long-time regimes. No time-stepping is required to compute the solution because $(n,t)$ appear as parameters in the associated RH problem. The solution of the Toda lattice is computed in long-time asymptotic regions where the asymptotics are not known rigorously.
\end{abstract}
\maketitle

\section{Introduction}\label{S:intro}
We consider the numerical solution of the Cauchy initial value problem for the doubly-infinite Toda lattice
\begin{align}\label{E:eom_ab}
\begin{cases}
\partial_t{a}_n(t) = a_n(t)\big(b_{n+1}(t) - b_n(t)\big)\\
\partial_t{b}_n(t) = 2\big(a_n(t)^2 - a_{n-1}(t)^2\big)\\
a_n(0) = a_n^0 > 0,\text{ and } b_n(0)=b_n^0,
\end{cases}
\end{align}
for $(n,t)\in \mathbb{Z} \times \mathbb{R}$ with solutions\footnote{We omit subscripts to refer to the functions defined on $\mathbb Z$.} $(a(t),b(t)) \in \ell(\mathbb{Z})\times \ell(\mathbb{Z})$ satisfying
\begin{equation}\label{E:decay}
\sum_{n\in\mathbb{Z}} \sigma(n) \left(\left |a_n(t) - \tfrac{1}{2} \right | + \left | b_n(t) \right | \right)<\infty,\text{ for all }t\in\mathbb{R}, ~~ \sigma(n) = e^{\delta |n|},
\end{equation}
for some\footnote{{  Many results for the Toda lattice hold with less restrictive choices for $\sigma(n)$ \cite{Tes}.  For example, the inverse scattering transform method described below can be applied for data in the so-called Marchenko class (\emph{i.e.}, $\sigma(n) = 1 + |n|$). We impose exponential decay for the convenience of the numerical implementation.} } $\delta > 0$.
\begin{figure}[!h]
\begin{overpic}[width=0.15\textwidth]{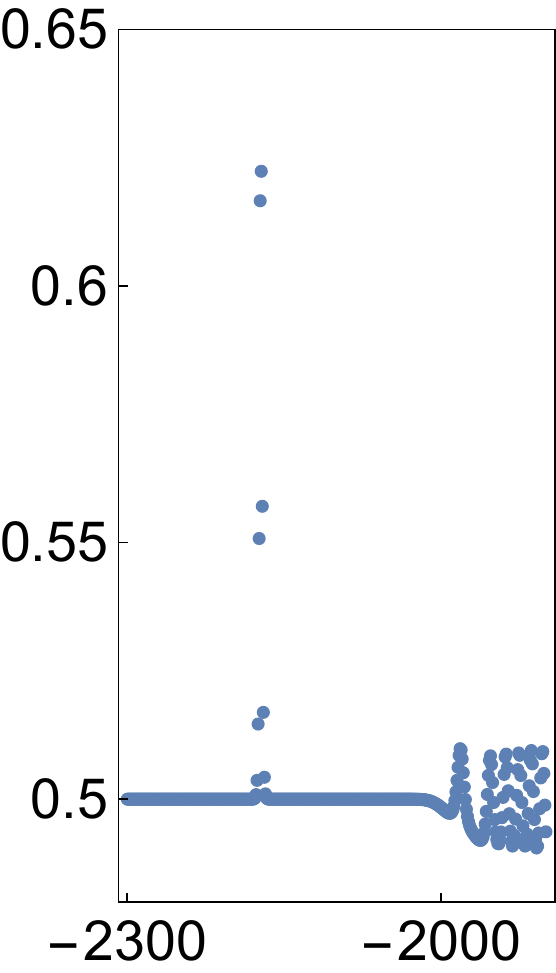}
  \put(-9,40){\rotatebox{90}{$a_n(2000)$}}
\end{overpic}
\begin{overpic}[width=0.62\textwidth]{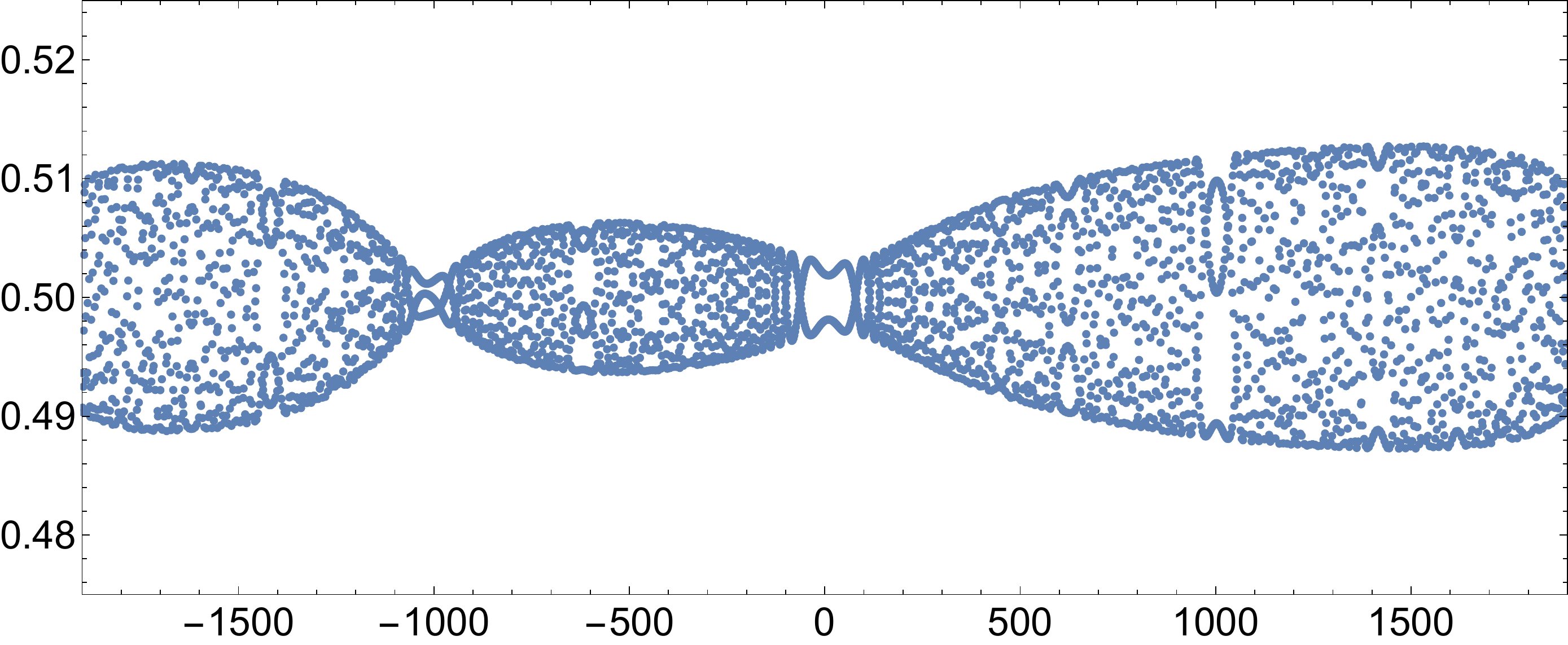}
  \put(52,-2){$n$}
\end{overpic}
\hspace{.05in}
\includegraphics[width=.15\textwidth]{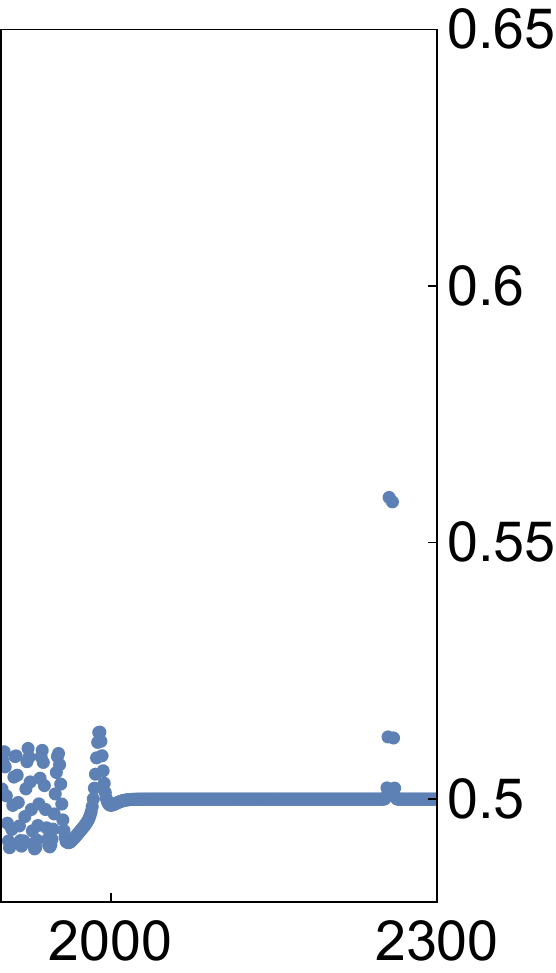}\\

\caption{An example solution of the Toda lattice computed at $t = 2000$ with the method presented here. The initial data is given by $a_n(0) = |1/2 - n e^{-n^2 + n}|$ and $b_n(0) = n \sech n$.  This initial data produces dispersive radiation (center panel) and four solitons, two traveling each direction (left and right panels).  }\label{quad}
\end{figure}

The Toda lattice was introduced by Morikazu Toda in \cite{Toda-orig} (see also \cite{Toda}). The Toda lattice is a completely integrable model for a one-dimensional crystal. The form \eqref{E:eom_ab} we use in this paper is the Toda lattice written in Flaschka's variables \cite{Fla1} (see also the work of S.~V.~Manakov \cite{Man}). This system has been studied in great detail because it is the prototypical discrete-space, continuous-time infinite-dimensional integrable system.

A consequence of the complete integrability of the Toda lattice is an associated inverse scattering transform method (ISTM). The ISTM first maps the initial data to a spectral plane where its time evolution is simple via a transformation called \emph{direct scattering}, see Section~\ref{S:scattering_data}. Then, at a given time $t$, the evolved spectral data is mapped back to the physical plane to find the solution values $a_n(t)$ and $b_n(t)$ by a transformation called \emph{inverse scattering}. This inverse problem is solved by considering an associated oscillatory Riemann--Hilbert (RH) problem. {\RHP}s are boundary-value problems in the complex plane for sectionally analytic functions. General references are \cite{AF,CG,Deift,SIE,thesis,mybook,Zhou}. From an analytical point of view, the benefit of studying the {\RHP} is that asymptotics can be extracted by the method of nonlinear steepest descent. Broadly, the method works by deforming the contours of the {\RHP} {as in the classical scalar method of steepest descent to turn oscillatory terms to exponentially decaying terms}. See \cite{Deift,DZ_RHP,DZ_PII} for some implementations of this method.

As for specific applications of the method of nonlinear steepest descent to the Toda lattice, we refer the reader to the work of Kamvissis \cite{Kam} and  Kr\"uger and Teschl \cite{KT_sol,KT_rev}. These works give explicit long-time asymptotics for the solution of the Toda lattice in the soliton and dispersive regions that we define in Section~\ref{S:regions}. The work of Kamvissis gives the asymptotic behavior in the Painlev\'e region (as defined in Section~\ref{S:regions}) for non-generic initial data.

We approach the ISTM from a numerical perspective.  See Figure~\ref{quad} for a sample solution computed with our method. Given initial data {with sufficient decay (see \eqref{E:decay})} we are able to compute the solution at a given point $(n,t)$, to a given accuracy, in a bounded number of operations by solving the {\RHP} numerically and incorporating the deformations used in the method of nonlinear steepest descent. Stated another way, for any given $(n,t)$ we give an $\mathcal O(1)$ algorithm to compute the two values $\{a_n(t),b_n(t)\}$. No time-stepping is required to compute these values. This methodology has been previously applied to the KdV and mKdV equations \cite{TOD_KdV}, the focusing and defocusing NLS equations \cite{TO_NLS}, the Painlev\'e II equation \cite{olver-PII,OT_RMT} and orthogonal polynomials on the line \cite{OT_RMT}. We compute the solution of the Toda lattice for arbitrarily large values of $t$. The complexity and the accuracy of the methodology is discussed in \cite{OT_RHP}. { The code we have developed is available at \cite{ISTPackage}.} With the current state of the art, the use of deformations appears to be necessary as the associated {\RHP} is increasingly oscillatory as $|n|$ or $t$ increase. Deformations can be avoided in some cases with oscillatory integral techniques \cite{GMRES} but unfortunately these methods are not currently general enough to use in the case of the Toda lattice.

The {\RHP} associated with the Toda lattice has fundamental differences from problems previously solved numerically. First, the fundamental domain is the unit circle as opposed to the real axis as in the cases of the KdV, mKdV and NLS equations. Second, the {\RHP}s in the previously solved cases have had their deformations worked out in detail. In this paper, we develop deformations for the Toda lattice in regions of the $(n,t)$-plane where deformations do not exist in the literature including the determination of the so-called $g$-function, see Appendix~\ref{A:g}. We believe our deformations will be required for the future asymptotic analysis of the Toda lattice.  Importantly, the asymptotic regions we define, and the deformations performed therein, cover the entire $(n,t)$ plane. This is something, to our knowledge, that has not been performed previously in the literature. Finally, we encounter some interesting technical challenges in computing the functions used in the deformations, see Appendices \ref{A:sing} and \ref{g-func-comp}.  Also, in light of the work in \cite{BN} we believe this numerical method will be useful in studying non-integrable, Hamiltonian perturbations of the Toda lattice.

In this paper we do not consider large amplitude data.  Large amplitude data induces singular behavior in the solution of the Toda lattice akin to the behavior in the small dispersion limit for the KdV equation, for example.  This affects numerical methods in a critical way.
Although, we do not consider large amplitude data, throughout the manuscript we include footnotes that highlight the complications that arise for larger initial data.  We also do not treat the case where poles in the Riemann--Hilbert problem are very close to the unit circle.  The methodology described here can be used to handle this case accurately with some extra work and deformations.

The paper is organized as follows. In Section~\ref{S:back} we give the background material on the ISTM for the Toda lattice. The direct scattering and inverse scattering maps are discussed along with a discussion of the (asymptotic) regions (soliton, dispersive, Painlev\'e, collisionless shock, and transition) of the Toda lattice. {We also describe the fundamental deformations, which are performed in all of these regions. Section~\ref{S:manual} provides a step by step guide outlining the computational procedure for obtaining the solution of the Toda lattice.} The majority of the paper is devoted to {Section~\ref{S:deformations}} where we discuss, and explicitly derive, the deformations of the {\RHP} in each region. {In Section~\ref{S:inverse} we discuss the numerical solution of {\RHP}s.} Finally, in Section~\ref{S:numerics} we give some numerical results including an error analysis. We include five appendices. Appendix~\ref{A:sing} discusses the numerical solution of singular, but diagonal, {\RHP}s. Appendix~\ref{A:1} gives a deeper discussion of computing the eigenvalues of Jacobi operators. Appendix~\ref{A:g} details the $g$-function that is used in both the collisionless shock and transition regions. Appendix~\ref{A:solve} contains the vanishing lemma and a discussion of the unique solvability of {\RHP}s considered in this work. Lastly, Appendix~\ref{A:genericity} gives a proof that Jacobi matrices whose reflection coefficient attains the value $-1$ at the edges of its continuous spectrum forms an open dense subset of the Marchenko class { (c.f., $\sigma(n) = 1 + |n|$ in \eqref{E:decay})} of Jacobi matrices.  This implies that for an open dense set of initial data, the long-time behavior of the solution of the Toda lattice exhibits a collisionless shock region: see Section~\ref{S:regions} below (see also \cite{AS-KdV,DVZ}).  If the reflection coefficient does not attain the value $-1$ at the edge of the continuous spectrum then the collisionless shock region is absent\footnote{This also implies that the transition region, defined in Section~\ref{S:regions}, is also absent.}.

\section{Background material}\label{S:back}
We use this section to cover theoretical background and fix notation.
\subsection{Integrability and Lax pairs}
The complete integrability of the Toda lattice was proved by H.~Flaschka in 1974 in a sequence of papers \cite{Fla1} and \cite{Fla2}, and independently by S.~V.~Manakov in \cite{Man}. Introduce the second-order linear difference operators $L$ and $P$ defined on $\ell^2(\mathbb{Z})$ by
\begin{align}
(Lf)_n &= a_{n-1}f_{n-1}+ b_nf_n + a_nf_{n+1}, \label{E:L}\\
(Pf)_n &= -a_{n-1}f_{n-1} + a_nf_{n+1}\,.\label{E:P}
\end{align}
and note that in the standard basis $L=L\big(\{a_n\}_{n\in\mathbb{Z}},\{b_n\}_{n\in\mathbb{Z}}\big)$ is a Jacobi matrix (symmetric, tridiagonal with \underline{positive} off-diagonal entries) and $P$ is a skew-symmetric matrix, i.e., $P^T = -P$:
\begin{equation*}
L=
\begin{pmatrix}
 \ddots  &\ddots & 	\ddots	 & 		  \\
  \ddots & b_{n-1} & a_{n-1} & 0      \\
 \ddots     & a_{n-1} & b_n  & a_n   & \ddots    	 \\
	  & 0     & a_n  & b_{n+1} & \ddots    	 \\
		  &		 & \ddots   &\ddots & \ddots     \\
\end{pmatrix}
\quad\text{and}\quad
P=
\begin{pmatrix}
 \ddots  &\ddots & 	\ddots	 & 		  \\
  \ddots & 0 & a_{n-1} & 0      \\
 \ddots     & -a_{n-1} & 0  & a_n   & \ddots    	 \\
	  & 0     & -a_n  &0 & \ddots    	 \\
		  &		 & \ddots   &\ddots & \ddots     \\
\end{pmatrix}.
\end{equation*}
The system of equations given in \eqref{E:eom_ab} is equivalent to
\begin{equation*}
\partial_t L(t) = [P(t), L(t)] = P(t)L(t) - L(t)P(t),
\end{equation*}
and $(P,L)$ is called a Lax pair. Its existence shows the complete integrability of the Toda lattice.  A consequence of complete integrability (or of the Lax pair) is the existence of an inverse scattering transform for the Toda lattice.
\subsection{Direct scattering: definition of the scattering data}\label{S:scattering_data}
Since $L$ is a bounded self-adjoint operator the spectrum $\sigma(L) \subset \mathbb{R}$. Furthermore, \eqref{E:decay} implies that the spectrum of $L$ consists of a purely absolutely continuous (a.c.) part
\begin{equation*}
\sigma_{\text{ac}}(L) = [-1, 1],
\end{equation*}
and a finite simple pure point part
\begin{equation*}
\sigma_{\text{pp}}(L) = \lbrace \lambda_j \colon j =1,2,\dots,N \rbrace \subset (-\infty, -1) \cup (1, +\infty)\,.
\end{equation*}
For convenience we map the spectrum via the Joukowski transformation:
\begin{equation*}
\lambda = \tfrac{1}{2}\left(z + z^{-1}\right), \phantom{x} z = \lambda - \sqrt{ \lambda^2 - 1 }, \phantom{x} \lambda \in \mathbb{C}, \phantom{x} |z| \leq 1\,.
\end{equation*}
Here the square root $\sqrt{\lambda^{2}-1}$ is defined to be positive for $\lambda>1$ and $\sigma_{\text{ac}}(L)$ is the branch cut. Under this transformation, the a.c.-spectrum, $[-1,1]$, is mapped to the unit circle $\mathbb{T}$ and the eigenvalues $\lambda_j$ are mapped to $\zeta^{\pm 1}_j$, with $\zeta_j \in (-1,0) \cup (0,1)$ via
\begin{equation}\label{E:zeta}
\lambda_j = \tfrac{1}{2}\left(\zeta_j + \zeta_j^{-1} \right),
\end{equation}
for $j=1,2,\dots, N$. For any $z$ with $0 < |z| \leq 1$, $z \neq \pm 1$ the equation
\begin{equation}\label{E:evalprob}
L \varphi = \frac{z + z^{-1}}{2} \varphi
\end{equation}
has two unique solutions, $\varphi_+$ and $\varphi_-$, normalized such that
\begin{equation}\label{E:normalize}
\lim_{n\to\pm\infty}z^{\mp n} \varphi_{\pm}(z; n) = 1.
\end{equation}
{ For fixed $n$, $\varphi_\pm(z;n)$ are analytic functions of $z$, $0  < |z| < 1$.  With the assumption \eqref{E:decay} of exponential decay in the initial data, the functions $\varphi_\pm(z;n)$ extend analytically to $0 < |z| < 1 + r(\delta)$ where $r(\delta)>0$ depends on the decay rate of $\sigma(n)$ in \eqref{E:decay}. } It follows from Green's formula that the Wronskians $W\left(\varphi_\pm(z; \cdot), \varphi_\pm(z^{-1}; \cdot)\right)$ are independent of $n$, and evaluating them at $\pm\infty$ we observe that $\lbrace \varphi_\pm(z; \cdot), \varphi_\pm(z^{-1}; \cdot) \rbrace$ are two sets of linearly independent solutions. We define the \emph{transmission coefficient} $T(z)$ and the \emph{reflection coefficients} $R_{\pm}(z)$ by the scattering relations for $|z| = 1$
\begin{equation}\label{E:scatrel}
\begin{aligned}
T(z)\varphi_+(z,n) &= R_{-}(z) \varphi_-(z;n) + \varphi_-\left(z^{-1}; n\right),\\
T(z)\varphi_-(z,n) &= R_{+}(z) \varphi_+(z;n) + \varphi_+\left(z^{-1}; n\right).
\end{aligned}
\end{equation}
  Also, for general data the transmission coefficient has a meromorphic extension inside the unit disk $|z|\leq 1$, with finitely many simple poles {at} $\zeta_j$, $|\zeta_j|<1$, $j=1,2,\dots,N$. The residues of $T(z)$ are given by:
\begin{equation}\label{E:proportionality}
\underset{z=\zeta_j}{\text{Res}}\,T(z) = -\zeta_j \gamma_{+, j}\mu_j^{- 1} =-\zeta_j \gamma_{-, j}\mu_j ,
\end{equation}
where
\begin{equation}\label{E:gamma}
\gamma_{\pm,j} = \frac{1}{\left\| \varphi_{\pm}(\zeta_j; \cdot)\right\|^2_{\ell^2(\mathbb{Z})}}
\end{equation}
are the $\emph{norming constants}$ and $\mu_j$ is the associated proportionality constant: $\varphi_{-}(\zeta_j,\cdot)=\mu_j \varphi_{+}(\zeta_j, \cdot)$. {Due to the assumption on the initial data \eqref{E:decay}, the relations \eqref{E:scatrel} remain valid in an annulus containing the unit circle and therefore $R_\pm(z)$ and $T(z)$ are meromorphic in this annulus. }

One reflection coefficient, one set of norming constants, and the set of eigenvalues is sufficient for reconstructing $L$ via the inverse scattering transform for Jacobi matrices whose coefficients decay sufficiently fast \cite{Tes}. Define
\begin{equation*}
R(z) = R_{+}(z)\quad\text{and}\quad \gamma_j = \gamma_{+,j},\phantom{x}j=1,2,\dots,N,
\end{equation*}
and the set
\begin{equation*}
{\bf S}(L) = \left\lbrace R(z), \left\lbrace \zeta_j\right\rbrace_{j=1}^{N}, \left\lbrace\gamma_j\right\rbrace_{j=1}^{N} \right\rbrace
\end{equation*}
to be the scattering data for the Lax operator $L$. For a more detailed account of {the} scattering theory for Jacobi matrices, see \cite{Tes} or \cite{KT_rev}.

\subsection{Inverse scattering: the Riemann--Hilbert problem}

We phrase the inverse problem in terms of a sectionally meromorphic {\RHP}. In what follows, plus ($+$) and minus ($-$) sides of a contour correspond to the left and right sides by orientation, respectively. {And} $m^{\pm}(z) {= m_\pm(z)}$ denote the boundary values of a function $m(z)$ as $z$ tends to the relevant contour from the $\pm$ side.

\begin{rhp}\label{rhp:m}
Let the unit circle $\mathbb{T}$ have counterclockwise orientation. As in \cite{KT_rev}, we seek a function $m \colon \mathbb C \setminus \mathbb{T} \to \mathbb{C}^{1\times 2}$ that is sectionally meromorphic, continuous\footnote{Throughout this paper, unless we specify otherwise, we look for solutions of the Riemann--Hilbert problems that are continuous up to their jump contours.} up to $\mathbb T$, with simple poles at $\zeta_j^{\pm 1}$, $j=1\dots,N$, and satisfies:
\begin{itemize}
\item \emph{the jump condition:}
\begin{equation*}
m^{+}(z;n,t) = m^{-}(z;n,t) J(z;n,t),\phantom{x} z\in\mathbb{T},\quad J(z;n,t) = \begin{pmatrix} 1 - |R(z)|^2 & -\overline{R(z)} e^{-\theta(z;n,t)} \\ R(z) e^{\theta(z;n,t)} & 1 \end{pmatrix},
\end{equation*}
\item \emph{the residue conditions:}
\begin{equation*}
\begin{aligned}
\underset{z=\zeta_j}{\text{Res}}\,m(z;n,t) &= \lim_{z\to\zeta_j} m(z;n,t)\begin{pmatrix}
0 & 0 \\
-\zeta_j \gamma_j e^{\theta(\zeta_j;n,t)} & 0
\end{pmatrix},\quad j=1,2, \dots, N,\\
\underset{z=\zeta_j^{-1}}{\text{Res}}\,m(z;n,t) &= \lim_{z\to\zeta_j^{-1}} m(z;n,t)\begin{pmatrix}
0 & \zeta_j^{-1} \gamma_j e^{\theta(\zeta_j;n,t)} \\
0 & 0
\end{pmatrix}, \quad j=1,2, \dots, N,
\end{aligned}
\end{equation*}
\item \emph{the symmetry condition:}
\begin{equation}\label{E:symm}
m\left(z^{-1}; n, t\right) = m(z;n,t)\begin{pmatrix} 0& 1 \\ 1 & 0 \end{pmatrix},
\end{equation}
\item \emph{the normalization condition:}
\begin{equation}\label{E:vector_normal}
\lim_{z\to \infty} m(z;n,t) = \begin{pmatrix} m_1 & m_2 \end{pmatrix},\text{ with } m_1 \cdot m_2 = 1 \text{ and } m_1 > 0\,.
\end{equation}
\end{itemize}
\end{rhp}
Here the exponent $\theta(z;n,t)$ in the jump matrix $J$ is given by:
\begin{equation}\label{E:theta}
\theta(z;n,t) = t\left(z - z^{-1}\right) + 2 n \log (z),
\end{equation}
and it can be shown that $\overline{R(z)} = R(z^{-1})$ \cite{Tes}. The symmetry condition ensures that \rhref{rhp:m} has a unique solution for all values of $(n,t)$ (see Section 3 in \cite{KT_sol}).

\begin{remark}\label{r:nosol}
{The associated matrix {\RHP} for \rhref{rhp:m} (the {\RHP} with a $2 \times 2$ unknown function satisfying the same jump condition, normalized to the identity matrix at infinity and no symmetry condition, see Definition~\ref{D:assoc} below) may not have a solution for some exceptional $(n,t)$ values.  Indeed, these exceptional values are guaranteed to exist when $N \neq 0$, see \cite[Lemma~2.6]{KT_sol} and the preceding discussion, for example.  But as illustrated by this example, such an exceptional value of $(n,t)$ occurs when $|\gamma_j e^{\theta(\zeta_j;n,t)}| \approx |\zeta^j - \zeta_j^{-1}|$, \emph{i.e.} near the peak of a soliton.  This phenomenon is a consideration in the numerical method developed in this work, see Remark~\ref{r:singular}.}
\end{remark}

We have the following well-known and important fact:
\begin{prop}\label{p:m-analytic}
{For generic initial data $R(\pm 1) = -1$ and hence $m^+_1(-1) = m_2^-(-1) = 0$.  If the potentials $(a^0-1/2,b^0)$ tend to zero exponentially as $|n| \to \infty$, $R(z)$ is analytic in a neighborhood of $\mathbb T$. Moreover, $m^\pm(z)$ have analytic extensions across $\mathbb T$ and hence $m_1^{+}(z)$ and  $m_2^{-}(z)$ have a zero of at least first order at $z=- 1$.}
\end{prop}
\begin{proof}
First, that generically $R(\pm 1) = -1$ is shown in Appendix~\ref{A:genericity}.  Let $m(z)$ be the solution of \rhref{rhp:m}. Then at $z = -1$
\begin{align*}
  	{m}^+(-1) = {m}^-(-1) \begin{pmatrix} 0 & 1 \\ -1 & 1  \end{pmatrix},
\end{align*}
because $R(-1) = -1$. The first component of this equation gives
\begin{align*}
	{m}_1^+(-1) = -{m}_2^-(-1).
	\end{align*}
But the symmetry condition \eqref{E:symm} gives that ${m}_1^+(-1) = {m}_2^-(-1)$ so that ${m}_1^+(-1) = {m}_2^-(-1) = 0$.  Analyticity follows from considering the Volterra summation equations \eqref{e:volt} which forces the zero to be of at least first order.
\end{proof}
We proceed with a lemma for recovering the potential from the unique solution of \rhref{rhp:m}.

\begin{lemma}\label{L:recover}
 The solution $(a(t), b(t)) \in \ell(\mathbb{Z}) \times \ell(\mathbb{Z})$ to the initial value problem \eqref{E:eom_ab} for the Toda lattice can be recovered from the asymptotic behavior of $m(z;n,t)$ near $z=0$,
\begin{equation*}
m(z;n,t) = \begin{pmatrix} A_n(t)\big( 1 - 2B_{n-1}(t)z\big) & \tfrac{1}{A_n(t)}\big( 1 + 2B_n(t)z\big)\end{pmatrix} + \mathcal{O}\left(z^2\right),
\end{equation*}
or $z=\infty$,
\begin{equation}\label{E:m_asym}
m(z;n,t) = \begin{pmatrix} \tfrac{1}{A_n(t)}\big( 1 + 2B_n(t)z^{-1}\big) &  A_n(t)\big( 1 - 2B_{n-1}(t)z^{-1}\big)\end{pmatrix} + \mathcal{O}\left(z^{-2}\right),
\end{equation}
where $A$ and $B$ are defined by
\begin{equation*}
A_n(t) = \prod_{j=n}^{\infty}2a_n(t)\phantom{x}\text{ and }\phantom{x}B_n(t) = - \sum_{j=n+1}^{\infty}b_j(t)\,.
\end{equation*}
Furthermore, \eqref{E:vector_normal} ensures $A_n(t) >0$ and hence $a_n(t) > 0$ for all $t \geq 0$.
\end{lemma}

\subsection{Asymptotic regions}\label{S:regions}
In this section we discuss asymptotic regions for the long-time asymptotics of the Toda lattice with decaying initial data. A rigorous study of long-time asymptotics for solutions of the Toda lattice equations was recently carried out in \cite{KT_sol} and \cite{KT_rev} in the soliton and the dispersive regions (see below), but the question of long-time asymptotics in the region $|n|/t \sim 1$ has not been addressed in generality so far. The long-time behavior of solutions in this region was studied in \cite{Kam} under the additional assumptions that no solitons are present, that is, the {\RHP} has no poles, and that $| R(\pm 1 )| < 1$. Under the latter assumption, the solution is given asymptotically in terms of a Painlev\'e II
transcendent in the region $|n|/t \sim 1$, \cite{Kam}. However, one generically has $R(z = \pm 1) = -1$ (We give a proof of this fact in Appendix~\ref{A:genericity}). In this case, an additional region called the collisionless shock region appears as the stationary phase points of the jump matrix coalesce at $z = \pm 1$, and one needs to introduce additional contour deformations, employing the so-called $g$-function method, to bridge the dispersive and the Painlev\'e regions.  In the current work, we present new deformations of the associated {\RHP} for this unstudied region $|n|/t \sim 1$ with generic initial data. These deformations are essential to compute solutions numerically.
\begin{figure}
\includegraphics[scale=0.6]{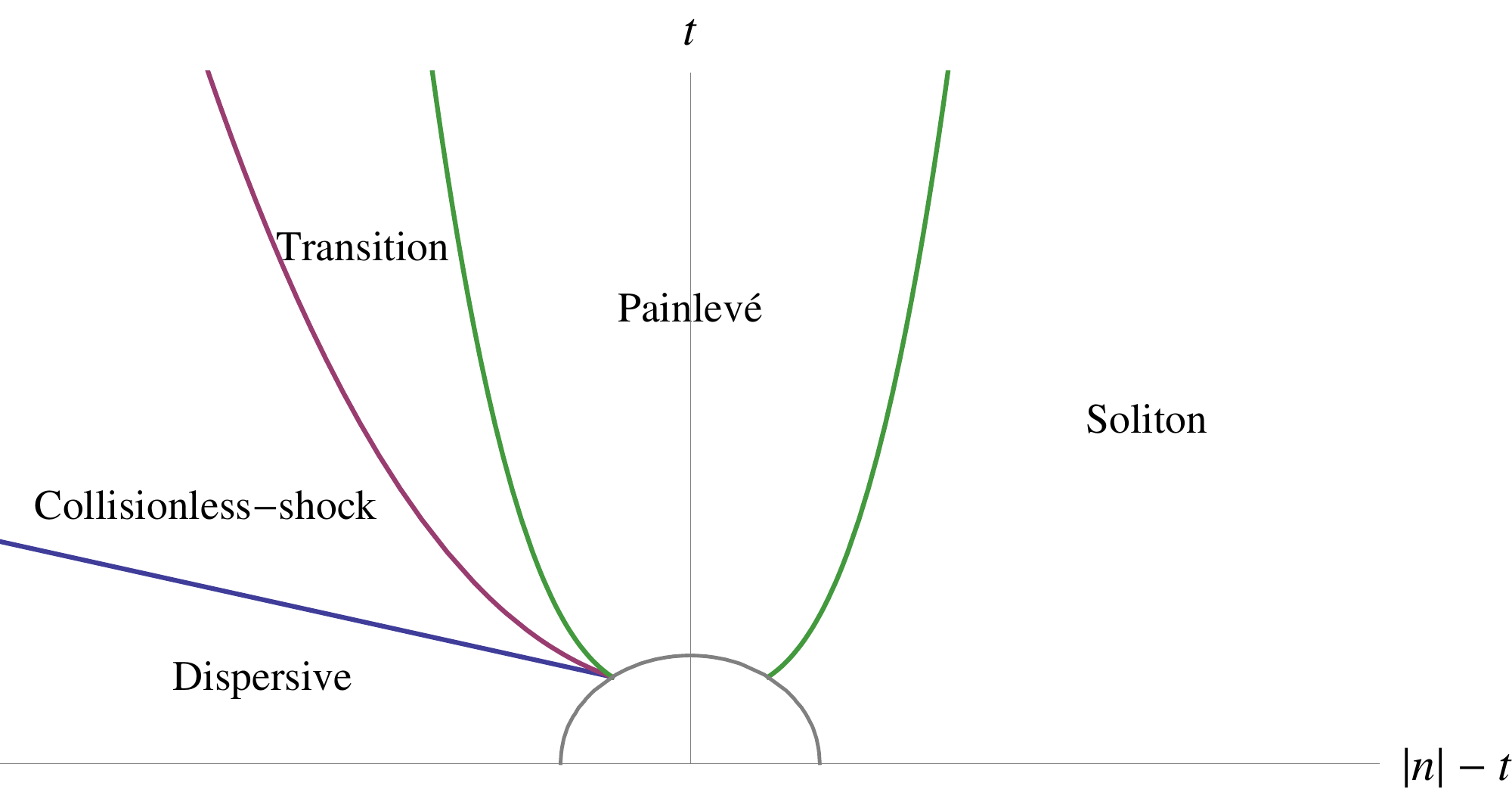}
\caption{Asymptotic regions.}
\end{figure}

Introduce constants, $c_j  > 0$, to divide asymptotic regions.
\begin{enumerate}
\item[1.] \emph{The dispersive region.} This region is defined for $|n| \leq c_1 t$, with $0<c_1<1$. Asymptotics in this region were obtained in \cite{KT_rev}.

\item[2.] \emph{The collisionless shock region.} This region, to the best of our knowledge, has not been addressed in the literature. It is defined by the relation $  c_1 t\leq |n|  \leq t - c_2 t^{1/3}\left(\log{t} \right)^{2/3} $. Asymptotics are not known in this region.

\item[3.] \emph{The transition region.} This region, to the best of our knowledge, is also not present in the literature. {The region is} defined by the relation $ t - c_2 t^{1/3}\left(\log{t} \right)^{2/3}   \leq |n| \leq t - c_3 t^{1/3}$. Asymptotics are not known in this region. An analogue of this region was first introduced for KdV in \cite{TOD_KdV}.

\item[4.] \emph{The Painlev\'e region.} This region is defined for $t - c_3 t^{1/3} \leq |n| \leq t + c_3 t^{1/3}$. Asymptotics in this region were obtained in \cite{Kam} in absence of solitons and under the additional assumption that $|R(z)|<1$.

\item[5.] \emph{The soliton region.} This region is defined for $|n| > t + c_3 t^{1/3}$. Let $v_k >1$ denote the velocity of the $k^{\text{th}}$ soliton and choose $\nu>0$ so that the intervals $(v_k - \nu, v_k +\nu)$, $k=1,2,\dots, N$, are disjoint. If $|n/t - v_k| < \nu$, the asymptotics in this region were obtained in \cite{KT_sol} and \cite{KT_rev}.  It will follow in Section~\ref{S:Painleve} that the deformation in the Painlev\'e region for $|n| \geq t$ is the same as that for soliton region although one should expect the long-time behavior to be different in each region.  We will see that from a numerical perspective the Painlev\'e region for $|n| \geq t$ can be identified with the soliton region.
\end{enumerate}

\subsection{Fundamental deformations of the inverse problem}\label{S:fundamentals}

As outlined in Section~\ref{S:manual} {below}, the procedure for numerical solution of \rhref{rhp:m} involves a sequence of deformations, dictated by the values of $(n,t)$, which result in a numerically tractable RH Problem satisfied by a sectionally analytic vector-valued function. In this section we present two deformations that are {performed} for all values of the parameters $(n,t)$. Our first step is to remove the poles (if any) from the sectionally meromorphic RH Problem \ref{rhp:m}. This is achieved by introducing small circles centered at each pole and using the appropriate jump conditions on these new contours \cite{DKKZ}. Fix $\varepsilon >0$ such that
\begin{equation}
\varepsilon < \frac{1}{2}\min\left\lbrace \min_{1\leq j\leq N}\{ | \zeta_j | \}, \min_{1\leq j\leq N}\left\{\left||\zeta_j| - 1\right|\right\}, \min_{\substack{1\leq j,l\leq N \\ j\neq l}}\left\{|\zeta_j - \zeta_l |\right\} \right\rbrace\,
\label{E:eps-choice}
\end{equation}
and define the circles $D_j^{\pm}$ by
\begin{equation*}
D^{\pm}_j = \left\lbrace z \colon \left |z^{\pm 1} - \zeta_j\right| = \varepsilon \right\rbrace,\quad j=1,2,\dots, N.
\end{equation*}
The choice \eqref{E:eps-choice} of $\varepsilon$ guarantees that the disks enclosed by the circles $D^{\pm}_j$, $j=1,2,\dots, N$, do not intersect each other or the unit circle and none of them contains the origin. We define $\widehat{m}(z;n,t)$ by
\begin{equation}\label{E:hat}
\widehat{m}(z;n,t) \defeq
\begin{cases}
m(z;n,t)\begin{pmatrix} 1 & 0 \\ \frac{\zeta_j \gamma_j e^{\theta(\zeta_j;n,t)}}{z-\zeta_j} & 1\end{pmatrix},\quad &|z - \zeta_j | < \varepsilon,\phantom{x} j =1,2,\dots, N, \\
m(z;n,t)\begin{pmatrix} 1 & -\frac{z\gamma_j e^{\theta(\zeta_j;n,t)}}{z-\zeta_j^{-1}} \\ 0 & 1\end{pmatrix},\quad &\left|z^{-1} - \zeta_j \right| < \varepsilon,\phantom{x} j =1,2,\dots, N, \\
m(z;n,t),\quad&\text{otherwise}.
\end{cases}
\end{equation}
It is straightforward to show that $\widehat{m}(z)$ solves the following sectionally analytic {\RHP}\footnote{From here on we state {\RHP}s only in terms of their jump condition, jump contour, symmetry condition, and normalization.}:
\begin{rhp}\label{rhp:hm}
\begin{equation}\label{E:hat_m}
\begin{aligned}
\widehat{m}^{+}(z;n,t) &=
\begin{cases}
\widehat{m}^{-}(z;n,t)J(z;n,t),\quad&z\in\mathbb{T},\\
\widehat{m}^{-}(z;n,t)\begin{pmatrix} 1 & 0 \\ \frac{\zeta_j \gamma_j e^{\theta(\zeta_j;n,t)}}{z-\zeta_j} & 1\end{pmatrix},\quad &z\in D^{+}_j,\phantom{x} j =1,2,\dots, N, \\
\widehat{m}^{-}(z;n,t)\begin{pmatrix} 1 & \frac{z\gamma_j e^{\theta(\zeta_j;n,t)}}{z-\zeta_j^{-1}} \\ 0& 1\end{pmatrix},\quad &z\in D^{-}_j,\phantom{x} j =1,2,\dots, N,
\end{cases}\\
\widehat{m}\left(z^{-1}; n, t\right) &= \widehat{m}(z;n,t)\begin{pmatrix} 0& 1 \\ 1 & 0 \end{pmatrix}, \quad |z| > 1,\\
\lim_{z\to \infty} \widehat{m}(z;n,t) &= \begin{pmatrix} m_1 & m_2 \end{pmatrix},\text{ with } m_1 \cdot m_2 = 1 \text{ and } m_1 > 0\,,
\end{aligned}
\end{equation}
where $\{D^{+}_j\}$ are oriented counter-clockwise and $\{D^{-}_j\}$ are oriented clockwise.\end{rhp}

This {deformation} brings in a possibility of exponential growth for $t>0$ in the new jump matrices on $\{D_j^\pm\}$. There are two cases to distinguish. If $\Re \theta(\zeta_j;n,t) < 0$ the jump matrices introduced in \rhref{rhp:hm} have exponential decay to the identity as $t\to\infty$, which is what we desire. If $\Re \theta(\zeta_j;n,t) > 0$ for some $\zeta_j$, however, the jumps around such poles are unbounded as $t \to \infty$. Following the approach in \cite{DKKZ} (see also \cite{KT_rev}) we employ a conjugation procedure to restate our problem so that the jump matrices tend to the identity exponentially fast as {either $|n|$ or $t$ tend to infinity.} Let $K_{n,t} \subseteq \lbrace 1,2,\dots, N \rbrace$ denote the index set for $\zeta_j$ (if any) such that $|\gamma_j(t)|\exp{(\Re\theta(\zeta_j;n,t))}>1$. We set

\begin{equation}\label{E:Q}
Q(z) =
\begin{pmatrix}
\displaystyle\prod_{j \in K_{n,t}} \frac{z-\zeta_j}{z - \zeta_j^{-1}} & 0 \\ 0 &\displaystyle\prod_{j \in K_{n,t}} \frac{z -\zeta_j^{-1}}{z - \zeta_j},
\end{pmatrix}
\end{equation}
{and $Q(z)$ is naturally defined as the identity matrix if $K_{n,t}=\emptyset$}. For $j\in K_{n,t}$ define
\begin{equation}\label{E:Phi_1d}
\widetilde{m}(z;n,t)=
\begin{cases}
\widehat{m}(z;n,t)\begin{pmatrix}1 & \frac{z-\zeta_j}{\gamma_j \zeta_j e^{\theta(\zeta_j; n, t)}} \\ -\frac{\zeta_j \gamma_j e^{\theta(\zeta_j;n,t)}}{z-\zeta_j} & 0\end{pmatrix}Q(z), \quad &|z-\zeta_j|<\varepsilon, ~ j\in K_{n,t},\\
\widehat{m}(z;n,t)\begin{pmatrix}0 & \frac{z\zeta_j\gamma_j e^{\theta(\zeta_j;n,t)}}{z\zeta_j - 1} \\ -\frac{z\zeta_j - 1}{z \zeta_j \gamma_j e^{\theta(\zeta_j; n, t)}} & 1\end{pmatrix}Q(z), \quad &\left|z^{-1}-\zeta_j\right|<\varepsilon, ~j\in K_{n,t},\\
\widehat{m}(z;n,t)Q(z), \quad &\text{otherwise.}
\end{cases}
\end{equation}

The matrices
\begin{equation*}
\begin{pmatrix}1 & \frac{z-\zeta_j}{\gamma_j \zeta_j e^{\theta(\zeta_j; n, t)}} \\ -\frac{\zeta_j \gamma_j e^{\theta(\zeta_j;n,t)}}{z-\zeta_j} & 0\end{pmatrix}Q(z)\quad\text{and}\quad \begin{pmatrix}0 & \frac{z\zeta_j\gamma_j e^{\theta(\zeta_j;n,t)}}{z\zeta_j - 1} \\ -\frac{z\zeta_j - 1}{z \zeta_j \gamma_j e^{\theta(\zeta_j; n, t)}} & 1\end{pmatrix}Q(z)
\end{equation*}
have removable singularities at $\zeta_j$ and $\zeta^{-1}_{j}$, respectively. Now, $\widetilde{m}(z)$ satisfies:
\begin{rhp}\label{rhp:disp}
\begin{equation}
\begin{aligned}
\widetilde{m}^{+}(z;n,t) &=
\begin{cases}
\widetilde{m}^{-}(z;n,t)Q^{-1}(z) J(z;n,t) Q(z),\quad & z\in\mathbb{T},\\
\widetilde{m}^{-}(z;n,t)Q^{-1}(z)\begin{pmatrix}1 & \frac{z-\zeta_j}{\zeta_j \gamma_j e^{\theta(\zeta_j; n, t)}} \\ 0 & 1\end{pmatrix}Q(z),\quad & z\in D^{+}_j,\quad j\in K_{n,t},\\
\widetilde{m}^{-}(z;n,t)Q^{-1}(z) \begin{pmatrix}1 & 0 \\ \frac{z\zeta_j - 1}{z \zeta_j\gamma_j e^{\theta(\zeta_j;n,t)}} & 1\end{pmatrix}Q(z),\quad & z\in D^{-}_j,\quad j\in K_{n,t},\\
\widetilde{m}^{-}(z;n,t)Q^{-1}(z)\begin{pmatrix}1 & 0 \\ \frac{\zeta_j \gamma_j e^{\theta(\zeta_j; n, t)}}{z-\zeta_j} & 1\end{pmatrix}Q(z),\quad & z\in D^{+}_j,\quad j\notin K_{n,t}\\
\widetilde{m}^{-}(z;n,t)Q^{-1}(z) \begin{pmatrix}1 & \frac{z \zeta_j\gamma_j e^{\theta(\zeta_j;n,t)}}{z\zeta_j - 1} \\ 0 & 1\end{pmatrix}Q(z),\quad & z\in D^{-}_j,\quad j\notin K_{n,t},
\end{cases}\\
\widetilde{m}(0; n, t) &= \widetilde{m}(\infty; n, t)\begin{pmatrix} 0 & 1 \\ 1 & 0 \end{pmatrix}Q(0).\\
\lim_{z\to \infty} \widetilde{m}(z; n, t) &= \begin{pmatrix} {m}_1 & {m}_2 \end{pmatrix},\text{ with } m_1\cdot m_2 = 1\text{ and } m_1 >0,
\end{aligned}
\end{equation}
where $m_j\in\mathbb{R}$ are given in \eqref{E:vector_normal}.
\end{rhp}
Note that we have relaxed the \emph{global} symmetry condition \eqref{E:symm} (required to hold for all $z\in\mathbb{C}$) present in \rhref{rhp:m} and \rhref{rhp:hm} to an \emph{asymptotic} symmetry condition (required to hold only at $z=\infty$) in \rhref{rhp:disp}. As we shall see in Section~\ref{S:inverse}, doing this does not cause a problem for numerical solution of \rhref{rhp:m} and it simplifies some calculations, see Remark~\ref{R:symm}. If $N$ is large\footnote{One would expect large $N$ if the data has large amplitude.} the entries in these jump matrices can also be large for finite $(n,t)$ even though we have decay to the identity as $|n|$ or $t$ becomes large.  We consider small values of $N$. Finally, to simplify the notation, let $X_{j,\pm}$ denote the jump matrices that are used to invert the exponentials for $j \in K_{n,t}$:
\begin{equation}\label{E:X}
X_{j,+}(z;n,t) = \begin{pmatrix} 1 & \frac{z-\zeta_j}{\zeta_j \gamma_j e^{\theta(\zeta_j;n,t)}} \\ 0 & 1\end{pmatrix}\quad\text{and}\quad X_{j,-}(z;n,t)= \begin{pmatrix} 1 & 0 \\ \frac{z\zeta_j - 1}{z\zeta_j\gamma_j e^{\theta(\zeta_j; n, t)}}& 1\end{pmatrix},
\end{equation}
and $Y_{j,\pm}(z;n,t)$ denote the jumps introduced in \rhref{rhp:disp} for $j \not \in K_{n,t}$:
\begin{equation}\label{E:Y}
Y_{j,+}(z;n,t) = \begin{pmatrix} 1 & 0 \\ \frac{\zeta_j \gamma_j e^{\theta(\zeta_j;n,t)}}{z-\zeta_j} & 1\end{pmatrix}\quad\text{and}\quad Y_{j,-}(z;n,t) = \begin{pmatrix} 1 & \frac{z\zeta_j \gamma_j e^{\theta(\zeta_j;n,t)}}{z\zeta_j-1} \\ 0& 1\end{pmatrix}\,.
\end{equation}

To summarize, for all values of $(n,t)$, we initially perform the following chain of deformations:
\begin{equation*}
\big[m(z;n,t);~\text{\rhref{rhp:m}}\big] \longmapsto\big[\widehat{m}(z;n,t);~\text{\rhref{rhp:hm}}\big]\longmapsto\big[\widetilde{m}(z;n,t);~\text{\rhref{rhp:disp}}\big].
\end{equation*}

\section{The Methodology: A step by step guide}\label{S:manual}
In order to compute solutions of the Toda lattice we must perform each of the procedures oulined in Section~\ref{S:back} numerically. Here we outline what this entails and give a brief discussion of each step. The sections that follow describe many of these steps in an increasing level of detail. We perform the following:
\begin{enumerate}
\item[3.1.] numerical computation of the scattering data,
\item[3.2.] deformation of the vector RH problem, and
\item[3.3.] numerical solution of the deformed vector RH problem.
\end{enumerate}
{Recall that we are computing with initial data that has exponential decay as described in \eqref{E:decay}.}

\subsection{Numerical computation of the scattering data} \label{S:data}

\begin{itemize}
\item\textbf{Computing $R(z)$.}\newline
For $z$ in a neighborhood of $\mathbb{T}$, {on which $R(z)$ is analytic,} we look for solutions of $L\varphi = \tfrac{1}{2}\left(z+z^{-2}\right)\varphi$ which behave like $z^{\pm n}$ as $n \to \pm \infty$. We define two new functions $f_{+}(z;n) = \varphi_{+}(z;n)z^{-n}$ and $f_{-}(z;n) = \varphi_{-}(z;n)z^{n}$ so that we have $f_{\pm}(z;n) \to 1$ as $n\to\pm\infty$. Then $f_{\pm}$ satisfies
\begin{equation*}
a_{n-1}z^{\mp 1}f_{\pm}(z;n-1) + \left(b_n - \tfrac{1}{2}\left(z + z^{-1}\right) \right)f_{\pm}(z;n) + a_nz^{\pm 1}f_{\pm}(z;n+1) = 0.
\end{equation*}
This can be effectively solved using back substitution on $n\geq 0$ using the appropriate boundary conditions for $f_{+}(z;\cdot)$: for $K$ large approximate $f_+$ by using the condition $f_+(z;K+1) = f_+(z;K) = 1$. The constant $K$ is chosen so that $|b_m|$, $|a_m-1/2|$ are both less than machine accuracy for $|m| \geq K$. For $n\leq 0$ a similar method works for $f_{-}(z; \cdot)$ by setting $f_-(z;-K-1)=f_+(z;-K) = 1$. Matching these approximate solutions at $n=0$ yields an approximation of the reflection coefficient.
\item\textbf{Computing $\left\lbrace\zeta^{\pm 1}_1,\zeta^{\pm 1}_2,\dots, \zeta^{\pm 1}_N\right\rbrace$.}\newline
Computing $\zeta^{\pm 1}_j$ is equivalent to computing the $\ell^{2}(\mathbb{Z})$ eigenvalues of the doubly-infinite Jacobi matrix $L$ that is defined in \eqref{E:L}. We approximate the eigenvalues of $L$ by computing eigenvalues of a $L_K$ that are outside the interval $[-1,1]$ for a large value of $K$ \cite{BN}. This method might fail to capture eigenvalues of $L$ that are close to its continuous spectrum. We present an example illustrating this case and provide the underlying spectral theory for Jacobi matrices in Appendix~\ref{A:1}. We check whether we successfully capture all of the eigenvalues by computing the inverse scattering transform at $t=0$ and comparing the reconstructed solution to the initial data. If these values differ by a user-prescribed tolerance, we employ Newton iteration to compute the (simple) zeros of $1/T(z)$, since, as mentioned in Section~\ref{S:scattering_data}, $\zeta_{j}^{\pm 1}$ are the (simple) poles of the transmission coefficient, $T(z)$.

\item\textbf{Computing $\left\lbrace\gamma_1,\gamma_2,\dots, \gamma_N\right\rbrace$.}\newline
At the points $z=\zeta_j$, the solutions $\varphi_{\pm}(z;\cdot)$ are proportional and thus both lie in $\ell^2 (\mathbb{Z})$ with exponential decay as $n\to\pm\infty$. We compute the norming constants by computing $\varphi_{+}(\zeta_j; n)$ for $n\geq 0$ and $\varphi_{-}(\zeta_j; n)$ for $n\leq 0$. We determine the proportionality constant $\mu_j$ (see \eqref{E:proportionality}) by matching these solutions at $n=0$. We then calculate the $\ell^2$-norm $\left\| \varphi_{+}(\zeta_j; \cdot) \right\|_{\ell^{2}(\mathbb{Z})}$ piecewise, using $\varphi_{+}(\zeta_j; n)$ for $n\geq 0$ and $\mu_j^{-1}\varphi_{-}(\zeta_j; n)$ for $n< 0$. Thus, by recalling \eqref{E:gamma}, we obtain the norming constant $\gamma_j$:
\begin{equation*}
\gamma_j^{-1} = \left\| \varphi_{+}(\zeta_j; \cdot) \right\|^2_{\ell^{2}(\mathbb{Z})}.
\end{equation*}
\end{itemize}

\subsection{Deformation of the vector RH problem}

 For most values of $(n,t)$, \rhref{rhp:disp} has high-oscillation in the jump matrix $J(z;n,t)$. To be able to accurately compute the function $\widetilde m(z)$ (or $m(z)$), the {\RHP} needs to be deformed to control these oscillations. This is the objective of the Deift--Zhou method of nonlinear steepest descent (see, for example, \cite{Deift,DVZ,DZ_RHP,DZ_PII}). The input to deform the vector RH problem is {both the parameter set $(n,t)$ and the numerically computed scattering data.}

\begin{itemize}
\item\textbf{Choose the (asymptotic) region for $(n,t)$.}\newline
  For each pair\footnote{We only consider $t > 0$ and $n > 0$. {A} transformation {is used} to treat $n < 0$, see Remark~\ref{r:symmetry}.} $(n,t)$, $n,t >0$, we need to associate a region. This region will dictate how to deform the {\RHP}. We use the five regions introduced in Section~\ref{S:regions} and one {additional} region:
  \begin{enumerate}
  \item[0.] a region where \underline{no deformation} is made.
  \end{enumerate}
  In our code, we use tom~\ref{Alg:Choose} to choose the region:

\renewcommand{\figurename}{Algorithm}
\begin{scheme}[h!]
\centering {\large \underline{Algorithm~\ref{Alg:Choose}:} Choosing a region.}
  \begin{algorithm}[H]
  \KwData{$(n,t)$}
  \KwResult{The region for deformation.}
  set $c_1 = .96$;  set $c_2 = .2$; set $c_3 = 3/2$; set $C_1 = 3$; set $\rho_0 = \sqrt{1-(n/t)^2}$\;
    \uIf{$n^2 + 4 t^2 < C_1$}{no deformation}
    \uElseIf{$n \geq t$}{use soliton region deformation}
    \uElseIf{$t - c_3 (t^{1/3}+1) \leq n \leq t + c_3 (t^{1/3}+1)$}{use Painlev\'e region deformation}
    \uElseIf{$n/t < c$}{use dispersive region deformation}
    \uElseIf{$n/t < 1$ \textbf{and} $- 2 \log \rho_0/(t \rho_0^3) \leq c_2$}{use collisonless shock region deformation}
    \Else{use the transition region deformation}
  \end{algorithm}
 \caption{This is the algorithm for choosing a region. See \eqref{E:rho} for the appearance of $\rho_0$ in the analysis. The constants defined in the algorithm are user specified and can be adjusted on case-by-case basis. In our code we leave them fixed as displayed.\label{Alg:Choose}}
  \end{scheme}
\renewcommand{\figurename}{Figure}

\item\textbf{Deform the jump contours and compute auxilliary functions. }\newline
Once the region has been chosen the deformation of the vector RH problem has to be implemented. This usually follows the ideas of the Deift--Zhou method of nonlinear steepest descent {applied} from a numerical perspective \cite{OT_RHP}. For each output of Algorithm~\ref{Alg:Choose} one has to determine and hard-code appropriate deformations which often introduce additional, or auxilliary, functions that must be computed. These deformations and auxilliary functions are discussed in great detail in Section~\ref{S:deformations}.

\end{itemize}

\subsection{Numerical solution of the deformed vector RH problem}

Once the deformed vector RH problem is in hand, we can proceed with its numerical solution. We must perform the following steps:
\begin{itemize}
\item\textbf{Compute the solution of the associated \underline{matrix} RH problem.}\newline
  The vector \rhref{rhp:disp} and its deformations have a normalization at infinity that is difficult to treat numerically. So, we numerically solve the associated matrix {\RHP} whose solution is defined to be a $2\times 2$ matrix-valued function that has the same jump condition, has no symmetry condition and is normalized to tend to the identity matrix at infinity.  { As discussed in Remark~\ref{r:nosol}, this matrix {\RHP} may fail to have a solution at some exceptional values of $(n,t)$ yet we can still reliably solve it numerically, see Remark~\ref{r:singular}.}

\item\textbf{Construct the solution of the deformed vector RH problem.}\newline
  Because the rows of the associated matrix {\RHP} are linearly independent the solution of the deformed vector {\RHP} must be a {linear combination} of the rows of the associated matrix {\RHP}. The {combination} is determined by solving a $2 \times 2$ eigenvalue problem.

\item\textbf{Extract the solution of the Toda lattice at $(n,t)$.}\newline
  Once the solution of the deformed vector {\RHP} is computed the deformations must be {reversed} and a Taylor expansion is performed to compute $A_n(t)$ and $B_n(t)$ using Lemma~\ref{L:recover}. This procedure is repeated for $(n+1,t)$. Using $A_n(t),A_{n+1}(t),B_n(t)$ and $B_{n-1}(t)$ we have
  \begin{align}\label{E:Atoa}
    a_n(t) = \half A_n(t)/A_{n+1}(t), \quad b_n(t) = B_{n}(t) - B_{n-1}(t).
  \end{align}

\end{itemize}

\section{Deformation of the vector RH problem}\label{S:deformations}
In this section we present the deformations of \rhref{rhp:disp} which are required in each asymptotic region. These deformations involve explicit functions that need not satisfy the global symmetry condition present in \rhref{rhp:m} (or \rhref{rhp:disp}). In each region, the deformations result in a vector {\RHP} with a sectionally analytic solution $m_{\sharp, \alpha}(z;n,t)$ {which becomes unique given a generically true technical assumption discussed in Section~\ref{S:inverse}, specifically Lemma~\ref{l:equiv}.} Here $\alpha$ stands for characters used to denote the asymptotic region $(n,t)$ lies in (\textit{e.g.} $\alpha = \text{cs}$ for the collisionless shock region.) We often suppress the $(n,t)$-dependence of these vector-valued functions $m_{\sharp,\alpha}(z;n,t) = m_{\sharp,\alpha}(z)$.

We use the notation $m_{\kappa,\alpha}(z)$ for unknown vector functions obtained through the deformation procedure where the {integer} $\kappa$ indicates how many deformations have been performed with $\kappa = 1$ being the first deformation of \rhref{rhp:disp}. The final deformation replaces the number $\kappa$ with the symbol $\sharp$. The subscript characters $\alpha$ are used to denote the region as described above. For example, we will have the following sequence of deformations in the dispersive region:
\begin{equation*}
􏰙\widetilde{m}(z)\longmapsto m_{1,\text{d}}(z)\longmapsto m_{2,\text{d}}(z)\longmapsto m_{\sharp,\text{d}}(z)\,.
\end{equation*}
{The functions $m_{\kappa,\alpha}(z)$ are always related to $\widetilde{m}(z)$ via explicit transformations but they may not satisfy an {\RHP} with continous boundary values. Then $m_{\sharp,\alpha}(z)$ will always solve an {\RHP} with continuous boundary values and it is computed numerically.}

The jump matrix $J(z;n,t)$ in \rhref{rhp:disp} has terms that are highly oscillatory for most values of $(n,t)$ and we need to control these oscillations in order to compute $\widetilde{m}(z)$ (or, equivalently $m(z)$) accurately. To do so, we employ Deift-Zhou method of nonlinear steepest descent and examine the phase $\theta(z)$ that appears in these expressions. Solving $\theta'(z)=0$ for $z$, the stationary phase points of $\theta(z)$ are found to be
\begin{equation*}
z_0^{\pm 1} = -\frac{n}{t} \pm \sqrt{\left(\frac{n}{t}\right)^2 - 1}\,.
\end{equation*}
Note that if $n=t$ or $n=-t$, the stationary phase points coalesce at $z=-1$, or $z=1$, respectively.

Also, we present the results and deformations for the case $n\geq 0$ and $t>0$. It is straightforward to obtain the solution for the case $n<0$ and $t>0$ by modifying the initial data and using $n \geq 0$ \cite{KT_rev}:

\begin{remark}\label{r:symmetry}
  If $(a(t), b(t))$ solves the Toda lattice with initial data $a^0_n$ and  $b_n^0$ and $(\tilde a(t), \tilde b(t))$ solves the Toda lattice with initial data $a_{-n}^0$ and  $- b_{-n + 1}^0$ then $a_{-n}(t) = \tilde a_{n}(t)$ and $b_{-n}(t) = -\tilde b_{n-1}(t)$. {An alternate approach would be to use the other reflection coefficient $R_-(z)$ for $n < 0$ and let the stationary phase points lie in the right-half plane.}
\end{remark}

We proceed with the details of the deformations used in each region.

\subsection{The Dispersive Region}\label{S:dispersive}

In this region, the stationary phase points $z_0^{\pm 1}$ of the exponent $\theta(z)$ lie on the unit circle, $\mathbb{T}$. We set $\Sigma = \big\lbrace{z \colon |z|=1,\, -1 \leq \Re z  < \Re z_0 }\big\rbrace$ {and $\Gamma = \mathbb{T}\setminus \Sigma$ as shown in Figure~\ref{F:disp} (this figure omits the contours $D^{\pm}_j$). Note that since $\Re \theta\left(z^{-1}\right) = - \Re \theta(z)$, the curves $\Re \theta(z) = 0$ are symmetric with respect to the mapping $z\mapsto z^{-1}$.
\begin{figure}[htp!]
\centering
\begin{tikzpicture}[scale=0.8,decoration={markings,
mark=at position 0.1 with {\arrow[line width =1.8pt]{>}},
mark=at position 0.5 with {\arrow[line width =1.8pt]{>}},
mark=at position 0.9 with {\arrow[line width =1.8pt]{>}}
}
]
\def\yL{4}
\def\xL{5}
\def\R{\xL-2}
\def\olens{\R+1}
\def\ilens{\R-0.5}

\coordinate (zo) at (120:\R);
\coordinate (zoi) at (240:\R);
\coordinate (o) at (0,0);
\coordinate (yaxisend) at (0,-\yL);
\draw[help lines,<->] (-\xL,0) -- (\xL,0) coordinate (xaxis);
\draw[help lines,<->] (0,-\yL) -- (0,\yL) coordinate (yaxis);

\path[draw,cyan,line width =1.8pt , postaction=decorate]
(\R,0) arc(0:360:\R);

\draw[dotted, gray, line width = 0.8](zoi) to [out=45,in=270] (o) to [out=90, in=135] (zo);
\draw[dotted, gray, line width = 0.8](zoi) to [out=225, in=10] (-\xL,-\yL);
\draw[dotted, gray, line width = 0.8] (zo) to [out=135, in=-10] (-\xL,\yL);

\foreach \Point in {(zo), (zoi), (o)}{
  \node at \Point {\textbullet};
}
\node[below] at (xaxis) {$\Re z$};
\node[above] at (yaxis) {$\Im z$};
\node[below left] {$0$};
\node[above ] at (zo) (z0) {$z_0$};
\node[below] at (zoi) {\small{$z_0^{-1}$}};
\node[below right] at (-\xL+0.3,\yL-1) {\tiny{$\Re \theta(z) < 0$}};
\node[right] at (30:\R) {$\Gamma$};
\node[left] at (150:\R) {$\Sigma$};
\node[above left] at (240:\R-2) {\tiny{$\Re \theta(z) > 0$}};
\node[below] at (0:\R-1.2) {\tiny{$\Re \theta(z) < 0$}};
\node[right] at (15:\R) {\tiny{$\Re \theta(z) > 0$}};
\node[below, gray] at (-\xL+1,-\yL+1) {\tiny{$\Re \theta(z) = 0$}};
\node[below] at (yaxisend) {\small Case: $n>0$};
\end{tikzpicture}
\caption{Sign of $\theta(z)$ and original jump contours for the \rhref{rhp:m} in the dispersive region, $n>0$.}
\label{F:disp}
\end{figure}
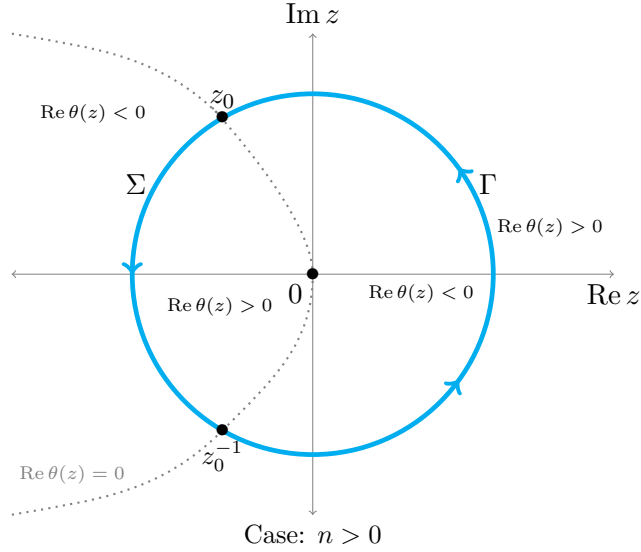

Assume that we have performed the initial deformations discussed in Section~\ref{S:fundamentals} and obtained the vector-valued unknown $\widetilde{m}(z;n,t)$. We then proceed with a deformation which will move the the oscillatory jumps along $\mathbb{T}$ into regions where the oscillatory terms decay exponentially. We define
\begin{equation*}
\tau(z) = 1 - R(z)R\big(z^{-1}\big)\,,
\end{equation*}
and note that the jump matrix $J(z;n,t)$ on $\mathbb{T}$ admits the following two factorizations:
\begin{equation}\label{E:MP}
J(z;n,t)=M(z;n,t)P(z;n,t),
\end{equation}
where
\begin{equation*}
M(z;n,t) = \begin{pmatrix} 1 & - R\left(z^{-1}\right)e^{-\theta(z;n,t)} \\ 0 & 1\end{pmatrix}~\text{and}~P(z;n,t) =\begin{pmatrix} 1 & 0 \\ R(z)e^{\theta(z;n,t)} & 1\end{pmatrix},
\end{equation*}
and
\begin{equation}\label{E:LDU}
J(z;n,t)=L(z;n,t)D(z)U(z;n,t),
\end{equation}
where
\begin{equation*}
L(z;n,t) = \begin{pmatrix} 1 & 0 \\ \frac{ R(z)e^{\theta(z;n,t)}}{\tau(z)} & 1\end{pmatrix},\, U(z;n,t) =\begin{pmatrix} 1 & \frac{- R\left(z^{-1}\right)e^{-\theta(z;n,t)}}{\tau(z)} \\ 0 & 1\end{pmatrix},~\text{and}~
D(z) = \begin{pmatrix} \tau(z) & 0 \\ 0 & \frac{1}{\tau(z)}\end{pmatrix}\,.
\end{equation*}
We use the $MP$-factorization~\eqref{E:MP} on $\Gamma$, and the $LDU$-factorization~\eqref{E:LDU} on $\Sigma$. $M$ (for `minus') will be deformed into the exterior (`minus' side) of the unit circle and $P$ (for `plus') will be deformed into the interior (`plus' side) of the unit circle. Then $L$ is lower triangular and will be deformed into exterior of the unit circle, $D$ is diagonal and will not be deformed, and $U$ is upper triangular and will be deformed into interior of the unit circle. We employ these factorizations so that only one of $e^{\theta(z)}$ or $e^{-\theta(z)}$ appears in each matrix, which in turn makes it possible to obtain exponential decay in different regions of the complex plane. We introduce ``ghost'' contours, $\Sigma_{\pm}$, deformed into $\pm$ side of $\Sigma$, and $\Gamma_{\pm}$ deformed into $\pm$ side of $\Gamma$. Note that these new contours pass locally along the directions of steepest descent for $e^{\pm\theta(z;n,t)}$. The first transformation in this region now follows. We define a new vector-valued function $m_{1,\text{d}}(z)$ based on the regions of the complex plane that emerge from this deformation, as shown in Figure~\ref{F:d_2}(a). Note that $m_{1,\text{d}}(z)$ satisfies the asymptotic symmetry condition and the quadratic normalization condition at infinity which are present in \rhref{rhp:disp}. When $j\in K_{n,t}$ we use the jumps $X_{j,\pm}$ on $D^{\pm}_j$, respectively, in order to turn exponential growth into exponential decay as $t\to +\infty$. We use the jumps $Y_{j,\pm}$ on $D^{\pm}_j$ otherwise.

\begin{figure}
\centering
\subfigure[]{
\begin{tikzpicture}
\def\R{4}
\def\s{0.3}
\def\yL{\R}
\def\xL{\R}
\def\olens{\R+0.8}
\def\ilens{\R-1}

\coordinate (zo) at (120:\R);
\coordinate (zoi) at (240:\R);
\coordinate (o) at (0,0);
\draw[help lines,<->] (-\R-3,0) -- (\R+3,0) coordinate (xaxis);

\begin{scope}[very thick,decoration={
  	 markings,
	 mark=at position 0.33 with {\arrow[line width =1.8pt]{>}},
	mark=at position 0.66 with {\arrow[line width =1.8pt]{>}}}
  ]
\path[draw,cyan,line width=2, postaction=decorate]
(zoi) arc(-120:120:\R);
\path[draw,cyan,line width=2, postaction=decorate]
(zo) arc(120:240:\R);
\end{scope}

\draw[dotted, gray,line width = 0.8] (zoi) to [out=45,in=270] (o) to [out=90, in=315] (zo);
\draw[dotted, gray,line width = 0.8] (zoi) to [out=225, in=10] (-\xL,-\yL);
\draw[dotted, gray,line width = 0.8] (zo) to [out=135, in=-10] (-\xL,\yL);
\begin{scope}[very thick,decoration={
  	 markings,
  	 mark=at position 0.05 with {\arrow[line width =1.8pt]{>}},
	mark=at position 0.5 with {\arrow[line width =1.8pt]{>}},
	mark=at position 0.95 with {\arrow[line width =1.8pt]{>}}}
  ]
\draw[black, line width =1.8, dashed, postaction=decorate] plot[smooth] coordinates {
   (zoi)
   ($(zoi) + (-75:1.13137)$)
   (-90:\olens)
   (-60:\olens)
   (-30:\olens)
   (0:\olens)
   (30:\olens)
   (60:\olens)
   (90:\olens)
   ($(zo)+(75:1.13137)$)
   (zo)
  };
  \draw[black, line width =1.8, dashed, postaction=decorate] plot[smooth] coordinates {
   (zoi)
   ($(zoi)+(15:1.41421)$)
   (-60:\ilens)
   (0:\ilens)
   (60:\ilens)
   ($(zo)+(-15:1.41421)$)
   (zo)
  };
  \end{scope}
\begin{scope}[very thick,decoration={
  	 markings,
  	 mark=at position 0.1 with {\arrow[line width =1.8pt]{>}},
	mark=at position 0.5 with {\arrow[line width =1.8pt]{>}},
	mark=at position 0.9 with {\arrow[line width =1.8pt]{>}}}
  ]
   \draw[black, line width =1.8, dashed, postaction=decorate] plot[smooth] coordinates {
    (zo)
    ($(zo) + (165:1.13137)$)
   (140:\olens)
   (180:\olens)
   (220:\olens)
   ($(zoi) + (-165:1.13137)$)
   (zoi)
   };

\draw[black, line width =1.8, dashed, postaction=decorate] plot[smooth] coordinates {
   	(zo)
   	($(zo)+(255:1.41421)$)
   	(180:\ilens)
	($(zoi)+(-255:1.41421)$)
	(zoi)
     };
\end{scope}

\foreach \Point in {(zo), (zoi), (o), (180:\R), (0:\R) }{
\node at \Point{$\times$};
}
\coordinate (zeta1) at (180:\R-2);
\coordinate (zeta2) at (180:\R-3);
\coordinate (zeta3) at (0:\R-2);
\coordinate (zeta4) at (0:\R-3);
\coordinate (zeta1i) at (180:\R+1.5);
\coordinate (zeta2i) at (180:\R+2.5);
\coordinate (zeta3i) at (0:\R+1.5);
\coordinate (zeta4i) at (0:\R+2.5);

\foreach \Point in {(zeta1), (zeta2), (zeta3), (zeta4), (zeta1i), (zeta2i), (zeta3i), (zeta4i)}{
\node at \Point{\textbullet};
}
\begin{scope}[very thick,decoration={
  	 markings,
	mark=at position 0.5 with {\arrow[line width =1.8pt]{>}}}
  ]

\draw[draw,cyan,line width =1.8, postaction=decorate]
(zeta1) + (0:\s) arc(0:360:\s);
\draw[draw,cyan,line width =1.8, postaction=decorate]
(zeta2) + (0:\s) arc(0:360:\s);
\draw[draw,cyan,line width =1.8, postaction=decorate]
(zeta3) + (0:\s) arc(0:360:\s);
\draw[draw,cyan,line width =1.8, postaction=decorate]
(zeta4) + (0:\s) arc(0:360:\s);
\draw[draw,cyan,line width =1.8, postaction=decorate]
(zeta1i) + (0:\s) arc(0:-360:\s);
\draw[draw,cyan,line width =1.8, postaction=decorate]
(zeta2i) + (0:\s) arc(0:-360:\s);
\draw[draw,cyan,line width =1.8, postaction=decorate]
(zeta3i) + (0:\s) arc(0:-360:\s);
\draw[draw,cyan,line width =1.8, postaction=decorate]
(zeta4i) + (0:\s) arc(0:-360:\s);
\end{scope}

\node[above,yshift=\s+7] at (zeta1) {\tiny{$Q^{-1}X_{j,+}Q$}};
\node[below,yshift=-\s-7] at (zeta2) {\tiny{$Q^{-1}X_{j,+}Q$}};
\node[above,yshift=\s+7] at (zeta1i) {\tiny{$Q^{-1}Y_{j,-}Q$}};
\node[below,yshift=-\s-7] at (zeta2i) {\tiny{$Q^{-1}Y_{j,-}Q$}};
\node[above,yshift=\s+7] at (zeta3) {\tiny{$Q^{-1}Y_{j,+}Q$}};
\node[below,yshift=-\s-7] at (zeta4) {\tiny{$Q^{-1}Y_{j,+}Q$}};
\node[above,yshift=\s+7] at (zeta3i) {\tiny{$Q^{-1}X_{j,-}Q$}};
\node[below,yshift=-\s-7] at (zeta4i) {\tiny{$Q^{-1}X_{j,-}Q$}};

\node (L) at (150:\R+0.2) {};
\coordinate (Ltext) at (150:\R+3);
\node[yshift=4pt] at (Ltext) {\tiny{$m_{1,\text{d}}\defeq\widetilde{m}Q^{-1} LQ$}};
\draw [->,thin ] (Ltext) -- (L);

\node (U) at (160:\R-0.6) {};
\coordinate (Utext) at (160:\R+2.5);
\node[yshift=4pt] at (Utext) {\tiny{$m_{1,\text{d}}\defeq\widetilde{m} (Q^{-1}U Q)^{-1}$}};
\draw [->,thin ] (Utext) -- (U);

\node (M) at (30:\R+0.2) {};
\coordinate (Mtext) at (30:\R+3);
\node[yshift=4pt] at (Mtext) {\tiny{$m_{1,\text{d}}\defeq\widetilde{m}Q^{-1} MQ$}};
\draw [->,thin ] (Mtext) -- (M);

\node (P) at (20:\R-0.6) {};
\coordinate (Ptext) at (20:\R+2.5);
\node[yshift=4pt] at (Ptext) {\tiny{$m_{1,\text{d}}\defeq\widetilde{m} (Q^{-1}P Q)^{-1}$}};
\draw [->,thin ] (Ptext) -- (P);

\node[below left] {$0$};
\node[above right] at (110:2) {\tiny{$m_{1,\text{d}} \defeq \widetilde{m}$}};
\node [right] at(200:\R) {\tiny{$Q^{-1} J Q$}};
\node [left,xshift=1pt] at(-20:\R) {\tiny{$Q^{-1} J Q$}};
\node[above,yshift=6pt,xshift=-8pt] at (zo) (z0) {$z_0$};
\node[below, yshift=-6pt,xshift=-8pt] at (zoi) {$z_0^{-1}$};
\node[above ] at (130:\R+1.5) {\tiny{$m_{1,\text{d}} \defeq \widetilde{m}$}};
\node at (-25:\R+2) {\tiny{$m_{1,\text{d}} \defeq \widetilde{m}$}};
\node [left, gray] at (-\xL,\yL) {\tiny{$\Re \theta(z) =0$}};
\node [above left] at (180:\R) {$-1$};
\node [above right] at (0:\R) {$1$};
\end{tikzpicture}
}
\subfigure[]{
\begin{tikzpicture}
\def\R{4}
\def\s{0.3}
\def\yL{\R}
\def\xL{\R}
\def\olens{\R+0.8}
\def\ilens{\R-1}

\coordinate (zo) at (120:\R);
\coordinate (zoi) at (240:\R);
\coordinate (o) at (0,0);
\draw[help lines,<->] (-\R-3,0) -- (\R+3,0) coordinate (xaxis);

\begin{scope}[very thick,decoration={
  	 markings,
	 mark=at position 0.33 with {\arrow[line width =1.8pt]{>}},
	mark=at position 0.66 with {\arrow[line width =1.8pt]{>}}}
  ]
\path[draw,dashed,gray,line width =1.8,postaction=decorate]
(zoi) arc(-120:120:\R);
\end{scope}
\begin{scope}[very thick,decoration={
  	 markings,
	 mark=at position 0.1 with {\arrow[line width =1.8pt]{>}},
	mark=at position 0.9 with {\arrow[line width =1.8pt]{>}}}
  ]
\path[draw,cyan,line width =1.8, postaction=decorate]
(zo) arc(120:240:\R);
\end{scope}

\draw[dotted, gray,line width = 0.8] (zoi) to [out=45,in=270] (o) to [out=90, in=315] (zo);
\draw[dotted, gray,line width = 0.8] (zoi) to [out=225, in=10] (-\xL,-\yL);
\draw[dotted, gray,line width = 0.8] (zo) to [out=135, in=-10] (-\xL,\yL);
\begin{scope}[very thick,decoration={
  	 markings,
  	 mark=at position 0.05 with {\arrow[line width =1.8pt]{>}},
	mark=at position 0.5 with {\arrow[line width =1.8pt]{>}},
	mark=at position 0.95 with {\arrow[line width =1.8pt]{>}}}
  ]
\draw[cyan, line width =1.8, postaction=decorate] plot[smooth] coordinates {
   (zoi)
   ($(zoi) + (-75:1.13137)$)
   (-90:\olens)
   (-60:\olens)
   (-30:\olens)
   (0:\olens)
   (30:\olens)
   (60:\olens)
   (90:\olens)
   ($(zo)+(75:1.13137)$)
   (zo)
  };
  \draw[cyan, line width =1.8, postaction=decorate] plot[smooth] coordinates {
   (zoi)
   ($(zoi)+(15:1.41421)$)
   (-60:\ilens)
   (0:\ilens)
   (60:\ilens)
   ($(zo)+(-15:1.41421)$)
   (zo)
  };
  \end{scope}
\begin{scope}[very thick,decoration={
  	 markings,
  	 mark=at position 0.1 with {\arrow[line width =1.8pt]{>}},
	mark=at position 0.5 with {\arrow[line width =1.8pt]{>}},
	mark=at position 0.9 with {\arrow[line width =1.8pt]{>}}}
  ]
   \draw[cyan, line width =1.8, postaction=decorate] plot[smooth] coordinates {
    (zo)
    ($(zo) + (165:1.13137)$)
   (140:\olens)
   (180:\olens)
   (220:\olens)
   ($(zoi) + (-165:1.13137)$)
   (zoi)
   };

\draw[cyan, line width =1.8, postaction=decorate] plot[smooth] coordinates {
   	(zo)
   	($(zo)+(255:1.41421)$)
   	(180:\ilens)
	($(zoi)+(-255:1.41421)$)
	(zoi)
     };
\end{scope}

\foreach \Point in {(zo), (zoi), (o), (180:\R), (0:\R) }{
\node at \Point{$\times$};
}
\coordinate (zeta1) at (180:\R-2);
\coordinate (zeta2) at (180:\R-3);
\coordinate (zeta3) at (0:\R-2);
\coordinate (zeta4) at (0:\R-3);
\coordinate (zeta1i) at (180:\R+1.5);
\coordinate (zeta2i) at (180:\R+2.5);
\coordinate (zeta3i) at (0:\R+1.5);
\coordinate (zeta4i) at (0:\R+2.5);

\foreach \Point in {(zeta1), (zeta2), (zeta3), (zeta4), (zeta1i), (zeta2i), (zeta3i), (zeta4i)}{
\node at \Point{\textbullet};
}
\begin{scope}[very thick,decoration={
  	 markings,
	mark=at position 0.5 with {\arrow[line width =1.8pt]{>}}}
  ]
\draw[draw,cyan,line width =1.8, postaction=decorate]
(zeta1) + (0:\s) arc(0:360:\s);
\draw[draw,cyan,line width =1.8, postaction=decorate]
(zeta2) + (0:\s) arc(0:360:\s);
\draw[draw,cyan,line width =1.8, postaction=decorate]
(zeta3) + (0:\s) arc(0:360:\s);
\draw[draw,cyan,line width =1.8, postaction=decorate]
(zeta4) + (0:\s) arc(0:360:\s);
\draw[draw,cyan,line width =1.8, postaction=decorate]
(zeta1i) + (0:\s) arc(0:-360:\s);
\draw[draw,cyan,line width =1.8, postaction=decorate]
(zeta2i) + (0:\s) arc(0:-360:\s);
\draw[draw,cyan,line width =1.8, postaction=decorate]
(zeta3i) + (0:\s) arc(0:-360:\s);
\draw[draw,cyan,line width =1.8, postaction=decorate]
(zeta4i) + (0:\s) arc(0:-360:\s);
\end{scope}

\node[above,yshift=\s+7] at (zeta1) {\tiny{$Q^{-1}X_{j,+}Q$}};
\node[below,yshift=-\s-7] at (zeta2) {\tiny{$Q^{-1}X_{j,+}Q$}};
\node[above,yshift=\s+7] at (zeta1i) {\tiny{$Q^{-1}Y_{j,-}Q$}};
\node[below,yshift=-\s-7] at (zeta2i) {\tiny{$Q^{-1}Y_{j,-}Q$}};
\node[above,yshift=\s+7] at (zeta3) {\tiny{$Q^{-1}Y_{j,+}Q$}};
\node[below,yshift=-\s-7] at (zeta4) {\tiny{$Q^{-1}Y_{j,+}Q$}};
\node[above,yshift=\s+7] at (zeta3i) {\tiny{$Q^{-1}X_{j,-}Q$}};
\node[below,yshift=-\s-7] at (zeta4i) {\tiny{$Q^{-1}X_{j,-}Q$}};
\node [left] (Ltext) at (150:\olens) {\tiny{$Q^{-1} LQ$}};
\node [right, xshift=3pt] at (150:\olens) {\color{cyan}\tiny{$\Sigma_+$}};
\node[right,xshift=2pt] (Utext) at (150:\ilens) {\tiny{$Q^{-1}U Q$}};
\node[left] at (150:\ilens) {\color{cyan}\tiny{$\Sigma_-$}};
\node[right] (Mtext) at (30:\olens) {\tiny{$Q^{-1} MQ$}};
\node[left] at (30:\olens) {\color{cyan}\tiny{$\Gamma_-$}};
\node[left,xshift=-6pt] (Ptext) at (30:\ilens) {\tiny{$Q^{-1}P Q$}};
\node[right,xshift=-3pt] at (30:\ilens) {\color{cyan}\tiny{$\Gamma_+$}};

\node[below left] {$0$};
\node [right] at(200:\R) {\tiny{$Q^{-1} D Q$}};
\node[above,yshift=6pt,xshift=-8pt] at (zo) (z0) {$z_0$};
\node[below, yshift=-6pt,xshift=-8pt] at (zoi) {$z_0^{-1}$};
\node [left, gray] at (-\xL,\yL) {\tiny{$\Re \theta(z) =0$}};
\node [above left] at (180:\R) {$-1$};
\node [above right] at (0:\R) {$1$};
\end{tikzpicture}
}
\caption{(a) Jump contours (blue) and matrices for \rhref{rhp:disp} with `ghost' contours (dashed black), (b) Jump contours and matrices for $m_{1,\text{d}}(z)$.  This figure contains the definitions of $\Gamma_\pm$ and $\Sigma_\pm$.}
\label{F:d_2}
\end{figure}
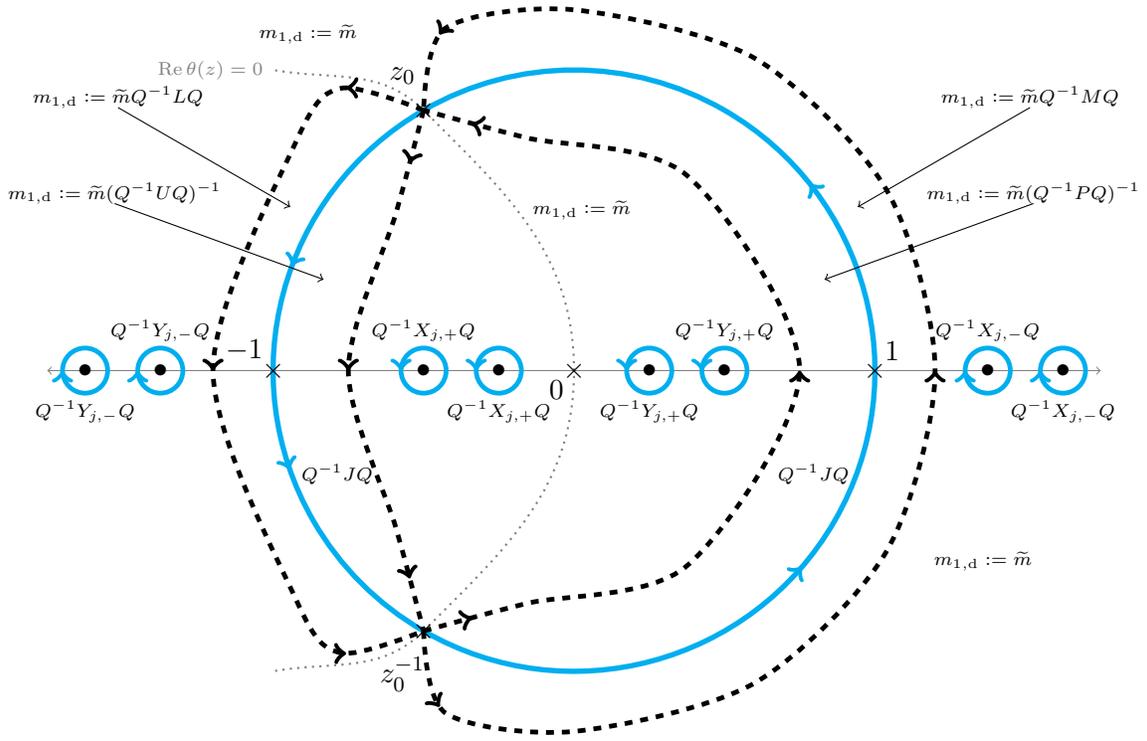
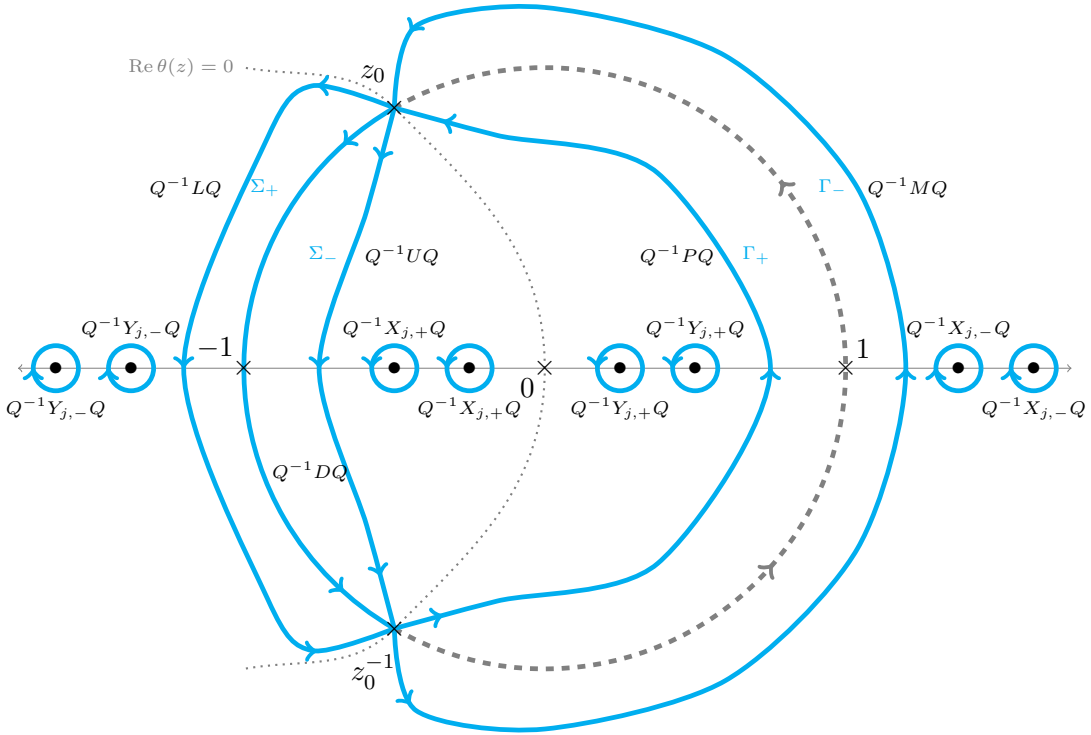

More precisely, $m_{1,\text{d}}(z)$ satisfies the following jump conditions:
\begin{equation*}
m^{+}_{1,\text{d}} (z;n,t)=
\begin{cases}
m^{-}_{1,\text{d}}(z;n,t)Q^{-1}(z)L(z;n,t)Q(z),\quad & z\in \Sigma_{-},\\
m^{-}_{1,\text{d}}(z;n,t)Q^{-1}(z)D(z)Q(z),\quad & z\in \Sigma,\\
m^{-}_{1,\text{d}}(z;n,t)Q^{-1}(z)U(z;n,t)Q(z),\quad & z\in \Sigma_{+},\\
m^{-}_{1,\text{d}}(z;n,t)Q^{-1}(z)M(z;n,t)Q(z),\quad &z\in\Gamma_{-},\\
m^{-}_{1,\text{d}}(z;n,t)Q^{-1}(z)P(z;n,t)Q(z),\quad &z\in\Gamma_{+},\\
m^{-}_{1,\text{d}}(z;n,t)Q^{-1}(z)X_{j,\pm}(z;n,t)Q(z),\quad &z\in D^{\pm}_j,\, j\in K_{n,t},\\
m^{-}_{1,\text{d}}(z;n,t)Q^{-1}(z)Y_{j,\pm}(z;n,t)Q(z),\quad &z\in D^{\pm}_j,\, j\notin K_{n,t},\\
\end{cases}
\end{equation*}
as seen in Figure~{\ref{F:d_2}}(b). Note that the definitions of $\Gamma_\pm$ and $\Sigma_\pm$ are given in the figure.

\begin{remark}\label{r:lens}
The procedure by which the analyticity of an algebraic factorization of the jump matrix is exploited to modify the contours of an {\RHP} (cf.~the transformation from $\widetilde{m}(z;n,t)$ to $m_{1,\text{d}}(z;n,t)$ shown in Figure~\ref{F:d_2}) is referred to here as \emph{lensing}. A full discussion of this can be found in \cite[p.~191]{Deift} although the term lensing does not appear there.
\end{remark}

\begin{remark}\label{r:deformations}
{Away from the points $z_0^{\pm 1}$ where the main contribution of the jump matrix is supported, the deformed contours $\Sigma_{\pm}$ and $\Gamma_{\pm}$ are deformed to stay as far away as possible from $\mathbb{T}$ where $J(z;n,t)$ is oscillatory for large $t$. More precisely, $\Sigma_{\pm}$ and $\Gamma_{\pm}$ pass from $z_0^{\pm 1}$ locally in the directions of steepest descent of $e^{\pm\theta(z;n,t)}$; and away from $z_0^{\pm 1}$, $\Sigma_{\pm}$ and $\Gamma_{\pm}$ are chosen to be close to the inner and outer boundaries of the strip of analyticity of the reflection coefficient, while avoiding any intersection with the finitely many circles $D^{\pm}_{j}$, $1\leq j \leq N$, or with the curves where $\Re \theta(z;n,t)$ changes sign. This arrangement helps us gain sufficient exponential decay (to the identity matrix) in the jump matrices that are defined on the deformed contours. The analogous deformations in the Painlev\'e, transition, collisionless shock, and soliton regions obey this principle. In the collisionless shock and transition regions, the points $\alpha^{\pm 1}$ and $\beta^{\pm 1}$ play the role of $z_0^{\pm 1}$ for the arrangement of the deformed contours.}
\end{remark}

Since generically $R(\pm 1) = -1$ (see Appendix~\ref{A:genericity}), the matrix $D(z)$ has a singularity at $z=-1$ (at $z=1$ in the case $n<0$) and we need to remove this singularity. We also need the jump matrix to approach the identity to achieve accuracy for large values of the parameters. For these reasons we must remove the jump on the contour $\Sigma $ (see Figure~\ref{F:disp}). This is achieved by solving a diagonal matrix {\RHP} with the jump contour $\Sigma$. We introduce the unique $2 \times 2$ matrix-valued function $\Delta(z) = \Delta(z;n,t)$ that solves the diagonal {\RHP}:
\begin{rhp}\label{rhp:delta}
\begin{align}\label{E:Delta}\begin{split}
  \Delta^{+}(z;n,t) = \Delta^{-}(z;n,t)D(z), \quad z&\in\Sigma, \quad \Delta(\infty;n,t) = I,\\
  \Delta(z;n,t)\diag\left(|z+1|^{-1},|z+1|\right) &= \mathcal O(1),~~ z \to -1, ~~ |z| < 1,\\
  \Delta(z;n,t)\diag\left(|z+1|,|z+1|^{-1}\right) &= \mathcal O(1),~~ z \to -1, ~~ |z| > 1,\end{split}
\end{align}
such that $\Delta(z)$ is bounded for $z$ in a neighborhood of $z_0$, $z^{-1}_0$ {and the boundary values $\Delta^\pm(z)$, $z \in \Sigma$, are not continuous only at $z_0,z_0^{-1}$ and $-1$}.
\end{rhp}
It follows from classical theory that $\Delta(z)$ is a diagonal matrix with the property $\Delta_{11}(z)=1/\Delta_{22}(z)$. The exact form of $\Delta(z)$, its properties, and a proof of the fact that it is unique can be found in the Appendix~\ref{A:sing}. Note that in general $\Delta(z)$ has singularities at the end points $z=z_0^{\pm 1}$ of the jump contour. To combat this issue we introduce circles around both $z_0^{\pm 1}$, see Figure~\ref{F:d_3}. {We omit the conjugations by $Q$ in Figure~\ref{F:d_3}(b) and Figure~\ref{F:d_3}(c) for the diagonal jump matrix $D(z)$ since $Q^{-1}(z)$, $D(z)$, and $Q(z)$ commute.}

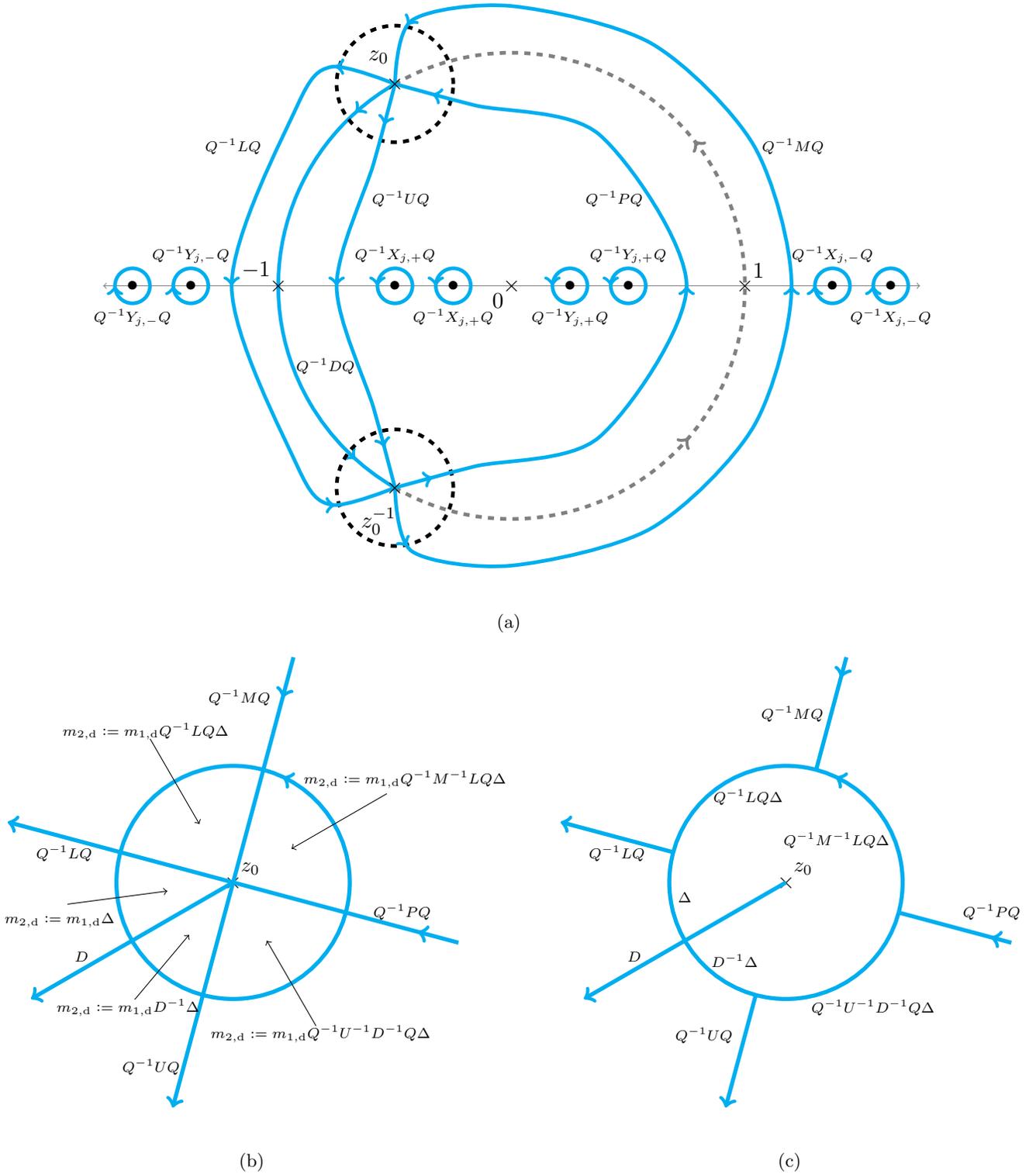
\begin{figure}[htp]
\subfigure[]{
\begin{tikzpicture}
\def\R{4}
\def\s{0.3}
\def\yL{\R}
\def\xL{\R}
\def\olens{\R+0.8}
\def\ilens{\R-1}

\coordinate (zo) at (120:\R);
\coordinate (zoi) at (240:\R);
\coordinate (o) at (0,0);

\draw[help lines,<->] (-\R-3,0) -- (\R+3,0) coordinate (xaxis);

\begin{scope}[very thick,decoration={
  	 markings,
	 mark=at position 0.33 with {\arrow[line width =1.8pt]{>}},
	mark=at position 0.66 with {\arrow[line width =1.8pt]{>}}}
  ]
\path[draw,dashed,gray,line width =1.8,postaction=decorate]
(zoi) arc(-120:120:\R);
\end{scope}
\begin{scope}[very thick,decoration={
  	 markings,
	 mark=at position 0.1 with {\arrow[line width =1.8pt]{>}},
	mark=at position 0.9 with {\arrow[line width =1.8pt]{>}}}
  ]
\path[draw,cyan,line width =1.8, postaction=decorate]
(zo) arc(120:240:\R);
\end{scope}

\path[draw,black,dashed,line width =1.8]
(zo) circle (1);
\path[draw,black,dashed,line width =1.8]
(zoi) circle (1);

\begin{scope}[very thick,decoration={
  	 markings,
  	 mark=at position 0.05 with {\arrow[line width =1.8pt]{>}},
	mark=at position 0.5 with {\arrow[line width =1.8pt]{>}},
	mark=at position 0.95 with {\arrow[line width =1.8pt]{>}}}
  ]
\draw[cyan, line width =1.8, postaction=decorate] plot[smooth] coordinates {
   (zoi)
   ($(zoi) + (-75:1.13137)$)
   (-90:\olens)
   (-60:\olens)
   (-30:\olens)
   (0:\olens)
   (30:\olens)
   (60:\olens)
   (90:\olens)
   ($(zo)+(75:1.13137)$)
   (zo)
  };
  \draw[cyan, line width =1.8, postaction=decorate] plot[smooth] coordinates {
   (zoi)
   ($(zoi)+(15:1.41421)$)
   (-60:\ilens)
   (0:\ilens)
   (60:\ilens)
   ($(zo)+(-15:1.41421)$)
   (zo)
  };
  \end{scope}
\begin{scope}[very thick,decoration={
  	 markings,
  	 mark=at position 0.1 with {\arrow[line width =1.8pt]{>}},
	mark=at position 0.5 with {\arrow[line width =1.8pt]{>}},
	mark=at position 0.9 with {\arrow[line width =1.8pt]{>}}}
  ]
   \draw[cyan, line width =1.8, postaction=decorate] plot[smooth] coordinates {
    (zo)
    ($(zo) + (165:1.13137)$)
   (140:\olens)
   (180:\olens)
   (220:\olens)
   ($(zoi) + (-165:1.13137)$)
   (zoi)
   };

\draw[cyan, line width =1.8, postaction=decorate] plot[smooth] coordinates {
   	(zo)
   	($(zo)+(255:1.41421)$)
   	(180:\ilens)
	($(zoi)+(-255:1.41421)$)
	(zoi)
     };
\end{scope}
\foreach \Point in {(zo), (zoi), (o), (180:\R), (0:\R) }{
\node at \Point{$\times$};
}
\coordinate (zeta1) at (180:\R-2);
\coordinate (zeta2) at (180:\R-3);
\coordinate (zeta3) at (0:\R-2);
\coordinate (zeta4) at (0:\R-3);
\coordinate (zeta1i) at (180:\R+1.5);
\coordinate (zeta2i) at (180:\R+2.5);
\coordinate (zeta3i) at (0:\R+1.5);
\coordinate (zeta4i) at (0:\R+2.5);

\foreach \Point in {(zeta1), (zeta2), (zeta3), (zeta4), (zeta1i), (zeta2i), (zeta3i), (zeta4i)}{
\node at \Point{\textbullet};
}
\begin{scope}[very thick,decoration={
  	 markings,
	mark=at position 0.5 with {\arrow[line width =1.8pt]{>}}}
  ]
\draw[draw,cyan,line width =1.8, postaction=decorate]
(zeta1) + (0:\s) arc(0:360:\s);
\draw[draw,cyan,line width =1.8, postaction=decorate]
(zeta2) + (0:\s) arc(0:360:\s);
\draw[draw,cyan,line width =1.8, postaction=decorate]
(zeta3) + (0:\s) arc(0:360:\s);
\draw[draw,cyan,line width =1.8, postaction=decorate]
(zeta4) + (0:\s) arc(0:360:\s);
\draw[draw,cyan,line width =1.8, postaction=decorate]
(zeta1i) + (0:\s) arc(0:-360:\s);
\draw[draw,cyan,line width =1.8, postaction=decorate]
(zeta2i) + (0:\s) arc(0:-360:\s);
\draw[draw,cyan,line width =1.8, postaction=decorate]
(zeta3i) + (0:\s) arc(0:-360:\s);
\draw[draw,cyan,line width =1.8, postaction=decorate]
(zeta4i) + (0:\s) arc(0:-360:\s);
\end{scope}

\node[above,yshift=\s+7] at (zeta1) {\tiny{$Q^{-1}X_{j,+}Q$}};
\node[below,yshift=-\s-7] at (zeta2) {\tiny{$Q^{-1}X_{j,+}Q$}};
\node[above,yshift=\s+7] at (zeta1i) {\tiny{$Q^{-1}Y_{j,-}Q$}};
\node[below,yshift=-\s-7] at (zeta2i) {\tiny{$Q^{-1}Y_{j,-}Q$}};
\node[above,yshift=\s+7] at (zeta3) {\tiny{$Q^{-1}Y_{j,+}Q$}};
\node[below,yshift=-\s-7] at (zeta4) {\tiny{$Q^{-1}Y_{j,+}Q$}};
\node[above,yshift=\s+7] at (zeta3i) {\tiny{$Q^{-1}X_{j,-}Q$}};
\node[below,yshift=-\s-7] at (zeta4i) {\tiny{$Q^{-1}X_{j,-}Q$}};
\node [left] (Ltext) at (150:\olens) {\tiny{$Q^{-1} LQ$}};
\node[right,xshift=2pt] (Utext) at (150:\ilens) {\tiny{$Q^{-1}U Q$}};
\node[right] (Mtext) at (30:\olens) {\tiny{$Q^{-1} MQ$}};
\node[left,xshift=-6pt] (Ptext) at (30:\ilens) {\tiny{$Q^{-1}P Q$}};
\node[below left] {$0$};
\node [right,xshift=-2pt] at(200:\R) {\tiny{$Q^{-1} D Q$}};
\node[above,yshift=6pt,xshift=-8pt] at (zo) (z0) {$z_0$};
\node[below, yshift=-6pt,xshift=-8pt] at (zoi) {$z_0^{-1}$};
\node [above left] at (180:\R) {$-1$};
\node [above right] at (0:\R) {$1$};
\end{tikzpicture}
}
\subfigure[]{
\begin{tikzpicture}
\def\R{4};
\node at (0,0) {$\times$};
\node[above right] at (0,0) {\small{$z_0$}};
\draw [cyan, line width=2,->] (0:0) -- (165:4);
\draw [cyan, line width=2,->] (0:0) -- (210:4);
\draw [cyan, line width=2,->] (0:0) -- (255:4);
\begin{scope}[very thick,decoration={
  markings,
  mark=at position 0.18 with {\arrow{>}}}
  ]
  \draw [cyan, line width=2,postaction={decorate}] (75:4) -- (0:0);
  \draw [cyan, line width=2,postaction={decorate}] (345:4) -- (0:0);
  \draw [cyan, line width = 2, postaction={decorate}] (0,0) circle (2);
\end{scope}
\node[below] at (165:3) {\tiny{$Q^{-1}LQ$}};
\node[above] at (210:3) {\tiny{$D$}};
\node[below left] at (255:3) {\tiny{$Q^{-1}UQ$}};
\node[above left] at (75:3) {\tiny{$Q^{-1}MQ$}};
\node[above] at (345:3) {\tiny{$Q^{-1}PQ$}};
\node at (120:3) {\tiny{$m_{2,\text{d}} \defeq m_{1,\text{d}}Q^{-1}LQ\Delta$}};
\node[yshift=5pt,xshift=5pt] at (30:3.2) {\tiny{$m_{2,\text{d}} \defeq m_{1,\text{d}}Q^{-1}M^{-1}LQ\Delta$}};
\node[below] at (187.5:3) {\tiny{$m_{2,\text{d}} \defeq m_{1,\text{d}}\Delta$}};
\node at (230:2.8) {\tiny{$m_{2,\text{d}} \defeq m_{1,\text{d}}D^{-1}\Delta$}};
\node at (300:3) {\tiny{$m_{2,\text{d}} \defeq m_{1,\text{d}}Q^{-1}U^{-1}D^{-1}Q\Delta$}};
\node (1) at (120:1){};
\node (2) at (30:1){};
\node (3) at (187.5:1){};
\node (4) at (230:1){};
\node (5) at (300:1){};
\node (1text) at (120:3){};
\node (2text) at (30:3.2){};
\node (3text) at (187.5:3){};
\node (4text) at (230:2.8){};
\node (5text) at (300:3){};

\draw [->,thin,anchor=south ] (1text) -- (1);
\draw [->,thin, anchor=south] (2text) -- (2);
\draw [->,thin ,anchor=east] (3text) -- (3);
\draw [->,thin,above ] (4text) -- (4);
\draw [->,thin,above ] (5text) -- (5);

\end{tikzpicture}
}
\quad
\subfigure[]{
\begin{tikzpicture}
\def\R{4};
\node at (0,0) {$\times$};
\node[above right] at (0,0) {\small{$z_0$}};
\draw [cyan, line width=2,->] (165:2) -- (165:4);
\draw [cyan, line width=2,->] (0:0) -- (210:4);
\draw [cyan, line width=2,->] (255:2) -- (255:4);
\begin{scope}[very thick,decoration={
  markings,
  mark=at position 0.18 with {\arrow{>}}}
  ]
  \draw [cyan, line width=2,postaction={decorate}] (75:4) -- (75:2);
  \draw [cyan, line width=2,postaction={decorate}] (345:4) -- (345:2);
  \draw [cyan, line width = 2, postaction={decorate}] (0,0) circle (2);
\end{scope}
\node[below] at (165:3) {\tiny{$Q^{-1}LQ$}};
\node[above] at (210:3) {\tiny{$D$}};
\node[above left] at (255:3) {\tiny{$Q^{-1}UQ$}};
\node[left] at (75:3) {\tiny{$Q^{-1}MQ$}};
\node[above right] at (345:3) {\tiny{$Q^{-1}PQ$}};

\node[below, xshift=10pt] at (120:2) {\tiny{$Q^{-1}LQ\Delta$}};
\node[below left, xshift =5pt] at (30:2) {\tiny{$Q^{-1}M^{-1}LQ\Delta$}};
\node[right] at (187.5:2) {\tiny{$\Delta$}};
\node[right,yshift=2pt] at (225:2) {\tiny{$D^{-1}\Delta$}};
\node[below,xshift=14pt,yshift=-4pt] at (300:2) {\tiny{$Q^{-1}U^{-1}D^{-1}Q\Delta$}};
\end{tikzpicture}
}
\caption{(a) `Ghost' circles in preparation for the singularities of $\Delta(z)$, (b) Definitions of $m_{2,\text{d}}(z)$ near $z_0$, (c) The jump contours and matrices for $m_{2,\text{d}}(z)$ near $z_0$.}
\label{F:d_3}
\end{figure}

We define $m_{2,\text{d}}(z)$ as shown in Figure~\ref{F:d_3}(b), where $m_{2,\text{d}}(z) \equiv m_{1,\text{d}}(z)$ when no definition is specified; and we see that $m_{2,\text{d}}(z)$ satisfies the {jump conditions that are} presented in Figure~\ref{F:d_3}(c). We apply the same procedure at $z_0^{-1}$. Along with the jump conditions in Figure~\ref{F:d_3}, $m_{2,\text{d}}(z)$ also satisfies an asymptotic symmetry condition
\begin{align}\label{E:m2d}
{m_{2,\text{d}}(0) = m_{2,\text{d}}(\infty) \begin{pmatrix} 0 & 1 \\ 1 & 0 \end{pmatrix} Q(0) \Delta^{-1}(0)}
\end{align}
and the quadratic normalization condition at infinity present in \rhref{rhp:disp}. Finally, we define $m_{\sharp,\text{d}}(z)$ by $m_{\sharp,\text{d}}(z) \equiv m_{2,\text{d}}(z) \Delta^{-1}(z)$ and see that the vector-valued function $m_{\sharp,\text{d}}(z)$ satisfies the jump conditions shown graphically in Figure~\ref{F:d_4}. This is the final deformation performed in the dispersive region.
\begin{figure}[htp]
\begin{tikzpicture}[decoration={markings,
mark=at position 0.35 with {\arrow[line width =1.8pt]{>}}
}
]
\def\R{4.5}
\def\s{0.3}
\def\m{1.2}
\def\yL{\R}
\def\xL{\R}
\def\olens{\R+0.8}
\def\ilens{\R-1}
\coordinate (zo) at (120:\R);
\coordinate (zoi) at (240:\R);
\coordinate (o) at (0,0);
\draw[help lines,<->] (-\R-3,0) -- (\R+3,0) coordinate (xaxis);

\draw[cyan, line width =1.8, postaction=decorate] (zo) circle (\m);
\draw[cyan, line width =1.8, postaction=decorate] (zoi) circle (\m);

\coordinate (zeta1) at (180:\R-1.5);
\coordinate (zeta2) at (180:\R-2.5);
\coordinate (zeta3) at (0:\R-1.5);
\coordinate (zeta4) at (0:\R-2.5);
\coordinate (zeta1i) at (180:\R+1.5);
\coordinate (zeta2i) at (180:\R+2.5);
\coordinate (zeta3i) at (0:\R+1.5);
\coordinate (zeta4i) at (0:\R+2.5);

\foreach \Point in {(zo), (zoi), (o), (180:\R), (0:\R) }{
\node at \Point{$\times$};
}
\foreach \Point in {(zeta1), (zeta2), (zeta3), (zeta4), (zeta1i), (zeta2i), (zeta3i), (zeta4i)}{
\node at \Point{\textbullet};
}
\coordinate (M) at ($(zo)+(75:\m)$);
\coordinate (Mout) at ($(M)+(75:\m)$);
\coordinate (P) at ($(zo)+(-15:\m)$);
\coordinate (Pout) at ($(P)+(-15:\m)$);
\coordinate (L) at ($(zo)+(165:\m)$);
\coordinate (Lout) at ($(L)+(165:\m)$);
\coordinate (U) at ($(zo)+(255:\m)$);
\coordinate (Uout) at ($(U)+(255:\m)$);
\coordinate (iM) at ($(zoi)+(-75:\m)$);
\coordinate (iMout) at ($(iM)+(-75:\m)$);
\coordinate (iP) at ($(zoi)+(15:\m)$);
\coordinate (iPout) at ($(iP)+(15:\m)$);
\coordinate (iL) at ($(zoi)+(-165:\m)$);
\coordinate (iLout) at ($(iL)+(-165:\m)$);
\coordinate (iU) at ($(zoi)+(-255:\m)$);
\coordinate (iUout) at ($(iU)+(-255:\m)$);

\coordinate (pt1) at (-3.19995, 3.16391);
\coordinate (pt2) at (-1.14005, 4.35319);
\coordinate (ipt1) at (-3.19995, -3.16391);
\coordinate (ipt2) at (-1.14005, -4.35319);
\tkzCircumCenter(pt1,zo,pt2)\tkzGetPoint{O};
\tkzDrawArc[gray,dashed,line width =1.8](O,pt2)(pt1);
\tkzCircumCenter(ipt1,zoi,ipt2)\tkzGetPoint{iO};
\tkzDrawArc[gray,dashed,line width =1.8](iO,ipt1)(ipt2);

\draw [cyan, line width  = 1.8, postaction = decorate] (Mout) -- (M);
\draw [cyan, line width  = 1.8, postaction = decorate] (Pout) -- (P);
\draw [cyan, line width  = 1.8, postaction = decorate] (U) -- (Uout);
\draw [cyan, line width  = 1.8, postaction = decorate] (L) -- (Lout);
\draw [cyan, line width  = 1.8, postaction = decorate] (iM) -- (iMout);
\draw [cyan, line width  = 1.8, postaction = decorate] (iP) -- (iPout);
\draw [cyan, line width  = 1.8, postaction = decorate] (iUout) -- (iU);
\draw [cyan, line width  = 1.8, postaction = decorate] (iLout) -- (iL);

\draw[draw,cyan,line width =1.8, postaction=decorate]
(zeta1) + (0:\s) arc(0:360:\s);
\draw[draw,cyan,line width =1.8, postaction=decorate]
(zeta2) + (0:\s) arc(0:360:\s);
\draw[draw,cyan,line width =1.8, postaction=decorate]
(zeta3) + (0:\s) arc(0:360:\s);
\draw[draw,cyan,line width =1.8, postaction=decorate]
(zeta4) + (0:\s) arc(0:360:\s);
\draw[draw,cyan,line width =1.8, postaction=decorate]
(zeta1i) + (0:\s) arc(0:-360:\s);
\draw[draw,cyan,line width =1.8, postaction=decorate]
(zeta2i) + (0:\s) arc(0:-360:\s);
\draw[draw,cyan,line width =1.8, postaction=decorate]
(zeta3i) + (0:\s) arc(0:-360:\s);
\draw[draw,cyan,line width =1.8, postaction=decorate]
(zeta4i) + (0:\s) arc(0:-360:\s);

\node[above,yshift=2pt] at ($(zo)+(122:\m)$) {\tiny{$\Delta Q^{-1}L Q$}};
\node[right] at ($(zo)+(30:\m)$) {\tiny{$\Delta Q^{-1}M^{-1} L Q$}};
\node[anchor=west] at ($(zo)+(185:\m)$) {\tiny{$\Delta$}};
\node[anchor=west] at ($(zo)+(220:\m)$) {\tiny{$\Delta D^{-1}$}};
\node[right] at ($(zo)+(300:\m)$) {\tiny{$\Delta Q^{-1}U^{-1}D^{-1} Q$}};

\node[below] at ($(zoi)+(-122:\m)$) {\tiny{$\Delta Q^{-1}L Q$}};
\node[right] at ($(zoi)+(330:\m)$) {\tiny{$\Delta Q^{-1}M^{-1} L Q$}};
\node[anchor=west] at ($(zoi)+(175:\m)$) {\tiny{$\Delta$}};
\node[anchor=west] at ($(zoi)+(138:\m)$) {\tiny{$\Delta D^{-1}$}};
\node[right, xshift=2pt] at ($(zoi)+(60:\m)$) {\tiny{$\Delta Q^{-1}U^{-1}D^{-1} Q$}};

\node[below right] at (Mout) {\tiny{$\Delta Q^{-1}M Q \Delta^{-1}$}};
\node[above right] at (iMout) {\tiny{$\Delta Q^{-1}M Q \Delta^{-1}$}};
\node[above, xshift=12pt,yshift=2pt] at (Pout) {\tiny{$\Delta Q^{-1}P Q \Delta^{-1}$}};
\node[below, xshift=6pt] at (iPout) {\tiny{$\Delta Q^{-1}P Q \Delta^{-1}$}};
\node[above] at (Lout) {\tiny{$\Delta Q^{-1}L Q \Delta^{-1}$}};
\node[below] at (iLout) {\tiny{$\Delta Q^{-1}L Q \Delta^{-1}$}};
\node[above right] at (Uout) {\tiny{$\Delta Q^{-1}U Q \Delta^{-1}$}};
\node[below right,xshift=2pt] at (iUout) {\tiny{$\Delta Q^{-1}U Q \Delta^{-1}$}};

\node[above,yshift=\s+7] at (zeta1) {\tiny{$\Delta Q^{-1}X_{j,+}Q\Delta^{-1}$}};
\node[below,yshift=-\s-7] at (zeta2) {\tiny{$\Delta Q^{-1}X_{j,+}Q\Delta^{-1}$}};
\node[above,yshift=\s+7] at (zeta1i) {\tiny{$\Delta Q^{-1}Y_{j,-}Q\Delta^{-1}$}};
\node[below,yshift=-\s-7] at (zeta2i) {\tiny{$\Delta Q^{-1}Y_{j,-}Q\Delta^{-1}$}};
\node[above,yshift=\s+7] at (zeta3) {\tiny{$\Delta Q^{-1}Y_{j,+}Q\Delta^{-1}$}};
\node[below,yshift=-\s-7] at (zeta4) {\tiny{$\Delta Q^{-1}Y_{j,+}Q\Delta^{-1}$}};
\node[above,yshift=\s+7] at (zeta3i) {\tiny{$\Delta Q^{-1}X_{j,-}Q\Delta^{-1}$}};
\node[below,yshift=-\s-7] at (zeta4i) {\tiny{$\Delta Q^{-1}X_{j,-}Q\Delta^{-1}$}};

\node[below] at (0:\R){$1$};
\node[below] at (180:\R){$-1$};
\node[above] {$0$};
\node[above] at (zo) (z0) {\small{$z_0$}};
\node[below] at (zoi) {\small{$z_0^{-1}$}};
\end{tikzpicture}
\caption{A zoomed view of the jump contours and matrices for the final {\RHP} in the dispersive region.}
\label{F:d_4}
\end{figure}
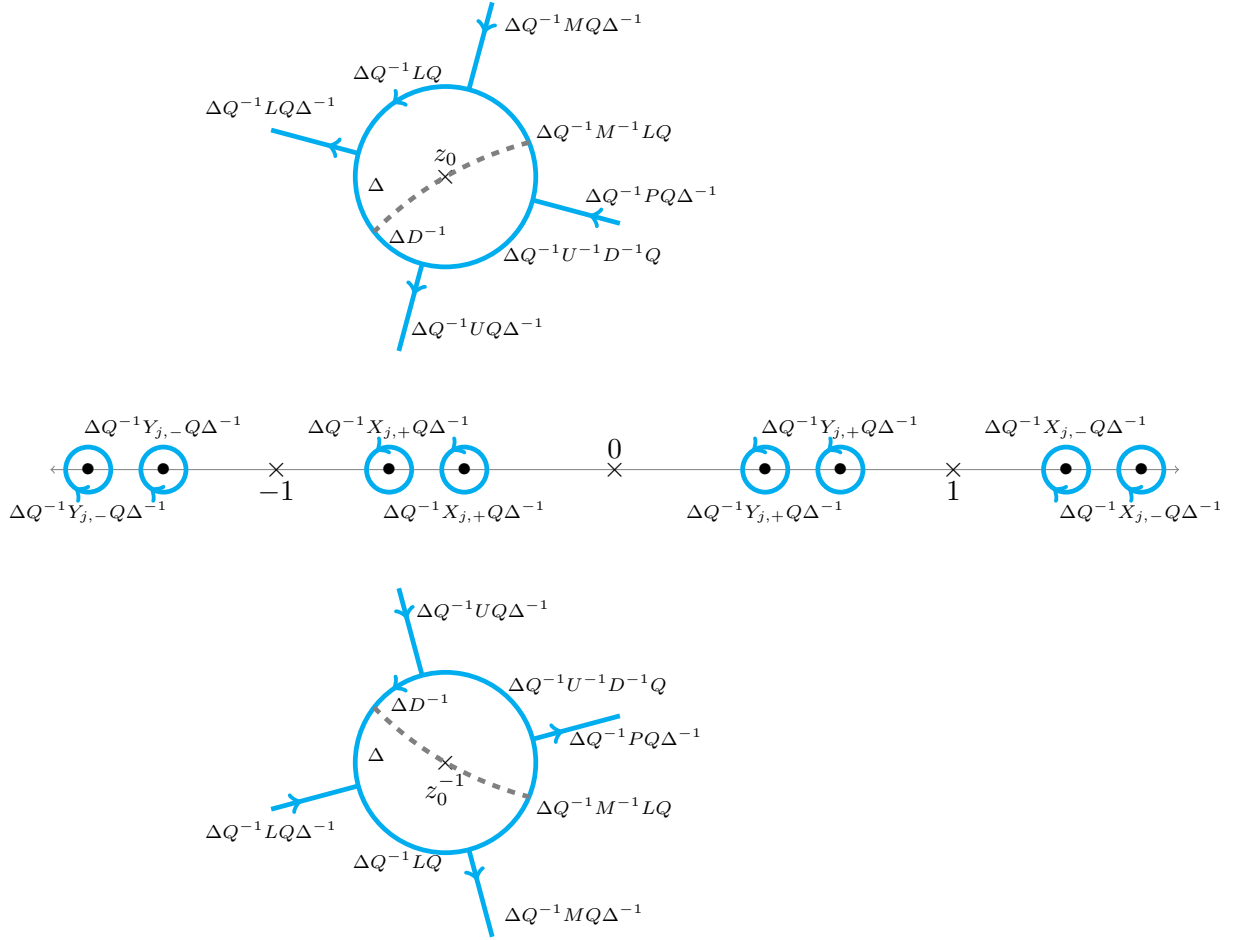
We have an important remark on boundary values of $m_{\sharp,\text{d}}(z)$ near $z=-1$.
\begin{remark}\label{r:sing}
  Due to the singularity in $L(z)$, $D(z)$, and $U(z)$ at $ z = -1$ it is not immediately clear in what sense $m_{2,\text{d}}(z)$ should satisfy the jump condition, or if a residue condition at $z = -1$ is needed for $m_{\sharp,\text{d}}(z)$.  Since we never solve for $m_{2,\text{d}}(z)$ we can ignore this issue if we understand what conditions we need on $m_{\sharp,\text{d}}(z)$.
{From Proposition~\ref{p:m-analytic} and the behavior of $\Delta_\pm(z)$ near $z = -1$} we conclude that $\widehat{m}(z) \Delta^{-1}(z)$ (and hence $\widetilde{m}(z) \Delta^{-1}(z)$) is continuous up to $\mathbb T$ with jump
\begin{align*}
	\widehat{m}^{+}(z) \Delta^{-1}_+(z) = \widehat{m}^-(z) \Delta^{-1}_-(z) \left[\Delta_{-}(z) L(z) \Delta_-^{-1}(z) \Delta_{+}(z) U(z) \Delta_+^{-1}(z) \right].
\end{align*}
It then follows that both $\Delta_{-}(z) L(z) \Delta^{-1}_-(z)$ and $\Delta_{+}(z) U(z) \Delta_+^{-1}(z)$ have analytic extensions to a strip that lies outside and inside the unit circle, respectively, and extend continuously up to the unit circle. Therefore, the jump contours and matrices of the vector problem (\rhref{rhp:disp}) can be deformed to those of $m_{\sharp,\text{d}}(z)$ leaving no singularity at $z = -1$.
\end{remark}

One final detail to be covered is the radius of the circles we have placed near $z_0$ and $z_0^{-1}$. We follow the methodology put forth in \cite{TOD_KdV} to determine this radius. Near $z_0$ we have
\begin{align*}
  \theta(z;n,t) = \theta (z_0;n,t) - 2 \frac{t+n z_0}{z_0^3} (z-z_0)^2 + \bigo (z-z_0)^3.
\end{align*}
We choose the radius of the circles to be proportional to $r(n,t) \defeq \sqrt{\frac{|z_0|^3}{|t+n z_0|}}$ so that for $s \defeq (z-z_0)/(c r(n,t))$, $|z-z_0| = |c r(n,t)|$
\begin{align}\label{E:scale}
  \tilde \theta(s;n,t) \defeq \theta\left( z_0 + s c r(n,t) ;n,t \right)) = \theta (z_0;n,t) - \tilde c s^2( 1 + o(1)),
\end{align}
where $\tilde c$ is proportional to $c^2$ and accounts for the phase.  Because $\theta(z_0;n,t)$ is purely imaginary, $e^{\tilde \theta(s;n,t)}$ is bounded if $s$ is bounded.  It also follows that for $|n|/t \leq c_1$, in the dispersive region, we have $r(n,t) = \mathcal O(t^{-1/2})$ so one may use $r(n,t)= c t^{-1/2}$, in practice.

{The function} $m_{\sharp,\text{d}}(z)$ satisfies a sectionally analytic {\RHP} with
\begin{itemize}
\item the jump conditions described in Figure~\ref{F:d_4},
\item the asymptotic symmetry condition given in \eqref{E:m2d}, and
\item the quadratic normalization condition present in \rhref{rhp:disp}.
\end{itemize}

\subsection{The Painlev\'e Region}\label{S:Painleve}

For $\frac{n}{t}>1$ this region intersects with the soliton region defined below, and we use that deformation (see Section~\ref{S:soliton}). For $\frac{n}{t}<1$, the saddle points are coalescing at $z=-1$ and this allows for a new deformation. Consider the arc $\Sigma$ that passes from $z=-1$ and the two stationary phase points $z^{\pm 1}_0$ as shown in Figure~\ref{F:disp}. Set $z = e^{i(\pi-\omega)}$ and $z_{0} = e^{i (\pi-\omega_0)}$ with $\cos (\pi-\omega_0 )= -\tfrac{n}{t}$. Thus $z\in\Sigma$ if and only if $|\omega| < |\omega_0|$, and $\omega_0 \to 0$ as $t\to \infty$. Choose the branch cut $(0, \infty)$ for the logarithm. Then $\theta(z;n,t)$ can be expressed in terms of $\omega$ as:
\begin{equation*}
\theta\big(e^{i(\pi - \omega)};n,t\big) = 2i\left( t\sin (\pi -\omega) + n(\pi- \omega)\right)\,.
\end{equation*}
Note that $\frac{|n|}{t} \geq 1 - ct^{-2/3}$ implies
\begin{equation*}
|\sin (\pi - \omega)| \leq \sin\omega_0 \leq \sqrt{2c}t^{-1/3}\sqrt{1 - \tfrac{c}{2} t^{-2/3}} = \sqrt{2c}t^{-1/3}\left( 1 - \frac{c}{4}t^{-2/3} + \frac{c}{16}t^{-4/3} + \mathcal{O}\left(t^{-2}\right)\right).
\end{equation*}
for large values of $t$. This together with $t - |n| <ct^{1/3}$ yields
\begin{equation*}
\left|\theta\big(e^{i(\pi - \omega)};n,t\big) - 2 \pi i n\right| = 2| t \sin (\pi - \omega) - n \omega | \leq 2|\sin \omega_0|\big( t - |n| \big) \leq 2\sqrt{2}c^2 + \mathcal{O}(t^{-2/3}),
\end{equation*}
{for $|\omega| < \omega_0$, which implies that the oscillations are controlled between the two stationary points. Therefore, the $LDU$-factorization that was used in the dispersive region is not needed. {Note that in this case the parametrix $\Delta(z)$, which has unbounded behavior at $z=-1$ in the dispersive region, is not used.} We perform a single deformation $\widetilde{m}(z)\longmapsto m_{\sharp,\text{P}}(z)$. Definition of the vector-valued function $m_{\sharp,\text{P}}(z)$ is given in Figure~\ref{F:p_1}(a). The jump contours and the jump matrices satisfied by $m_{\sharp,\text{P}}(z)$ are described in Figure~\ref{F:p_1}.
\begin{figure}[htp]
\centering
\subfigure[]{
\begin{tikzpicture}
\def\R{4}
\def\s{0.3}
\def\yL{\R}
\def\xL{\R}
\def\olens{\R+0.8}
\def\ilens{\R-1}

\coordinate (zo) at (120:\R);
\coordinate (zoi) at (240:\R);
\coordinate (o) at (0,0);
\draw[help lines,<->] (-\R-3,0) -- (\R+3,0) coordinate (xaxis);

\begin{scope}[very thick,decoration={
     markings,
   mark=at position 0.33 with {\arrow[line width =1.8pt]{>}},
  mark=at position 0.66 with {\arrow[line width =1.8pt]{>}}}
  ]
\path[draw,cyan,line width=2, postaction=decorate]
(zoi) arc(-120:120:\R);
\path[draw,cyan,line width=2, postaction=decorate]
(zo) arc(120:240:\R);
\end{scope}

\draw[dotted, gray,line width = 0.8] (zoi) to [out=45,in=270] (o) to [out=90, in=315] (zo);
\draw[dotted, gray,line width = 0.8] (zoi) to [out=225, in=10] (-\xL,-\yL);
\draw[dotted, gray,line width = 0.8] (zo) to [out=135, in=-10] (-\xL,\yL);
\begin{scope}[very thick,decoration={
     markings,
     mark=at position 0.05 with {\arrow[line width =1.8pt]{>}},
  mark=at position 0.5 with {\arrow[line width =1.8pt]{>}},
  mark=at position 0.95 with {\arrow[line width =1.8pt]{>}}}
  ]
\draw[black, line width =1.8, dashed, postaction=decorate] plot[smooth] coordinates {
   (zoi)
   ($(zoi) + (-75:1.13137)$)
   (-90:\olens)
   (-60:\olens)
   (-30:\olens)
   (0:\olens)
   (30:\olens)
   (60:\olens)
   (90:\olens)
   ($(zo)+(75:1.13137)$)
   (zo)
  };
  \draw[black, line width =1.8, dashed, postaction=decorate] plot[smooth] coordinates {
   (zoi)
   ($(zoi)+(15:1.41421)$)
   (-60:\ilens)
   (0:\ilens)
   (60:\ilens)
   ($(zo)+(-15:1.41421)$)
   (zo)
  };
\end{scope}
\begin{scope}[very thick,decoration={
     markings,
     mark=at position 0.1 with {\arrow[line width =1.8pt]{>}},
  mark=at position 0.5 with {\arrow[line width =1.8pt]{>}},
  mark=at position 0.9 with {\arrow[line width =1.8pt]{>}}}
  ]

\end{scope}

\foreach \Point in {(zo), (zoi), (o), (180:\R), (0:\R) }{
\node at \Point{$\times$};
}
\coordinate (zeta1) at (180:\R-2);
\coordinate (zeta2) at (180:\R-3);
\coordinate (zeta3) at (0:\R-2);
\coordinate (zeta4) at (0:\R-3);
\coordinate (zeta1i) at (180:\R+1.5);
\coordinate (zeta2i) at (180:\R+2.5);
\coordinate (zeta3i) at (0:\R+1.5);
\coordinate (zeta4i) at (0:\R+2.5);

\foreach \Point in {(zeta1), (zeta2), (zeta3), (zeta4), (zeta1i), (zeta2i), (zeta3i), (zeta4i)}{
\node at \Point{\textbullet};
}
\begin{scope}[very thick,decoration={
     markings,
  mark=at position 0.5 with {\arrow[line width =1.8pt]{>}}}
  ]

\draw[draw,cyan,line width =1.8, postaction=decorate]
(zeta1) + (0:\s) arc(0:360:\s);
\draw[draw,cyan,line width =1.8, postaction=decorate]
(zeta2) + (0:\s) arc(0:360:\s);
\draw[draw,cyan,line width =1.8, postaction=decorate]
(zeta3) + (0:\s) arc(0:360:\s);
\draw[draw,cyan,line width =1.8, postaction=decorate]
(zeta4) + (0:\s) arc(0:360:\s);
\draw[draw,cyan,line width =1.8, postaction=decorate]
(zeta1i) + (0:\s) arc(0:-360:\s);
\draw[draw,cyan,line width =1.8, postaction=decorate]
(zeta2i) + (0:\s) arc(0:-360:\s);
\draw[draw,cyan,line width =1.8, postaction=decorate]
(zeta3i) + (0:\s) arc(0:-360:\s);
\draw[draw,cyan,line width =1.8, postaction=decorate]
(zeta4i) + (0:\s) arc(0:-360:\s);
\end{scope}

\node[above,yshift=\s+7] at (zeta1) {\tiny{$Q^{-1}X_{j,+}Q$}};
\node[below,yshift=-\s-7] at (zeta2) {\tiny{$Q^{-1}X_{j,+}Q$}};
\node[above,yshift=\s+7] at (zeta1i) {\tiny{$Q^{-1}Y_{j,-}Q$}};
\node[below,yshift=-\s-7] at (zeta2i) {\tiny{$Q^{-1}Y_{j,-}Q$}};
\node[above,yshift=\s+7] at (zeta3) {\tiny{$Q^{-1}Y_{j,+}Q$}};
\node[below,yshift=-\s-7] at (zeta4) {\tiny{$Q^{-1}Y_{j,+}Q$}};
\node[above,yshift=\s+7] at (zeta3i) {\tiny{$Q^{-1}X_{j,-}Q$}};
\node[below,yshift=-\s-7] at (zeta4i) {\tiny{$Q^{-1}X_{j,-}Q$}};



\node (M) at (30:\R+0.2) {};
\coordinate (Mtext) at (30:\R+3);
\node[yshift=4pt] at (Mtext) {\tiny{$m_{\sharp,\text{P}}\defeq\widetilde{m}Q^{-1} MQ$}};
\draw [->,thin ] (Mtext) -- (M);

\node (P) at (20:\R-0.6) {};
\coordinate (Ptext) at (20:\R+2.5);
\node[yshift=4pt] at (Ptext) {\tiny{$m_{\sharp,\text{P}}\defeq\widetilde{m} (Q^{-1}P Q)^{-1}$}};
\draw [->,thin ] (Ptext) -- (P);

\node[below left] {$0$};
\node[above right] at (110:2) {\tiny{$m_{\sharp,\text{P}} \defeq \widetilde{m}$}};
\node [right] at(200:\R) {\tiny{$Q^{-1} J Q$}};
\node [left,xshift=1pt] at(-20:\R) {\tiny{$Q^{-1} J Q$}};
\node[above,yshift=6pt,xshift=-8pt] at (zo) (z0) {$z_0$};
\node[below, yshift=-6pt,xshift=-8pt] at (zoi) {$z_0^{-1}$};
\node[above ] at (130:\R+1.5) {\tiny{$m_{\sharp,\text{P}} \defeq \widetilde{m}$}};
\node at (-25:\R+2) {\tiny{$m_{\sharp,\text{P}} \defeq \widetilde{m}$}};
\node [left, gray] at (-\xL,\yL) {\tiny{$\Re \theta(z) =0$}};
\node [above left] at (180:\R) {$-1$};
\node [above right] at (0:\R) {$1$};
\end{tikzpicture}
}
\subfigure[]{
\begin{tikzpicture}
\def\R{4}
\def\s{0.3}
\def\yL{\R}
\def\xL{\R}
\def\olens{\R+0.8}
\def\ilens{\R-1}

\coordinate (zo) at (120:\R);
\coordinate (zoi) at (240:\R);
\coordinate (o) at (0,0);
\draw[help lines,<->] (-\R-3,0) -- (\R+3,0) coordinate (xaxis);

\begin{scope}[very thick,decoration={
     markings,
   mark=at position 0.33 with {\arrow[line width =1.8pt]{>}},
  mark=at position 0.66 with {\arrow[line width =1.8pt]{>}}}
  ]
\path[draw,dashed,gray,line width =1.8,postaction=decorate]
(zoi) arc(-120:120:\R);
\end{scope}
\begin{scope}[very thick,decoration={
     markings,
   mark=at position 0.1 with {\arrow[line width =1.8pt]{>}},
  mark=at position 0.9 with {\arrow[line width =1.8pt]{>}}}
  ]
\path[draw,cyan,line width =1.8, postaction=decorate]
(zo) arc(120:240:\R);
\end{scope}
\begin{scope}[very thick,decoration={
     markings,
     mark=at position 0.05 with {\arrow[line width =1.8pt]{>}},
  mark=at position 0.5 with {\arrow[line width =1.8pt]{>}},
  mark=at position 0.95 with {\arrow[line width =1.8pt]{>}}}
  ]
\draw[cyan, line width =1.8, postaction=decorate] plot[smooth] coordinates {
   (zoi)
   ($(zoi) + (-75:1.13137)$)
   (-90:\olens)
   (-60:\olens)
   (-30:\olens)
   (0:\olens)
   (30:\olens)
   (60:\olens)
   (90:\olens)
   ($(zo)+(75:1.13137)$)
   (zo)
  };
  \draw[cyan, line width =1.8, postaction=decorate] plot[smooth] coordinates {
   (zoi)
   ($(zoi)+(15:1.41421)$)
   (-60:\ilens)
   (0:\ilens)
   (60:\ilens)
   ($(zo)+(-15:1.41421)$)
   (zo)
  };
\end{scope}

\foreach \Point in {(zo), (zoi), (o), (180:\R), (0:\R) }{
\node at \Point{$\times$};
}
\coordinate (zeta1) at (180:\R-2);
\coordinate (zeta2) at (180:\R-3);
\coordinate (zeta3) at (0:\R-2);
\coordinate (zeta4) at (0:\R-3);
\coordinate (zeta1i) at (180:\R+1.5);
\coordinate (zeta2i) at (180:\R+2.5);
\coordinate (zeta3i) at (0:\R+1.5);
\coordinate (zeta4i) at (0:\R+2.5);

\foreach \Point in {(zeta1), (zeta2), (zeta3), (zeta4), (zeta1i), (zeta2i), (zeta3i), (zeta4i)}{
\node at \Point{\textbullet};
}
\begin{scope}[very thick,decoration={
     markings,
  mark=at position 0.5 with {\arrow[line width =1.8pt]{>}}}
  ]
\draw[draw,cyan,line width =1.8, postaction=decorate]
(zeta1) + (0:\s) arc(0:360:\s);
\draw[draw,cyan,line width =1.8, postaction=decorate]
(zeta2) + (0:\s) arc(0:360:\s);
\draw[draw,cyan,line width =1.8, postaction=decorate]
(zeta3) + (0:\s) arc(0:360:\s);
\draw[draw,cyan,line width =1.8, postaction=decorate]
(zeta4) + (0:\s) arc(0:360:\s);
\draw[draw,cyan,line width =1.8, postaction=decorate]
(zeta1i) + (0:\s) arc(0:-360:\s);
\draw[draw,cyan,line width =1.8, postaction=decorate]
(zeta2i) + (0:\s) arc(0:-360:\s);
\draw[draw,cyan,line width =1.8, postaction=decorate]
(zeta3i) + (0:\s) arc(0:-360:\s);
\draw[draw,cyan,line width =1.8, postaction=decorate]
(zeta4i) + (0:\s) arc(0:-360:\s);
\end{scope}

\node[above,yshift=\s+7] at (zeta1) {\tiny{$Q^{-1}X_{j,+}Q$}};
\node[below,yshift=-\s-7] at (zeta2) {\tiny{$Q^{-1}X_{j,+}Q$}};
\node[above,yshift=\s+7] at (zeta1i) {\tiny{$Q^{-1}Y_{j,-}Q$}};
\node[below,yshift=-\s-7] at (zeta2i) {\tiny{$Q^{-1}Y_{j,-}Q$}};
\node[above,yshift=\s+7] at (zeta3) {\tiny{$Q^{-1}Y_{j,+}Q$}};
\node[below,yshift=-\s-7] at (zeta4) {\tiny{$Q^{-1}Y_{j,+}Q$}};
\node[above,yshift=\s+7] at (zeta3i) {\tiny{$Q^{-1}X_{j,-}Q$}};
\node[below,yshift=-\s-7] at (zeta4i) {\tiny{$Q^{-1}X_{j,-}Q$}};
\node[right] (Mtext) at (30:\olens) {\tiny{$Q^{-1} MQ$}};
\node[left] (Mtext) at (30:\olens) {\color{cyan}\tiny{$\Gamma_-$}};
\node[left,xshift=-6pt] (Ptext) at (30:\ilens) {\tiny{$Q^{-1}P Q$}};
\node[right,xshift=-3pt] (Ptext) at (30:\ilens) {\color{cyan}\tiny{$\Gamma_+$}};

\node[below left] {$0$};
\node [right] at(200:\R) {\tiny{$Q^{-1} J Q$}};
\node[above,yshift=6pt,xshift=-8pt] at (zo) (z0) {$z_0$};
\node[below, yshift=-6pt,xshift=-8pt] at (zoi) {$z_0^{-1}$};
\node [above left] at (180:\R) {$-1$};
\node [above right] at (0:\R) {$1$};
\end{tikzpicture}
}
\caption{(a) Definition of $m_{\sharp,\text{P}}(z)$ and the `ghost' contours in preparation for the deformation, (b) The jump contours and matrices for the final {\RHP} in the Painlev\'{e} region with $0 <n < t$.}
\label{F:p_1}
\end{figure}
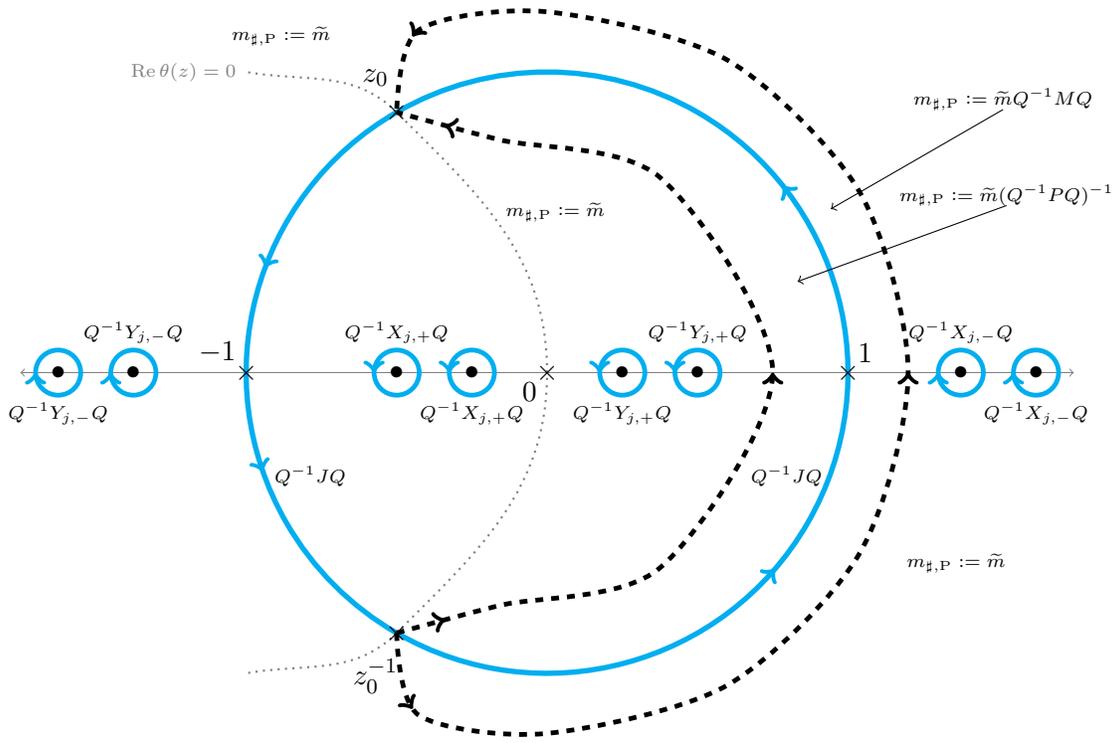
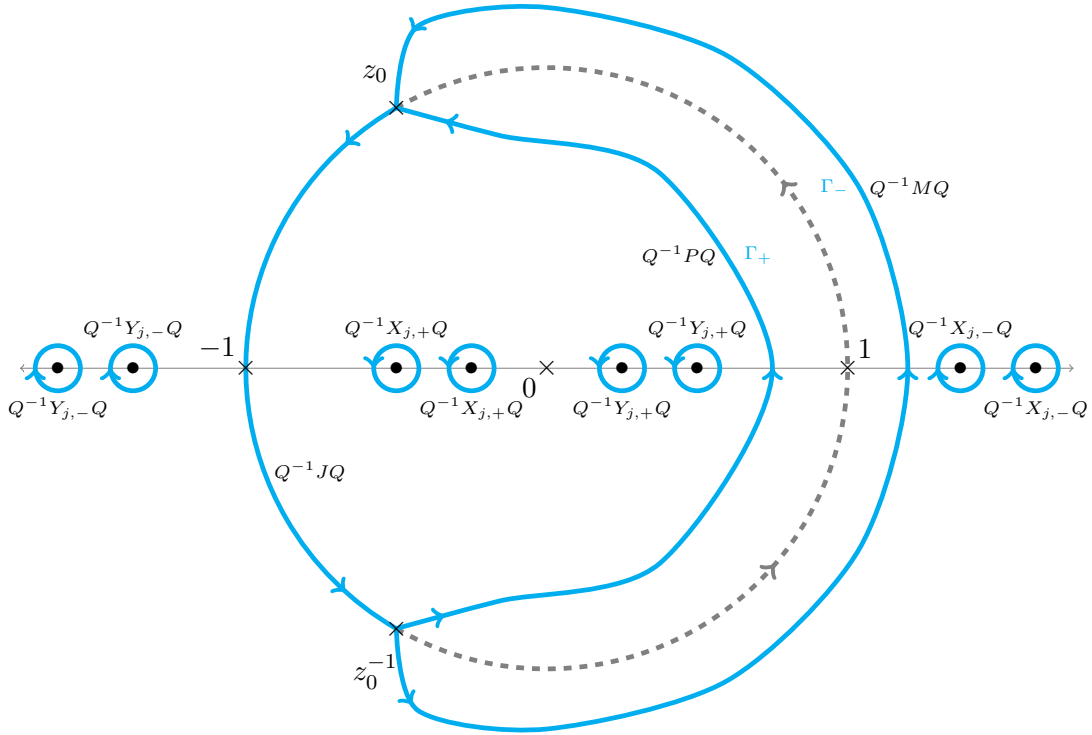
We note that $m_{\sharp,\text{P}}(z)$ takes continuous boundary values on its jump contour, and it is analytic on the arc of the jump contour where the jump matrix has been turned into the identity matrix. $m_{\sharp,\text{P}}(z)$ satisfies a sectionally analytic {\RHP} with
\begin{itemize}
\item the jump conditions described in Figure~\ref{F:p_1}(b),
\item the asymptotic symmetry condition present in \rhref{rhp:disp}, and
\item the quadratic normalization condition present in \rhref{rhp:disp}.
\end{itemize}

\subsection{The Collisionless Shock Region}\label{S:col_shock}
Recall that we use a deformation which involved $\Delta(z)$ in the dispersive region. The singularity at $z=-1$ in the matrix $D(z)$ destroys the boundedness of the parametrix\footnote{We use the term parametrix in a different way than is typical in the asymptotic analysis of {\RHP}s.  We use the term for any function that solves, or regularizes, any portion of the {\RHP}.} $\Delta(z;n,t)$. As $z \to -1$, the matrices $\Delta(z) Q^{-1}(z)M(z)Q(z)\Delta^{-1}(z)$ and $\Delta(z) Q^{-1}(z)P(z)Q(z)\Delta^{-1}(z)$ are unbounded and we cannot bridge the dispersive region and the Painlev\'{e} region. By adjusting the constants that determine the asymptotic regions we can make the dispersive and Painlev\'{e} regions overlap up to some finite $t$, but we wish to obtain a method which is stable for large values of $t$. To achieve this stability, we need to introduce additional deformations. The analogous region for the KdV equation has been introduced in \cite{SA} and the deformations were derived in \cite{DVZ}. The asymptotic analysis of the solutions, the scaling, and the needed deformations for the Toda lattice in this region, to the best of our knowledge, are not present in the literature.

As $n$ increases in the dispersive region, the stationary phase points of $e^{\theta(z)}$ approach the singularity ($z=-1$) of the parametrix $\Delta(z)$. To prevent this, we replace the exponent $\theta(z;n,t)$ by a so-called $g$-function as was done for KdV in \cite{DZV_EXT} (see also \cite{TOD_KdV}). In what follows, we define the $g$-function $g(z)$ as the solution of an {\RHP} with the properties that {mollify the unboundedness of $D(z)$.} Having done that, we introduce the needed deformations in this region. We leave the implementation details for solution of the {\RHP} given in \eqref{E:g_rhp} to Appendix~\ref{g-func-comp}. In Appendix~\ref{A:construct_g}, we explicitly construct the $g$-function, but it is more convenient to compute it numerically from the {\RHP} formulation given in \eqref{E:g_rhp}.

For $n/t \leq 1$, we let $\lambda_0 \in[-1,0]$ and $\rho_0 \in [0,1]$ denote the real and imaginary parts of the stationary phase point $z_0$, respectively. {Explicitly,}
\begin{equation}\label{E:rho}
\lambda_0 = -\frac{n}{t} ~\text{ and }~\rho_0 = \sqrt{1 - \left(\frac{n}{t}\right)^2}\,.
\end{equation}
For $z_1,\,z_2 \in \mathbb T$, $z_1 \neq z_2$, with $0\leq \arg z_j < 2\pi$, $j = 1,2$, define
\begin{align*}
  \oarc{z_1,z_2} = \left\lbrace e^{i \theta} : \min\{ \arg z_1,~ \arg z_2\} < \theta < \max\{ \arg z_1,~ \arg z_2\}\right\rbrace
\end{align*}
oriented from $z_1$ to $z_2$.  Then $\carc{z_1,z_2}$ is {defined to be} the closure of $\oarc{z_1,z_2}$. For $\alpha,\beta \in \mathbb{T}$ with $-1 \leq \Re{\alpha} \leq \lambda_0 \leq \Re{\beta} \leq 1$, we label $\Sigma_u=\carc{\beta,\alpha},~\Sigma_c=\carc{\alpha, \alpha^{-1}}$, and $\Sigma_l =\carc{\alpha^{-1}, \beta^{-1}}$, as shown in Figure~\ref{F:g}.
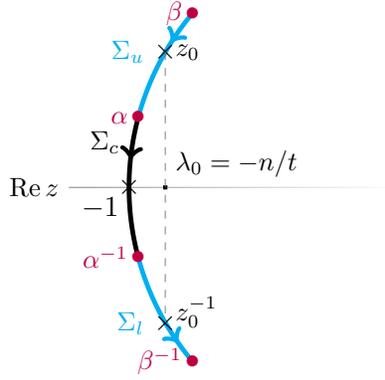
\begin{figure}[htp]
\centering
\begin{tikzpicture}[scale=0.4]
\def\L{9}
\coordinate (zo) at (150:\L);
\coordinate (zoi) at (210:\L);
\coordinate (o) at (0,0);
\coordinate (b) at (140:\L);
\coordinate (bi) at (220:\L);
\coordinate (a) at (165:\L);
\coordinate (ai) at (195:\L);
\coordinate (xaxisend) at (-\L-2,0);
\pgfmathsetmacro\lo{cos(150)*\L};

\draw[help lines, path fading=east] (0,0) -- (-\L-2,0) coordinate (xaxis) ++(0,1pt);

\draw [line width =1.8, cyan, domain=140:165, decoration={markings, mark=at position 0.3 with {\arrow[line width=2pt]{>}}},postaction=decorate] plot ({\L*cos(\x)}, {\L*sin(\x)});
\draw [line width =1.8, black, domain=165:195, decoration={markings,mark=at position 0.3 with {\arrow[line width=2pt]{>}}},postaction=decorate] plot ({\L*cos(\x)}, {\L*sin(\x)});
\draw [line width =1.8, cyan, domain=195:220, decoration={markings, mark=at position 0.8 with {\arrow[line width=2pt]{>}}},postaction=decorate] plot ({\L*cos(\x)}, {\L*sin(\x)});

\foreach \Point in {(zo), (zoi)}{
  \node at \Point {$\times$};
}
\foreach \Point in {(a), (ai),(b),(bi)}{
  \node at \Point {{\color{purple}\textbullet}};
}

\draw [ultra thick] (\lo,-2pt) -- (\lo,2pt);

\draw [very thin, dashed, gray] (zo) -- (zoi);

\node[left] at (xaxisend) {\small{$\Re z$}};
\node[right ] at (zo) (z0) {\small{$z_0 $}};
\node[right, yshift=4pt ] at (zoi) {\small{$z_0^{-1}$}};
\node [below left] at (180:\L) {$-1$};
\node[left] at (b) {\color{purple} \small{$\beta$}};
\node[left] at (a) {\color{purple} \small{$\alpha$}};

\node[left] at (ai) {\color{purple} \small{$\alpha^{-1}$}};

\node [above left] at (155:\L) {\color{black} \small{\color{cyan}$\Sigma_u$}};
\node [below left] at (205:\L) {\color{black} \small{\color{cyan}$\Sigma_l$}};
\node [above left] at (175:\L) {\color{black} \small{$\Sigma_c$}};
\node[left] at (bi) {\color{purple} \small{$\beta^{-1}$}};

\node at (-\L,0) {$\times$};
\node [above, anchor=south west] at (\lo,0) {\small{$\lambda_0=-n/t$}};

\end{tikzpicture}
\caption{The cuts $\Sigma_u$, $\Sigma_c$, and $\Sigma_l$ for $0<n<t$.}
\label{F:g}
\end{figure}

Before describing the sequence of deformations used in this region, we proceed with the {\RHP} that determines the $g$-function. Following the approach in \cite{TOD_KdV}, we determine $\alpha$ and $\beta$ on the unit circle so that there exists a function $g(z)$ that satisfies the following properties, for some complex constants $\delta_1$ and $\delta_2$
\begin{equation}\label{E:g_rhp}
\begin{aligned}
&g^{+}(z) + g^{-}(z) =
\begin{cases}
\begin{aligned}
{\delta_1/t}&\quad\text{if } z\in \Sigma_u, \\
{-\delta_1/t}&\quad\text{if } z\in \Sigma_l,
\end{aligned}
\end{cases}\\
&g^{+}(z) - g^{-}(z) = {\delta_2/t},\phantom{x} z\in \Sigma_c,\\
&g(z) - \frac{1}{2t}\theta(z) \text{ analytic in $z$ for } z \not\in \carc{\beta,\beta^{-1}}= \Sigma_u \cup \Sigma_c \cup \Sigma_l,\\
&g(z)\text{ is bounded at }z=\alpha^{\pm 1} \text{ and } z=\beta^{\pm 1},\\
&g(z) = \tfrac{1}{2}z - \lambda_0\log z + \mathcal{O}\left(z^{-1}\right)\text{ as } z\rightarrow\infty\,,
\end{aligned}
\end{equation}
The constants {$\delta_1$ and $\delta_2$} depend on $\alpha$ and $\beta$ and, as we will see, they have the desired properties to eliminate the singularities. We leave the details of our method to solve this {\RHP} to Appendix~\ref{g-func-comp}. Once $g(z)$ is obtained, define the scalar function
\begin{equation*}\mathfrak{g}(z) \defeq t\left(g(z) - \frac{1}{2t}\theta(z) \right),
\end{equation*}
and construct the matrix function
\begin{equation}\label{e:phi}
\phi(z) \defeq \begin{pmatrix}
e^{\mathfrak{g}(z)} & 0 \\ 0 & e^{-\mathfrak{g}(z)}
\end{pmatrix},
\end{equation}
which has the asymptotic behavior
\begin{equation*}
\phi(z)\to I \text{ as }z\to\infty.
\end{equation*}
Note that, for $z\in\mathbb{T}$ the jump condition satisfied by $\widetilde{m}(z)\phi(z)$ is $\phi_{-}^{-1}(z)Q^{-1}(z)J(z) Q(z)\phi_{+}(z)$. {In order to determine $\delta_1$ and $\delta_2$ we proceed as if $Q(z) = I$.  It will be clear that this is sufficient.} In this case, the jump condition (satisfied by the vector function $\widetilde{m}(z)\phi(z)$) on $\mathbb{T}$ is given explicitly by
\begin{equation*}
\phi_{-}^{-1}(z)J(z)\phi_{+}(z) = \begin{pmatrix}
\left[1 - R(z)R\left(z^{-1}\right) \right]e^{\mathfrak{g}^{+}(z) -\mathfrak{g}^{-}(z) } & - R\left(z^{-1}\right)e^{-\theta(z) - \mathfrak{g}^{+}(z) - \mathfrak{g}^{-}(z)}\\[2pt]
R(z)e^{\theta(z) + \mathfrak{g}^{+}(z) + \mathfrak{g}^{-}(z)} & e^{- \mathfrak{g}^{+}(z) + \mathfrak{g}^{-}(z)}
\end{pmatrix},
\end{equation*}
because $\mathfrak{g}(z)$ satisfies
\begin{equation*}
\begin{aligned}
\mathfrak{g}^{+}(z) - \mathfrak{g}^{-}(z) &= t\left(g^{+}(z) - g^{-}(z)\right) = 0,\text{ for } z\notin\carc{\beta,\beta^{-1}}\\
2\mathfrak{g}(z) & = 2 t g(z)  - \theta(z) \to 0 \text{ as } z\to \infty.
\end{aligned}
\end{equation*}
{We write}
\begin{equation}\label{E:cs-jump}
\phi_{-}^{-1}(z)J(z)\phi^{+}(z) =
\begin{cases}
\begin{pmatrix} 1 - R(z)R\left(z^{-1}\right) & - R\left(z^{-1}\right) e^{-2tg(z)} \\ R(z) e^{2 tg(z)} & 1 \end{pmatrix},\quad & z \in \oarc{1,\beta},\\[12pt]
\begin{pmatrix} \left[1 - R(z)R\left(z^{-1}\right)\right]e^{t\left(g^{+}(z) - g^{-}(z)\right)} & - R\left(z^{-1}\right) e^{-\delta_1} \\ R(z) e^{\delta_1} & e^{t\left(-g^{+}(z) + g^{-}(z)\right)} \end{pmatrix},\quad & z \in \oarc{\beta,\alpha},\\[12pt]
\begin{pmatrix} \left[1 - R(z)R\left(z^{-1}\right)\right]e^{\delta_2} & - R\left(z^{-1}\right) e^{-t \left(g^{+}(z) + g^{-}(z)\right)} \\ R(z) e^{t\left(g^{+}(z) + g^{-}(z)\right)} & e^{-\delta_2} \end{pmatrix},\quad & z \in \oarc{\alpha,\alpha^{-1}},\\[12pt]
\begin{pmatrix} \left[1 - R(z)R\left(z^{-1}\right)\right]e^{t\left(g^{+}(z) - g^{-}(z)\right)} & - R\left(z^{-1}\right) e^{\delta_1} \\ R(z) e^{-\delta_1} & e^{t\left(-g^{+}(z) + g^{-}(z)\right)} \end{pmatrix},\quad & z \in \oarc{\alpha^{-1},\beta^{-1}},\\[12pt]
\begin{pmatrix} 1 - R(z)R\left(z^{-1}\right) & - R\left(z^{-1}\right) e^{-2tg(z)} \\ R(z) e^{2t g(z)} & 1 \end{pmatrix},\quad & z \in \oarc{\beta^{-1}, 1}.
\end{cases}
\end{equation}
Here $\delta_1 /t = g^{+}(z) + g^{-}(z) $ for $z\in\Sigma_u$, and $\delta_2 / t = g^{+}(z) - g^{-}(z) $ for $z\in\Sigma_c$. As can be seen in \eqref{E:cs-jump}, this conjugation removes $\theta(z;n,t)$ from the problem.

We now present the initial deformation $\widetilde{m}(z) \longmapsto m_{1,\text{cs}}(z)$ in the collisionless shock region. As in the dispersive region, we use the $J(z)=L(z)D(z)U(z)$ factorization on $\Sigma_{c}$ and define $m_{1,\text{cs}}(z)$, see Figure~\ref{F:cs_1}(a). {Here we used the lensing process (see Remark~\ref{r:lens}) to deform the {\RHP}}. The jumps and contours near $\alpha^{-1}$ and $\beta^{-1}$ are given in Figure~\ref{F:cs_1}(b). What happens near $\alpha$ and $\beta$ is clear by symmetry.
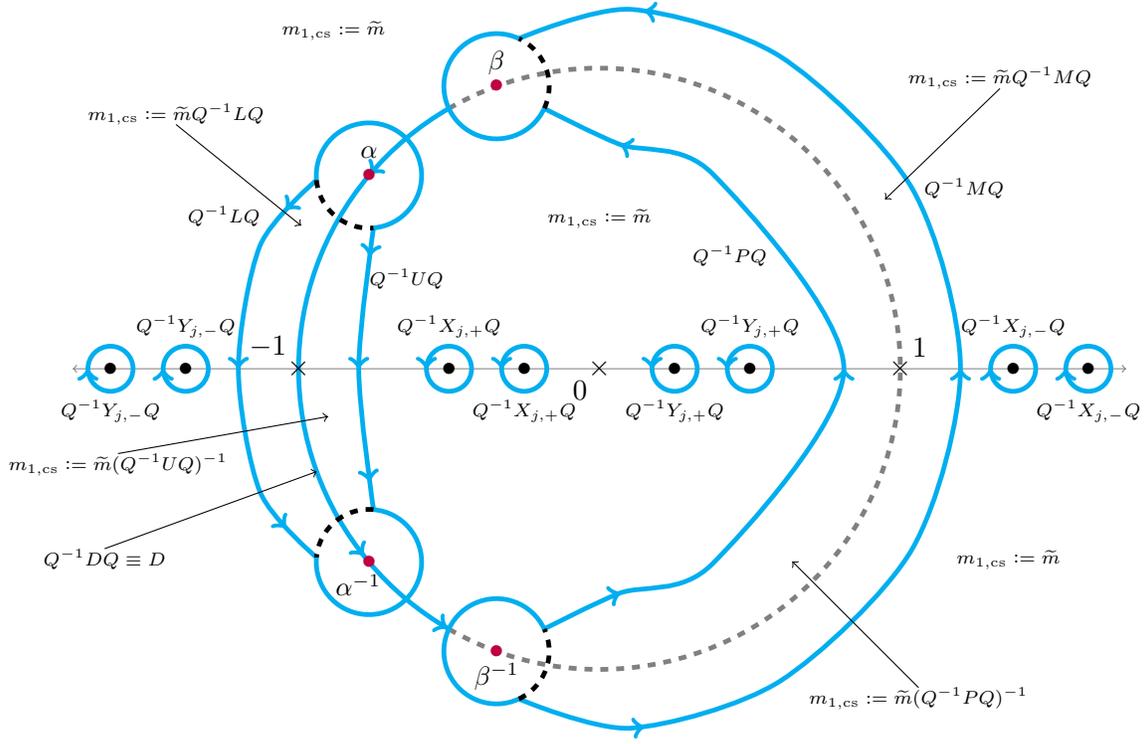
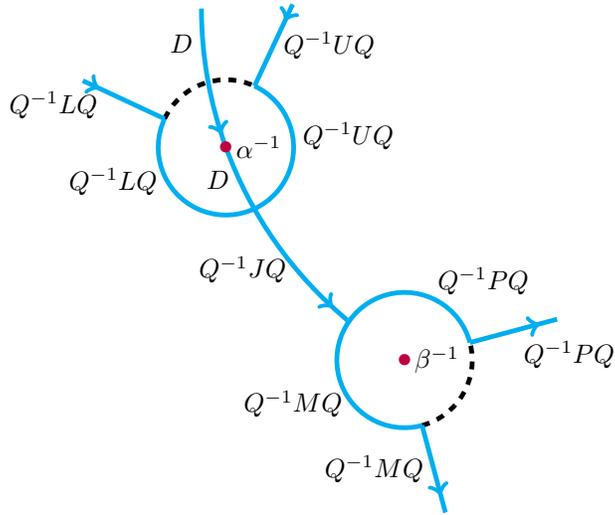
\begin{figure}[htp]
\centering
\subfigure[]{
\begin{tikzpicture}
\def\R{4}
\def\s{0.3}
\def\m{0.7}
\def\yL{\R}
\def\xL{\R}
\def\olens{\R+0.8}
\def\ilens{\R-1}

\coordinate (a) at (140:\R);
\coordinate (ai) at (-140:\R);
\coordinate (b) at (110:\R);
\coordinate (bi) at (-110:\R);
\coordinate (zo) at (120:\R);
\coordinate (zoi) at (240:\R);
\coordinate (o) at (0,0);

\draw[help lines,<->] (-\R-3,0) -- (\R+3,0) coordinate (xaxis);


\coordinate (pt1) at (-2.00239, 3.46272);
\coordinate (pt2) at (-2.00239, -3.46272);
\coordinate (pt3) at (-0.69187, 3.93971);
\coordinate (pt4) at (-0.69187, -3.93971);
\coordinate (end) at (180:\R);
\coordinate (begin) at (0:\R);

\begin{scope}[decoration={
  markings,
  mark=at position 0.98 with {\arrow[line width =1.8pt]{>}}}
  ]
  \tkzCircumCenter(pt1,end,pt2)\tkzGetPoint{O};
  \tkzDrawArc[color=cyan,line width =1.8,postaction=decorate](O,pt1)(a);
  \tkzDrawArc[color=cyan,line width =1.8,postaction=decorate](O,a)(ai);
  \tkzDrawArc[color=cyan,line width =1.8,postaction=decorate](O,ai)(pt2);
\end{scope}
\tkzCircumCenter(pt3,begin,pt4)\tkzGetPoint{O};
\tkzDrawArc[color=gray,dashed, line width =1.8](O,pt2)(pt1);

\begin{scope}[decoration={markings,
mark=at position 0.1 with {\arrow[line width =1.8pt]{>}},
mark=at position 0.5 with {\arrow[line width =1.8pt]{>}},
mark=at position 0.9 with {\arrow[line width =1.8pt]{>}}
}
]

\draw[cyan, line width =1.8, postaction=decorate] plot[smooth] coordinates {
   ($(bi) + (-65:\m)$)
   (-85:\olens)
   (-60:\olens)
   (-30:\olens)
   (0:\olens)
   (30:\olens)
   (60:\olens)
   (85:\olens)
   ($(b) +(65:\m)$)
  };

\draw[cyan, line width =1.8, postaction=decorate] plot[smooth] coordinates {
   ($(bi) + (25:\m)$)
   (-85:\ilens)
   (-60:\ilens)
   (0:\ilens+0.25)
   (60:\ilens)
   (85:\ilens)
   ($(b) +(-25:\m)$)
  };
\draw[cyan, line width =1.8, postaction=decorate] plot[smooth] coordinates {
    ($(a)+(185:\m)$)
   (160:\olens)
   (180:\olens)
   (200:\olens)
  ($(ai)+(175:\m)$)
   };
\draw[cyan, line width =1.8, postaction=decorate] plot[smooth] coordinates {
    ($(a)+(275:\m)$)
   (180:\ilens+0.2)
   ($(ai)+(85:\m)$)
     };
 \end{scope}
\draw[draw,cyan,line width =1.8]
(a) + (275:\m) arc(-85:185:\m);
\draw[draw,dashed,line width =1.8]
(a) + (185:\m) arc(185:275:\m);

\draw[draw,cyan,line width =1.8]
(ai) + (175:\m) arc(-185:85:\m);
\draw[draw,dashed,line width =1.8]
(ai) + (85:\m) arc(85:175:\m);

\draw[draw,dashed,line width =1.8]
(b) + (-25:\m) arc(-25:65:\m);
\draw[draw,cyan,line width =1.8]
(b) + (65:\m) arc(65:335:\m);

\draw[draw,cyan,line width =1.8]
(bi) + (25:\m) arc(25:295:\m);
\draw[draw,dashed,line width =1.8]
(bi) + (-65:\m) arc(-65:25:\m);

\foreach \Point in {(o), (180:\R), (0:\R) }{
\node at \Point{$\times$};
}
\foreach \Point in{(a), (ai),(b),(bi)}{
\node at \Point {\color{purple}\textbullet};
}
\coordinate (zeta1) at (180:\R-2);
\coordinate (zeta2) at (180:\R-3);
\coordinate (zeta3) at (0:\R-2);
\coordinate (zeta4) at (0:\R-3);
\coordinate (zeta1i) at (180:\R+1.5);
\coordinate (zeta2i) at (180:\R+2.5);
\coordinate (zeta3i) at (0:\R+1.5);
\coordinate (zeta4i) at (0:\R+2.5);

\foreach \Point in {(zeta1), (zeta2), (zeta3), (zeta4), (zeta1i), (zeta2i), (zeta3i), (zeta4i)}{
\node at \Point{\textbullet};
}
\begin{scope}[decoration={markings,
mark=at position 0.5 with {\arrow[line width =1.8pt]{>}}
}
]
\draw[draw,cyan,line width =1.8, postaction=decorate]
(zeta1) + (0:\s) arc(0:360:\s);
\draw[draw,cyan,line width =1.8, postaction=decorate]
(zeta2) + (0:\s) arc(0:360:\s);
\draw[draw,cyan,line width =1.8, postaction=decorate]
(zeta3) + (0:\s) arc(0:360:\s);
\draw[draw,cyan,line width =1.8, postaction=decorate]
(zeta4) + (0:\s) arc(0:360:\s);
\draw[draw,cyan,line width =1.8, postaction=decorate]
(zeta1i) + (0:\s) arc(0:-360:\s);
\draw[draw,cyan,line width =1.8, postaction=decorate]
(zeta2i) + (0:\s) arc(0:-360:\s);
\draw[draw,cyan,line width =1.8, postaction=decorate]
(zeta3i) + (0:\s) arc(0:-360:\s);
\draw[draw,cyan,line width =1.8, postaction=decorate]
(zeta4i) + (0:\s) arc(0:-360:\s);
\end{scope}
\node[above,yshift=\s+7] at (zeta1) {\tiny{$Q^{-1}X_{j,+}Q$}};
\node[below,yshift=-\s-7] at (zeta2) {\tiny{$Q^{-1}X_{j,+}Q$}};
\node[above,yshift=\s+7] at (zeta1i) {\tiny{$Q^{-1}Y_{j,-}Q$}};
\node[below,yshift=-\s-7] at (zeta2i) {\tiny{$Q^{-1}Y_{j,-}Q$}};
\node[above,yshift=\s+7] at (zeta3) {\tiny{$Q^{-1}Y_{j,+}Q$}};
\node[below,yshift=-\s-7] at (zeta4) {\tiny{$Q^{-1}Y_{j,+}Q$}};
\node[above,yshift=\s+7] at (zeta3i) {\tiny{$Q^{-1}X_{j,-}Q$}};
\node[below,yshift=-\s-7] at (zeta4i) {\tiny{$Q^{-1}X_{j,-}Q$}};
\node [left] at (155:\olens) {\tiny{$Q^{-1} LQ$}};
\node[right,xshift=2pt] at (160:\ilens+0.5) {\tiny{$Q^{-1}U Q$}};
\node[right] at (30:\olens) {\tiny{$Q^{-1} MQ$}};
\node[left,xshift=-6pt] at (30:\ilens) {\tiny{$Q^{-1}P Q$}};
\node[below left] {$0$};

\coordinate (dtext) at(200:\R+3);
\node[yshift=-4pt] at (dtext) {\tiny{$Q^{-1} D Q \equiv D$}};
\coordinate (d) at (200:\R);
\draw [->,thin] (dtext) -- (d);

\node[above,yshift=2pt] at (a) {\small$\alpha$};
\node[below,xshift=-4pt] at (ai) {\small$\alpha^{-1}$};
\node[above] at (b) {\small$\beta$};
\node[below] at (bi) {\small$\beta^{-1}$};
\node [above left] at (180:\R) {$-1$};
\node [above right] at (0:\R) {$1$};

\node at (90:2) {\tiny{$m_{1,\text{cs}}\defeq \widetilde{m}$}};

\node (L) at (155:\R+0.2) {};
\coordinate (Ltext) at (150:\R+2.5);
\node[yshift=4pt] at (Ltext) {\tiny{$m_{1,\text{cs}}\defeq\widetilde{m}Q^{-1} LQ$}};
\draw [->,thin ] (Ltext) -- (L);

\node (U) at (190:\R-0.5) {};
\coordinate (Utext) at (190:\R+2.5);
\node[yshift=-4pt] at (Utext) {\tiny{$m_{1,\text{cs}}\defeq\widetilde{m} (Q^{-1}U Q)^{-1}$}};
\draw [->,thin ] (Utext) -- (U);

\node (M) at (30:\R+0.2) {};
\coordinate (Mtext) at (35:\R+2.5);
\node[yshift=4pt] at (Mtext) {\tiny{$m_{1,\text{cs}}\defeq\widetilde{m}Q^{-1} MQ$}};
\draw [->,thin ] (Mtext) -- (M);

\node (P) at (-45:\R-0.6) {};
\coordinate (Ptext) at (-45:\R+2);
\node[yshift=-4pt] at (Ptext) {\tiny{$m_{1,\text{cs}}\defeq\widetilde{m} (Q^{-1}P Q)^{-1}$}};
\draw [->,thin] (Ptext) -- (P);

\node[above] at (130:\R+1.5) {\tiny{$m_{1,\text{cs}} \defeq \widetilde{m}$}};
\node at (-25:\R+2) {\tiny{$m_{1,\text{cs}} \defeq \widetilde{m}$}};
\end{tikzpicture}
}\\
\subfigure[]{
\begin{tikzpicture}[scale=0.6]
\def\R{9}
\def\m{1.5}
\coordinate (begin) at (180:\R);
\coordinate (a) at (200:\R);
\coordinate (b) at (240:\R);
\coordinate (end) at (255:\R);
\coordinate (zo) at (-6.47862,-6.2472); 
\coordinate (left) at (-5.73202, -6.93858);
\coordinate (right) at (-3.14298, -8.43337);
\node[right] at (a) {\small{$\alpha^{-1}$}};
\node[right] at (b) {\small{$\beta^{-1}$}};
\tkzCircumCenter(a,zo,b)\tkzGetPoint{O};
\begin{scope}[decoration={markings,
mark=at position 0.9 with {\arrow[line width =1.8pt]{>}}
}
]
\tkzDrawArc[color=cyan, line width =1.8,postaction=decorate](O,begin)(a);
\tkzDrawArc[color=cyan, line width =1.8,postaction=decorate](O,a)(left);
\end{scope}
\draw[draw,dashed,line width =1.8]
(a) + (155:\m) arc(155:65:\m);
\draw[draw,cyan,line width =1.8]
(a) + (155:\m) arc(155:425:\m);

\draw[draw,dashed,line width =1.8]
(b) + (285:\m) arc(285:375:\m);
\draw[draw,cyan,line width =1.8]
(b) + (285:\m) arc(285:15:\m);

\foreach \Point in{(a), (b)}{
\node at \Point {\color{purple}\textbullet};
}
\begin{scope}[decoration={markings,
mark=at position 0.2 with {\arrow[line width =1.8pt]{>}}
}
]
\draw [cyan, line width =1.8,postaction=decorate] ($(a) + (155:\m+2)$)--($(a) + (155:\m)$);
\draw [cyan, line width =1.8,postaction=decorate] ($(a) + (65:\m+2)$)--($(a) + (65:\m)$);
\end{scope}
\begin{scope}[decoration={markings,
mark=at position 0.8 with {\arrow[line width =1.8pt]{>}}
}
]
\draw [cyan, line width =1.8,postaction=decorate] ($(b) + (15:\m)$)--($(b) + (15:\m+2)$);
\draw [cyan, line width =1.8,postaction=decorate] ($(b) + (285:\m)$)--($(b) + (285:\m+2)$);
\end{scope}
\node[left] at (185:\R) {\small{$D$}};
\node[left,yshift=-2pt,xshift=-6pt] at ($(a) + (155:\m+1)$) {\small{$Q^{-1}LQ$}};
\node[left] at ($(a) + (210:\m)$) {\small{$Q^{-1}LQ$}};
\node[right] at ($(a) + (65:\m+1)$) {\small{$Q^{-1}UQ$}};
\node[right] at ($(a) + (10:\m)$) {\small{$Q^{-1}UQ$}};
\node[left] at (205:\R) {\small{$D$}};
\node[left] at (220:\R) {\small{$Q^{-1}JQ$}};
\node[below right] at ($(b) + (15:\m+1)$) {\small{$Q^{-1}PQ$}};
\node[right,yshift=6pt] at ($(b) + (70:\m)$) {\small{$Q^{-1}PQ$}};
\node[ left] at ($(b) + (285:\m+1)$) {\small{$Q^{-1}MQ$}};
\node[ left] at ($(b) + (220:\m)$) {\small{$Q^{-1}MQ$}};

\end{tikzpicture}
}
\caption{The initial deformation of the {\RHP} in the collisionless shock region. (a) Definition of $m_{1,\text{cs}}(z)$ with its jump contours and matrices, (b) The initial jump contours and matrices near $\alpha^{-1}$, $\beta^{-1}$.}
\label{F:cs_1}
\end{figure}

We now perform our second deformation, $m_{1,\text{cs}}(z)\longmapsto m_{2,\text{cs}}(z)$, in this region. Define $m_{2,\text{cs}}(z)$ inside the circles centered at $\alpha$ and $\alpha^{-1}$ as shown in Figure~\ref{F:cs_2}(a) and leave $m_{2,\text{cs}}(z)\equiv m_{1,\text{cs}}(z)$ everywhere else. The jump conditions satisfied by $m_{2,\text{cs}}(z)$ near the points $\alpha^{-1}$ and $\beta^{-1}$ are shown in Figure~\ref{F:cs_2}(b).
\begin{figure}[htp]
\centering
\subfigure[]{
\begin{tikzpicture}[scale=0.65]
\def\R{9}
\def\m{1.5}
\coordinate (begin) at (180:\R);
\coordinate (a) at (200:\R);
\coordinate (b) at (240:\R);
\coordinate (end) at (255:\R);
\coordinate (zo) at (-6.47862,-6.2472); 
\coordinate (left) at (-5.73202, -6.93858);
\coordinate (right) at (-3.14298, -8.43337);
\node[right] at (a) {\small{$\alpha^{-1}$}};
\node[right] at (b) {\small{$\beta^{-1}$}};
\tkzCircumCenter(a,zo,b)\tkzGetPoint{O};
\begin{scope}[decoration={markings,
mark=at position 0.9 with {\arrow[line width =1.8pt]{>}}
}
]
\tkzDrawArc[color=cyan, line width =1.8,postaction=decorate](O,begin)(a);
\tkzDrawArc[color=cyan, line width =1.8,postaction=decorate](O,a)(left);
\end{scope}
\draw[draw,dashed,line width =1.8]
(a) + (155:\m) arc(155:65:\m);
\draw[draw,cyan,line width =1.8]
(a) + (155:\m) arc(155:425:\m);

\draw[draw,dashed,line width =1.8]
(b) + (285:\m) arc(285:375:\m);
\draw[draw,cyan,line width =1.8]
(b) + (285:\m) arc(285:15:\m);

\foreach \Point in{(a), (b)}{
\node at \Point {\color{purple}\textbullet};
}
\begin{scope}[decoration={markings,
mark=at position 0.2 with {\arrow[line width =1.8pt]{>}}
}
]
\draw [cyan, line width =1.8,postaction=decorate] ($(a) + (155:\m+2)$)--($(a) + (155:\m)$);
\draw [cyan, line width =1.8,postaction=decorate] ($(a) + (65:\m+2)$)--($(a) + (65:\m)$);
\end{scope}
\begin{scope}[decoration={markings,
mark=at position 0.8 with {\arrow[line width =1.8pt]{>}}
}
]
\draw [cyan, line width =1.8,postaction=decorate] ($(b) + (15:\m)$)--($(b) + (15:\m+2)$);
\draw [cyan, line width =1.8,postaction=decorate] ($(b) + (285:\m)$)--($(b) + (285:\m+2)$);
\end{scope}

\coordinate (phi_L_text) at ($(a)+(230:\m+3)$);
\node[yshift=-4pt] at (phi_L_text) {\small{$m_{2,\text{cs}} \defeq m_{1,\text{cs}} D$}};
\node(phi_L) at ($(a)+(210:\m/2) $){};
\draw [->,thin] (phi_L_text) -- (phi_L);
\node(phi_R_text) at ($(a)+(-15:\m+3.5)$) {\small{$m_{2,\text{cs}} \defeq m_{1,\text{cs}} $}};
\node(phi_R) at ($(a)+(-35:\m/2) $){};
\draw [->,thin] (phi_R_text) -- (phi_R);

\node[left] at (185:\R) {\small{$D$}};
\node[left,yshift=-2pt,xshift=-6pt] at ($(a) + (155:\m+1)$) {\small{$Q^{-1}LQ$}};
\node[left] at ($(a) + (210:\m)$) {\small{$Q^{-1}LQ$}};
\node[right] at ($(a) + (65:\m+1)$) {\small{$Q^{-1}UQ$}};
\node[right] at ($(a) + (10:\m)$) {\small{$Q^{-1}UQ$}};
\node[left] at (205:\R) {\small{$D$}};
\node[left] at (220:\R) {\small{$Q^{-1}JQ$}};
\node[below right] at ($(b) + (15:\m+1)$) {\small{$Q^{-1}PQ$}};
\node[right,yshift=6pt] at ($(b) + (70:\m)$) {\small{$Q^{-1}PQ$}};
\node[ left] at ($(b) + (285:\m+1)$) {\small{$Q^{-1}MQ$}};
\node[ left] at ($(b) + (220:\m)$) {\small{$Q^{-1}MQ$}};
\end{tikzpicture}
}\\
\subfigure[]{
\begin{tikzpicture}[scale=0.65]
\def\R{9}
\def\m{1.5}
\coordinate (begin) at (180:\R);
\coordinate (a) at (200:\R);
\coordinate (b) at (240:\R);
\coordinate (end) at (255:\R);
\coordinate (zo) at (-6.47862,-6.2472); 
\coordinate (top_left) at (-8.85102, -1.63079);
\coordinate (top_right) at (-7.82853, -4.44006);
\coordinate (left) at (-5.73202, -6.93858);
\coordinate (right) at (-3.14298, -8.43337);
\coordinate (topcircle_L) at ($(a)+(155:\m)$);
\coordinate (topcircle_U) at ($(a)+(65:\m)$);
\coordinate (bottomcircle_M) at ($(b)+(285:\m)$);
\coordinate (bottomcircle_P) at ($(b)+(15:\m)$);
\node[right] at (a) {\small{$\alpha^{-1}$}};
\node[right] at (b) {\small{$\beta^{-1}$}};
\tkzCircumCenter(a,zo,b)\tkzGetPoint{O};
\begin{scope}[decoration={markings,
mark=at position 0.6 with {\arrow[line width =1.8pt]{>}}
}
]
\tkzDrawArc[color=cyan, line width =1.8,postaction=decorate](O,begin)(top_left);
\tkzDrawArc[color=cyan, line width =1.8,postaction=decorate](O,top_right)(left);
\end{scope}

\begin{scope}[decoration={markings,
mark=at position 0.5 with {\arrow[line width =1.8pt]{>}}
}
]
\tkzDrawArc[color=cyan, line width =1.8, postaction=decorate](a,top_left)(topcircle_L);
\tkzDrawArc[color=cyan, line width =1.8, postaction=decorate](a,topcircle_L)(top_right);
\tkzDrawArc[color=cyan, line width =1.8, postaction=decorate](a,top_right)(topcircle_U);
\tkzDrawArc[color=black, line width =1.8, dashed](a,topcircle_U)(top_left);

\tkzDrawArc[color=cyan, line width =1.8, postaction=decorate](b,bottomcircle_P)(left);
\tkzDrawArc[color=cyan, line width =1.8, postaction=decorate](b,left)(bottomcircle_M);
\tkzDrawArc[color=black, line width =1.8, dashed](b,bottomcircle_M)(bottomcircle_P);

\end{scope}

\foreach \Point in{(a), (b)}{
\node at \Point {\color{purple}\textbullet};
}
\begin{scope}[decoration={markings,
mark=at position 0.2 with {\arrow[line width =1.8pt]{>}}
}
]
\draw [cyan, line width =1.8,postaction=decorate] ($(a) + (155:\m+2)$)--($(a) + (155:\m)$);
\draw [cyan, line width =1.8,postaction=decorate] ($(a) + (65:\m+2)$)--($(a) + (65:\m)$);
\end{scope}
\begin{scope}[decoration={markings,
mark=at position 0.8 with {\arrow[line width =1.8pt]{>}}
}
]
\draw [cyan, line width =1.8,postaction=decorate] ($(b) + (15:\m)$)--($(b) + (15:\m+2)$);
\draw [cyan, line width =1.8,postaction=decorate] ($(b) + (285:\m)$)--($(b) + (285:\m+2)$);
\end{scope}

\node[left] at (185:\R) {\small{$D$}};
\node[left,yshift=-2pt,xshift=-6pt] at ($(a) + (155:\m+1)$) {\small{$Q^{-1}LQ$}};
\node[left] at ($(a) + (210:\m)$) {\small{$Q^{-1}LD Q$}};
\node[right] at ($(a) + (65:\m+1)$) {\small{$Q^{-1}UQ$}};
\node[right] at ($(a) + (10:\m)$) {\small{$Q^{-1}U^{-1}Q$}};
\node[above,yshift=4pt] at ($(a) +(140:\m)$) {\small{$D$}};
\node[left] at (220:\R) {\small{$Q^{-1}JQ$}};
\node[below right] at ($(b) + (15:\m+1)$) {\small{$Q^{-1}PQ$}};
\node[right,yshift=6pt] at ($(b) + (70:\m)$) {\small{$Q^{-1}P^{-1}Q$}};
\node[ left] at ($(b) + (285:\m+1)$) {\small{$Q^{-1}MQ$}};
\node[ left, xshift=-4pt] at ($(b) + (220:\m)$) {\small{$Q^{-1}MQ$}};
\end{tikzpicture}}
\caption{(a) Definition of $m_{2,\text{cs}}(z)$ near $\alpha^{-1}$, $\beta^{-1}$, (b) The jump contours and matrices for $m_{2,\text{cs}}(z)$ near $\alpha^{-1}$, $\beta^{-1}$.}
\label{F:cs_2}
\end{figure}
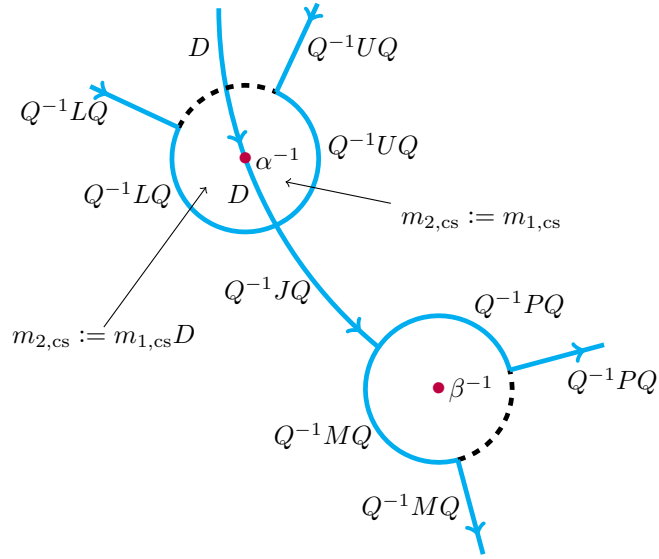
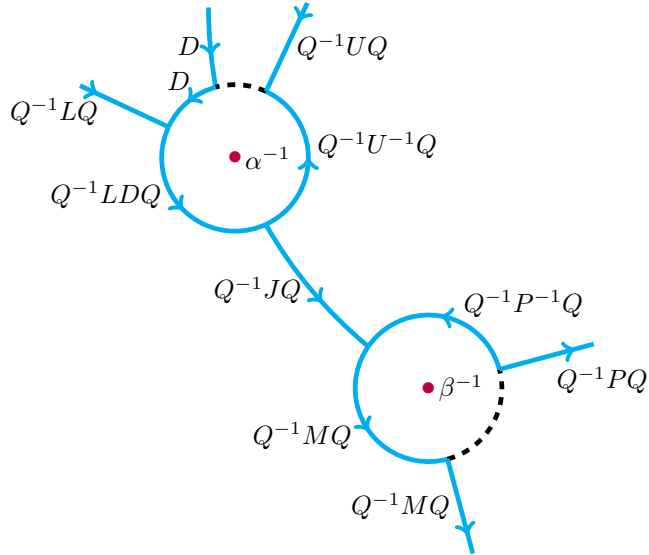
This deformation turned the jumps inside the circles surrounding $\alpha$ and $\alpha^{-1}$ to the identity jump, \emph{i.e.} no jump. {We remove} $\theta(z;n,t)$ from the problem using $\phi(z)$, as discussed in the beginning of this section, by defining
$m_{3,\text{cs}}(z)$ to be
\begin{equation*}
m_{3,\text{cs}}(z) =
\begin{cases}
m_{2,\text{cs}}(z),\quad&\text{inside the circles centered at $\alpha^{\pm 1}$ and $\beta^{\pm 1}$}, \\
m_{2,\text{cs}}(z)\phi(z),\quad&\text{outside the circles centered at $\alpha^{\pm 1}$ and $\beta^{\pm 1}$}.
\end{cases}
\end{equation*}
The jump condition satisfied by $m_{3,\text{cs}}(z)$ near the points $\alpha^{-1}$ and $\beta^{-1}$ is shown in Figure~\ref{F:cs_3}.
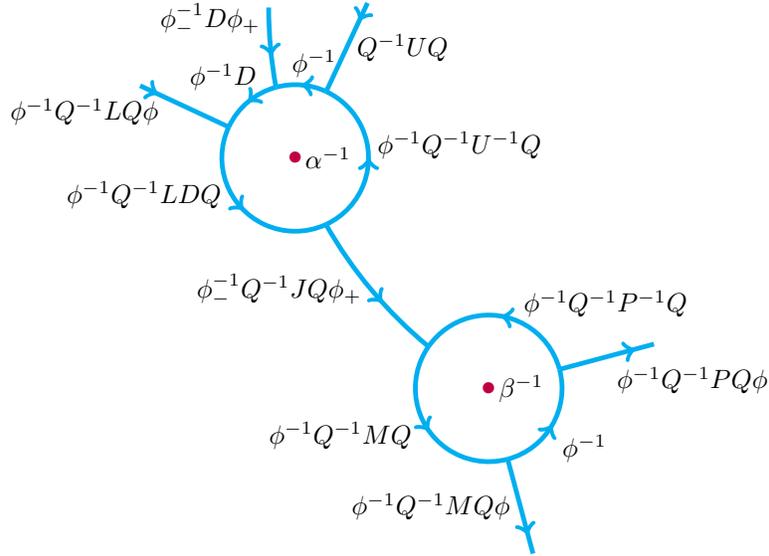
\begin{figure}[htp]
\centering
\begin{tikzpicture}[scale=0.65]
\def\R{9}
\def\m{1.5}
\coordinate (begin) at (180:\R);
\coordinate (a) at (200:\R);
\coordinate (b) at (240:\R);
\coordinate (end) at (255:\R);
\coordinate (zo) at (-6.47862,-6.2472); 
\coordinate (top_left) at (-8.85102, -1.63079);
\coordinate (top_right) at (-7.82853, -4.44006);
\coordinate (left) at (-5.73202, -6.93858);
\coordinate (right) at (-3.14298, -8.43337);
\coordinate (topcircle_L) at ($(a)+(155:\m)$);
\coordinate (topcircle_U) at ($(a)+(65:\m)$);
\coordinate (bottomcircle_M) at ($(b)+(285:\m)$);
\coordinate (bottomcircle_P) at ($(b)+(15:\m)$);
\node[right] at (a) {\small{$\alpha^{-1}$}};
\node[right] at (b) {\small{$\beta^{-1}$}};
\tkzCircumCenter(a,zo,b)\tkzGetPoint{O};
\begin{scope}[decoration={markings,
mark=at position 0.6 with {\arrow[line width =1.8pt]{>}}
}
]
\tkzDrawArc[color=cyan, line width =1.8,postaction=decorate](O,begin)(top_left);
\tkzDrawArc[color=cyan, line width =1.8,postaction=decorate](O,top_right)(left);
\end{scope}

\begin{scope}[decoration={markings,
mark=at position 0.5 with {\arrow[line width =1.8pt]{>}}
}
]
\tkzDrawArc[color=cyan, line width =1.8, postaction=decorate](a,top_left)(topcircle_L);
\tkzDrawArc[color=cyan, line width =1.8, postaction=decorate](a,topcircle_L)(top_right);
\tkzDrawArc[color=cyan, line width =1.8, postaction=decorate](a,top_right)(topcircle_U);
\tkzDrawArc[color=cyan, line width =1.8, postaction=decorate](a,topcircle_U)(top_left);

\tkzDrawArc[color=cyan, line width =1.8, postaction=decorate](b,bottomcircle_P)(left);
\tkzDrawArc[color=cyan, line width =1.8, postaction=decorate](b,left)(bottomcircle_M);
\tkzDrawArc[color=cyan, line width =1.8, postaction=decorate](b,bottomcircle_M)(bottomcircle_P);
\end{scope}

\foreach \Point in{(a), (b)}{
\node at \Point {\color{purple}\textbullet};
}
\begin{scope}[decoration={markings,
mark=at position 0.2 with {\arrow[line width =1.8pt]{>}}
}
]
\draw [cyan, line width =1.8,postaction=decorate] ($(a) + (155:\m+2)$)--($(a) + (155:\m)$);
\draw [cyan, line width =1.8,postaction=decorate] ($(a) + (65:\m+2)$)--($(a) + (65:\m)$);
\end{scope}
\begin{scope}[decoration={markings,
mark=at position 0.8 with {\arrow[line width =1.8pt]{>}}
}
]
\draw [cyan, line width =1.8,postaction=decorate] ($(b) + (15:\m)$)--($(b) + (15:\m+2)$);
\draw [cyan, line width =1.8,postaction=decorate] ($(b) + (285:\m)$)--($(b) + (285:\m+2)$);
\end{scope}

\node[above left] at (185:\R) {\small{$\phi_{-}^{-1}D\phi_+$}};
\node[left,yshift=-2pt,xshift=-6pt] at ($(a) + (155:\m+1)$) {\small{$\phi^{-1}Q^{-1}LQ\phi$}};
\node[left] at ($(a) + (210:\m)$) {\small{$\phi^{-1}Q^{-1}LD Q$}};
\node[right] at ($(a) + (65:\m+1)$) {\small{$Q^{-1}UQ$}};
\node[right] at ($(a) + (10:\m)$) {\small{$\phi^{-1}Q^{-1}U^{-1}Q$}};
\node[above,yshift=4pt,xshift=-6pt] at ($(a) +(140:\m)$) {\small{$\phi^{-1} D$}};
\node[above] at ($(a) +(75:\m)$) {\small{$\phi^{-1}$}};
\node[below right] at ($(b) +(330:\m)$) {\small{$\phi^{-1}$}};
\node[left] at (220:\R) {\small{$\phi_{-}^{-1}Q^{-1}JQ\phi_{+}$}};
\node[below right] at ($(b) + (15:\m+1)$) {\small{$\phi^{-1}Q^{-1}PQ\phi$}};
\node[right,yshift=6pt] at ($(b) + (70:\m)$) {\small{$\phi^{-1}Q^{-1}P^{-1}Q$}};
\node[ left] at ($(b) + (285:\m+1)$) {\small{$\phi^{-1}Q^{-1}MQ\phi$}};
\node[ left, xshift=-4pt] at ($(b) + (220:\m)$) {\small{$\phi^{-1}Q^{-1}MQ$}};
\end{tikzpicture}
\caption{The jump contours and matrices for $m_{3,\text{cs}}(z)$ near $\alpha^{-1}$, $\beta^{-1}$.}
\label{F:cs_3}
\end{figure}
As in the dispersive region, the diagonal matrix $D(z)$ has a singularity at $z=-1\in\Sigma_c$ since $R(\pm 1)=-1$ and this singularity has to be removed.

We proceed with analyzing the jump matrix on $\Sigma_c$ in the limit $z_0\to-1$ to determine the constants $\delta_1$ and $\delta_2$ introduced in \eqref{E:g_rhp} (or in \eqref{E:cs-jump}) so that the singularity disappears. On $\Sigma_c = \oarc{ \alpha, \alpha^{-1} }$, the jump matrix is given by
\begin{equation}\label{E:cs_diag_jump}
\phi^{-1}_{-} (z) D(z) \phi_{+}(z)=
\begin{pmatrix}
\left[1 - R(z)R\left(z^{-1}\right)\right]e^{\delta_2} & 0 \\ 0 &\left( \left[1 - R(z)R\left(z^{-1}\right)\right]e^{\delta_2} \right)^{-1}
\end{pmatrix}\,.
\end{equation}
Using $R(-1) = -1$ and the analyticity of $R(z)$ in a neighborhood around $z= -1$, we observe that
\begin{equation}\label{E:ref_expand}
1 - R(z)R\left(z^{-1}\right) = \nu (z+1)^2 + \mathcal{O}\left( (z+1)^4 \right)~\text{near}~z=-1\,,
\end{equation}
for some constant $\nu\in\mathbb{C}$. So far, we have left $\alpha$ and $\beta$ mostly arbitrary. It follows that (see Appendix~\ref{A:construct_g}) the prescribed asymptotic behavior in \eqref{E:g_rhp} for $g(z)$ {as $z \to \infty$} requires $\Re \alpha + \Re \beta = 2 \lambda_ 0$, leaving us with single degree of freedom. Now consider the affine transformation, $k=K(z)$, defined by
\begin{equation}\label{E:cov}
K(z) = i \left(\frac{z-\lambda_0}{\rho_0}\right),\quad z(k) = K^{-1}(k) = \lambda_0 - i \rho_0 k\,.
\end{equation}
Note that this transformation fixes the stationary phase points: $K\left(z_0^{\pm 1}\right) = \mp 1$, and the image of the contour $\oarc{\beta, \beta^{-1}}$ under the mapping $K$ flattens as $z_0^{\pm 1} \to -1$ (see Figure~\ref{F:Kmap}).
\begin{figure}[htp]
\centering
\begin{minipage}{.40\textwidth}
\centering
\begin{tikzpicture}[
decoration={markings,
mark=at position 0.4 with {\arrow[line width=1pt]{>}}
}
]
\def\L{4}
\coordinate (zo) at (155:\L);
\coordinate (zoi) at (205:\L);
\coordinate (o) at (0,0);
\coordinate (b) at (150:\L);
\coordinate (bi) at (210:\L);
\coordinate (a) at (165:\L);
\coordinate (ai) at (195:\L);
\coordinate (xaxisend) at (-\L-2,0);
\pgfmathsetmacro\lo{cos(155)*\L};
\pgfmathsetmacro\la{cos(150)*\L};
\pgfmathsetmacro\lb{cos(165)*\L};

\draw[help lines, path fading=east] (0,0) -- (-\L-2,0) coordinate (xaxis) ++(0,1pt);

\draw [color=cyan, line width=1.8,domain=150:165] plot ({\L*cos(\x)}, {\L*sin(\x)});
\draw [color=cyan, line width=1.8,domain=165:195, postaction=decorate] plot ({\L*cos(\x)}, {\L*sin(\x)});
\draw [color= cyan, line width=1.8,domain=195:210] plot ({\L*cos(\x)}, {\L*sin(\x)});
\foreach \Point in {(zo), (zoi)}{
  \node at \Point {\textbullet};
}

\foreach \Point in {(a), (ai), (b), (bi)}{
  \node at \Point {\color{purple}\textbullet};
}

\draw [very thick] (\lo,-2pt) -- (\lo,2pt);
\draw [very thin, dashed, gray] (zo) -- (zoi);
\node at (150:\L-1.2) {\small $z$-plane};
\node[below] at (xaxisend) {\small{$\Re z$}};
\node[left ] at (zo) (z0) {\small{$z_0 $}};
\node[left ] at (zoi) {\small{$z_0^{-1}$}};
\node [below left] at (180:\L) {$-1$};
\node[ above left] at (b) {\color{purple} \small{$\beta$}};
\node[above left] at (a) {\color{purple} \small{$\alpha$}};

\node[above left] at (ai) {\color{purple} \small{$\alpha^{-1}$}};

\node[below left] at (bi) {\color{purple} \small{$\beta^{-1}$}};

\node at (-\L,0) {$\times$};
\node [above, anchor=south west] at (\lo,0) {\small{$\lambda_0=-n/t$}};

\end{tikzpicture}
\end{minipage}
$\hspace{-2em}\overset{k=K(z)}\longmapsto$
\begin{minipage}{.40\textwidth}
\centering
\begin{tikzpicture}[
decoration={markings,
mark=at position 0.45 with {\arrow[line width=1pt]{>}}
}
]

\def\L{4.6}
\pgfmathsetmacro\lo{cos(155)*\L};
\pgfmathsetmacro\betao{sin(155)*\L};
\pgfmathsetmacro\la{cos(150)*\L};
\pgfmathsetmacro\lb{cos(165)*\L};
\coordinate (one) at ([shift=(90:\L)]0,\lo);
\coordinate (ko) at ([shift=(65:\L)]0,\lo);
\coordinate (koi) at ([shift=(115:\L)]0,\lo);
\coordinate (Ai) at ([shift=(-80:\L)]0,-\lo);
\coordinate (A) at ([shift=(-100:\L)]0,-\lo);
\coordinate (Bi) at ([shift=(-50:\L)]0,-\lo);
\coordinate (B) at ([shift=(-130:\L)]0,-\lo);
\coordinate (xaxisend) at (-\L+1,0);
\draw[line width=1.8, color=cyan, postaction=decorate] ([shift=(-130:\L)]0,-\lo) arc (-130:-50:\L);

\draw[help lines,->] (-\L+1,0) -- (\L-1,0) coordinate (xaxis);
\draw[help lines,->] (0,-\betao) -- (0,\betao) coordinate (yaxis);

\foreach \Point in {(ko), (koi)}{
  \node at \Point {\textbullet};
}

\foreach \Point in {(A), (Ai), (B), (Bi)}{
  \node at \Point {\color{purple}\textbullet};
}

\node at (150:\L-1.2) {\small $k$-plane};
\node[below] at (xaxisend) {\small{$\Re k$}};
\node[below] at (Ai) {\color{purple} \small $-\overline{A}$};
\node[above] at (Bi) {\color{purple} \small $-\overline{B}$};
\node[below] at (A) {\color{purple} \small $A$};
\node[above] at (B) {\color{purple} \small $B$};
\node[below right] at (ko) {\small $1$};
\node[below left] at (koi) { \small $-1$};

\node at ([shift=(-90:\L)]0,-\lo){$\times$};
\node[above right] at ([shift=(-90:\L)]0,-\lo){\small $k^*$};
\end{tikzpicture}
\end{minipage}
\caption{Cuts mapped under the affine mapping $z \mapsto k=K(z)$.}
\label{F:Kmap}
\end{figure}
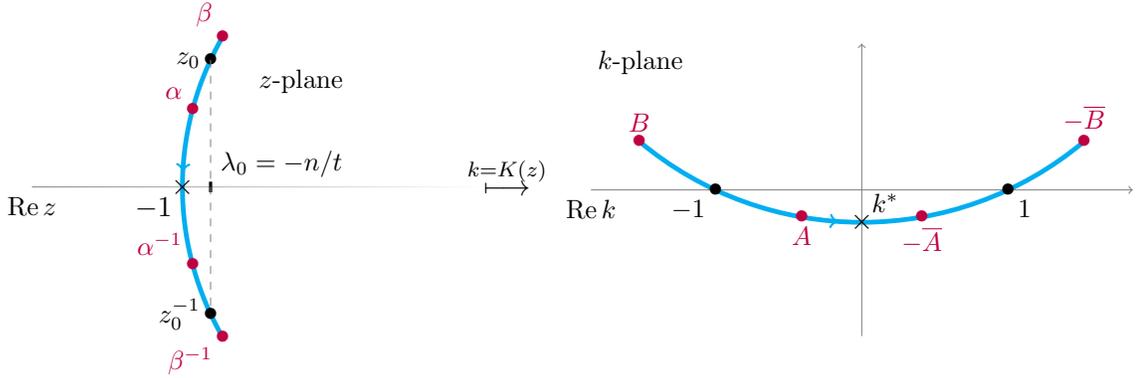

To remove the singularity of $D(z)$ at $z=-1$, we need to obtain a parametrix $\psi(z)$ by solving the following diagonal {\RHP}:
\begin{rhp}\label{rhp:psi}
\begin{align*}\begin{split}
  \psi^{+}(z) = \psi^{-}(z)\phi^{-1}_{-} (z) D(z) \phi_{+}(z),~~~z&\in \Sigma_c, ~~ \psi(\infty) = I,\\
  \psi(z) \diag\left(|z+1|^{-1},|z+1|\right) &= \mathcal O(1), ~~~ z \to -1, ~~~ |z| < 1,\\
  \psi(z) \diag\left(|z+1|,|z+1|^{-1}\right) &= \mathcal O(1), ~~~ z \to -1, ~~~ |z| > 1,
\end{split}\end{align*}
    {such that $\psi(z)$ is bounded for $z$ in a neighborhood of the endpoints $\alpha$, $\alpha^{-1}$  and the boundary values $\psi^\pm(z)$ are not continuous at $\alpha,\alpha^{-1}$ and $-1$}.
    \end{rhp}
In new variables \eqref{E:cov}, {the jump condition in \rhref{rhp:psi}} reads
\begin{equation}
H^{+}(k) = H^{-}(k) \widetilde{D}(k),\,~k\in K\left(\Sigma_c\right),
\end{equation}
where $K^{-1}(k) = z(k)$, $H(k) = \psi\left(K^{-1}(k)\right)$, and $\widetilde{D}(k) = \phi^{-1}_{-} (K^{-1}\left(k)\right) D(K^{-1}\left(k)\right) \phi_{+}\left(K^{-1}(k)\right)$.
Let $k^{*} = K(-1)$ so that $z(k)+1 = -i\rho_0 (k-k^{*})$. We choose $\delta_2$ to enforce $\rho_0^2 e^{\delta_2} = 1$ so that the $(1,1)$-entry of the diagonal jump matrix $\widetilde{D}(k)$ satisfies
\begin{equation}
\left[1 - R\left(z(k)\right)R\left(z(k)^{-1}\right)\right]e^{\delta_2} = \nu (k - k^{*})^2 + \mathcal{O} \left( \rho_0^2 (k-k^{*})^4 \right),~\text{ near }~z=-1,~k=k^{*},
\end{equation}
hence removing, up to second order, the dependence on $\rho_0$. This indicates that $\widetilde{D}(k)$ remains bounded as $\rho_0 \to 0$ away from $k^*$ which ensures the boundedness of $H$, and hence of $\psi$ as $t\to\infty$. In Appendix~\ref{A:construct_g} it is shown that $\alpha$ and $\beta$ can (and should) be chosen so that
\begin{align}\label{E:cs-ab}
\frac{-\log \rho_0^2}{t} = \int_{\beta}^{\alpha}\frac{1}{p^2}\sqrt{\left(p-\alpha\right)\left(p-\alpha^{-1}\right)\left(p-\beta\right)\left(p-\beta^{-1}\right)}^{+}\, dp.
\end{align}
See Appendix~\ref{A:construct_g} for the definition of the square root in \eqref{E:cs-ab}. Loosely speaking, the collisionless shock region is defined to be the region in the $(n,t)$-plane where \eqref{E:cs-ab} is solvable for $\alpha$ and $\beta$ and $\alpha+1$ is not too small. This reasoning gives the asymptotic condition $n = t - c_2 t^{1/3} (\log t)^{2/3}$.  See Appendix~\ref{A:scaling} for more detail.

Once $\psi(z)$ is obtained (see Appendix~\ref{A:g}), we conjugate the problem by $\psi(z)$ as was done with $\Delta(z)$ in Section~\ref{S:dispersive}. Define $m_{\sharp,\text{cs}}(z)$ by
\begin{equation*}
m_{\sharp,\text{cs}}(z) =
\begin{cases}
m_{3,\text{cs}}(z)\psi^{-1}(z),\quad &\text{outside the circles centered at $\alpha$ and $\alpha^{-1}$},\\
m_{3,\text{cs}}(z),\quad &\text{inside the circles centered at $\alpha$ and $\alpha^{-1}$}.
\end{cases}
\end{equation*}
The final deformation for this region and the {\RHP} satisfied by $m_{\sharp,\text{cs}}(z)$ is shown in Figure~\ref{F:cs_4}. Similar to the case addressed in Remark~\ref{r:sing} the jump contours and matrices of the vector problem (\rhref{rhp:disp}) can be deformed to those of $m_{\sharp,\text{cs}}(z)$ leaving no singularity at $z=-1$, {despite the fact that $D(z)$ and $\psi(z)$ are singular at $z = -1$.}
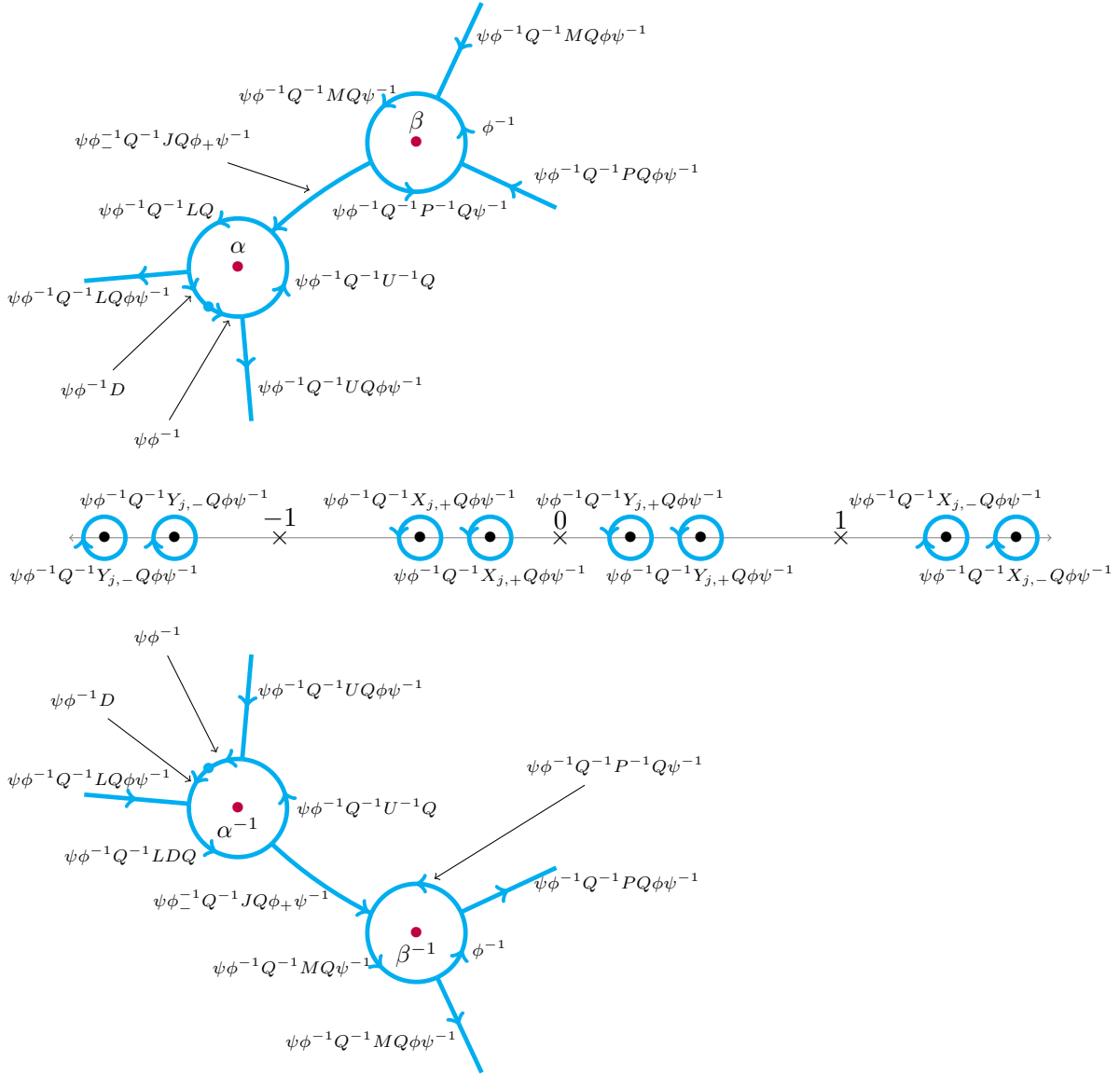
\begin{figure}[htp]
\centering
\begin{tikzpicture}
\def\R{4}
\def\L{\R+2}
\def\s{0.3}
\def\m{0.7}

\coordinate (a) at (140:\L);
\coordinate (ai) at (-140:\L);
\coordinate (b) at (110:\L);
\coordinate (bi) at (-110:\L);
\coordinate (zo) at (120:\L);
\coordinate (zoi) at (240:\L);
\coordinate (o) at (0,0);
\coordinate (to_b) at (-1.38149, 5.83879);
\coordinate (from_b) at (-2.69482, 5.36078);
\coordinate (to_a) at (-4.1158, 4.3658);
\coordinate (from_a) at (-5.01417, 3.29516);
\coordinate (from_bi) at (-1.38149, -5.83879);
\coordinate (to_bi) at (-2.69482, -5.36078);
\coordinate (from_ai) at (-4.1158, -4.3658);
\coordinate (to_ai) at (-5.01417, -3.29516);
\coordinate (top_M) at ($(b) + (65:\m)$);
\coordinate (top_P) at ($(b) + (-25:\m)$);
\coordinate (top_L) at ($(a) + (185:\m)$);
\coordinate (top_U) at ($(a) + (275:\m)$);
\coordinate (bottom_M) at ($(bi) + (295:\m)$);
\coordinate (bottom_P) at ($(bi) + (25:\m)$);
\coordinate (bottom_L) at ($(ai) + (175:\m)$);
\coordinate (bottom_U) at ($(ai) + (85:\m)$);
\draw[help lines,<->] (-\R-3,0) -- (\R+3,0) coordinate (xaxis);

\coordinate (pt1) at (-2.00239, 3.46272);
\coordinate (pt2) at (-2.00239, -3.46272);
\coordinate (pt3) at (-0.69187, 3.93971);
\coordinate (pt4) at (-0.69187, -3.93971);
\coordinate (end) at (180:\R);
\coordinate (begin) at (0:\R);
\begin{scope}[decoration={markings, mark=at position 0.5 with {\arrow[line width =1.8pt]{>}}}]
\tkzDrawArc[color=cyan,line width =1.8,postaction=decorate](b,top_M)(from_b);
\tkzDrawArc[color=cyan,line width =1.8,postaction=decorate](b,from_b)(top_P);
\tkzDrawArc[color=cyan,line width =1.8,postaction=decorate](b,top_P)(top_M);

\tkzDrawArc[color=cyan,line width =1.8,postaction=decorate](a,to_a)(top_L);
\tkzDrawArc[color=cyan,line width =1.8,postaction=decorate](a,top_L)(from_a)
\tkzDrawArc[color=cyan,line width =1.8,postaction=decorate](a,from_a)(top_U);
\tkzDrawArc[color=cyan,line width =1.8,postaction=decorate](a,top_U)(to_a);

\tkzDrawArc[color=cyan,line width =1.8,postaction=decorate](bi,to_bi)(bottom_M);
\tkzDrawArc[color=cyan,line width =1.8,postaction=decorate](bi,bottom_M)(bottom_P);
\tkzDrawArc[color=cyan,line width =1.8,postaction=decorate](bi,bottom_P)(to_bi);

\tkzDrawArc[color=cyan,line width =1.8,postaction=decorate](ai,to_ai)(bottom_L);
\tkzDrawArc[color=cyan,line width =1.8,postaction=decorate](ai,bottom_L)(from_ai)
\tkzDrawArc[color=cyan,line width =1.8,postaction=decorate](ai,from_ai)(bottom_U);
\tkzDrawArc[color=cyan,line width =1.8,postaction=decorate](ai,bottom_U)(to_ai);
\end{scope}

\begin{scope}[decoration={markings, mark=at position 0.98 with {\arrow[line width =1.8pt]{>}}}]
\tkzDrawArc[color=cyan,line width =1.8,postaction=decorate](o,from_b)(to_a);
\tkzDrawArc[color=cyan,line width =1.8,postaction=decorate](o,from_ai)(to_bi);
\end{scope}


\foreach \Point in {(o), (180:\R), (0:\R) }{
\node at \Point{$\times$};
}
\foreach \Point in{(a), (ai),(b),(bi)}{
\node at \Point {\color{purple}\textbullet};
}
\coordinate (zeta1) at (180:\R-2);
\coordinate (zeta2) at (180:\R-3);
\coordinate (zeta3) at (0:\R-2);
\coordinate (zeta4) at (0:\R-3);
\coordinate (zeta1i) at (180:\R+1.5);
\coordinate (zeta2i) at (180:\R+2.5);
\coordinate (zeta3i) at (0:\R+1.5);
\coordinate (zeta4i) at (0:\R+2.5);

\foreach \Point in {(zeta1), (zeta2), (zeta3), (zeta4), (zeta1i), (zeta2i), (zeta3i), (zeta4i)}{
\node at \Point{\textbullet};
}
\begin{scope}[decoration={markings,
mark=at position 0.5 with {\arrow[line width =1.8pt]{>}}
}
]
\draw[draw,cyan,line width =1.8, postaction=decorate]
(zeta1) + (0:\s) arc(0:360:\s);
\draw[draw,cyan,line width =1.8, postaction=decorate]
(zeta2) + (0:\s) arc(0:360:\s);
\draw[draw,cyan,line width =1.8, postaction=decorate]
(zeta3) + (0:\s) arc(0:360:\s);
\draw[draw,cyan,line width =1.8, postaction=decorate]
(zeta4) + (0:\s) arc(0:360:\s);
\draw[draw,cyan,line width =1.8, postaction=decorate]
(zeta1i) + (0:\s) arc(0:-360:\s);
\draw[draw,cyan,line width =1.8, postaction=decorate]
(zeta2i) + (0:\s) arc(0:-360:\s);
\draw[draw,cyan,line width =1.8, postaction=decorate]
(zeta3i) + (0:\s) arc(0:-360:\s);
\draw[draw,cyan,line width =1.8, postaction=decorate]
(zeta4i) + (0:\s) arc(0:-360:\s);

\draw [cyan, line width =1.8,postaction=decorate] ($(b) + (65:\m+1.5)$)--($(b) + (65:\m)$);
\draw [cyan, line width =1.8,postaction=decorate] ($(b) + (-25:\m+1.5)$)--($(b) + (-25:\m)$);

\draw [cyan, line width =1.8,postaction=decorate] ($(a) + (275:\m)$)--($(a) + (275:\m+1.5)$);
\draw [cyan, line width =1.8,postaction=decorate] ($(a) + (185:\m)$)--($(a) + (185:\m+1.5)$);

\draw [cyan, line width =1.8,postaction=decorate] ($(ai) + (175:\m+1.5)$)--($(ai) + (175:\m)$);
\draw [cyan, line width =1.8,postaction=decorate] ($(ai) + (85:\m+1.5)$)--($(ai) + (85:\m)$);

\draw [cyan, line width =1.8,postaction=decorate] ($(bi) + (25:\m)$)--($(bi) + (25:\m+1.5)$);
\draw [cyan, line width =1.8,postaction=decorate] ($(bi) + (295:\m)$)--($(bi) + (295:\m+1.5)$);
\end{scope}

\node[above,yshift=\s+7] at (zeta1) {\tiny{$\psi\phi^{-1}Q^{-1}X_{j,+}Q\phi\psi^{-1}$}};
\node[below,yshift=-\s-7] at (zeta2) {\tiny{$\psi\phi^{-1}Q^{-1}X_{j,+}Q\phi\psi^{-1}$}};
\node[above,yshift=\s+7] at (zeta1i) {\tiny{$\psi\phi^{-1}Q^{-1}Y_{j,-}Q\phi\psi^{-1}$}};
\node[below,yshift=-\s-7] at (zeta2i) {\tiny{$\psi\phi^{-1}Q^{-1}Y_{j,-}Q\phi\psi^{-1}$}};
\node[below,yshift=-\s-7] at (zeta3) {\tiny{$\psi\phi^{-1}Q^{-1}Y_{j,+}Q\phi\psi^{-1}$}};
\node[above,yshift=\s+7] at (zeta4) {\tiny{$\psi\phi^{-1}Q^{-1}Y_{j,+}Q\phi\psi^{-1}$}};
\node[above,yshift=\s+7] at (zeta3i) {\tiny{$\psi\phi^{-1}Q^{-1}X_{j,-}Q\phi\psi^{-1}$}};
\node[below,yshift=-\s-7] at (zeta4i) {\tiny{$\psi\phi^{-1}Q^{-1}X_{j,-}Q\phi\psi^{-1}$}};

\node at (from_a) {\color{cyan}\textbullet};
\node at (to_ai) {\color{cyan}\textbullet};
\node[right,xshift=4pt] at ($(b) + (20:\m)$) {\tiny{$\phi^{-1}$}};
\node[right] at ($(b) + (65:\m+1)$) {\tiny{$\psi\phi^{-1} Q^{-1}M Q \phi \psi^{-1}$}};
\node[above right] at ($(b) + (-25:\m+1)$) {\tiny{$\psi\phi^{-1} Q^{-1}P Q \phi \psi^{-1}$}};
\node[left] at ($(b) + (100:\m)$) {\tiny{$\psi\phi^{-1} Q^{-1}M Q \psi^{-1}$}};
\node[below] at ($(b) + (-85:\m)$) {\tiny{$\psi\phi^{-1} Q^{-1}P^{-1} Q \psi^{-1}$}};
\node[right] at ($(a) + (-15:\m)$) {\tiny{$\psi\phi^{-1} Q^{-1}U^{-1} Q $}};
\node[left,yshift=4pt,xshift=-4pt] at ($(a) + (95:\m)$) {\tiny{$\psi\phi^{-1} Q^{-1}L Q $}};
\node[right] at ($(a) + (275:\m+1)$) {\tiny{$\psi\phi^{-1} Q^{-1}U Q \phi \psi^{-1}$}};
\node[yshift=2pt,xshift=2pt](phiD_up) at ($(a) + (210:\m)$) {};
\node(phiD_up_text) at ($(a) + (220:\m+2)$) {\tiny{$\psi\phi^{-1}D$}};
\draw [->,thin] (phiD_up_text) -- (phiD_up);
\node[xshift=2pt, yshift=2pt](phi_up) at ($(a) + (260:\m)$) {};
\node(phi_up_text) at ($(a) + (245:\m+2)$) {\tiny{$\psi\phi^{-1}$}};
\draw [->,thin] (phi_up_text) -- (phi_up);
\node[below,xshift=-12pt ] at ($(a) + (185:\m+1)$) {\tiny{$\psi\phi^{-1} Q^{-1}L Q \phi \psi^{-1}$}};
\node(J_up) at (125:\L) {};
\node(J_up_text) at (135:\L+2) {\tiny{$\psi \phi_{-}^{-1}Q^{-1}JQ\phi_{+}\psi^{-1}$}};
\draw [->,thin] (J_up_text) -- (J_up);

\node[xshift=2pt,yshift=-2pt](phiD_down) at ($(ai) + (150:\m)$) {};
\node(phiD_down_text) at ($(ai) + (145:\m+2)$) {\tiny{$\psi\phi^{-1}D$}};
\draw [->,thin] (phiD_down_text) -- (phiD_down);
\node(phi_down) at ($(ai) + (110:\m)$) {};
\node(phi_down_text) at ($(ai) + (115:\m+2)$) {\tiny{$\psi\phi^{-1}$}};
\draw [->,thin] (phi_down_text) -- (phi_down);

\node[right] at ($(ai) + (85:\m+1)$) {\tiny{$\psi\phi^{-1} Q^{-1}U Q \phi \psi^{-1}$}};
\node[above, xshift=-12pt] at ($(ai) + (175:\m+1)$) {\tiny{$\psi\phi^{-1} Q^{-1}L Q \phi \psi^{-1}$}};
\node[right] at ($(ai) + (0:\m)$) {\tiny{$\psi\phi^{-1} Q^{-1}U^{-1} Q $}};
\node[left,yshift=-4pt] at ($(ai) + (230:\m)$) {\tiny{$\psi\phi^{-1} Q^{-1}L DQ$}};
\node[left,xshift=-4pt] at (-120:\L) {\tiny{$\psi \phi_{-}^{-1}Q^{-1}JQ\phi_{+}\psi^{-1}$}};
\node[right] at ($(bi) + (340:\m)$) {\tiny{$\phi^{-1}$}};
\node[left] at ($(bi) + (225:\m)$) {\tiny{$\psi\phi^{-1} Q^{-1}M Q \psi^{-1}$}};
\node[left] at ($(bi) + (295:\m+1)$) {\tiny{$\psi\phi^{-1} Q^{-1}M Q \phi \psi^{-1}$}};
\node[right] at ($(bi) + (25:\m+1)$) {\tiny{$\psi\phi^{-1} Q^{-1}P Q \phi \psi^{-1}$}};
\node (P_down_text) at ($(bi)+(40:\m+3)$) {\tiny{$\psi\phi^{-1} Q^{-1}P^{-1} Q \psi^{-1}$}};
\node(P_down) at ($(bi)+(80:\m)$) {};
\draw [->,thin] (P_down_text) -- (P_down);

\node[above] {$0$};
\node[above,yshift=2pt] at (a) {\small$\alpha$};
\node[below] at (ai) {\small$\alpha^{-1}$};
\node[above,] at (b) {\small$\beta$};
\node[below] at (bi) {\small$\beta^{-1}$};
\node [above ] at (180:\R) {$-1$};
\node [above ] at (0:\R) {$1$};
\end{tikzpicture}
\caption{A zoomed view of the jump contours and matrices of the final deformation of the {\RHP} in the collisionless shock region. Note that $\phi(z)$ and $\psi(z)$ commute.}
\label{F:cs_4}
\end{figure}

Finally, the choice of the radii of the circles round $\alpha$, $\beta$, $\alpha^{-1}$ and $\beta^{-1}$ must be specified.  It is easily seen from \eqref{E:def_g} that $g'(z)$ vanishes as a square root at each of these points and $g(z) = a + b(z-c)^{3/2}$ for $c = \alpha$, $\beta$, $\alpha^{-1}$ or $\beta^{-1}$ and $a,b$ depend on the choice of $c$. Following the arguments in \eqref{E:scale} we choose the radius of these circles to be proportional to $t^{-2/3}$, of course, under the constraint that the circles should not intersect one another.

The function $m_{\sharp,\text{cs}}(z)$ satisfies a sectionally analytic {\RHP} with
\begin{itemize}
\item the jump conditions described in Figure~\ref{F:cs_4},
\item the asymptotic symmetry condition
  \begin{align}\label{E:cs-sym}
    m_{\sharp,\text{cs}}(0) = m_{\sharp,\text{cs}}(\infty) \begin{pmatrix} 0 & 1 \\ 1 & 0 \end{pmatrix} Q(0) \phi(0) \psi^{-1}(0),
  \end{align}
  and
\item the quadratic normalization condition present in \rhref{rhp:disp}.
\end{itemize}

\subsection{Transition Region}\label{S:transition}
Similar to the case for the KdV equation (see \cite{TOD_KdV}), the deformations in the collisionless shock region extends the values of $(n,t)$ for which there exists a well-behaved {\RHP} {beyond} the dispersive region. However, this is not asymptotically reliable as we approach the Painlev\'{e} region: as $|n| - t$ decreases, $\alpha$ and $\alpha^{-1}$ approach the singularity of the parametrix $\psi(z)$ (see RH Problem~\ref{rhp:psi}) at $z=-1$. To avoid this issue, we collapse the lensing on $\Sigma_c=\oarc{\alpha, \alpha^{-1}}$ that was introduced in the collisionless shock region (see Figure~\ref{F:cs_1}) . Thus the $LDU$-factorization of the jump matrix $J(z;n,t)$ is not used in this region. In order to maintain numerical accuracy, we choose $\alpha$ to ensure that the oscillations are controlled on $\carc{\beta,\beta^{-1}}$. The first deformation $\widetilde{m}(z) \longmapsto m_{1,\text{t}}(z)$ we perform in this region is similar to the first deformation $\widetilde{m}(z) \longmapsto m_{1,\text{cs}}(z)$ in the collisionless shock region, but without the lensing on $\oarc{\alpha, \alpha^{-1}}$. Definition of $m_{1,\text{t}}(z)$ and the jump conditions it satisfies are given in Figure~\ref{F:trans_1}.
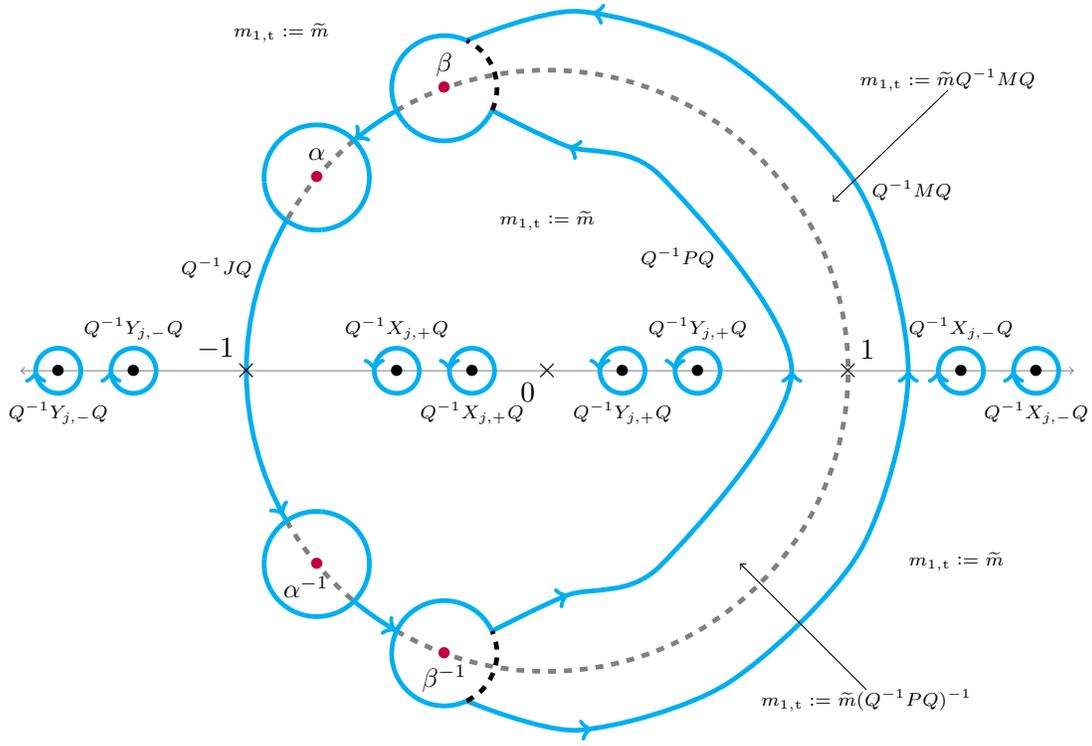
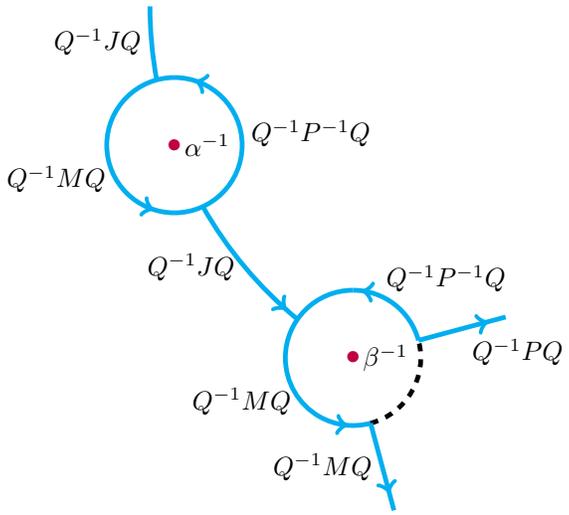
\begin{figure}[htp]
\centering
\subfigure[]{
\begin{tikzpicture}
\def\R{4}
\def\s{0.3}
\def\m{0.7}
\def\yL{\R}
\def\xL{\R}
\def\olens{\R+0.8}
\def\ilens{\R-1}

\coordinate (a) at (140:\R);
\coordinate (ai) at (-140:\R);
\coordinate (b) at (110:\R);
\coordinate (bi) at (-110:\R);
\coordinate (zo) at (120:\R);
\coordinate (zoi) at (240:\R);
\coordinate (o) at (0,0);

\draw[help lines,<->] (-\R-3,0) -- (\R+3,0) coordinate (xaxis);

\coordinate (pt1) at (-2.00239, 3.46272);
\coordinate (pt1i) at (-2.00239, -3.46272);
\coordinate (pt3) at (-0.69187, 3.93971);
\coordinate (pt3i) at (-0.69187, -3.93971);
\coordinate (end) at (180:\R);
\coordinate (begin) at (0:\R);
\coordinate (pt2) at (-2.56903,3.06595);
\coordinate (pt2i) at (-2.56903,-3.06595);
\coordinate (pt4) at (-3.46548,1.99761);
\coordinate (pt4i) at (-3.46548,-1.99761);

\begin{scope}[decoration={
  markings,
  mark=at position 0.98 with {\arrow[line width =1.8pt]{>}}}
  ]
  \tkzCircumCenter(pt1,end,pt2)\tkzGetPoint{O};
  \tkzDrawArc[color=cyan,line width =1.8,postaction=decorate](O,pt1)(pt2);
  \tkzDrawArc[color=cyan,line width =1.8,postaction=decorate](O,pt4)(pt4i);
  \tkzDrawArc[color=cyan,line width =1.8,postaction=decorate](O,pt2i)(pt1i);
\end{scope}
\tkzCircumCenter(pt3,begin,pt3i)\tkzGetPoint{O};
\tkzDrawArc[color=gray,dashed, line width =1.8](O,pt1i)(pt1);
\tkzDrawArc[color=gray,dashed, line width =1.8](O,pt2)(pt4);
\tkzDrawArc[color=gray,dashed, line width =1.8](O,pt4i)(pt2i);

\begin{scope}[decoration={markings,
mark=at position 0.1 with {\arrow[line width =1.8pt]{>}},
mark=at position 0.5 with {\arrow[line width =1.8pt]{>}},
mark=at position 0.9 with {\arrow[line width =1.8pt]{>}}
}
]

\draw[cyan, line width =1.8, postaction=decorate] plot[smooth] coordinates {
   ($(bi) + (-65:\m)$)
   (-85:\olens)
   (-60:\olens)
   (-30:\olens)
   (0:\olens)
   (30:\olens)
   (60:\olens)
   (85:\olens)
   ($(b) +(65:\m)$)
  };

\draw[cyan, line width =1.8, postaction=decorate] plot[smooth] coordinates {
   ($(bi) + (25:\m)$)
   (-85:\ilens)
   (-60:\ilens)
   (0:\ilens+0.25)
   (60:\ilens)
   (85:\ilens)
   ($(b) +(-25:\m)$)
  };
 \end{scope}
\draw[draw,cyan,line width =1.8]
(a) + (275:\m) arc(-85:185:\m);
\draw[draw,cyan,line width =1.8]
(a) + (185:\m) arc(185:275:\m);

\draw[draw,cyan,line width =1.8]
(ai) + (175:\m) arc(-185:85:\m);
\draw[draw,cyan,line width =1.8]
(ai) + (85:\m) arc(85:175:\m);

\draw[draw,dashed,line width =1.8]
(b) + (-25:\m) arc(-25:65:\m);
\draw[draw,cyan,line width =1.8]
(b) + (65:\m) arc(65:335:\m);

\draw[draw,cyan,line width =1.8]
(bi) + (25:\m) arc(25:295:\m);
\draw[draw,dashed,line width =1.8]
(bi) + (-65:\m) arc(-65:25:\m);

\foreach \Point in {(o), (180:\R), (0:\R) }{
\node at \Point{$\times$};
}
\foreach \Point in{(a), (ai),(b),(bi)}{
\node at \Point {\color{purple}\textbullet};
}
\coordinate (zeta1) at (180:\R-2);
\coordinate (zeta2) at (180:\R-3);
\coordinate (zeta3) at (0:\R-2);
\coordinate (zeta4) at (0:\R-3);
\coordinate (zeta1i) at (180:\R+1.5);
\coordinate (zeta2i) at (180:\R+2.5);
\coordinate (zeta3i) at (0:\R+1.5);
\coordinate (zeta4i) at (0:\R+2.5);

\foreach \Point in {(zeta1), (zeta2), (zeta3), (zeta4), (zeta1i), (zeta2i), (zeta3i), (zeta4i)}{
\node at \Point{\textbullet};
}
\begin{scope}[decoration={markings,
mark=at position 0.5 with {\arrow[line width =1.8pt]{>}}
}
]
\draw[draw,cyan,line width =1.8, postaction=decorate]
(zeta1) + (0:\s) arc(0:360:\s);
\draw[draw,cyan,line width =1.8, postaction=decorate]
(zeta2) + (0:\s) arc(0:360:\s);
\draw[draw,cyan,line width =1.8, postaction=decorate]
(zeta3) + (0:\s) arc(0:360:\s);
\draw[draw,cyan,line width =1.8, postaction=decorate]
(zeta4) + (0:\s) arc(0:360:\s);
\draw[draw,cyan,line width =1.8, postaction=decorate]
(zeta1i) + (0:\s) arc(0:-360:\s);
\draw[draw,cyan,line width =1.8, postaction=decorate]
(zeta2i) + (0:\s) arc(0:-360:\s);
\draw[draw,cyan,line width =1.8, postaction=decorate]
(zeta3i) + (0:\s) arc(0:-360:\s);
\draw[draw,cyan,line width =1.8, postaction=decorate]
(zeta4i) + (0:\s) arc(0:-360:\s);
\end{scope}
\node[above,yshift=\s+7] at (zeta1) {\tiny{$Q^{-1}X_{j,+}Q$}};
\node[below,yshift=-\s-7] at (zeta2) {\tiny{$Q^{-1}X_{j,+}Q$}};
\node[above,yshift=\s+7] at (zeta1i) {\tiny{$Q^{-1}Y_{j,-}Q$}};
\node[below,yshift=-\s-7] at (zeta2i) {\tiny{$Q^{-1}Y_{j,-}Q$}};
\node[above,yshift=\s+7] at (zeta3) {\tiny{$Q^{-1}Y_{j,+}Q$}};
\node[below,yshift=-\s-7] at (zeta4) {\tiny{$Q^{-1}Y_{j,+}Q$}};
\node[above,yshift=\s+7] at (zeta3i) {\tiny{$Q^{-1}X_{j,-}Q$}};
\node[below,yshift=-\s-7] at (zeta4i) {\tiny{$Q^{-1}X_{j,-}Q$}};
\node[right] (Mtext) at (30:\olens) {\tiny{$Q^{-1} MQ$}};
\node[left,xshift=-6pt] (Ptext) at (30:\ilens) {\tiny{$Q^{-1}P Q$}};
\node[below left] {$0$};
\node[left] at(160:\R) {\tiny{$Q^{-1}JQ$}};

\node[above,yshift=2pt] at (a) {\small$\alpha$};
\node[below,xshift=-4pt] at (ai) {\small$\alpha^{-1}$};
\node[above,] at (b) {\small$\beta$};
\node[below] at (bi) {\small$\beta^{-1}$};
\node [above left] at (180:\R) {$-1$};
\node [above right] at (0:\R) {$1$};
\node at (90:2) {\tiny{$m_{1,\text{t}}\defeq \widetilde{m}$}};

\node (M) at (30:\R+0.2) {};
\coordinate (Mtext) at (35:\R+2.5);
\node[yshift=4pt] at (Mtext) {\tiny{$m_{1,\text{t}}\defeq\widetilde{m}Q^{-1} MQ$}};
\draw [->,thin ] (Mtext) -- (M);

\node (P) at (-45:\R-0.6) {};
\coordinate (Ptext) at (-45:\R+2);
\node[yshift=-4pt] at (Ptext) {\tiny{$m_{1,\text{t}}\defeq\widetilde{m} (Q^{-1}P Q)^{-1}$}};
\draw [->,thin] (Ptext) -- (P);

\node[above] at (130:\R+1.5) {\tiny{$m_{1,\text{t}} \defeq \widetilde{m}$}};
\node at (-25:\R+2) {\tiny{$m_{1,\text{t}} \defeq \widetilde{m}$}};
\end{tikzpicture}
}\\
\subfigure[]{
\begin{tikzpicture}[scale=0.6]
\def\R{9}
\def\m{1.5}
\coordinate (begin) at (180:\R);
\coordinate (a) at (200:\R);
\coordinate (b) at (240:\R);
\coordinate (end) at (255:\R);
\coordinate (zo) at (-6.47862,-6.2472); 
\coordinate (a1) at (-8.85102, -1.63079);
\coordinate (a2) at (-7.82853, -4.44006);
\coordinate (left) at (-5.73202, -6.93858);
\coordinate (right) at (-3.14298, -8.43337);
\node[right] at (a) {\small{$\alpha^{-1}$}};
\node[right] at (b) {\small{$\beta^{-1}$}};
\tkzCircumCenter(a,zo,b)\tkzGetPoint{O};
\begin{scope}[decoration={markings,
mark=at position 0.9 with {\arrow[line width =1.8pt]{>}}
}
]
\tkzDrawArc[color=cyan, line width =1.8](O,begin)(a1); 
\tkzDrawArc[color=cyan, line width =1.8,postaction=decorate](O,a2)(left);
\draw[draw,cyan,line width =1.8, postaction=decorate]
(a) + (90:\m) arc(90:270:\m);
\draw[draw,cyan,line width =1.8, postaction=decorate]
(a) + (270:\m) arc(270:450:\m);
\draw[draw,cyan,line width =1.8, postaction=decorate]
(b) + (15:\m) arc(15:90:\m);
\draw[draw,cyan,line width =1.8, postaction=decorate]
(b) + (90:\m) arc(90:285:\m);
\end{scope}

\draw[draw,dashed,line width =1.8]
(b) + (285:\m) arc(285:375:\m);

\foreach \Point in{(a), (b)}{
\node at \Point {\color{purple}\textbullet};
}
\begin{scope}[decoration={markings,
mark=at position 0.8 with {\arrow[line width =1.8pt]{>}}
}
]
\draw [cyan, line width =1.8,postaction=decorate] ($(b) + (15:\m)$)--($(b) + (15:\m+2)$);
\draw [cyan, line width =1.8,postaction=decorate] ($(b) + (285:\m)$)--($(b) + (285:\m+2)$);
\end{scope}
\node[left] at (185:\R) {\small{$Q^{-1}JQ$}};
\node[left] at ($(a) + (210:\m)$) {\small{$Q^{-1}MQ$}};
\node[right] at ($(a) + (10:\m)$) {\small{$Q^{-1}P^{-1}Q$}};
\node[left] at (220:\R) {\small{$Q^{-1}JQ$}};
\node[below right] at ($(b) + (15:\m+1)$) {\small{$Q^{-1}PQ$}};
\node[right,yshift=6pt] at ($(b) + (70:\m)$) {\small{$Q^{-1}P^{-1}Q$}};
\node[ left] at ($(b) + (285:\m+1)$) {\small{$Q^{-1}MQ$}};
\node[ left] at ($(b) + (220:\m)$) {\small{$Q^{-1}MQ$}};
\end{tikzpicture}
}
\caption{The initial deformation of the {\RHP} in the collisionless shock region. (a) Definition of $m_{1,\text{t}}(z)$ with its jump contours and matrices, (b) The initial jump contours and matrices near $\alpha^{-1}$, $\beta^{-1}$.}
\label{F:trans_1}
\end{figure}

Our second deformation involves conjugation by $\phi(z)$ as in the collisionless shock region. We define $m_{2,\text{t}}(z)$ by $m_{2,\text{t}}(z)=m_{1,\text{t}}(z)\phi(z)$, where $\phi(z)$ is defined as in \eqref{e:phi} {but we modify the definition $g(z)$ below}. The jump contours and the jump matrices for $m_{2,\rm{t}}(z)$ near $\alpha^{-1}$ and $\beta^{-1}$ are presented in Figure~\ref{F:trans_2}. We will now show that collapsing the lensing on $\Sigma_c$ and conjugating by $\phi(z)$ results in a well-behaved {\RHP} when the values $(n,t)$ lie in this region. In the analysis that follows, we omit the factors that come from conjugation by $Q(z)$ to simplify the notation. As in the collisionless shock region, doing this has no effect on the result.

Let $n = t - t^{1/3}r(t)$, where $r(t)$ satisfies
\begin{equation*}
\lim_{t\to\infty}\frac{r(t)}{(\log t)^{2/3}}= 0,~\text{ and }~\lim_{t\to\infty} r(t) = \infty.
\end{equation*}
Given a positive bounded function $f(n,t)$ we choose $\alpha$ and $\beta$ by enforcing (recall that $\Re \alpha + \Re \beta = 2 \lambda_0$)
\begin{equation}\label{E:alpha_transition}
\frac{f(n,t)}{t}= -i \int_{-1}^{\alpha}\frac{1}{p^2}\sqrt{\left(p-\alpha\right)\left(p-\alpha^{-1}\right)\left(p-\beta\right)\left(p-\beta^{-1}\right)}\, dp.
\end{equation}
In light of \eqref{E:g-plus}, this is equivalent to the conditions
\begin{equation}
\begin{aligned}\label{E:cs-offd1}
t\big(g^+(z) + g^-(z)\big) &= i f(n,t)~\text{ for }~z\in\oarc{\beta,\alpha},\\
t\big(g^+(z) + g^-(z)\big) &= -i f(n,t)~\text{ for }~z\in\oarc{\alpha^{-1},\beta^{-1}}.
\end{aligned}
\end{equation}
By adjusting $f$, \eqref{E:alpha_transition} can be solved\footnote{In practice we use $f(n,t) \equiv 2$.  Other choices may result in more efficient computations.} for $\alpha$ since the right hand side is a monotone function of $\Re \alpha$ under the constraint $\Re \alpha + \Re \beta = 2\Re \lambda_0$.

\begin{figure}[htp]
\begin{tikzpicture}[scale=0.6]
\def\R{9}
\def\m{1.5}
\coordinate (begin) at (180:\R);
\coordinate (a) at (200:\R);
\coordinate (b) at (240:\R);
\coordinate (end) at (255:\R);
\coordinate (zo) at (-6.47862,-6.2472); 
\coordinate (a1) at (-8.85102, -1.63079);
\coordinate (a2) at (-7.82853, -4.44006);
\coordinate (amid1) at ($(a1)!.5!(a)$);
\coordinate (amid2) at ($(a)!.5!(a2)$);
\coordinate (bmid) at ($(left)!.5!(b)$);
\coordinate (left) at (-5.73202, -6.93858);
\coordinate (right) at (-3.14298, -8.43337);
\node[right] at (a) {\small{$\alpha^{-1}$}};
\node[right] at (b) {\small{$\beta^{-1}$}};
\tkzCircumCenter(a,zo,b)\tkzGetPoint{O};
\begin{scope}[decoration={markings,
mark=at position 0.9 with {\arrow[line width =1.8pt]{>}},
mark=at position 0.45 with {\arrow[line width =1.8pt]{>}}
}
]
\tkzDrawArc[color=cyan, line width =1.8,postaction=decorate](O,begin)(a); 
\end{scope}

\begin{scope}[decoration={markings,
mark=at position 0.9 with {\arrow[line width =1.8pt]{>}}
}
]
\tkzDrawArc[color=cyan, line width =1.8,postaction=decorate](O,a)(a2);
\tkzDrawArc[color=cyan, line width =1.8,postaction=decorate](O,a2)(left);
\tkzDrawArc[color=cyan, line width =1.8,postaction=decorate](O,left)(b);
\draw[draw,cyan,line width =1.8, postaction=decorate]
(a) + (90:\m) arc(90:270:\m);
\draw[draw,cyan,line width =1.8, postaction=decorate]
(a) + (270:\m) arc(270:450:\m);
\draw[draw,cyan,line width =1.8, postaction=decorate]
(b) + (15:\m) arc(15:90:\m);
\draw[draw,cyan,line width =1.8, postaction=decorate]
(b) + (90:\m) arc(90:285:\m);
\end{scope}

\draw[draw,dashed,line width =1.8]
(b) + (285:\m) arc(285:375:\m);

\foreach \Point in{(a), (b)}{
\node at \Point {\color{purple}\textbullet};
}
\begin{scope}[decoration={markings,
mark=at position 0.8 with {\arrow[line width =1.8pt]{>}}
}
]
\draw [cyan, line width =1.8,postaction=decorate] ($(b) + (15:\m)$)--($(b) + (15:\m+2)$);
\draw [cyan, line width =1.8,postaction=decorate] ($(b) + (285:\m)$)--($(b) + (285:\m+2)$);
\end{scope}
\node[left] at (185:\R) {\small{$\phi_{-}^{-1}Q^{-1}JQ\phi_{+}$}};
\node[left] at ($(a) + (210:\m)$) {\small{$\phi^{-1}Q^{-1}MQ\phi$}};
\node[right] at ($(a) + (10:\m)$) {\small{$\phi^{-1}Q^{-1}P^{-1}Q\phi$}};
\node[left] at (220:\R) {\small{$\phi_{-}^{-1}Q^{-1}JQ\phi_{+}$}};
\node[below right] at ($(b) + (15:\m+1)$) {\small{$\phi^{-1}Q^{-1}P^{-1}Q\phi$}};
\node[right,yshift=6pt] at ($(b) + (70:\m)$) {\small{$\phi^{-1}Q^{-1}P^{-1}Q\phi$}};
\node[ left] at ($(b) + (285:\m+1)$) {\small{$\phi^{-1}Q^{-1}MQ\phi$}};
\node[ left] at ($(b) + (220:\m)$) {\small{$\phi^{-1}Q^{-1}MQ\phi$}};
\node(a_up_text) at ($(a) + (165:\m+2)$) {\small{$\phi^{-1}_{-}\phi_{+}$}};
\draw [->,thin] (a_up_text) -- (amid1);
\node(a_down_text) at ($(b) + (90:\m+2)$) {\small{$\phi^{-1}_{-}\phi_{+}$}};
\draw [->,thin] (a_down_text) -- (amid2);
\draw [->,thin] (a_down_text) -- (bmid);
\end{tikzpicture}
\caption{The jump contours and matrices for $m_{2,\text{t}}(z)$ near $\alpha^{-1}$, $\beta^{-1}$.}
\label{F:trans_2}
\end{figure}
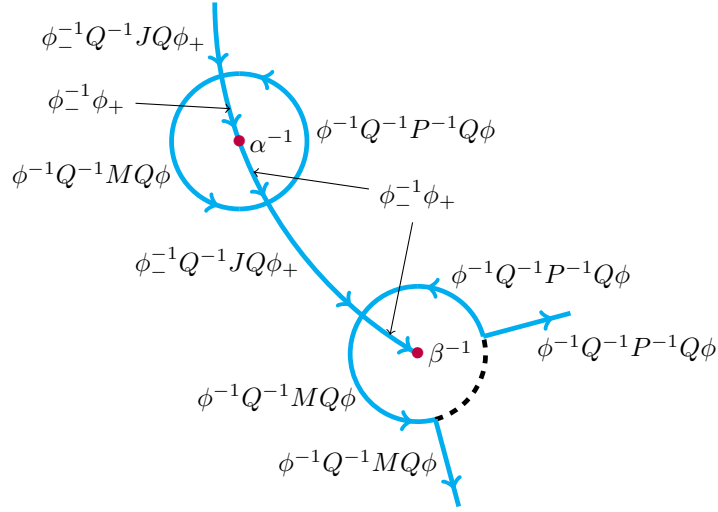

Define $h(n,t)$ by
\begin{equation*}
\frac{h(n,t)}{t} =  \int_{\alpha}^{\beta}\frac{1}{p^2}\sqrt{\left(p-\alpha\right)\left(p-\alpha^{-1}\right)\left(p-\beta\right)\left(p-\beta^{-1}\right)}^{+}\, dp,
\end{equation*}
and we have the following properties for $g(z)$:
\begin{equation}\label{E:g_rhp_2}
\begin{aligned}
&g^{+}(z) + g^{-}(z) =
\begin{cases}
\begin{aligned}
{i f(n,t)/t}&\quad\text{if } z\in \Sigma_u, \\
{-i f(n,t)/t}&\quad\text{if } z\in \Sigma_l,
\end{aligned}
\end{cases}\\
&g^{+}(z) - g^{-}(z) = {h(n,t)/t},\phantom{x} z\in \Sigma_c,\\
&g(z) - \frac{1}{2t}\theta(z) \text{ analytic in $z$ for } z \not\in \carc{\beta,\beta^{-1}}= \Sigma_u \cup \Sigma_c \cup \Sigma_l,\\
&g(z)\text{ is bounded at }z=\alpha^{\pm 1} \text{ and } z=\beta^{\pm 1},\\
&g(z) = \tfrac{1}{2}z - \lambda_0\log z + \mathcal{O}\left(z^{-1}\right)\text{ as } z\rightarrow\infty\,,
\end{aligned}
\end{equation}

Now, after applying the conjugation by $\phi(z)$, again assuming for simplicity that $Q(z) = I$, as in the collisionless shock region, the jump matrix on $\mathbb T$ in this region is of the form
\begin{equation}\label{E:t-jump}
\phi_{-}^{-1}(z)J(z)\phi^{+}(z) =
\begin{cases}
\begin{pmatrix} 1 - R(z)R\left(z^{-1}\right) & - R\left(z^{-1}\right) e^{-2tg(z)} \\ R(z) e^{2 tg(z)} & 1 \end{pmatrix},\quad & z \in \oarc{1,\beta},\\[12pt]
\begin{pmatrix} \left[1 - R(z)R\left(z^{-1}\right)\right]e^{t\left(g^{+}(z) - g^{-}(z)\right)} & - R\left(z^{-1}\right) e^{-if(n,t)} \\ R(z) e^{if(n,t)} & e^{t\left(-g^{+}(z) + g^{-}(z)\right)} \end{pmatrix},\quad & z \in \oarc{\beta,\alpha},\\[12pt]
\begin{pmatrix} \left[1 - R(z)R\left(z^{-1}\right)\right]e^{h(n,t)} & - R\left(z^{-1}\right) e^{-t \left(g^{+}(z) + g^{-}(z)\right)} \\ R(z) e^{t\left(g^{+}(z) + g^{-}(z)\right)} & e^{-h(n,t)} \end{pmatrix},\quad & z \in \oarc{\alpha,\alpha^{-1}},\\[12pt]
\begin{pmatrix} \left[1 - R(z)R\left(z^{-1}\right)\right]e^{t\left(g^{+}(z) - g^{-}(z)\right)} & - R\left(z^{-1}\right) e^{if(n,t)} \\ R(z) e^{-if(n,t)} & e^{t\left(-g^{+}(z) + g^{-}(z)\right)} \end{pmatrix},\quad & z \in \oarc{\alpha^{-1},\beta^{-1}},\\[12pt]
\begin{pmatrix} 1 - R(z)R\left(z^{-1}\right) & - R\left(z^{-1}\right) e^{-2tg(z)} \\ R(z) e^{2t g(z)} & 1 \end{pmatrix},\quad & z \in \oarc{\beta^{-1}, 1}.
\end{cases}
\end{equation}
Note that \eqref{E:cs-offd1}, along with the fact that
\begin{equation*}
t\big|g^+(z) + g^-(z)\big| \leq |f(n,t)|~\text{ for }~z\in\oarc{\alpha,\alpha^{-1}},
\end{equation*}
implies that oscillations in the off-diagonal entries of the jump matrix are controlled on $\oarc{\beta,\beta^{-1}}$. To analyze the situation concerning the diagonal entries, we find that $t\big|g^+(z) - g^-(z)\big| \leq |h(n,t)|$ for $z\in\oarc{\beta,\alpha}$. Using the change of variables $z(k) = K^{-1}(k)$ given in \eqref{E:cov}, one can see that there exists a constant $C>1$ such that
\begin{equation*}
\frac{1}{C} \leq \frac{f(n,t)}{t\rho_0^3} + \frac{h(n,t)}{t\rho_0^3} \leq C
\end{equation*}
in this region, where $\rho_0 = \Im z_0$ as before. Now, note that
\begin{equation*}
t\rho_0^3 \sim t \sqrt{8}\left( 1- \dfrac{t-t^{1/3}r(t)}{t} \right)^{3/2} = t\sqrt{8}t^{-1}r(t)^{3/2} =\sqrt{8}r(t)^{3/2} \to \infty,~\text{ as }~t\to\infty,
\end{equation*}
by the assumptions on $r(t)$. This implies that
\begin{equation*}
\frac{f(n,t)}{t\rho_0^3}\to 0~\text{ as }~t\to\infty,
\end{equation*}
and \eqref{E:alpha_transition} is solvable for sufficiently large $t$. Furthermore,
\begin{equation*}
h(n,t)\sim C t\rho_0^3 \to \infty~\text{ as }~t\to\infty,
\end{equation*}
which implies that the $(2,2)$-entries of the jump matrix given in \eqref{E:cs-jump} all tend to $0$ as $t\to\infty$ in this region. We are now left with the analysis of the $(1,1)$-entries of the jump matrix. We examine
\begin{equation*}
\left[1-R(z)R\left(z^{-1}\right) \right]e^{h(n,t)}~\text{ on }~\oarc{\beta,\beta^{-1}},
\end{equation*}
using the change of variables $z(k)=K^{-1}(k)$. Observe that
\begin{equation*}
\left[1-R\big(z(k)\big)R\big(z(k)^{-1}\big) \right]e^{h(n,t)} = \nu\big(z(k)+1\big)^2\big(1 + \mathcal{O}(z(k)+1)^2\big)e^{h(n,t)} = \nu\rho_0^2\left(1+\mathcal{O}\left(\rho_0^2\right) \right)e^{h(n,t)}
\end{equation*}
uniformly in $k$ for $A = K(\alpha)$ and $B = K(\beta)$ bounded as $t\to\infty$ and $K$ is defined in \eqref{E:cov}. Thus we are led to examine the behavior of $\rho_0^2 e^{h(n,t)}$ for large values of $t>0$. Note that $h(n,t)=\mathcal{O}(r(t)^{3/2})$, $\rho_0^2 = \mathcal{O}(t^{-2/3}r(t))$, and that for any $c>0$ there exists $T$ such that $r(t)^{3/2} \leq c \log(t)$ for $t>T$. There exists $C_1$, $C_2  > 0$
\begin{equation*}
\rho_0^2 e^{h(n,t)}\leq C_1 t^{-2/3} r(t) e^{C_2 c \log{t}}\leq C_1 t^{-2/3}t^{C_2 c} r(t)\to 0,~\text{ as }~t\to\infty,
\end{equation*}
for $c$ chosen sufficiently small. This implies that the $(1,1)$-entries of the jump matrix all tend to zero. Therefore the entries of the jump matrix remain bounded and this gives us an asymptotically well-behaved {\RHP} without any lensing on $\carc{\beta,\beta^{-1}}$.

We now proceed with the final deformation in this region. We define
\begin{equation*}
m_{\sharp,\text{t}}(z) =
\begin{cases}
m_{2,\text{t}}(z)\phi^{-1}(z),\quad &\text{inside the circles centered at $\alpha^{\pm 1}$ and $\beta^{\pm1}$},\\
m_{2,\text{t}}(z), \quad &\text{otherwise}.
\end{cases}
\end{equation*}
The jump contours and the jump matrices for the final {\RHP} for $m_{\sharp,\text{t}}(z)$ is given in Figure~\ref{F:trans_3}. The scaling of the circles around $\alpha$, $\beta$, $\alpha^{-1}$ and $\beta$ is the same as in the collisionless shock region: $\propto t^{-2/3}$.
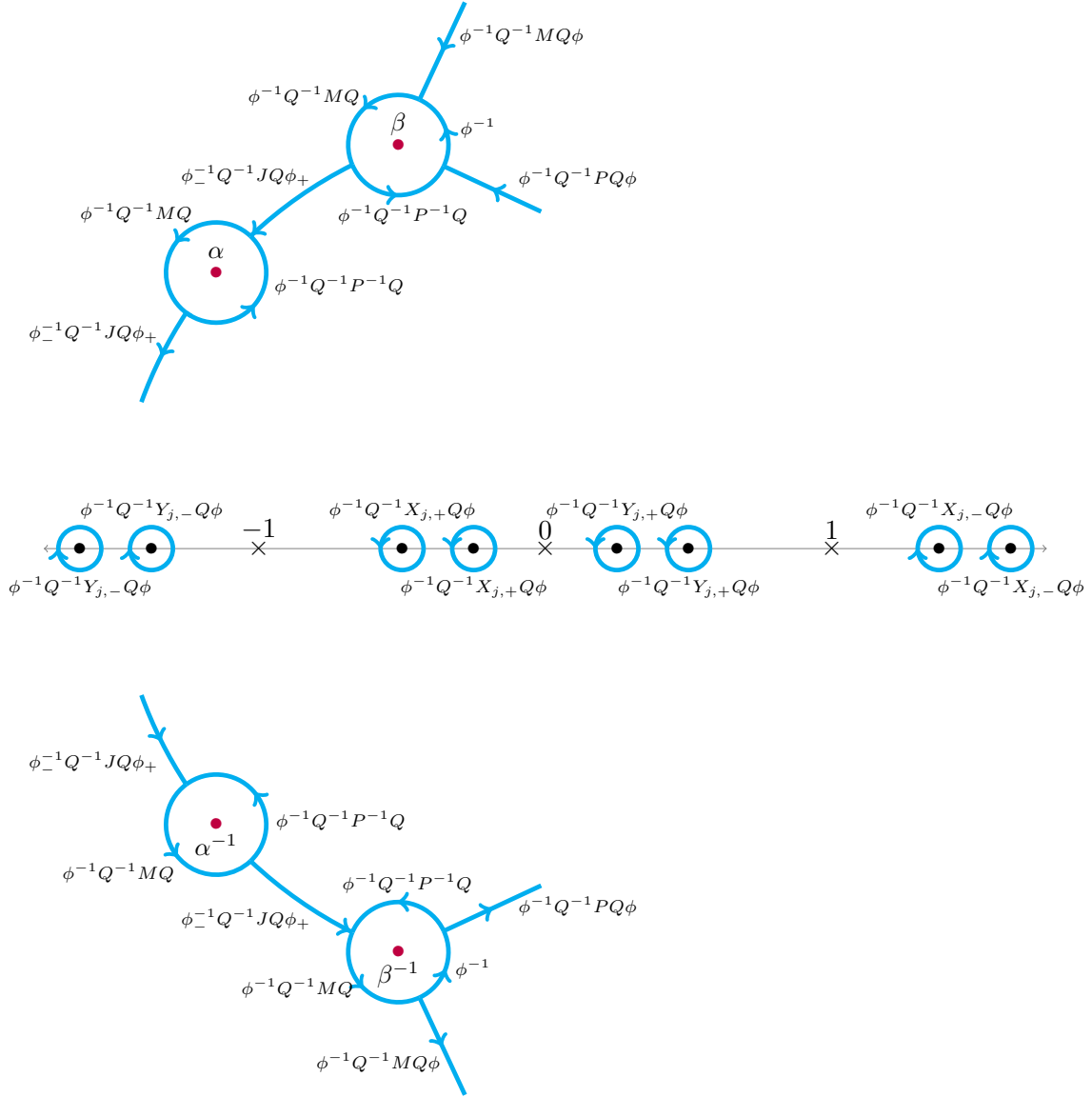
\begin{figure}[htp]
\centering
\begin{tikzpicture}
\def\R{4}
\def\L{\R+2}
\def\s{0.3}
\def\m{0.7}

\coordinate (a) at (140:\L);
\coordinate (ai) at (-140:\L);
\coordinate (b) at (110:\L);
\coordinate (bi) at (-110:\L);
\coordinate (zo) at (120:\L);
\coordinate (zoi) at (240:\L);
\coordinate (o) at (0,0);
\coordinate (to_b) at (-1.38149, 5.83879);
\coordinate (from_b) at (-2.69482, 5.36078);
\coordinate (to_a) at (-4.1158, 4.3658);
\coordinate (from_a) at (-5.01417, 3.29516);
\coordinate (from_bi) at (-1.38149, -5.83879);
\coordinate (to_bi) at (-2.69482, -5.36078);
\coordinate (from_ai) at (-4.1158, -4.3658);
\coordinate (to_ai) at (-5.01417, -3.29516);
\coordinate (top_M) at ($(b) + (65:\m)$);
\coordinate (top_P) at ($(b) + (-25:\m)$);
\coordinate (top_L) at ($(a) + (185:\m)$);
\coordinate (top_U) at ($(a) + (275:\m)$);
\coordinate (bottom_M) at ($(bi) + (295:\m)$);
\coordinate (bottom_P) at ($(bi) + (25:\m)$);
\coordinate (bottom_L) at ($(ai) + (175:\m)$);
\coordinate (bottom_U) at ($(ai) + (85:\m)$);
\coordinate (a_end) at ($(o) + (160:\L)$);
\coordinate (ai_begin) at ($(o) + (-160:\L)$);
\draw[help lines,<->] (-\R-3,0) -- (\R+3,0) coordinate (xaxis);

\coordinate (pt1) at (-2.00239, 3.46272);
\coordinate (pt2) at (-2.00239, -3.46272);
\coordinate (pt3) at (-0.69187, 3.93971);
\coordinate (pt4) at (-0.69187, -3.93971);
\coordinate (end) at (180:\R);
\coordinate (begin) at (0:\R);
\begin{scope}[decoration={markings, mark=at position 0.5 with {\arrow[line width =1.8pt]{>}}}]
\tkzDrawArc[color=cyan,line width =1.8,postaction=decorate](b,top_M)(from_b);
\tkzDrawArc[color=cyan,line width =1.8,postaction=decorate](b,from_b)(top_P);
\tkzDrawArc[color=cyan,line width =1.8,postaction=decorate](b,top_P)(top_M);

\tkzDrawArc[color=cyan,line width =1.8,postaction=decorate](a,to_a)(a_end);
\tkzDrawArc[color=cyan,line width =1.8,postaction=decorate](a,from_a)(to_a);

\tkzDrawArc[color=cyan,line width =1.8,postaction=decorate](bi,to_bi)(bottom_M);
\tkzDrawArc[color=cyan,line width =1.8,postaction=decorate](bi,bottom_M)(bottom_P);
\tkzDrawArc[color=cyan,line width =1.8,postaction=decorate](bi,bottom_P)(to_bi);

\tkzDrawArc[color=cyan,line width =1.8,postaction=decorate](ai,to_ai)(from_ai);
\tkzDrawArc[color=cyan,line width =1.8,postaction=decorate](ai,from_ai)(to_ai);

\tkzDrawArc[color=cyan,line width =1.8,postaction=decorate](o,from_a)(a_end);
\tkzDrawArc[color=cyan,line width =1.8,postaction=decorate](o,ai_begin)(to_ai);
\end{scope}

\begin{scope}[decoration={markings, mark=at position 0.98 with {\arrow[line width =1.8pt]{>}}}]
\tkzDrawArc[color=cyan,line width =1.8,postaction=decorate](o,from_b)(to_a);
\tkzDrawArc[color=cyan,line width =1.8,postaction=decorate](o,from_ai)(to_bi);
\end{scope}

\foreach \Point in {(o), (180:\R), (0:\R) }{
\node at \Point{$\times$};
}
\foreach \Point in{(a), (ai),(b),(bi)}{
\node at \Point {\color{purple}\textbullet};
}
\coordinate (zeta1) at (180:\R-2);
\coordinate (zeta2) at (180:\R-3);
\coordinate (zeta3) at (0:\R-2);
\coordinate (zeta4) at (0:\R-3);
\coordinate (zeta1i) at (180:\R+1.5);
\coordinate (zeta2i) at (180:\R+2.5);
\coordinate (zeta3i) at (0:\R+1.5);
\coordinate (zeta4i) at (0:\R+2.5);

\foreach \Point in {(zeta1), (zeta2), (zeta3), (zeta4), (zeta1i), (zeta2i), (zeta3i), (zeta4i)}{
\node at \Point{\textbullet};
}
\begin{scope}[decoration={markings,
mark=at position 0.5 with {\arrow[line width =1.8pt]{>}}
}
]
\draw[draw,cyan,line width =1.8, postaction=decorate]
(zeta1) + (0:\s) arc(0:360:\s);
\draw[draw,cyan,line width =1.8, postaction=decorate]
(zeta2) + (0:\s) arc(0:360:\s);
\draw[draw,cyan,line width =1.8, postaction=decorate]
(zeta3) + (0:\s) arc(0:360:\s);
\draw[draw,cyan,line width =1.8, postaction=decorate]
(zeta4) + (0:\s) arc(0:360:\s);
\draw[draw,cyan,line width =1.8, postaction=decorate]
(zeta1i) + (0:\s) arc(0:-360:\s);
\draw[draw,cyan,line width =1.8, postaction=decorate]
(zeta2i) + (0:\s) arc(0:-360:\s);
\draw[draw,cyan,line width =1.8, postaction=decorate]
(zeta3i) + (0:\s) arc(0:-360:\s);
\draw[draw,cyan,line width =1.8, postaction=decorate]
(zeta4i) + (0:\s) arc(0:-360:\s);

\draw [cyan, line width =1.8,postaction=decorate] ($(b) + (65:\m+1.5)$)--($(b) + (65:\m)$);
\draw [cyan, line width =1.8,postaction=decorate] ($(b) + (-25:\m+1.5)$)--($(b) + (-25:\m)$);



\draw [cyan, line width =1.8,postaction=decorate] ($(bi) + (25:\m)$)--($(bi) + (25:\m+1.5)$);
\draw [cyan, line width =1.8,postaction=decorate] ($(bi) + (295:\m)$)--($(bi) + (295:\m+1.5)$);
\end{scope}

\node[above,yshift=\s+7] at (zeta1) {\tiny{$\phi^{-1}Q^{-1}X_{j,+}Q\phi$}};
\node[below,yshift=-\s-7] at (zeta2) {\tiny{$\phi^{-1}Q^{-1}X_{j,+}Q\phi$}};
\node[above,yshift=\s+7] at (zeta1i) {\tiny{$\phi^{-1}Q^{-1}Y_{j,-}Q\phi$}};
\node[below,yshift=-\s-7] at (zeta2i) {\tiny{$\phi^{-1}Q^{-1}Y_{j,-}Q\phi$}};
\node[below,yshift=-\s-7] at (zeta3) {\tiny{$\phi^{-1}Q^{-1}Y_{j,+}Q\phi$}};
\node[above,yshift=\s+7] at (zeta4) {\tiny{$\phi^{-1}Q^{-1}Y_{j,+}Q\phi$}};
\node[above,yshift=\s+7] at (zeta3i) {\tiny{$\phi^{-1}Q^{-1}X_{j,-}Q\phi$}};
\node[below,yshift=-\s-7] at (zeta4i) {\tiny{$\phi^{-1}Q^{-1}X_{j,-}Q\phi$}};

\node[right,xshift=2pt] at ($(b) + (20:\m)$) {\tiny{$\phi^{-1}$}};
\node[right] at ($(b) + (65:\m+1)$) {\tiny{$\phi^{-1} Q^{-1}M Q \phi$}};
\node[above right] at ($(b) + (-25:\m+1)$) {\tiny{$\phi^{-1} Q^{-1}P Q \phi$}};
\node[left,xshift=-8pt] at ($(b) + (100:\m)$) {\tiny{$\phi^{-1} Q^{-1}M Q $}};
\node[below] at ($(b) + (-85:\m)$) {\tiny{$\phi^{-1} Q^{-1}P^{-1} Q$}};
\node[right] at ($(a) + (-15:\m)$) {\tiny{$\phi^{-1} Q^{-1}P^{-1} Q $}};
\node[left,yshift=4pt,xshift=-4pt] at ($(a) + (95:\m)$) {\tiny{$\phi^{-1} Q^{-1}M Q $}};
\node[left,xshift=-4pt] at (120:\L) {\tiny{$ \phi_{-}^{-1}Q^{-1}JQ\phi_{+}$}};
\node[left,xshift=-2pt] at ($(o)+(150:\L)$) {\tiny{$ \phi_{-}^{-1}Q^{-1}JQ\phi_{+}$}};
\node[left,xshift=-2pt] at ($(o)+(-150:\L)$) {\tiny{$ \phi_{-}^{-1}Q^{-1}JQ\phi_{+}$}};
\node[right] at ($(ai) + (5:\m)$) {\tiny{$\phi^{-1} Q^{-1}P^{-1} Q $}};
\node[left,yshift=-4pt] at ($(ai) + (230:\m)$) {\tiny{$\phi^{-1} Q^{-1}M Q$}};
\node[left,xshift=-4pt] at (-120:\L) {\tiny{$\phi_{-}^{-1}Q^{-1}JQ\phi_{+}$}};
\node[right] at ($(bi) + (340:\m)$) {\tiny{$\phi^{-1}$}};
\node[left] at ($(bi) + (225:\m)$) {\tiny{$\phi^{-1} Q^{-1}M Q $}};
\node[left] at ($(bi) + (295:\m+1)$) {\tiny{$\phi^{-1} Q^{-1}M Q \phi $}};
\node[right,yshift=-2pt] at ($(bi) + (25:\m+1)$) {\tiny{$\phi^{-1} Q^{-1}P Q \phi $}};
\node[above] at ($(bi)+(80:\m)$) {\tiny{$\phi^{-1} Q^{-1}P^{-1} Q $}};

\node[above] {$0$};
\node[above,yshift=2pt] at (a) {\small$\alpha$};
\node[below] at (ai) {\small$\alpha^{-1}$};
\node[above,] at (b) {\small$\beta$};
\node[below] at (bi) {\small$\beta^{-1}$};
\node [above ] at (180:\R) {$-1$};
\node [above ] at (0:\R) {$1$};
\end{tikzpicture}
\caption{A zoomed view of the jump contours and matrices of the final deformation of the {\RHP} (for $m_{\sharp,\text{t}}(z)$) in the transition region.}
\label{F:trans_3}
\end{figure}

The function $m_{\sharp,\text{t}}(z)$ satisfies a sectionally analytic {\RHP} with
\begin{itemize}
\item the jump conditions described in Figure~\ref{F:trans_3},
\item the asymptotic symmetry condition {given in \eqref{E:cs-sym}, and}
\item the quadratic normalization condition present in \rhref{rhp:disp}.
\end{itemize}

\begin{remark}In the collisionless shock region  when $\mathfrak t \defeq - \frac{\log \rho_0^2}{t \rho_0^3} \approx 0$, $\alpha \approx \beta$ and the $g$-function is essentially zero.  This is the degeneration of the collisionless shock region to the dispersive region and hence the two regions overlap.  The transition from the collisionless shock region to the transition region can be seen as $t \goto \infty$, when
\begin{align*}
    \mathfrak t  = - i \int_{B}^{A}\frac{\sqrt{(q-A)(q-B)\left(q+\bar{A}\right)\left(q+\bar{B}\right)}}{(\lambda_0 - i\rho_0 q)^2}^+ \; dq \approx -i \int_B^A \sqrt{(q-A)(q-B)\left(q+\bar{A}\right)\left(q+\bar{B}\right)}^+\,dq,
\end{align*}
and $A \approx 0$ (see Appendix~\ref{A:scaling} for the definition of $A$ and $B$).  This occurs when $ n = t - C_1 t^{1/3}(\log t)^{2/3}$ as $C_1 \downarrow 0$ and it can been seen that the jump matrices in the transition region are regularized for $r(t) = \epsilon (\log t)^{2/3}$ for $\epsilon$ sufficiently small.  Thus the collisionless shock region and the transition region overlap.  In the transition region as $t \goto \infty$ we have
\begin{align*}
\frac{f(n,t)}{t \rho_0^3}  =  -\int_{A}^{0}\frac{\sqrt{(q-A)(q-B)\left(q+\bar{A}\right)\left(q+\bar{B}\right)}}{(\lambda_0 - i\rho_0 q)^2} \, dq \approx \int_{0}^{A} \sqrt{(q-A)(q-B)\left(q+\bar{A}\right)\left(q+\bar{B}\right)}\,dq.
\end{align*}
Thus for $f(n,t)/(t \rho_0^3)$ sufficiently small, this equation is solvable for $A$.  Then note that for $n = t - c_3 t^{1/3}$, $t \rho^3 \sim (2 c_3)^{3/2}$.  Thus choosing $c_3$  sufficiently large, we see that the Painlev\'e region and the transition region overlap and the $g$-function degenerates to zero as $A \goto 0$ (or $\alpha \goto -1$). In this way, our deformations can be seen to bridge all regions. An animation showing the deformations is given in the supplementary material.  \end{remark}

\subsection{Soliton Region}\label{S:soliton}
In this region, we have $|n| > t$.  Note that $t = 0$, $|n| > 0$ is within the region. Therefore, the stationary phase points are no longer on the unit circle. Instead, for $n>0$:
\begin{equation*}
z_0 = -\frac{n}{t} + \sqrt{\left( \frac{n}{t}\right)^2 - 1} \in (-1, 0)\,,
\end{equation*}
and hence $z_0^{-1} \in (-\infty, -1)$. Let $\mathcal{A}_{\nu}=\left\lbrace { re^{i\omega}}\colon \omega\in[0,2\pi),\; r\in(1-\nu,1+\nu)\right\rbrace$, $\nu >0$, be the strip where $R(z)$ is analytic. Note that such a strip exists as a consequence of our exponential decay assumption on $\left(a_n^{0} -1/2, b_n^0\right)$. We have only one deformation to perform. We use the $MP$-factorization $J=MP$ within $\mathcal{A}_\nu$ and deform $\mathbb{T}$ into two contours, $\Gamma_{+}$ and $\Gamma_{-}$. As described in Remark~\ref{r:deformations}, if $z_0$ and $z_0^{-1}$ lie inside $\mathcal{A}_\nu$, $\Gamma_{+}$ and $\Gamma_{-}$ pass through $z_0$ and $z_0^{-1}$, respectively, locally in the directions of steepest descent of $e^{\pm\theta(z;n,t)}$. Away from $z_0$ and $z_0^{-1}$, $\Gamma_{+}$ and $\Gamma_{-}$ are concentric circles that stay close to the inner and outer boundaries of $\mathcal{A}_\nu$, respectively. Once $z_0^{\pm 1}$ leave the strip $\mathcal{A}_{\nu}$ deformed contours $\Gamma_{\pm}$ truncate to the concentric circles placed close to the boundaries of $\mathcal{A}_{\nu}$. Choices of $\Gamma_{\pm}$ are depicted in Figure~\ref{F:s_1}.

As in the dispersive region, we use the jump
\begin{equation*}
M(z;n,t) =\begin{pmatrix} 1 & - R\left(z^{-1}\right)e^{-\theta(z;n,t)} \\ 0 & 1\end{pmatrix}
\end{equation*}
on the contour $\Gamma_{-}$, and
\begin{equation*}
P(z;n,t) =\begin{pmatrix} 1 & 0 \\ R(z)e^{\theta(z;n,t)} & 1\end{pmatrix}
\end{equation*}
on the contour $\Gamma_{+}$. We define the vector-valued function $m_{\sharp,\text{s}}(z)$ in this region as given in Figure~\ref{F:s_1}(a).

There is one final detail left to be covered concerning the signature of the real part of the exponent $\theta(z;nt)$. Define $\zeta_0\in(-1,0)$ by $\Re \theta(z;n,t) = 0$, that is, by
\begin{equation}
\frac{n}{t} = - \frac{\zeta_0 - \zeta_0^{-1}}{2\log \left|\zeta_0\right|}\,.
\label{E:zeta0curve}
\end{equation}
Note that for $n>t>0$, we have $z_0<\zeta_0$ (see \cite{KT_rev}). In light of the discussion in Remark~\ref{r:deformations}, we arrange $\Gamma^{\pm} \subset \mathcal{A}_{\nu}$ in a way that they do not intersect the curve given in \eqref{E:zeta0curve}. This ensures that the exponents in $P$ and $M$ have negative real parts on their domains. Consequently, $P$ and $M$ tend to the identity matrix exponentially fast as $t\to\infty$.

The strip of analyticity, jump matrices, and jump contours for the {\RHP} satisfied by $m_{\sharp,\text{s}}(z)$ in this region are presented in Figure~\ref{F:s_1}(b) in the absence of poles $\{\zeta_j\}$. When poles are present, we make sure that $\Gamma_{\pm}$ do not intersect with the circles, $D_j^{\pm}$, around each pole, $\zeta_j^{\pm 1}$, $j=1,2,\dots,N$. We use the jumps that include $X_j^{\pm}$ and $Y_j^{\pm}$ (defined in \eqref{E:X} and \eqref{E:Y}) on $D_j^{\pm}$ in presence of the poles precisely as described in Section~\ref{S:dispersive}.
\begin{figure}[htp]
\centering
\subfigure[]{
\begin{tikzpicture}
\def\R{4}
\def\rin{3.6}
\def\rout{4.5}
\def\s{0.3}
\def\yL{\R}
\def\xL{\R}
\def\strip{1.1}
\def\olens{\R+0.8}
\def\ilens{\R-1}
\def\snapR{0.25}
\def\epsilon{0.15}

\draw[purple, dotted, line width = 1, fill=gray!10] (0,0) circle (\R+\strip);
\draw[purple, dotted, line width = 1, fill=white] (0,0) circle (\R-\strip);

\coordinate (zo) at (180:\rin);
\coordinate (zoi) at (180:\rout);
\coordinate (zetao) at (-1.6,0);
\coordinate (zetaoi) at (-5.8,0);
\coordinate (o) at (0,0);
\coordinate (snapUpperzo) at (170:\R-\strip+\epsilon);
\coordinate (snapLowerzo) at (190:\R-\strip+\epsilon);
\coordinate (snapUpperzoi) at (175:\R+\strip-\epsilon);
\coordinate (snapLowerzoi) at (185:\R+\strip-\epsilon);
\coordinate (innerEdgeStrip) at (180:\R-\strip+\epsilon);
\coordinate (outerEdgeStrip) at (180:\R+\strip+\epsilon);
\coordinate (snapInnerCircleCenter) at (180:\rin-\snapR);
\coordinate (snapOuterCircleCenter) at (180:\rout+\snapR);
\draw[help lines,<->] (-\R-3,0) -- (\R+3,0) coordinate (xaxis);

\begin{scope}[very thick,decoration={
     markings,
   mark=at position 0.33 with {\arrow[line width =1.8pt]{>}},
  mark=at position 0.66 with {\arrow[line width =1.8pt]{>}}}
  ]
\path[draw,cyan,line width =1.8,postaction=decorate]
(\R,0) arc(0:360:\R);
\end{scope}
\begin{scope}[very thick,decoration={
     markings,
   mark=at position 0.16 with {\arrow[line width =1.8pt]{>}},
   mark=at position 0.33 with {\arrow[line width =1.8pt]{>}},
   mark=at position 0.66 with {\arrow[line width =1.8pt]{>}},
   mark=at position 0.82 with {\arrow[line width =1.8pt]{>}}}
  ]
\path[draw,dashed,black,line width =1.8, postaction=decorate]
(\R-\strip+\epsilon,0) arc(0:170.3:\R-\strip+\epsilon);
\path[draw,dashed,black,line width =1.8, postaction=decorate]
(189.7:\R-\strip+\epsilon) arc(189.70:360:\R-\strip+\epsilon);
\path[draw,dashed,black,line width =1.8, postaction=decorate]
(\R+\strip-\epsilon,0) arc(0:175.3:\R+\strip-\epsilon);
\path[draw,dashed,black,line width =1.8, postaction=decorate]
(184.7:\R+\strip-\epsilon) arc(184.7:360:\R+\strip-\epsilon);
\end{scope}
\draw[dashed, black, line width =1.8, postaction=decorate] plot[smooth] coordinates {
  (snapUpperzo)
  ($(snapInnerCircleCenter) + (\snapR,+\snapR)$)
  ($(snapInnerCircleCenter) + (90:\snapR)$)
  ($(snapInnerCircleCenter) + (120:\snapR)$)
  ($(snapInnerCircleCenter) + (150:\snapR)$)
  ($(snapInnerCircleCenter) + (170:\snapR)$)
  ($(snapInnerCircleCenter) + (175:\snapR)$)
  (zo)
  ($(snapInnerCircleCenter) + (185:\snapR)$)
  ($(snapInnerCircleCenter) + (190:\snapR)$)
  ($(snapInnerCircleCenter) + (220:\snapR)$)
  ($(snapInnerCircleCenter) + (240:\snapR)$)
  ($(snapInnerCircleCenter) + (270:\snapR)$)
  ($(snapInnerCircleCenter) + (\snapR,-\snapR)$)
  (snapLowerzo)
   };

\draw[dashed, black, line width =1.8, postaction=decorate] plot[smooth] coordinates {
  (snapUpperzoi)
  ($(snapOuterCircleCenter) + (-\snapR+\epsilon/2,+\snapR)$)
  ($(snapOuterCircleCenter) + (90:\snapR)$)
  ($(snapOuterCircleCenter) + (60:\snapR)$)
  ($(snapOuterCircleCenter) + (30:\snapR)$)
  ($(snapOuterCircleCenter) + (10:\snapR)$)
  ($(snapOuterCircleCenter) + (5:\snapR)$)
  (zoi)
  ($(snapOuterCircleCenter) + (-5:\snapR)$)
  ($(snapOuterCircleCenter) + (-10:\snapR)$)
  ($(snapOuterCircleCenter) + (-30:\snapR)$)
  ($(snapOuterCircleCenter) + (-60:\snapR)$)
  ($(snapOuterCircleCenter) + (-90:\snapR)$)
  ($(snapOuterCircleCenter) + (-\snapR+\epsilon/2,-\snapR)$)
  (snapLowerzoi)
   };

\foreach \Point in {(zo), (zoi), (o),(zetao),(zetaoi) (180:\R), (0:\R) }{
\node at \Point{$\times$};
}
\coordinate (zeta1) at (180:\R-2);
\coordinate (zeta2) at (180:\R-3);
\coordinate (zeta3) at (0:\R-2);
\coordinate (zeta4) at (0:\R-3);
\coordinate (zeta1i) at (180:\R+1.5);
\coordinate (zeta2i) at (180:\R+2.5);
\coordinate (zeta3i) at (0:\R+1.5);
\coordinate (zeta4i) at (0:\R+2.5);

\draw[dotted, black, line width = 0.8](o) to [out=100,in=0] (-0.5,0.75) to [out=180, in=90] (zetao);
\draw[dotted, black, line width = 0.8](zetao) to [out=270,in=180] (-0.5,-0.75) to [out=0, in=260] (o);
\draw[dotted, black, line width = 0.8](zetaoi) to [out=100,in=315] (-\R-3,\R-2);
\draw[dotted, black, line width = 0.8](zetaoi) to [out=260,in=45] (-\R-3,-\R+2);
\node[below] (Mtext) at (90:\R+\strip-\epsilon) {\small{$\Gamma_{-}$}};
\node[above] (Ptext) at (90:\R-\strip+\epsilon) {\small{$\Gamma_{+}$}};
\node[below right] {$0$};
\node [below,yshift=2pt] at(270:\R) {\color{gray}\small{\color{cyan}$\mathbb{T}$}};
\node[above left,xshift=5pt] at (zo) {$z_0$};
\node[above right, xshift=-5pt] at (zoi) {$z_0^{-1}$};
\node [below left, xshift=4pt] at (180:\R) {$-1$};
\node [below right] at (0:\R) {$1$};
\node [below, xshift=-2pt] at (zetao) {$\zeta_0$};
\node [below, xshift=8pt] at (zetaoi) {$\zeta_0^{-1}$};
\node [above] at (-0.8, 0) {\tiny{\color{black}$\Re \theta(z) > 0$}};
\node at (0, -1.5) {\tiny{\color{black}$\Re \theta(z) < 0$}};
\node [above left] at (zetaoi) {\tiny{\color{black}$\Re \theta(z) < 0$}};
\node [above] at ($(zetaoi)+(0,2)$) {\tiny{\color{black}$\Re \theta(z) > 0$}};
\node[right] at (\R+\strip,\R-\strip) {$\mathcal{A}_{\nu}$};
\draw [<->,purple] (45:\R-\strip) -- (45:\R+\strip);
\draw [->,thin] (45:\R+0.1) -- (\R+\strip,\R-\strip);

\node[left] at (125:\R+1.5){\tiny{$m_{\sharp,\text{s}} \defeq \widetilde{m}$}};

\node (M) at (140:\R+0.25) {};
\node (Mtext) at (145:\R+3) {\tiny{$m_{\sharp,\text{s}} \defeq \widetilde{m}Q^{-1}MQ$}};
\draw [->,thin] (Mtext) -- (M);

\node (P) at (220:\R-0.7) {};
\node (Ptext) at (220:\R+3) {\tiny{$m_{\sharp,\text{s}} \defeq \widetilde{m}(Q^{-1}PQ)^{-1}$}};
\draw [->,thin] (Ptext) -- (P);

\node at (90:2){\tiny{$m_{\sharp,\text{s}} \defeq \widetilde{m}$}};

\end{tikzpicture}
}\\
\subfigure[]{
\begin{tikzpicture}
\def\R{4}
\def\rin{3.6}
\def\rout{4.5}
\def\s{0.3}
\def\yL{\R}
\def\xL{\R}
\def\strip{1.1}
\def\olens{\R+0.8}
\def\ilens{\R-1}
\def\snapR{0.25}
\def\epsilon{0.15}

\draw[purple, dotted, line width = 1, fill=gray!10] (0,0) circle (\R+\strip);
\draw[purple, dotted, line width = 1, fill=white] (0,0) circle (\R-\strip);

\coordinate (zo) at (180:\rin);
\coordinate (zoi) at (180:\rout);
\coordinate (zetao) at (-1.6,0);
\coordinate (zetaoi) at (-5.8,0);
\coordinate (o) at (0,0);
\coordinate (snapUpperzo) at (170:\R-\strip+\epsilon);
\coordinate (snapLowerzo) at (190:\R-\strip+\epsilon);
\coordinate (snapUpperzoi) at (175:\R+\strip-\epsilon);
\coordinate (snapLowerzoi) at (185:\R+\strip-\epsilon);
\coordinate (innerEdgeStrip) at (180:\R-\strip+\epsilon);
\coordinate (outerEdgeStrip) at (180:\R+\strip+\epsilon);
\coordinate (snapInnerCircleCenter) at (180:\rin-\snapR);
\coordinate (snapOuterCircleCenter) at (180:\rout+\snapR);
\draw[help lines,<->] (-\R-3,0) -- (\R+3,0) coordinate (xaxis);

\begin{scope}[very thick,decoration={
     markings,
   mark=at position 0.33 with {\arrow[line width =1.8pt]{>}},
  mark=at position 0.66 with {\arrow[line width =1.8pt]{>}}}
  ]
\path[draw,dashed,gray,line width =1.8,postaction=decorate]
(\R,0) arc(0:360:\R);
\end{scope}
\begin{scope}[very thick,decoration={
     markings,
   mark=at position 0.16 with {\arrow[line width =1.8pt]{>}},
   mark=at position 0.33 with {\arrow[line width =1.8pt]{>}},
   mark=at position 0.66 with {\arrow[line width =1.8pt]{>}},
   mark=at position 0.82 with {\arrow[line width =1.8pt]{>}}}
  ]
\path[draw,cyan,line width =1.8, postaction=decorate]
(\R-\strip+\epsilon,0) arc(0:170.3:\R-\strip+\epsilon);
\path[draw,cyan,line width =1.8, postaction=decorate]
(189.7:\R-\strip+\epsilon) arc(189.70:360:\R-\strip+\epsilon);
\path[draw,cyan,line width =1.8, postaction=decorate]
(\R+\strip-\epsilon,0) arc(0:175.3:\R+\strip-\epsilon);
\path[draw,cyan,line width =1.8, postaction=decorate]
(184.7:\R+\strip-\epsilon) arc(184.7:360:\R+\strip-\epsilon);
\end{scope}
\draw[cyan, line width =1.8, postaction=decorate] plot[smooth] coordinates {
  (snapUpperzo)
  ($(snapInnerCircleCenter) + (\snapR,+\snapR)$)
  ($(snapInnerCircleCenter) + (90:\snapR)$)
  ($(snapInnerCircleCenter) + (120:\snapR)$)
  ($(snapInnerCircleCenter) + (150:\snapR)$)
  ($(snapInnerCircleCenter) + (170:\snapR)$)
  ($(snapInnerCircleCenter) + (175:\snapR)$)
  (zo)
  ($(snapInnerCircleCenter) + (185:\snapR)$)
  ($(snapInnerCircleCenter) + (190:\snapR)$)
  ($(snapInnerCircleCenter) + (220:\snapR)$)
  ($(snapInnerCircleCenter) + (240:\snapR)$)
  ($(snapInnerCircleCenter) + (270:\snapR)$)
  ($(snapInnerCircleCenter) + (\snapR,-\snapR)$)
  (snapLowerzo)
   };

\draw[cyan, line width =1.8, postaction=decorate] plot[smooth] coordinates {
  (snapUpperzoi)
  ($(snapOuterCircleCenter) + (-\snapR+\epsilon/2,+\snapR)$)
  ($(snapOuterCircleCenter) + (90:\snapR)$)
  ($(snapOuterCircleCenter) + (60:\snapR)$)
  ($(snapOuterCircleCenter) + (30:\snapR)$)
  ($(snapOuterCircleCenter) + (10:\snapR)$)
  ($(snapOuterCircleCenter) + (5:\snapR)$)
  (zoi)
  ($(snapOuterCircleCenter) + (-5:\snapR)$)
  ($(snapOuterCircleCenter) + (-10:\snapR)$)
  ($(snapOuterCircleCenter) + (-30:\snapR)$)
  ($(snapOuterCircleCenter) + (-60:\snapR)$)
  ($(snapOuterCircleCenter) + (-90:\snapR)$)
  ($(snapOuterCircleCenter) + (-\snapR+\epsilon/2,-\snapR)$)
  (snapLowerzoi)
   };

\foreach \Point in {(zo), (zoi), (o),(zetao),(zetaoi) (180:\R), (0:\R) }{
\node at \Point{$\times$};
}
\coordinate (zeta1) at (180:\R-2);
\coordinate (zeta2) at (180:\R-3);
\coordinate (zeta3) at (0:\R-2);
\coordinate (zeta4) at (0:\R-3);
\coordinate (zeta1i) at (180:\R+1.5);
\coordinate (zeta2i) at (180:\R+2.5);
\coordinate (zeta3i) at (0:\R+1.5);
\coordinate (zeta4i) at (0:\R+2.5);

\draw[dotted, black, line width = 0.8](o) to [out=100,in=0] (-0.5,0.75) to [out=180, in=90] (zetao);
\draw[dotted, black, line width = 0.8](zetao) to [out=270,in=180] (-0.5,-0.75) to [out=0, in=260] (o);
\draw[dotted, black, line width = 0.8](zetaoi) to [out=100,in=315] (-\R-3,\R-2);
\draw[dotted, black, line width = 0.8](zetaoi) to [out=260,in=45] (-\R-3,-\R+2);
\node[below] (Mtext) at (90:\R+\strip-\epsilon) {\small{$Q^{-1} MQ$}};
\node[above] (Ptext) at (90:\R-\strip+\epsilon) {\small{$Q^{-1}P Q$}};
\node[below right] {$0$};
\node [below,yshift=2pt] at(270:\R) {\color{gray}\small{$Q^{-1} J Q$}};
\node[above left,xshift=5pt] at (zo) {$z_0$};
\node[above right, xshift=-5pt] at (zoi) {$z_0^{-1}$};
\node [below left, xshift=4pt] at (180:\R) {$-1$};
\node [below right] at (0:\R) {$1$};
\node [below, xshift=-2pt] at (zetao) {$\zeta_0$};
\node [below, xshift=8pt] at (zetaoi) {$\zeta_0^{-1}$};
\node [above] at (-0.8, 0) {\tiny{\color{black}$\Re \theta(z) > 0$}};
\node at (0, -1.5) {\tiny{\color{black}$\Re \theta(z) < 0$}};
\node [above left] at (zetaoi) {\tiny{\color{black}$\Re \theta(z) < 0$}};
\node [above] at ($(zetaoi)+(0,2)$) {\tiny{\color{black}$\Re \theta(z) > 0$}};
\node[right] at (\R+\strip,\R-\strip) {$\mathcal{A}_{\nu}$};
\draw [<->,purple] (45:\R-\strip) -- (45:\R+\strip);
\draw [->,thin] (45:\R+0.1) -- (\R+\strip,\R-\strip);
\end{tikzpicture}
}
\caption{(a) Jump contour (blue) and matrices for the initial {\RHP} with `ghost' contours (dashed black) that guide the deformation,(b) The jump contours and matrices for the {\RHP} satisfied by $m_{\sharp,\text{s}}(z)$ in the soliton region. This figure contains the definitions of the contours $\Gamma_{\pm}$.}
\label{F:s_1}
\end{figure}
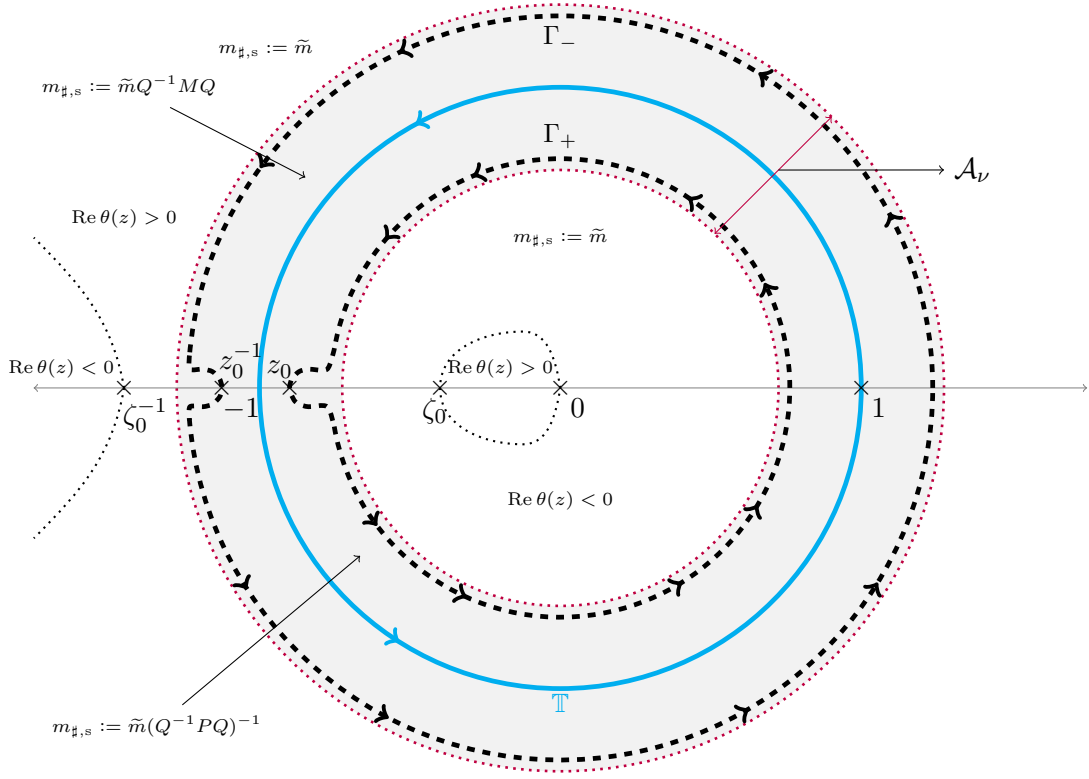
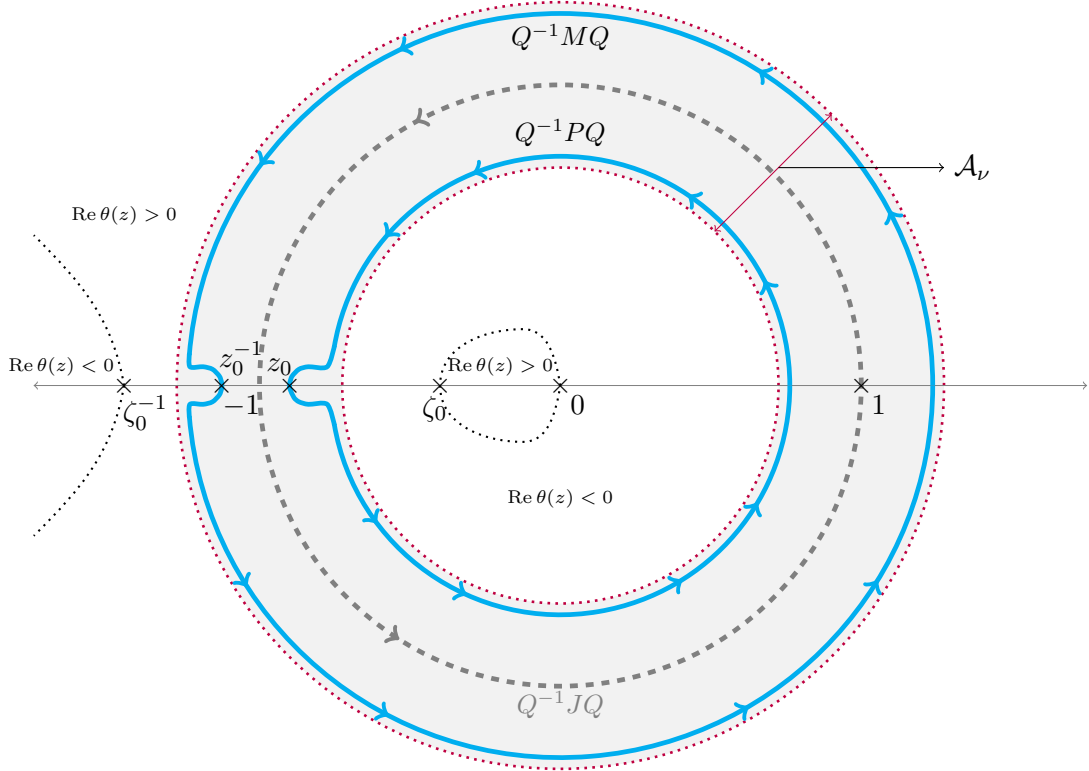

{The function} $m_{\sharp,\text{s}}(z)$ satisfies a sectionally analytic {\RHP} with
\begin{itemize}
\item the jump conditions described in Figure~\ref{F:s_1}(b),
\item the asymptotic symmetry condition present in \rhref{rhp:disp}, and
\item the quadratic normalization condition present in \rhref{rhp:disp}.
\end{itemize}

\section{Numerical solution of the deformed vector RH Problem}\label{S:inverse}
In this section we describe, in detail, the methodology used to numerically approximate the solution $m_{\sharp,\alpha}(z)$ of the deformed vector {\RHP} and then produce the associated approximation of the solution of the Toda lattice.

\subsection{The numerical solution of the associated matrix {\RHP}}\label{S:convert}

The asymptotic condition \eqref{E:vector_normal} is not convenient for numerical methods because it is nonlinear.   But, as pointed out in \cite{DKKZ}, we can convert a vector {\RHP} such as \rhref{rhp:disp} to a ($2\times 2$)-matrix {\RHP} with the same jump conditions and a standard (linear) condition at $\infty$ provided this matrix problem has a solution.  Then the solution of the vector {\RHP} can be reconstructed from the solution of the matrix problem.  This reconstruction is discussed in the following section.

Consider an {\RHP}
\begin{align}\label{test-RHP}
\Phi^+(z) = \Phi^-(z) J(z), ~~ z \in \Gamma, ~~\Phi(\infty) = I,
\end{align}
which has smooth solutions (see Section~2.7 in~\cite{mybook} for the requisite conditions on $J$).  We use $[J;\Gamma]$ to refer to this matrix {\RHP} with the identity matrix condition at infinity.

\begin{definition}\label{D:assoc}
If a vector {\RHP} has the same jump condition as \eqref{test-RHP} then $[J;\Gamma]$ is called the \emph{associated matrix {\RHP}}.  For example, \rhref{rhp:Psi} below is the associated matrix {\RHP} for \rhref{rhp:v}.
\end{definition}

Given an oriented, piecewise-smooth contour $\Gamma$ we define the Cauchy integral
\begin{equation*}
\mathcal{C}_\Gamma u(z) = \frac{1}{2\pi i}\int_{\Gamma}\frac{u(s)}{s-z}\,ds.
\end{equation*}
It is well known that the operators defined by
\begin{equation*}
\mathcal{C}_{\Gamma}^{\pm}u(z) = \left(\mathcal{C}_\Gamma u(z) \right)^{\pm}
\end{equation*}
are bounded operators from $L^{2}(\Gamma)$ to itself. Moreover, these operators satisfy the identity
\begin{equation}\label{E:plem}
\mathcal{C}_{\Gamma}^{+} - \mathcal{C}_\Gamma^{-} = I\,.
\end{equation}
If we assume that the solution to a matrix {\RHP} is of the form $\Phi = I + C_{\Gamma} u$, we can substitute this into the jump condition $\Phi^{+} = \Phi^{-} J$ and use this identity to obtain
\begin{equation}\label{E:sing_int}
\mathcal C[G;\Gamma] \defeq u - \mathcal{C}^{-}_{\Gamma}u \cdot(J-I) = J - I\,.
\end{equation}
This is a singular integral equation (SIE) for $u$. This motivates the following definition.
\begin{definition}
The matrix {\RHP}  $[J;\Gamma]$ is said to be well-posed if $\mathcal{C}[J;\Gamma]$ is invertible with a bounded inverse on $L^{2}(\Gamma)$ and $J-I \in L^{2}(\Gamma)$.
\end{definition}
This singular integral equation is critical in both the numerical and asymptotic solution of {\RHP}s.  A reader looking for a more in-depth discussion of {\RHP}s should look to \cite{AF} for an introduction and \cite{CG, Deift,Zhou} for a more advanced discussion. For numerical solution of {\RHP}s, we refer the reader to \cite{OT_RHP, mybook, TOD_KdV,TO_NLS}. A comprehensive discussion of the inverse scattering transform can be found in \cite{AC, AS}.

{We also point out that the singular integral equation formulation we use differs from that of \cite{Zhou}.  Our formulation has the benefit that the operator $u \mapsto \mathcal{C}^{-}_{\Gamma}u \cdot(J-I)$ can be applied exactly to a chosen basis where the operator considered in \cite{Zhou} given by $u \mapsto \mathcal{C}^{-}_{\Gamma}(u(J-I))$ does not have this property.  This seems to give a mild increase in the convergence rate.  See \cite[Chapter 2]{mybook} for a comparison of the theory for these two formulations.}

Consider the contour $\Gamma = \bigcup_{j=1}^n \Gamma_j$ where each
$\Gamma_j$ is either a line segment or a circular arc. Thus we restrict to considering contours where a
sequence of M\"obius transformations $M_1,\ldots,M_n$ are known such that
$M_k([-1,1]) = \Gamma_k$. Let $\mathbb P_m = \{\cos(j\pi/m): j =
0,1,\ldots,m\}$ be the Chebyshev points and let $T_m(x)$ denote the
$m^{\text{th}}$ Chebyshev polynomial of the first kind.  {The points $\bigcup_j M_j(\mathbb P_{n_j})$ are called the collocation points.  Note that if $a$ is an intersection point of a subset $\{\Gamma_{j_1},\ldots,\Gamma_{j_m}\}$ of the contours $\Gamma_1,\ldots,\Gamma_j$ then it will be included $m$ times in this union.  We include it $m$ times by using the notation $a + 0e^{i\theta_{j_k}}$ where $\theta_{j_k}$ is the angle at which $\Gamma_{j_k}$ leaves/approaches $a$.  Additionally,  $f \circeq g$ is used if and only if $f(a) = g(a)$ for all $a \in \bigcup_j M_j(\mathbb P_{n_j})$ that is not a point of self-intersection and $\lim_{\epsilon\downarrow 0} f(a +  \epsilon e^{i\theta_{j_k}}) = \lim_{\epsilon\downarrow 0} g(a + \epsilon e^{i\theta_{j_k}})$ if $a$ is a point of intersection.

 A function $\Psi$ is said to satisfy the {\RHP} \eqref{test-RHP} at the collocation points $\bigcup_j M_j(\mathbb P_{n_j})$ if $\Psi$ has continuous boundary values and $\Psi^+ \circeq \Psi^- J$.}  The framework of Olver \cite{olver-frame} implemented in \cite{RHPackage} (see also \cite{mybook}) is designed to return a vector $V_j$ of function values at the mapped points $M_j(\mathbb P_{n_j})$ (with directions attached at intersection points), so that the function $U: \Gamma \goto \mathbb C^{i \times j}$ defined piecewise by
\begin{align}
U(z)|_{\Gamma_j} &= \sum_{i=0}^{n_j} \alpha_i T_i(M_j^{-1}(z)),\label{E:coeff}\\
U(M(\mathbb P_{n_j})) &= V_j,\label{E:coeffdet}
\end{align}
satisfies
\begin{itemize}
\item $I + {\mathcal C}_{\Gamma}U$ is a bounded function in $\mathbb C\setminus \Gamma$, and
\item $I+ {\mathcal C}_{\Gamma}U$ satisfies the {\RHP} \eqref{test-RHP} exactly at $M_j(\mathbb P_{n_j})$.
\end{itemize}
Here the coefficients $\alpha_i$ in the definition of $U$ in \eqref{E:coeff} are determined by the condition \eqref{E:coeffdet}.

We describe the method in more detail. Similar to before, substituting
\begin{equation}\label{E:integral-form}
\Phi \approx I + \mathcal{C}_{\Gamma}U
\end{equation}
into the {\RHP} and using \eqref{E:plem} gives a linear equation for $U$:
\begin{align}\label{E:SIE}
  U - {\mathcal C}_\Gamma^- U \cdot (J-I) \circeq J - I.
\end{align}

 A closed-form expression for the Cauchy
 transform of the basis $T_i\big(M_j^{-1}(k)\big)$ \cite[Section 4]{Andre} (see also \cite{olver-hilbert}) allows the discretization
 of this linear equation by evaluating the Cauchy transform of the
 basis at the points $M_j(\mathbb P_{n_j})$. However, a modified
 definition for the Cauchy transform is required at the self-intersection (or junction) points
 points $\Gamma_0$ (which are included in the collocation points
 $M_j\big(\mathbb P_{n_j}\big)$), at which the Cauchy transform of this basis
 is unbounded.  By assuming that the computed $U$ is in the class of
 functions for which $I + {\mathcal C}_{\Gamma}U$ is bounded, we can
 define the bounded contribution of the Cauchy transform of each basis
 element at the points $\Gamma_0$. It can be shown, with appropriate assumptions on $J$ (see the \emph{product condition}, \cite[Definition~3.8.3]{thesis}) that the numerically calculated $U$ must
 be in this class of functions. Therefore $I + {\mathcal C}_\Gamma U$
 will be bounded and satisfies the {\RHP} at $\Gamma_0$, hence at all points in $M_j(\mathbb P_{n_j})$.

We use \eqref{E:integral-form} to show that if $U \in L^1(\Gamma)$ then
\begin{align*}
\lim_{z\goto \infty} z (\Phi(z)-I) = - \frac{1}{2\pi i} \int_{\Gamma} U(t) dt,
\end{align*}
by the Dominated Convergence Theorem provided $z$ is bounded away from $\Gamma$. The integral on the right-hand side can be computed using Clenshaw--Curtis quadrature. This relationship is needed in what follows to reconstruct the solution to the Toda lattice from the solution of the {\RHP}. Observe the complimentary fact: if
\begin{align*}
  \Phi(z) = A_1(z)A_2(z), \quad A_i(z) = I + A_{i,1} z^{-1} + \bigo(z^{-2}), \quad z \goto \infty,
\end{align*}
then
\begin{align}\label{product-recover}
  \Phi(z;n,t) = I + (A_{1,1} + A_{2,1}) z^{-1} + \bigo(z^{-2}), \quad z \goto \infty.
\end{align}

Another aspect of the numerical solution of {\RHP}s is contour truncation.  Note that if $J(z^*) = I$ then the linear system \eqref{E:SIE} at this point becomes $U(z^*) = 0$.  From this one can rigorously justify the removal of contours from the {\RHP} on which $\|J-I\|_{L^1 \cap L^\infty}$ is small at the cost of a small error.  This is discussed in more detail in \cite[Chapter 2]{mybook} and see \cite{Gauss} for a discussion of implementing this idea.

\begin{remark}\label{convergence} From the results in \cite{olver-frame} it follows that spectral convergence (\emph{i.e.} convergence that is faster than $n^{-k}$ for any $k$ where $n$ is the number of collocation points) can be verified \emph{a posteriori} for a well-posed {\RHP} by checking that the norm of the inverse of  the discretization of $\mathcal C[J;\Gamma]$ grows at most algebraically with respect to the number of collocation points. In the computations for this paper, we noticed at most logarithmic growth of the condition number for this collocation matrix, with a maximum on the order of $10^3$. \end{remark}

\subsection{Construct the solution of the deformed vector {\RHP}}
 With a method in hand to compute the solution of a matrix {\RHP}, normalized to be the identity matrix at infinity, we show that in order to solve the vector {\RHP} one can first solve the associated matrix {\RHP} and then take an appropriate linear combination of the rows of the solution of the associate matrix {\RHP}.  This is the generic situation but there are some technicalities so we take care in the following developments.

Let $J: \Gamma \to \mathbb C^{2\times 2}$ where the contour $\Gamma$ satisfies $\Gamma^{-1} \defeq \{z^{-1}: z \in \Gamma\} = - \Gamma$ and the minus sign refers to a reversal of orientation.  Assume the symmetry condition
  \begin{align}\label{E:J-symm}
  	J(z) = \begin{pmatrix} 0 & 1 \\ 1 & 0 \end{pmatrix} J^{-1}\left(z^{-1}\right) \begin{pmatrix} 0 & 1 \\ 1 & 0 \end{pmatrix}.
  \end{align}
  The {\RHP}s we want to solve are of the following form (compare with \rhref{rhp:hm}) which is assumed to be uniquely solvable:
  \begin{rhp}\label{rhp:v-full}
  For an oriented contour $\Gamma$, we seek a function $v \colon \mathbb C \setminus {\Gamma} \to \mathbb{C}^{1\times 2}$ that is sectionally analytic, continuous up to $ \Gamma$ and satisfies:
\begin{itemize}
\item \emph{the jump condition:}
\begin{equation}\label{E:matrix_jump-full}
v^{+}(z) = v^{-}(z) J(z),\phantom{x} z\in {\Gamma},
\end{equation}
\item \emph{the symmetry condition:}
\begin{equation}\label{E:symm-full}
v\left(z\right) = v(z^{-1})\begin{pmatrix} 0& 1 \\ 1 & 0 \end{pmatrix},
\end{equation}
\item  and \emph{the normalization condition:}
\begin{equation}\label{E:normal-full}
v(\infty)=\lim_{z\to \infty} v(z) = \begin{pmatrix} v_1 & v_2 \end{pmatrix}, \quad v_1\cdot v_2 = 1,~v_1 > 0\,.
\end{equation}
\end{itemize}
\end{rhp}
In the background, throughout all of our deformations of \rhref{rhp:hm} is a problem of the form of \rhref{rhp:v-full}. But we \emph{do not} preserve the symmetry condition through the deformations, mainly for convenience and ease of numerical implementation.  We introduce the notion of a non-singular deformation to encapsulate this:
\begin{definition}
A vector or matrix function $\tilde v$ is a \emph{non-singular deformation} of a sectionally analytic function $v : \mathbb C \setminus \Gamma \to \mathbb C^{j \times 2}$, for $j=1$ or $j=2$, with continuous boundary values if there exists a sectionally analytic matrix function $H: \mathbb C \setminus \Gamma' \to \mathbb C^{2\times 2}$, $\det H(z) = 1$, also with continuous boundary values, such that
\begin{align}\label{E:Hdiag}
\lim_{z \to \infty} H(z) = \diag(c_1,c_2)
\end{align}
exists and
\begin{align*}
\tilde v(z) = v(z) H(z), \quad z \in \mathbb C \setminus \tilde \Gamma, \quad \tilde \Gamma \defeq (\Gamma \cup \Gamma').
\end{align*}
\end{definition}
\noindent So, if $\tilde v: \mathbb C \setminus \tilde \Gamma \to\mathbb{C}^{j\times 2}$, $j=1$ or $j=2$, is a non-singular deformation of $v(z)$ then for
\begin{align}\label{E:tildeJ}
\tilde J(z) \defeq H^{-1}_-(z) J(z) H_+(z),
\end{align}
$\tilde v(z)$ satisfies:
\begin{rhp}\label{rhp:v}
  For an oriented contour $\tilde \Gamma$, we seek a function $\tilde v \colon \mathbb C \setminus {\tilde \Gamma} \to \mathbb{C}^{1\times 2}$ that is sectionally analytic, continuous up to $ \tilde \Gamma$ and satisfies:
\begin{itemize}
\item \emph{the jump condition:}
\begin{equation}\label{E:matrix_jump}
\tilde v^{+}(z) = \tilde v^{-}(z) \tilde J(z),\phantom{x} z\in {\tilde \Gamma},
\end{equation}
\item \emph{the {asymptotic} symmetry condition:}
\begin{equation}\label{E:symm-asym}
\tilde v\left(0\right) H^{-1}(0) = \tilde v(\infty)H^{-1}(\infty)\begin{pmatrix} 0& 1 \\ 1 & 0 \end{pmatrix},
\end{equation}
\item  and \emph{the normalization condition:}
\begin{equation}\label{E:normal}
\tilde v(\infty)H^{-1}(\infty) = \begin{pmatrix}v_1 & v_2\end{pmatrix}, \quad v_1\cdot v_2 = 1, ~v_1 > 0\,.
\end{equation}
\end{itemize}
\end{rhp}
\noindent  A solution of \rhref{rhp:v-full} clearly produces a solution of \rhref{rhp:v} and so \rhref{rhp:v} is solvable (we assume \rhref{rhp:v-full} is always uniquely solvable).  Since \rhref{rhp:v} turns out to be a bit more numerically tractable, we want to know when the solution of \rhref{rhp:v} is unique.  So, consider the associated matrix {\RHP}:
  \begin{rhp}\label{rhp:Psi}
  For a bounded, oriented contour $\tilde \Gamma$, bounded away from the origin, we seek a function $\Psi \colon \mathbb C \setminus {\tilde \Gamma} \to \mathbb{C}^{2\times 2}$ that is sectionally analytic, continuous up to $\tilde \Gamma$ and satisfies:
\begin{itemize}
\item \emph{the jump condition:}
\begin{equation*}
\Psi^{+}(z) = \Psi^{-}(z) \tilde J(z),\phantom{x} z\in {\tilde \Gamma},
\end{equation*}
\item  and \emph{the normalization condition:}
\begin{equation*}
\Psi(\infty)= I\,.
\end{equation*}
\end{itemize}
\end{rhp}

\begin{lemma}\label{l:equiv}
Suppose that \rhref{rhp:Psi} is uniquely solvable.  Then \rhref{rhp:v} is uniquely solvable if and only if \rhref{rhp:v-full} is uniquely solvable.
\end{lemma}
\begin{proof}
We begin with a straightforward calculation. Let $\tilde v(z)$ be a solution of \rhref{rhp:v}.  Define $\hat v(z) = \tilde v(z)H^{-1}(z)$, a solution of \rhref{rhp:v-full} with an asymptotic symmetry condition (not the global symmetry condition \eqref{E:symm-full}):
\begin{align}\label{E:v2-symm}
\hat v(0) = \hat v(\infty) \begin{pmatrix}
0 & 1 \\ 1 & 0
\end{pmatrix}.
\end{align}
Then consider
\begin{align*}
v(z) \defeq \hat v(z^{-1}) \begin{pmatrix} 0 & 1 \\ 1 & 0 \end{pmatrix}.
\end{align*}
We see from \eqref{E:J-symm}
\begin{align*}
v^+(z) & = \hat v^-(z^{-1}) \begin{pmatrix} 0 & 1 \\ 1 & 0 \end{pmatrix} = \hat v^+(z^{-1}) J^{-1}(z^{-1}) \begin{pmatrix} 0 & 1 \\ 1 & 0 \end{pmatrix} = v^-(z) J(z).
\end{align*}
Therefore $v(z)$  satisfies the jump condition in \rhref{rhp:v-full} with the asymptotic symmetry condition \eqref{E:v2-symm}.

Now, assume \rhref{rhp:v} is uniquely solvable, with solution $\tilde v(z)$.  Define $k(z)$ by
\begin{align*}
k(z) = \half \tilde v (z)H^{-1}(z) + \half \tilde v(z^{-1})H^{-1}(z^{-1}) \begin{pmatrix} 0 & 1 \\ 1 & 0 \end{pmatrix}.
\end{align*}
It follows that $k(z)$ satisfies \eqref{E:normal-full} and therefore $k(z)$ is a solution of \rhref{rhp:v-full}.  It is clear that any solution of \rhref{rhp:v-full} gives a solution, via $H(z)$, of \rhref{rhp:v} and so \rhref{rhp:v-full} is uniquely solvable.

Assume \rhref{rhp:v-full} is uniquely solvable with solution $v(z)$.  Then $v(z)H(z)$ is clearly a solution of \rhref{rhp:v}.  Assume then that $k(z)$ is a new solution to \rhref{rhp:v} with $k(z) \neq v(z)H(z)$.  Define
\begin{align*}
\tilde v(z) = \half k(z)H^{-1}(z) + \half k(z^{-1})H^{-1}(z^{-1}) \begin{pmatrix} 0 & 1 \\ 1 & 0 \end{pmatrix},
\end{align*}
which is a solution of \rhref{rhp:v-full} and by uniqueness $v(z) = \tilde v(z)$.  From \eqref{E:symm-asym}
\begin{align*}
k(0)H^{-1}(0) = k(\infty) H^{-1}(\infty) \begin{pmatrix} 0 & 1 \\ 1 & 0 \end{pmatrix},
\end{align*}
so that
\begin{align*}
\tilde v(\infty) = \half k(\infty) H^{-1}(\infty) + \half k(0) H^{-1}(0)\begin{pmatrix} 0 & 1 \\ 1 & 0 \end{pmatrix} = k(\infty) H^{-1}(\infty) = v(\infty).
\end{align*}
So, let $w(z) \defeq (v(z) - k(z)H^{-1}(z))H(z) = v(z)H(z) - k(z)$, the difference of two solutions of \rhref{rhp:v}. It follows that $w(z)$ satisfies the jump condition \eqref{E:matrix_jump} with
\begin{align*}
w(\infty) &= \begin{pmatrix}
0 & 0
\end{pmatrix}.
\end{align*}
From the uniqueness of solutions of \rhref{rhp:Psi}, $w$ is identically zero.
\end{proof}

Given \rhref{rhp:v}, we call \rhref{rhp:Psi} the associated matrix {\RHP}.  The goal is to compute the coefficients in the expansion
\begin{align*}
v(z) = \begin{pmatrix} v_1 + v_{1,1}z & v_2 + v_{2,1} z \end{pmatrix} + \mathcal O(z^{2}), \quad z \to 0.
\end{align*}
Since we know that \rhref{rhp:v} and \rhref{rhp:v-full} are equivalent if \rhref{rhp:Psi} is uniquely solvable, we use the following:
\begin{lemma}[From matrix solution to vector solution]\label{L:MatToVec}
Assume {RH Problem}s~\ref{rhp:Psi} and \ref{rhp:v-full} have unique solutions $\Psi$ and $v$, respectively.  Then the left nullspace of $\tilde \Psi \defeq \Psi(0)H^{-1}(0) - H^{-1}(\infty)\begin{pmatrix} 0 & 1 \\ 1 & 0 \end{pmatrix}$ is one dimensional and
\begin{align*}
v(z) \defeq Y \Psi(z) H^{-1}(z)
\end{align*}
solves \rhref{rhp:v-full} where $Y$ is the left null vector of $\tilde \Psi$ chosen so that $X\defeq Y H^{-1}(\infty)$ satisfies  $X_1 > 0$ and $Y_1 \cdot Y_2 = 1$.
\end{lemma}
\begin{proof}
Let $v(z)$ be the unique solution of \rhref{rhp:v-full} and we must show that $v(z) = Y \Psi(z)H^{-1}(z)$ for a unique vector $Y$. As we assume \rhref{rhp:Psi} is uniquely solvable with solution $\Psi$,  from Lemma~\ref{l:equiv} it suffices to enforce the asymptotic symmetry condition \eqref{E:symm-asym} because the unique solutions to \rhref{rhp:v} and \rhref{rhp:v-full} coincide.   By enforcing \eqref{E:symm-asym}
  \begin{align}\label{E:enforce-symm}
    Y \Psi(0)H^{-1}(0) = Y \Psi(\infty) H^{-1}(\infty) \begin{pmatrix} 0 & 1 \\ 1 & 0 \end{pmatrix}.
  \end{align}
  Then $Y$, if it exists, must be in the left nullspace of $\tilde \Psi \defeq \Psi(0)H^{-1}(0) - H^{-1}(\infty)\begin{pmatrix} 0 & 1 \\ 1 & 0 \end{pmatrix}$.  Now, to see that we can find a left null vector that can be normalized by  $X\defeq Y H^{-1}(\infty)$,  $X_1 > 0$ and $X_1 \cdot X_2 = 1$. Consider the matrix function
\begin{align*}
  K(z) \defeq \begin{pmatrix} \Psi_1(z)\\
    v(z)H(z) \end{pmatrix}
\end{align*}
where subscript refers to the first row.  Because neither $v_1$ nor $v_2$ in $v(\infty) = \begin{pmatrix}
v_1 & v_2
\end{pmatrix}$ can vanish and $\Psi_1(\infty) = \begin{pmatrix}
1 & 0
\end{pmatrix}$, $K(\infty)$ is invertible by \eqref{E:Hdiag}. We set $Z(z) = K^{-1}(\infty)K(z)$ so that $Z(\infty)= I$ and by uniqueness $\Psi = Z$.  It suffices to take $Y = \begin{pmatrix} 0 & 1 \end{pmatrix} K(\infty)$ and $Y$ must exist.

Now, assume there is another left null vector $\tilde Y$ of $\tilde \Psi$ that is not a multiple of $Y$.  Then $\tilde \Psi = 0$ and $Y_1 \cdot Y_2 > 0$ so $\tilde Y = \begin{pmatrix} Y_2^* & 0 \end{pmatrix}$ is linearly independent of $Y$. Let $\alpha > 0$ and $\beta > 0$, and consider
\begin{align*}
\hat Y H^{-1}(\infty) = \beta(Y + \alpha \tilde Y) H^{-1}(\infty) = \begin{pmatrix} \beta(Y_1 + \alpha Y_2^*)/c & c\beta Y_2 \end{pmatrix}, \quad H^{-1}(\infty) = \diag(1/c,c).
\end{align*}
As $Y_1 \cdot Y_2 = 1$, $Y_1 = Y_2^*/|Y_2|^2$ so that $\arg Y_1 = \arg Y_2^*$.  Then,  choosing $\alpha > 0$ sufficiently small
\begin{align*}
\frac{Y_1}{c}> 0 \Rightarrow \frac{(Y_1 + \alpha Y_2^*)}{c} > 0.
\end{align*}
Then by choosing $\beta$ so that $\hat Y_1 \cdot \hat Y_2 = 1$ we find that
\begin{align*}
\hat v(z) = \hat Y \Psi(z)H^{-1}(z),
\end{align*}
is another solution of \rhref{rhp:v-full} with $\hat v(\infty) \neq v(\infty)$, violating uniqueness.  Therefore $Y$ is uniquely defined.
\end{proof}

There is a subtlety here that will become apparent as we proceed.  We often know that the non-singular deformation carried through by $H(z)$ exists but we will not want to compute it.  We will also encounter singular deformations, \emph{i.e.}\ we will multiply the solution of \rhref{rhp:hm} by matrix functions that have singularities.  Assume that
\begin{align*}
\tilde v(z) = v(z)H(z) = \widehat m(z) S(z) H(z),
\end{align*}
where $S(z)$ is a, possibly singular, matrix function with $\det S(z) = 1$.  The product $S(z)H(z)$ is something we will know explicitly and be able to compute (\emph{i.e.}, $S(z)H(z) = \Delta^{-1}(z)Q(z)$ in  Section~\ref{S:dispersive}). Because $\widehat m(z)$ and $v(z)$ both satisfy the symmetry condition it follows that
\begin{align*}
S(z) = \begin{pmatrix}
0 & 1 \\ 1 & 0
\end{pmatrix} S(z^{-1}) \begin{pmatrix}
0 & 1 \\ 1 & 0
\end{pmatrix}.
\end{align*}
From \eqref{E:enforce-symm}, assuming $S(z)$ is analytic in a neighborhood of $z = 0$,
\begin{align*}
   v(0) &= Y \Psi(0)H^{-1}(0) = Y \Psi(\infty) H^{-1}(\infty) \begin{pmatrix} 0 & 1 \\ 1 & 0 \end{pmatrix},\\
    Y \Psi(0)H^{-1}(0)S^{-1}(0) &= Y \Psi(\infty) H^{-1}(\infty) \begin{pmatrix} 0 & 1 \\ 1 & 0 \end{pmatrix} S^{-1}(0),\\
     \widehat m(0) &= Y \Psi(0)H^{-1}(0)S^{-1}(0) = Y \Psi(\infty) H^{-1}(\infty) S^{-1}(\infty) \begin{pmatrix} 0 & 1 \\ 1 & 0 \end{pmatrix}.
 \end{align*}
  To compute $Y$ it suffices to know the product $S(z)H(z)$ at infinity and at zero.  Define $P(z) = S(z)H(z)$, $\Psi^1 \defeq \lim_{z\goto\infty} z(\Psi(z)-I)$ and $P^1 \defeq \lim_{z\goto\infty} z(P(z)-P(\infty))$. Then
\begin{align*}
\widehat m(z) &= Y\Psi(z)P^{-1}(z) =  Y\Psi(z^{-1})P^{-1}(z^{-1}) \begin{pmatrix}
0 & 1 \\ 1 & 0
\end{pmatrix} \\
&=  \left[ Y\Psi(\infty)P^{-1}(\infty) +   \left[ z Y  \Psi^1-P^{-1}(\infty)P^1 P^{-1}(\infty) \right] \right] \begin{pmatrix}
0 & 1 \\ 1 & 0
\end{pmatrix} + \bigo (z^2), ~~ z \to 0.
\end{align*}

 One should think of \rhref{rhp:v-full} as being an abstraction of \rhref{rhp:hm}. Our basic assumption is that the associated matrix {\RHP} to a non-singular deformation of a problem of the form of \rhref{rhp:v-full} is uniquely solvable.  This assumption is not violated in practice but one cannot rule out exceptional cases, see Remark~\ref{r:singular}.  To compute the solution of a vector {\RHP} with the normalization \eqref{E:normal}, we first deform the vector problem at hand, then solve the associated $2\times 2$ matrix {\RHP} and then use that matrix solution to construct the solution of the vector {\RHP}.  This process fails only if the solution of the associated matrix {\RHP} fails to exist.

\subsection{Extracting the solution of the Toda lattice}

The numerical procedure here returns an approximation of $m_{\sharp,\alpha}(z;n,t) = Y \Psi(z)$ for a matrix-valued function $\Psi(z)$ and a row vector $Y$. In a neighborhood of infinity $m(z;n,t) = m_{\sharp,\alpha}(z;n,t) Z(z) = Y \Psi(z) Z(z)$ for a (locally) analytic function $Z(z)$ such that $Z(\infty) = I$. From \eqref{product-recover} with
\begin{align*}
        \Psi(z) = I + \Psi_1z^{-1} + \bigo(z^{-1}), \quad Z(z) = I + Z_1z^{-1} + \bigo(z^{-1}), \quad z \to \infty,
\end{align*}
we have
\begin{align*}
m(z;n,t) = Y(I + (\Psi_1 + Z_1)z^{-1}) + \bigo(z^{-1}).
\end{align*}
This combined with Lemma~\ref{L:recover} and \eqref{E:Atoa} is enough to compute the solution of the Toda lattice.

\begin{remark} \label{R:symm}
In choosing $Z(\infty) = I$, which simplifies this calculation, we cannot perserve the symmetry condition $m(z;n,t) = m(z^{-1};n,t) \begin{pmatrix} 0 & 1 \\ 1 & 0 \end{pmatrix}$.
\end{remark}

\section{Numerical results}\label{S:numerics}
\subsection{Direct scattering}
In this section we present some numerical results on the computation of the scattering data.  We study two choices of initial data in detail:
\begin{itemize}
\item[(TS)]  A choice of initial data giving rise to two solitons (TS) is
  \begin{align*}
    a_n(0) &= \frac{1}{2} + \frac{4}{5} n e^{-n^2},\\
    b_n(0) &= \frac{1}{10} \sech(n).
  \end{align*}
  \item[(NS)] A choice of non-solitonic (NS) (\emph{i.e.}, $\sigma_{\text{pp}}(L) = \varnothing$) initial data is
  \begin{align*}
    a_n(0) &= \frac{1}{2} - \frac{1}{4} e^{-n^2},\\
    b_n(0) &= \frac{1}{10} \sech(n).
  \end{align*}
\end{itemize}

The reflection coefficient on $\mathbb T$ is shown in Figure~\ref{TS-R} for TS initial data and in Figure~\ref{NS-R} for NS initial data. With $K$ in Section~\ref{S:data} sufficiently large ($K = 30$ is sufficient), accuracy is guaranteed. In the case of the TS data, we find
\begin{align*}
  (\zeta_1,\zeta_2) &\approx (0.596142,-0.704859),\\
  (\gamma_1,\gamma_2) &\approx (3.25791,1.43054).
\end{align*}

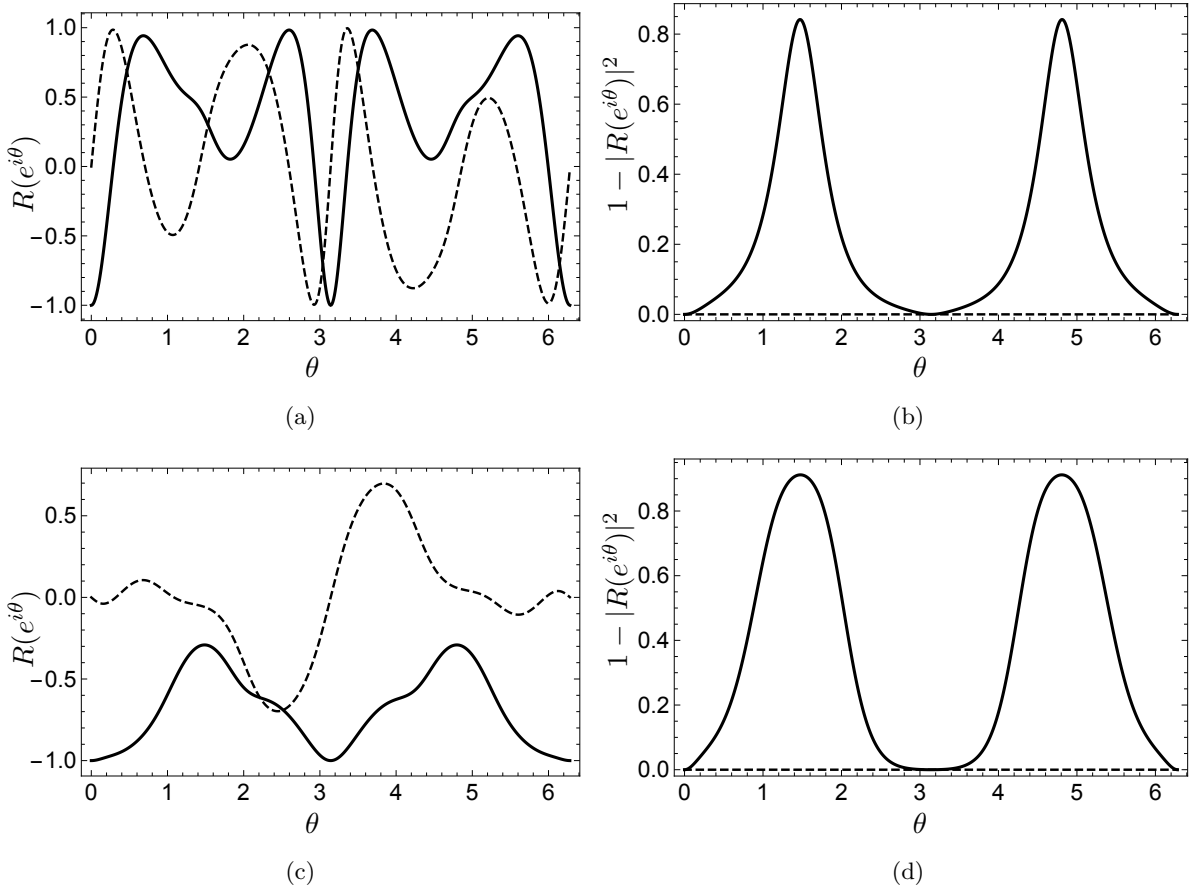
\begin{figure}[htp]
\subfigure[]{
  \begin{overpic}[width=0.4\textwidth]{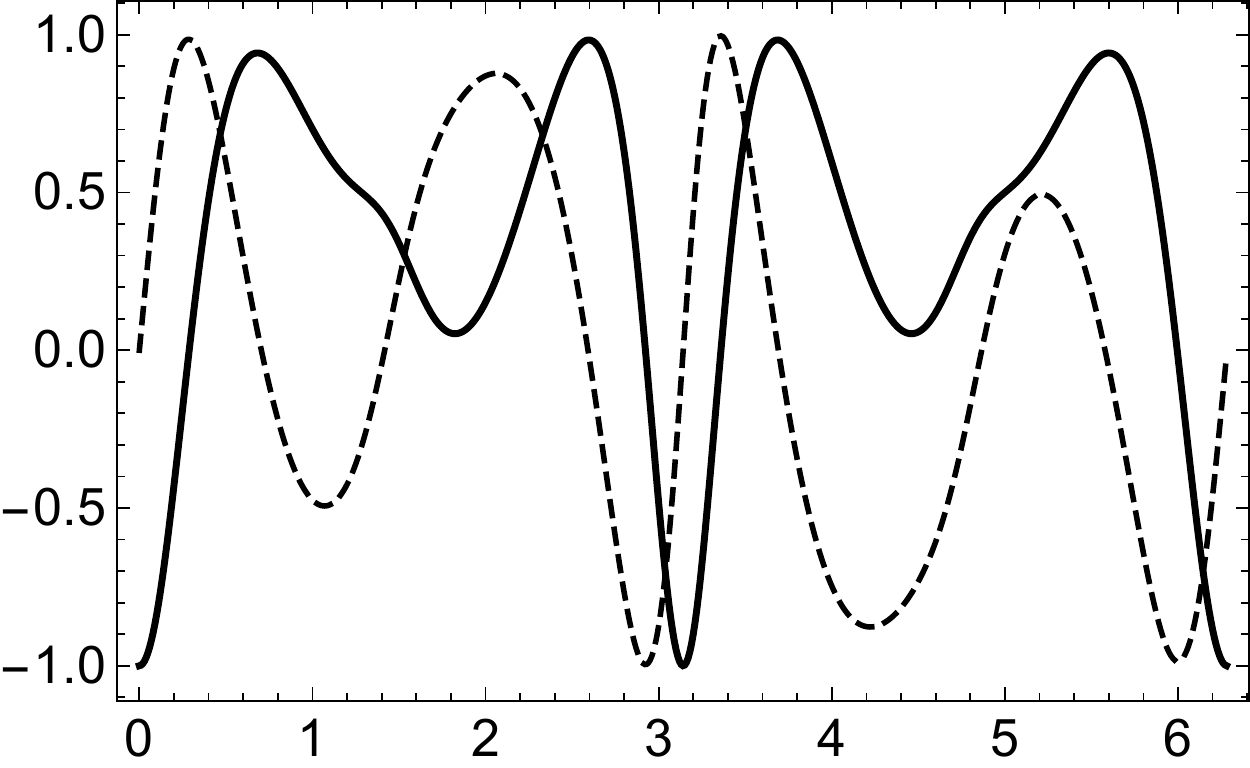}
    \put(50,-5){$\theta$}
    \put(-4,25){\rotatebox{90}{$R(e^{i \theta})$}}
  \end{overpic}\label{TS-R}
}\hspace{.15in}
\subfigure[]{
  \begin{overpic}[width=0.4\textwidth]{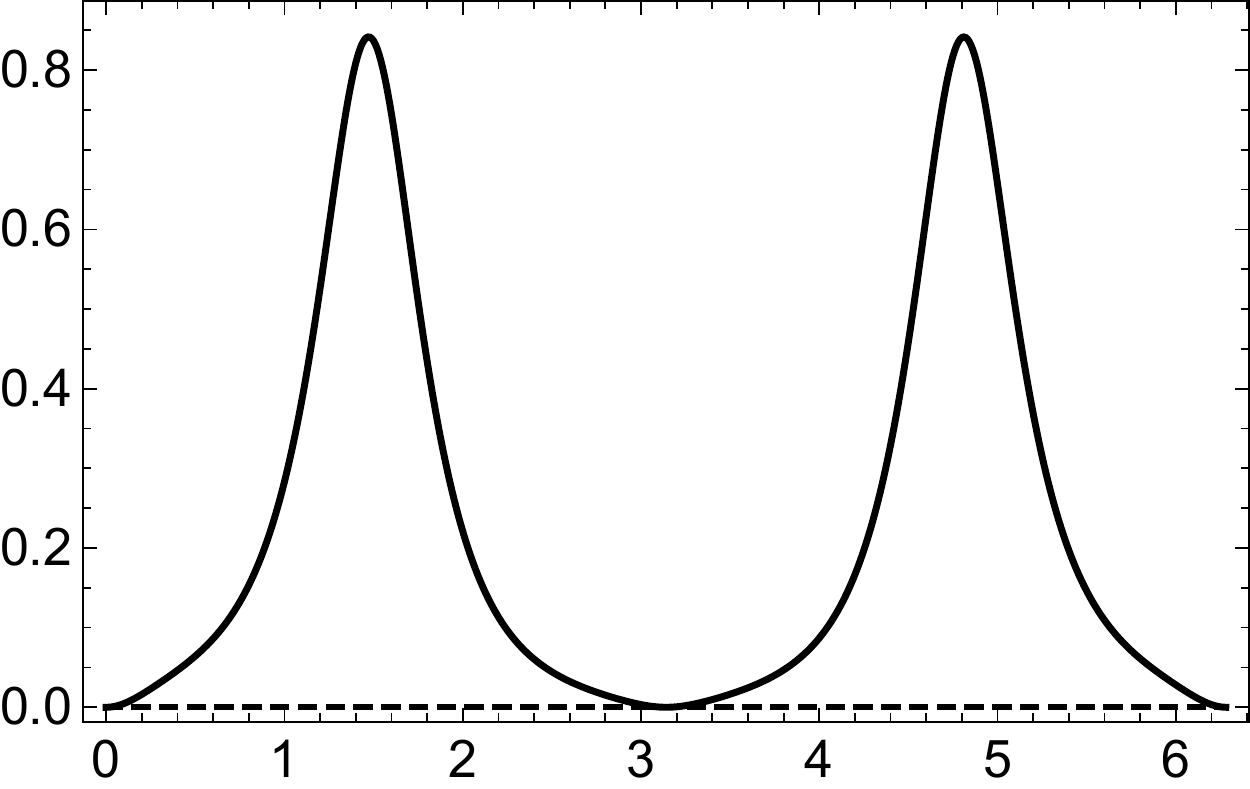}
    \put(50,-5){$\theta$}
    \put(-6.5,25){\rotatebox{90}{$1-|R(e^{i \theta})|^2$}}
  \end{overpic}
}
\subfigure[]{
 \begin{overpic}[width=0.4\textwidth]{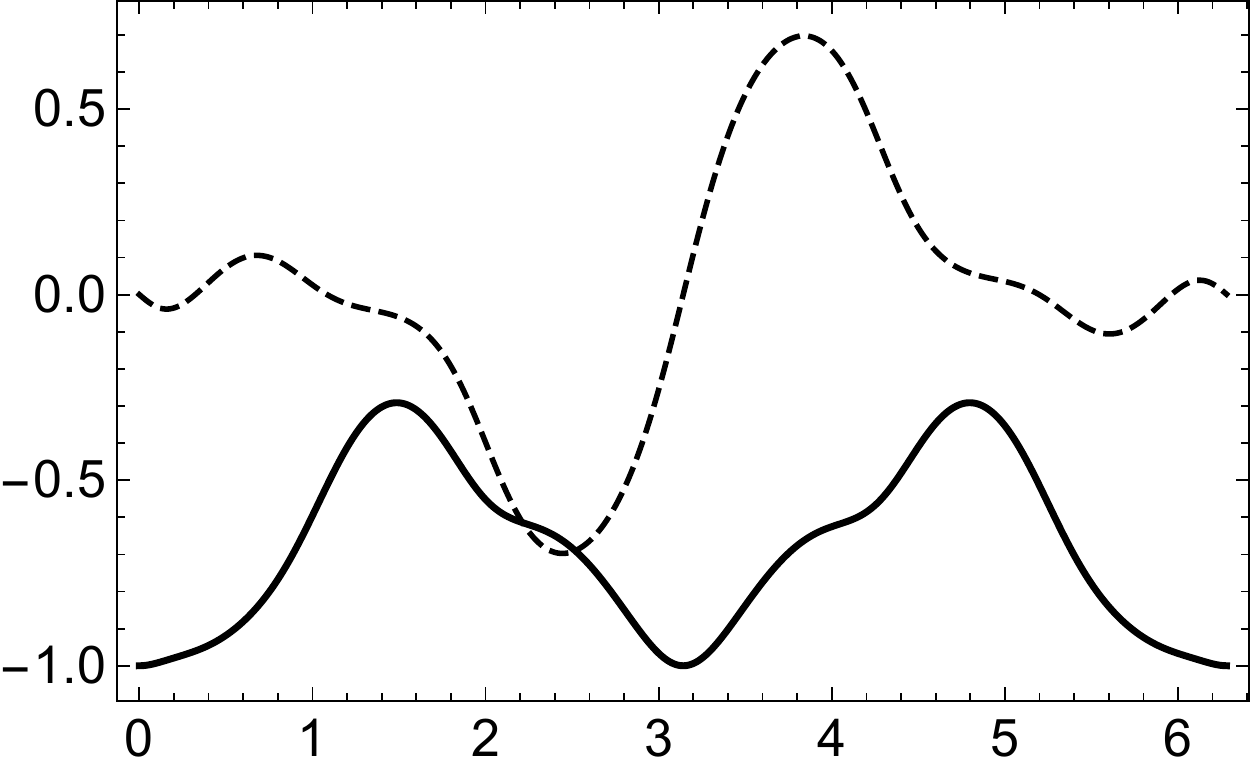}
    \put(50,-5){$\theta$}
    \put(-4,25){\rotatebox{90}{$R(e^{i \theta})$}}
  \end{overpic}\label{NS-R}
}\hspace{.15in}
\subfigure[]{
  \begin{overpic}[width=0.4\textwidth]{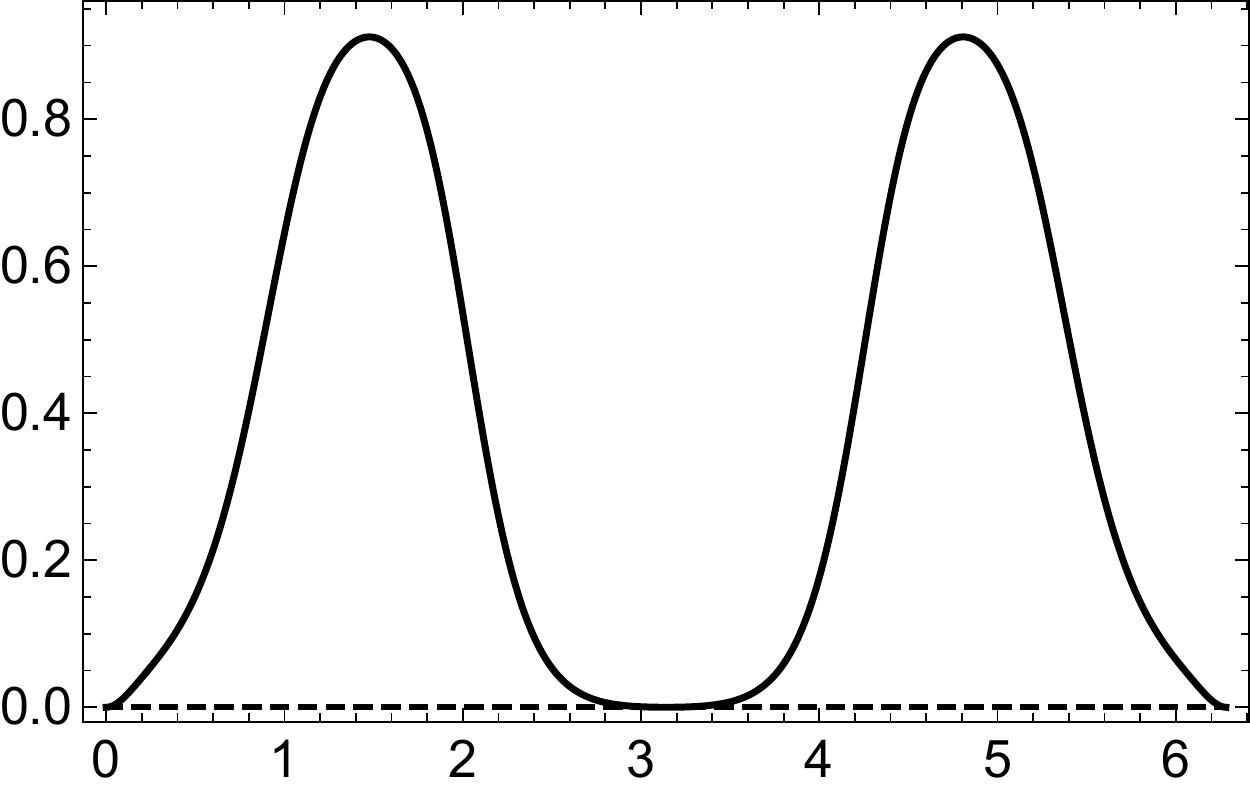}
    \put(50,-5){$\theta$}
    \put(-6.5,25){\rotatebox{90}{$1-|R(e^{i \theta})|^2$}}
  \end{overpic}
}

\caption{Numerical computation of the reflection coefficient $R(z)$ ($\Re R(z)$ solid curve, $\Im R(z)$ dashed graph) with initial data where (a) two eigenvalues are present (TS) and (c) $\sigma_{\text{pp}}(L)$ is empty (NS). We plot $1-|R(z)|^2$ for TS and NS initial data in (b) and (d), respectively. }

\label{F:reflection}
\end{figure}
\subsection{Inverse scattering}
In this section we present numerical results for the computation of the inverse scattering transform.  These results are of three flavors:
\begin{itemize}
\item Example solution plots,
\item error analysis, and
\item numerical asymptotics.
\end{itemize}

\subsubsection{Example solutions}
Here we present plots of the solution of the Toda lattice with both TS and NS initial data.  See Figures~\ref{F:Solution-DS-a} and \ref{F:Solution-DS-b} for plots of $a_n(t)$ and $b_n(t)$ when $ t= 100, 1000$ in the case of TS initial data.  Two solitons traveling in opposite directions are clearly visible.  Indeed, this is anticipated because $\zeta_1$ and $\zeta_2$ have opposite signs.  See Figures~\ref{F:Solution-NoS-a} and \ref{F:Solution-NoS-b} for plots of $a_n(t)$ and $b_n(t)$ when $t = 100, 1000, 10000$ in the case of NS initial data.  As stated above, no solitons are present in the solution and the high oscillation in the solution is apparent, especially at $t = 10000$.

\begin{figure}[htp]
  \subfigure[]{
    \begin{overpic}[width=0.8\textwidth]{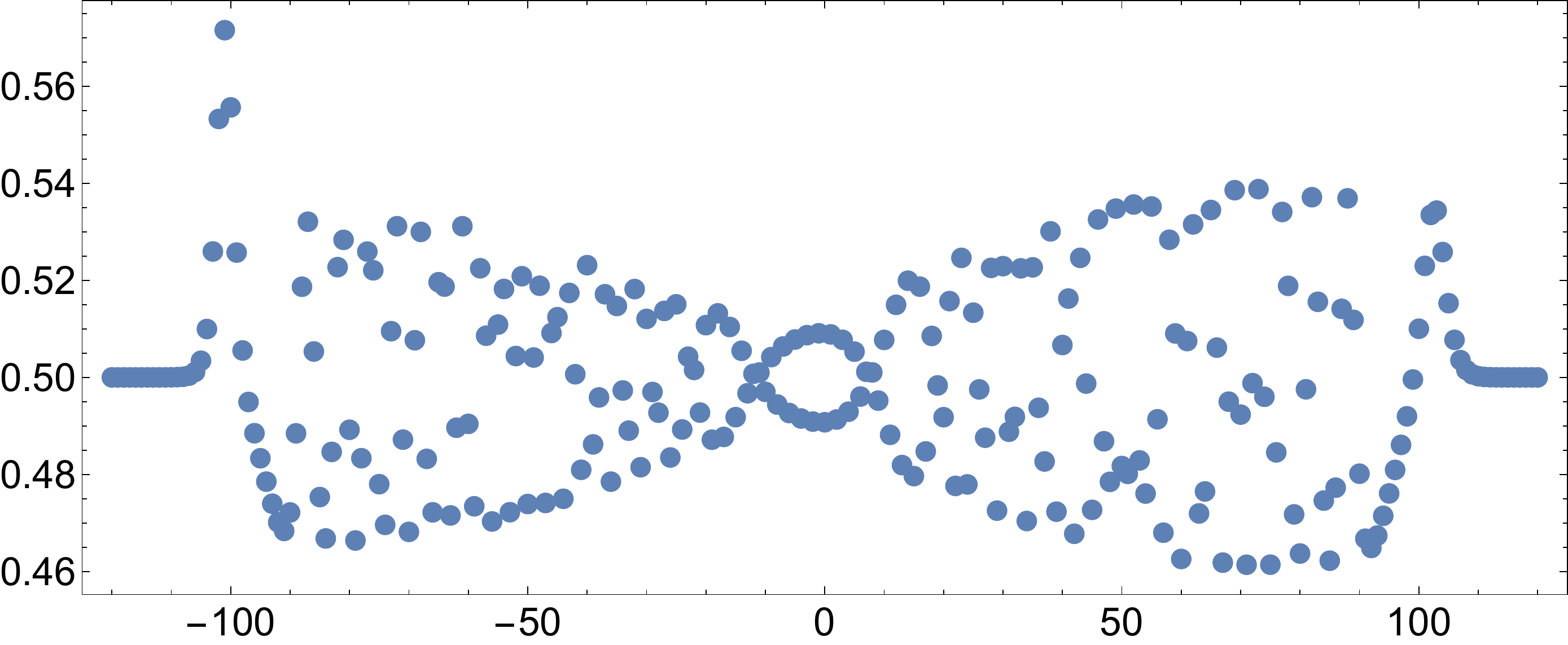}
    \put(52,-3){\large $n$}
    \put(-4,15){\rotatebox{90}{\large $a_n(100)$}}
  \end{overpic}
}
  \subfigure[]{
    \begin{overpic}[width=0.8\textwidth]{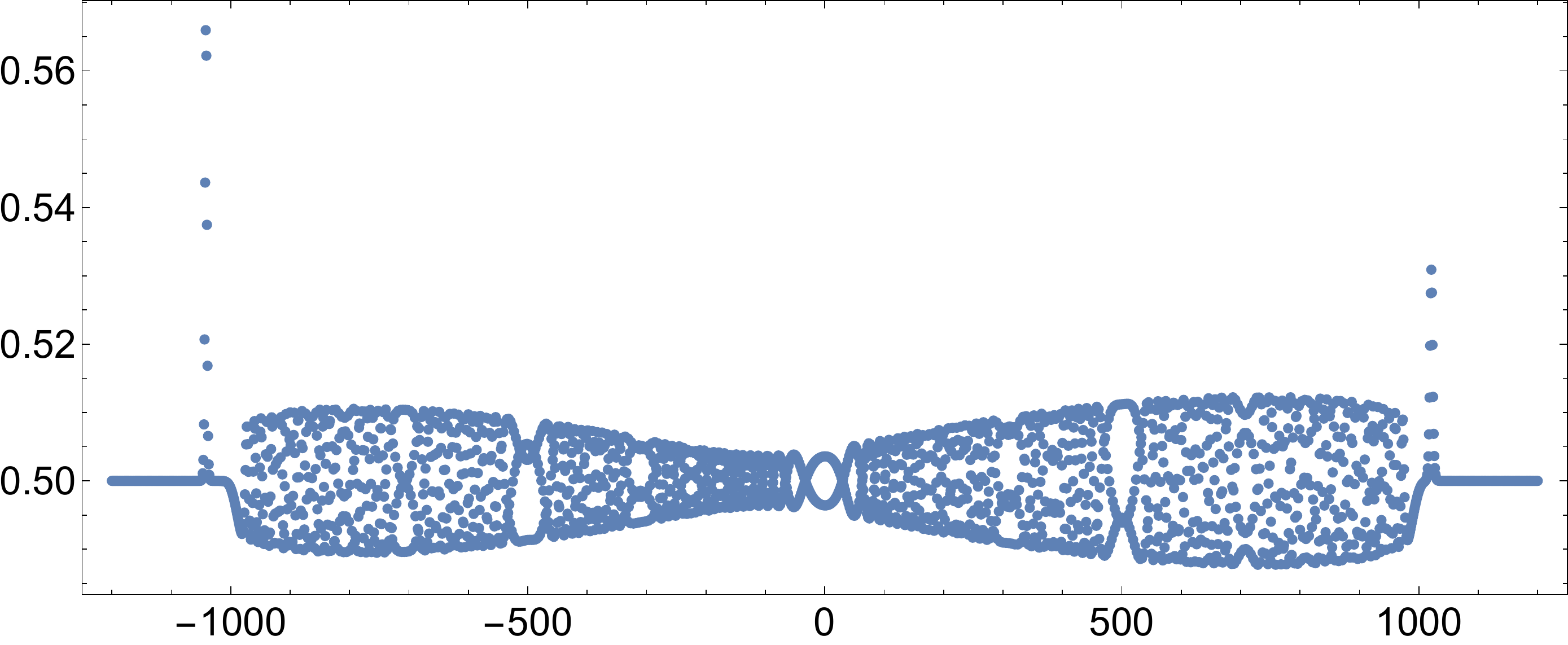}
    \put(52,-3){\large $n$}
    \put(-4,18){\rotatebox{90}{\large $a_n(1000)$}}
  \end{overpic}
}
\caption{Solution $a_n(t)$ obtained by numerical inverse scattering (a) at $t=100$, (b) at $t=1000$. The horizontal axis denotes the spatial parameter $n\in\mathbb{Z}$. The initial data produces two eigenvalues on opposite sides of the a.c.-spectrum giving rise to two solitons traveling in opposite directions.}
\label{F:Solution-DS-a}
\end{figure}

\begin{figure}[htp]
  \subfigure[]{
    \begin{overpic}[width=0.8\textwidth]{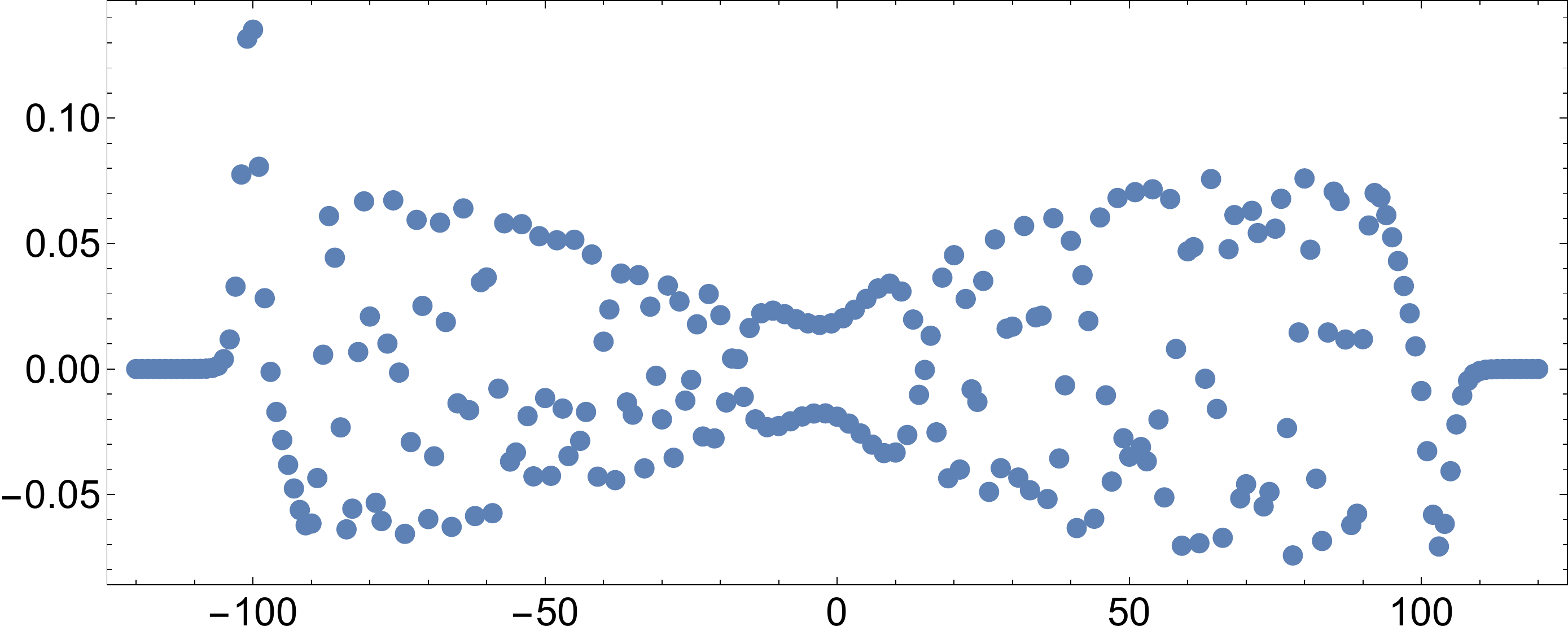}
    \put(52,-3){\large $n$}
    \put(-4,15){\rotatebox{90}{\large $b_n(100)$}}
  \end{overpic}
}
  \subfigure[]{
    \begin{overpic}[width=0.8\textwidth]{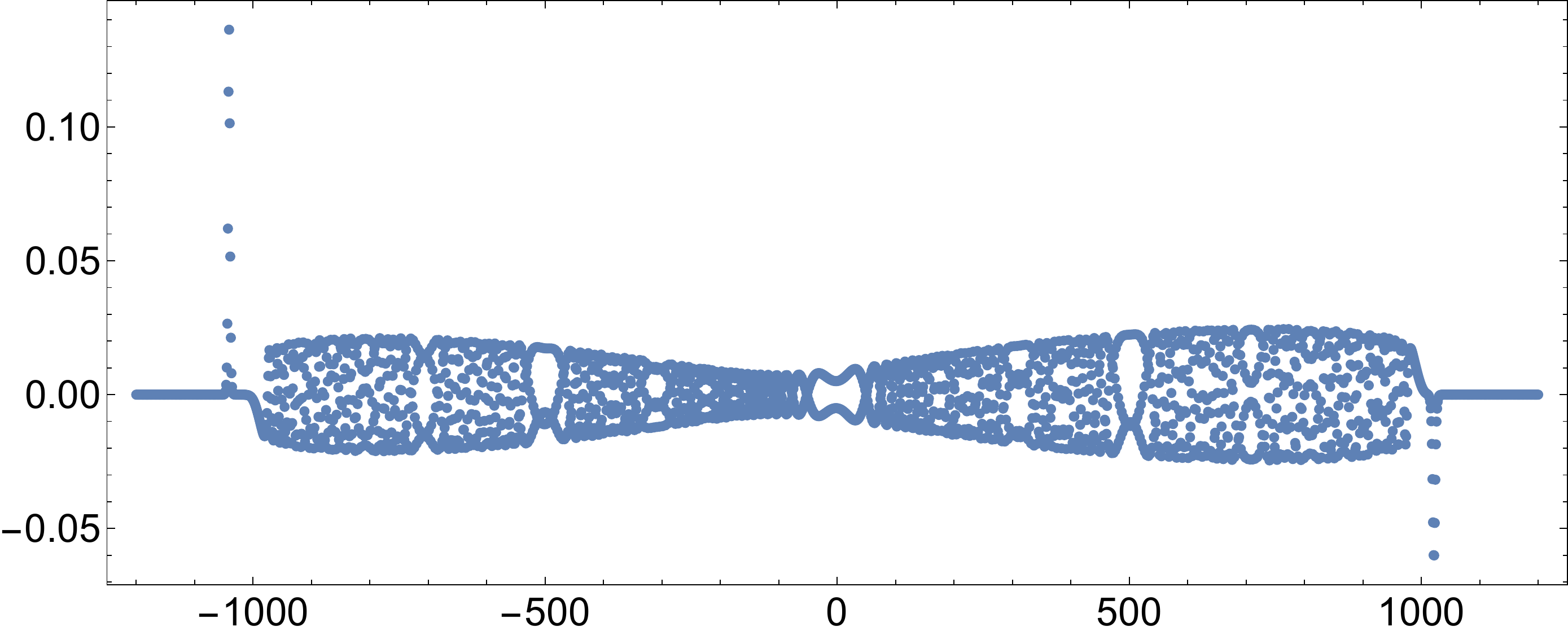}
    \put(52,-3){\large $n$}
    \put(-4,18){\rotatebox{90}{\large $b_n(1000)$}}
  \end{overpic}
}
\caption{Solution $b_n(t)$ obtained by numerical inverse scattering (a) at $t=100$, (b) at $t=1000$. The horizontal axis denotes the spatial parameter $n\in\mathbb{Z}$. The initial data produces two eigenvalues on opposite sides of the a.c.-spectrum giving rise to two solitons traveling in opposite directions.}
\label{F:Solution-DS-b}
\end{figure}

\begin{figure}[htp]
  \subfigure[]{
    \begin{overpic}[width=0.8\textwidth]{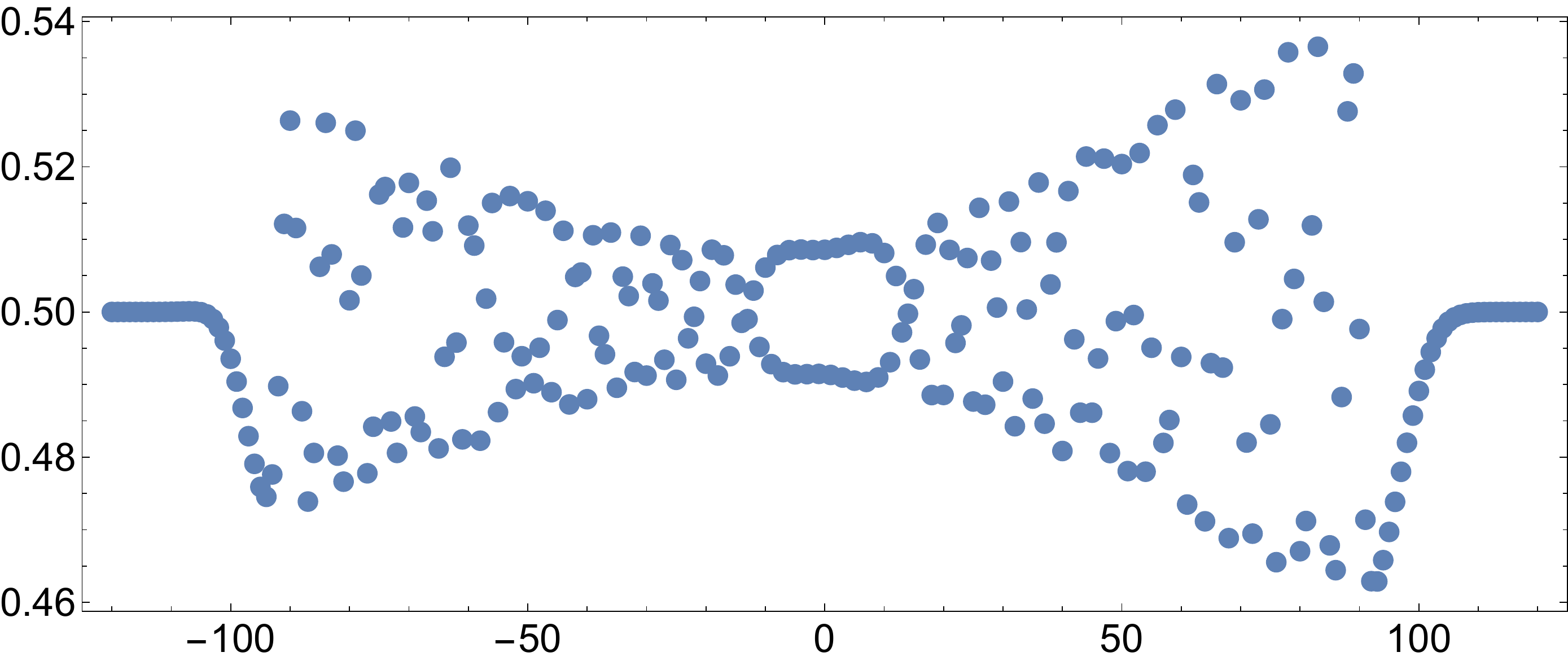}
    \put(52,-3){\large $n$}
    \put(-4,15){\rotatebox{90}{\large $a_n(100)$}}
  \end{overpic}
}
  \subfigure[]{
    \begin{overpic}[width=0.8\textwidth]{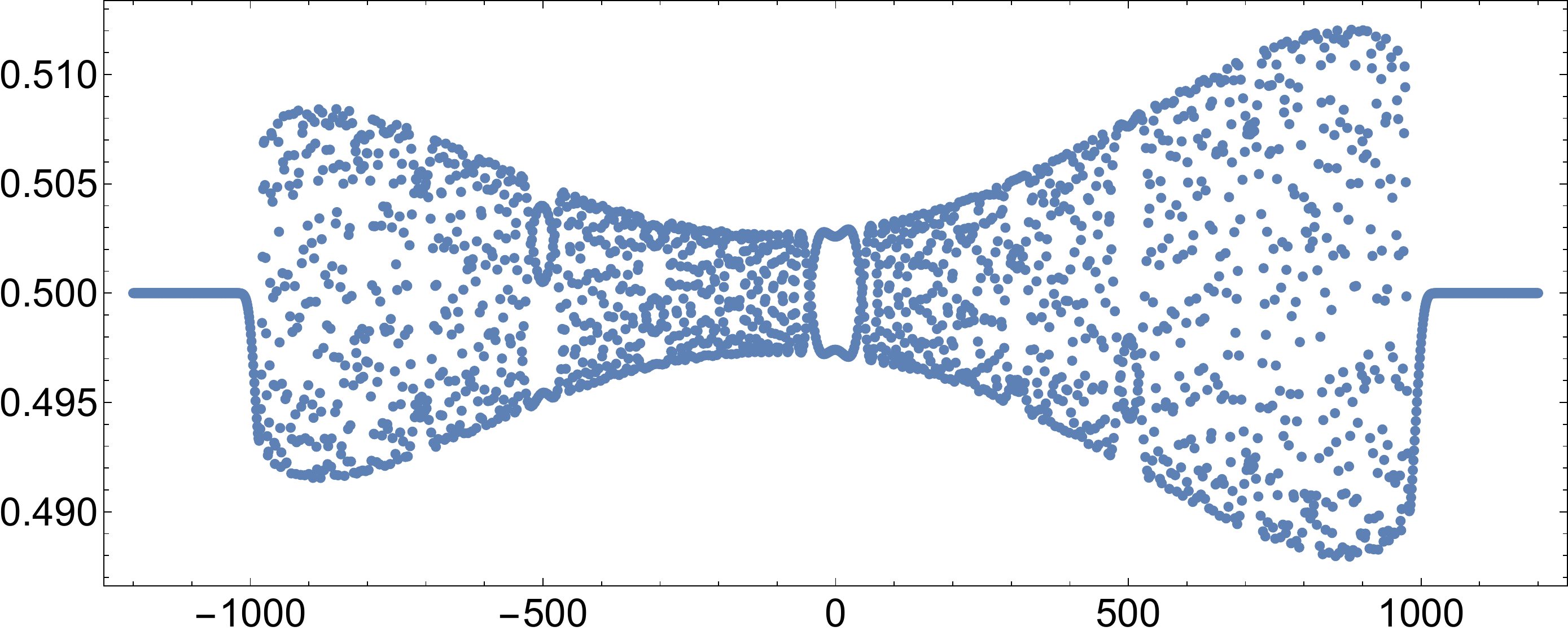}
    \put(52,-3){\large $n$}
    \put(-4,15){\rotatebox{90}{\large $a_n(1000)$}}
  \end{overpic}
}
  \subfigure[]{
    \begin{overpic}[width=0.8\textwidth]{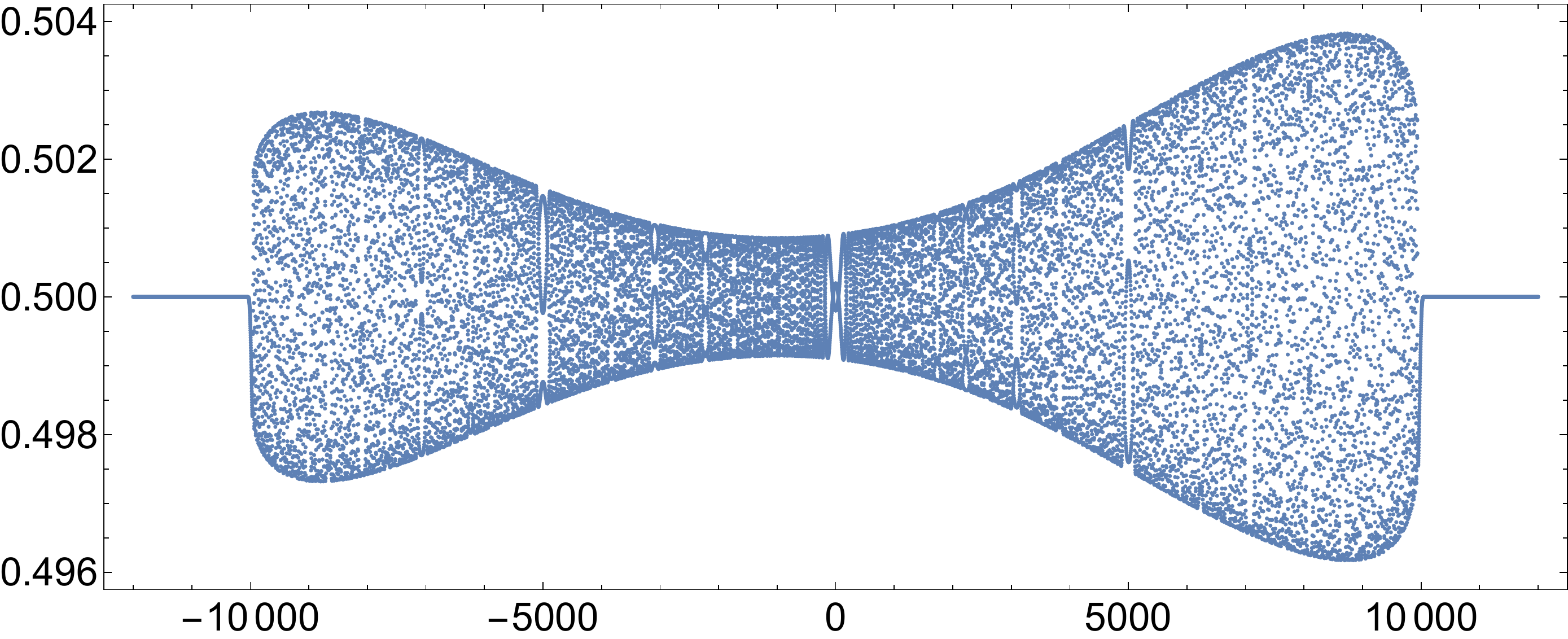}
    \put(52,-3){\large $n$}
    \put(-4,15){\rotatebox{90}{\large $a_n(10000)$}}
  \end{overpic}
}
\caption{Solution $a_n(t)$ obtained by numerical inverse scattering (a) at $t=100$, (b) at $t=1000$, (b) at $t=10000$. The horizontal axis denotes the spatial parameter $n\in\mathbb{Z}$. The initial data produces no eigenvalues and therefore no solitons are present in the solution.}
\label{F:Solution-NoS-a}
\end{figure}
\begin{figure}[htp]
 \subfigure[]{
    \begin{overpic}[width=0.8\textwidth]{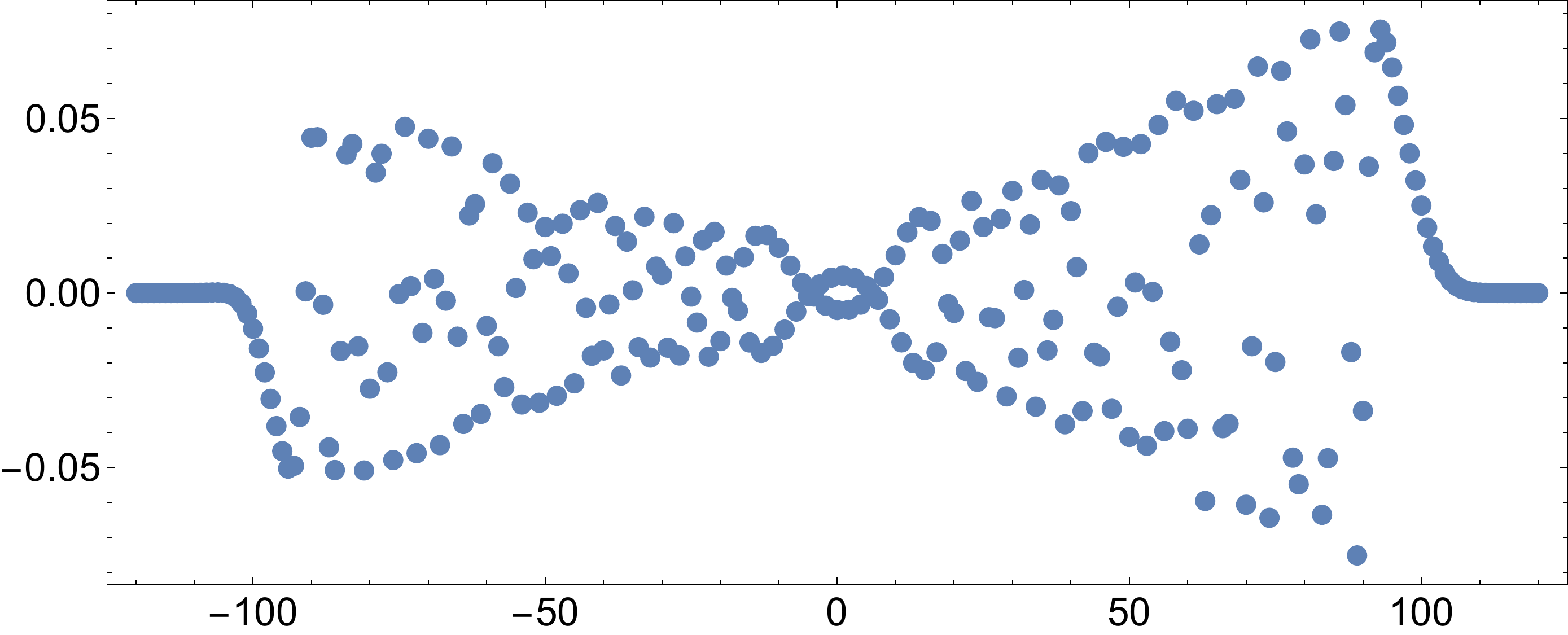}
    \put(52,-3){\large $n$}
    \put(-4,15){\rotatebox{90}{\large $b_n(100)$}}
  \end{overpic}
}
  \subfigure[]{
    \begin{overpic}[width=0.8\textwidth]{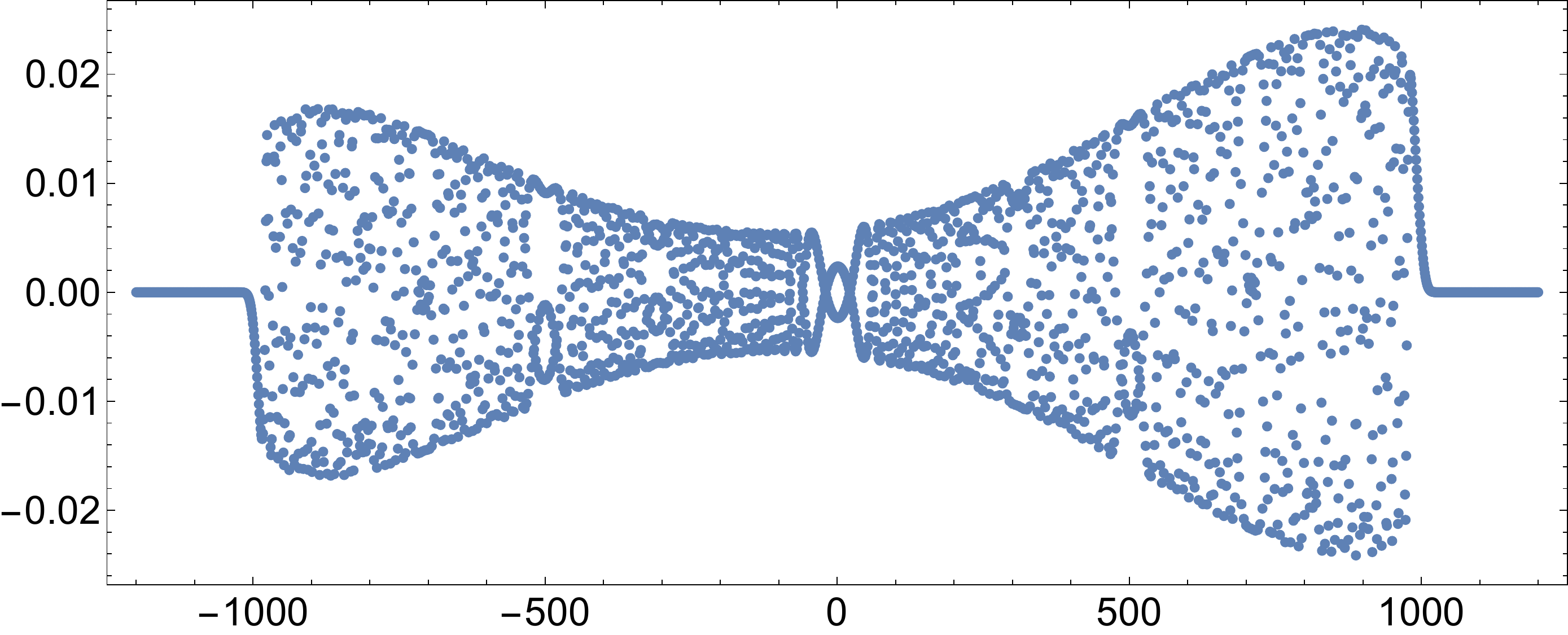}
    \put(52,-3){\large $n$}
    \put(-4,15){\rotatebox{90}{\large $b_n(1000)$}}
  \end{overpic}
}
  \subfigure[]{
    \begin{overpic}[width=0.8\textwidth]{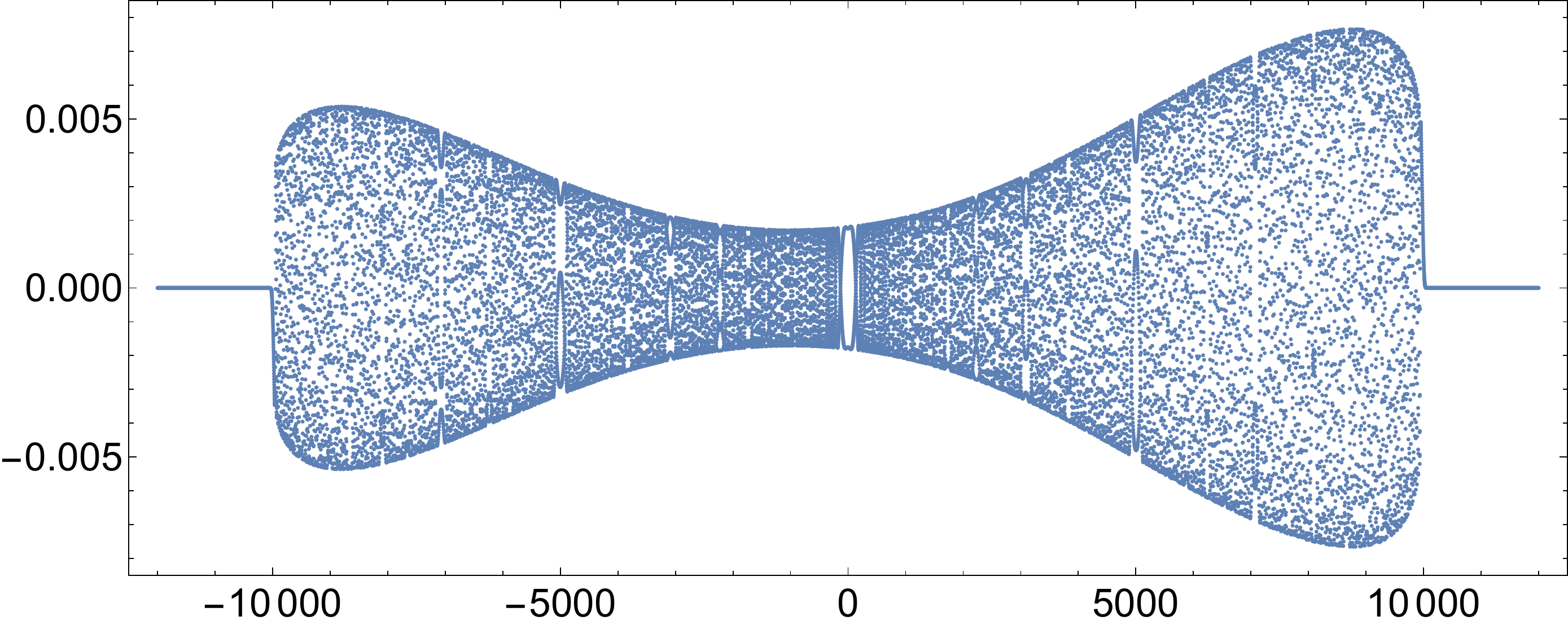}
    \put(52,-3){\large $n$}
    \put(-4,15){\rotatebox{90}{\large $b_n(10000)$}}
  \end{overpic}
}
\caption{Solution $b_n(t)$ obtained by numerical inverse scattering (a) at $t=100$, (b) at $t=1000$, (b) at $t=10000$. The horizontal axis denotes the spatial parameter $n\in\mathbb{Z}$. The initial data produces no eigenvalues and therefore no solitons are present in the solution.}
\label{F:Solution-NoS-b}
\end{figure}

\subsubsection{Error analysis}

To examine the accuracy of our numerical inverse scattering transform (IST) \emph{a posteriori} we compare it with a naive time-stepping method. A more detailed description of a related method can be found in \cite{BN}.   {Here we just use out-of-the-box Runge--Kutta 4.}  Fix $K > 0$ and consider the Toda lattice with Dirichlet boundary conditions: $a_{-K}(t) = a_K(t) = 1/2$, $b_{-K} = b_{K} = 0$.  Here $K$ is chosen sufficiently large so that the solution remains flat at the boundary $\pm K$ for all times simulated.  If $t < T$, it suffices to take $K = c T$ for some $c > 1$ which is larger than the speed of the fastest soliton present in the solution.  This is a finite-dimensional system of ODEs and can be integrated in time using the fourth-order Runge--Kutta method.  We compare the time-stepped solution at $t = 30$ with $\Delta t = 10^{-5}$ with the solution computed via the inverse scattering transform in Figure~\ref{F:Error}. As the number of collocation points in increased, the numerical inverse scattering solution converges exponentially to the true solution, with these errors saturating at approximately $10^{-11}$.  It is reasonable to expect that at this point the numerical IST gives a more accurate solution than the time-stepping method.   Furthermore, on a standard laptop it takes $\approx$ 6 seconds to compute the solution at $t = 30$, $n = 30$ with 720 collocation points using the IST and $\approx$ $2\times 10^5$ seconds with the naive time-stepping method\footnote{It should be noted that the time-stepping method produces an approximation of the entire solution profile in this time while the numerical IST gives the solution at only one point. Even so, it would take $\approx 1.2 \times 10^4$ second to compute the entire solution profile at this rate with the numerical IST.} {implemented in {\tt Mathematica}.  Presumably, by using more efficient integrators and software packages this time can be reduced by at least an order of magnitude but when computing at sufficiently long times, the numerical IST is guaranteed to have a shorter runtime.}  This comparison is pessimistic as contour truncation (see Figure~\ref{F:truncate} and the next paragraph) reduces IST computation times as $t$ increases while time-stepping methods see their complexity increase.

We emphasize that the number of collocation points required to solve the Toda lattice using the numerical inverse scattering transform to a given accuracy is typically decreasing with respect to $t$.  This is because the deformations performed on the original {\RHP} force the jump matrices on some contours to tend exponentially fast to the identity matrix as $t$ increases.  Thus, after truncation, discussed in Section~\ref{S:numerics}, fewer contours need to be discretized in the {\RHP} resulting in fewer collocation points.  To see this in action, using the contour truncation algorithm described in \cite{Gauss}, see Figure~\ref{F:truncate}.  Thus the errors seen in Figure~\ref{F:Error} are pessimistic for large values of $t$.

\begin{remark}\label{r:singular}
Given a uniquely solvable vector {\RHP} such as \rhref{rhp:v}, it does not follow that the associated matrix solution exists.  It is well-known (see \cite{ZhouIS}, for example) that the operator $\mathcal C[J;\Gamma]$ associated with \rhref{rhp:hm} is Fredholm with index zero.  Furthermore $\mathcal C[J^{-1};\Gamma]\mathcal C[J;\Gamma] = I + K$ where $K$ is a compact operator.  For fixed $n$, $K = K(t)$ is analytic in $t$ and then by the analytic Fredholm theorem $\mathcal C[J;\Gamma]$ is either never invertible or invertible on the compliment of a discrete set values of $t$.  From the work of \cite{DKKZ}, it can be deduced that for each fixed $n$, there exists $t^*(n)> 0$ such that if $t > t^*(n)$ then $\mathcal C[J;\Gamma]$ is invertible.  Therefore, we know that there is only a discrete set of possible values of $(n,t)$ where $\mathcal C[J;\Gamma]$ may fail to have an inverse.  This is discussed in Appendix~\ref{A:solve}.

To investigate the possibility of encountering a point $(n,t)$ we perform the following computation.  Fix $t \geq 0$ and vary $n$.  For each value of $(n,t)$ we plot the smallest singular value of the discretization of $\mathcal C[J;\Gamma]$ found using the numerical method described in this section.  It is known that $\mathcal C[J;\Gamma]$ is always invertible when no solitons are present --- when no poles are present in \rhref{rhp:m}.  In Figure~\ref{f:singular} we perform this experiment for initial data with and without solitons.  Because $n$ is discrete, and we will only ever evaluate at discrete times, we do not ever expect to encounter a singular operator $\mathcal C[J;\Gamma]$.
\end{remark}
\begin{figure}[htp]
\centering
\subfigure[]{\includegraphics[width=.49\linewidth]{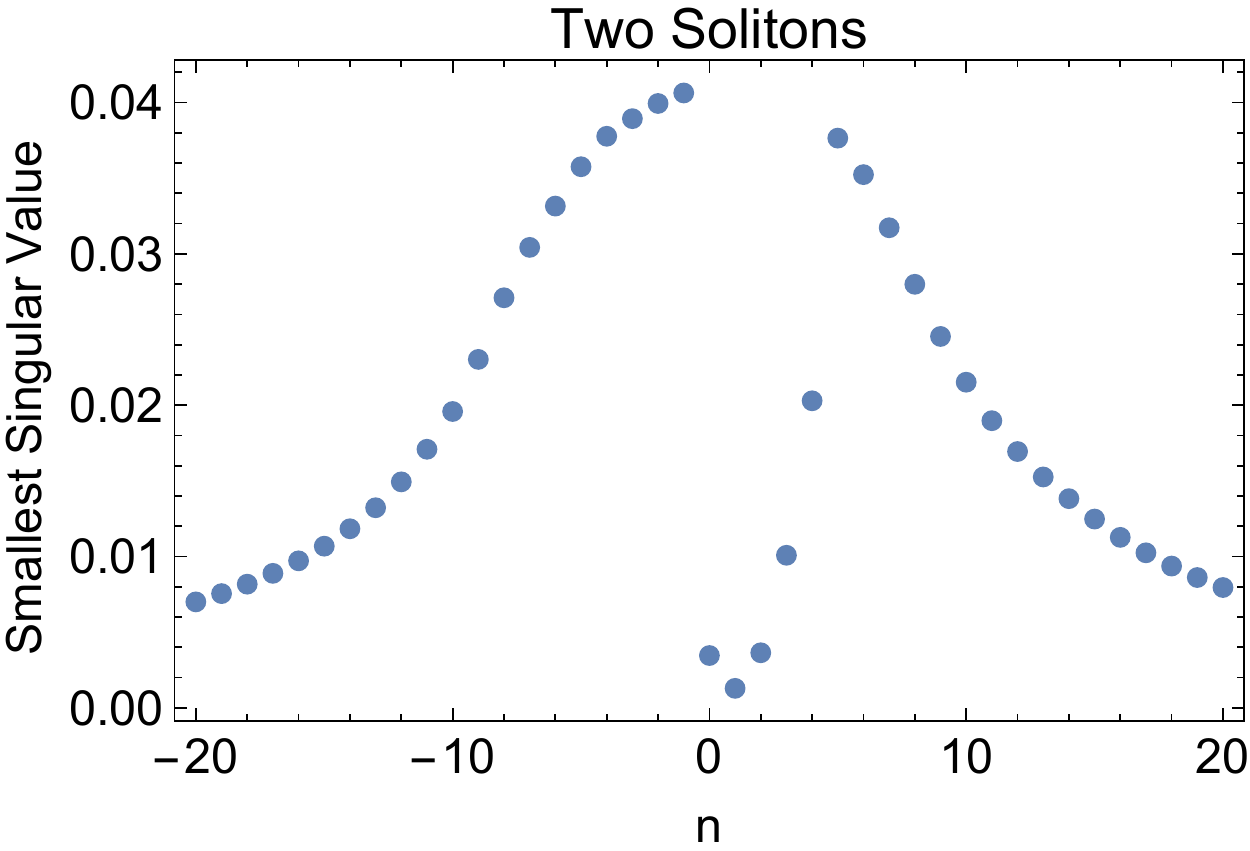}}
\subfigure[]{\includegraphics[width=.49\linewidth]{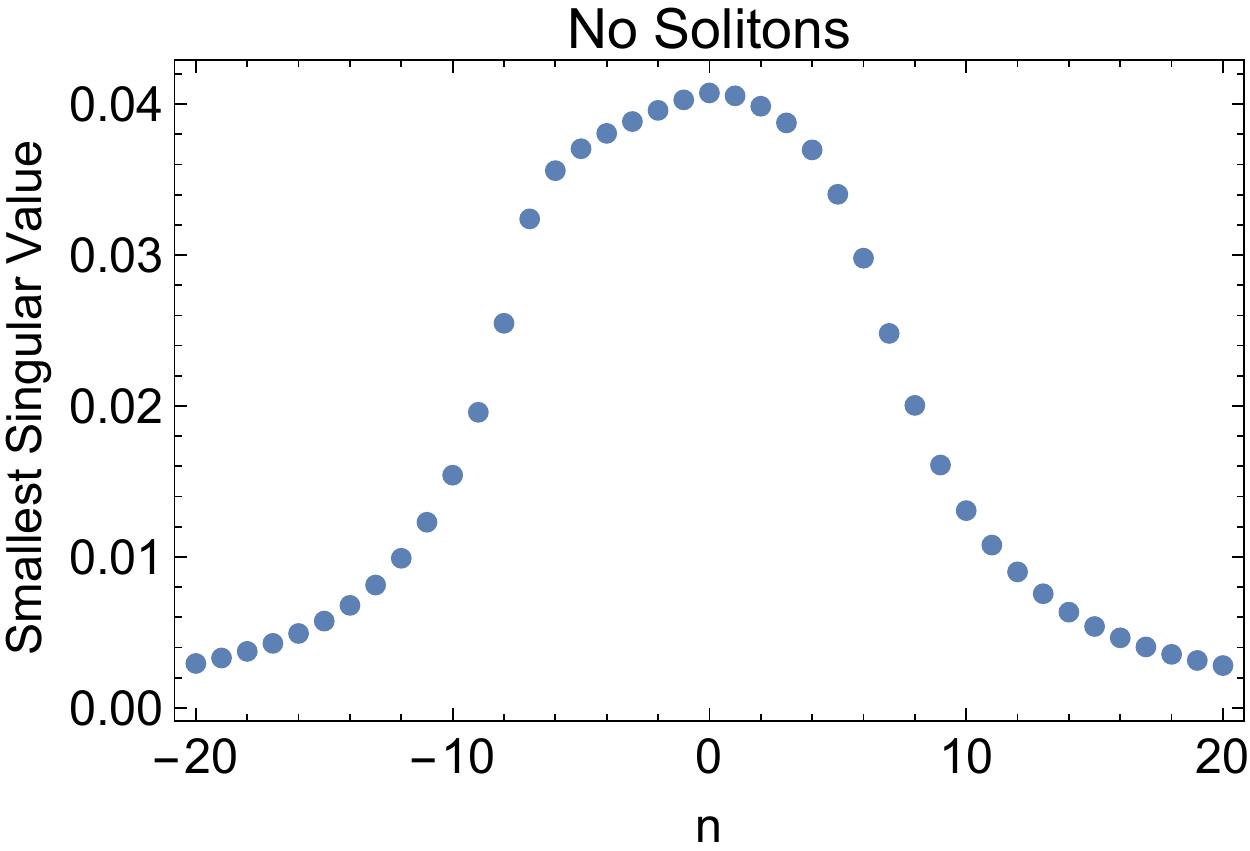}}
\caption{\label{f:singular} Plotting singular values for $t = 0$.  (a) A plot of the singular values of the discretization of the operator $\mathcal C[J;\Gamma]$from \rhref{rhp:hm} when the initial data produces two solitons.  (b) A plot of the singular values of the discretization of the operator $\mathcal C[J;\Gamma]$ from \rhref{rhp:hm} when the initial data produces no solitons.  In comparing these two plots, it is clear that the smallest singular value abruptly approaches zero, but due to the discrete nature, a zero singular value is not encountered.  In panel (a), the minimum at $n = 1$ is $\approx 0.0013$.  For the exact form of the initial data see (NS) and (TS) initial data in Section~\ref{S:numerics}.}
\end{figure}

\begin{figure}[htp]
  \begin{overpic}[width=0.5\textwidth]{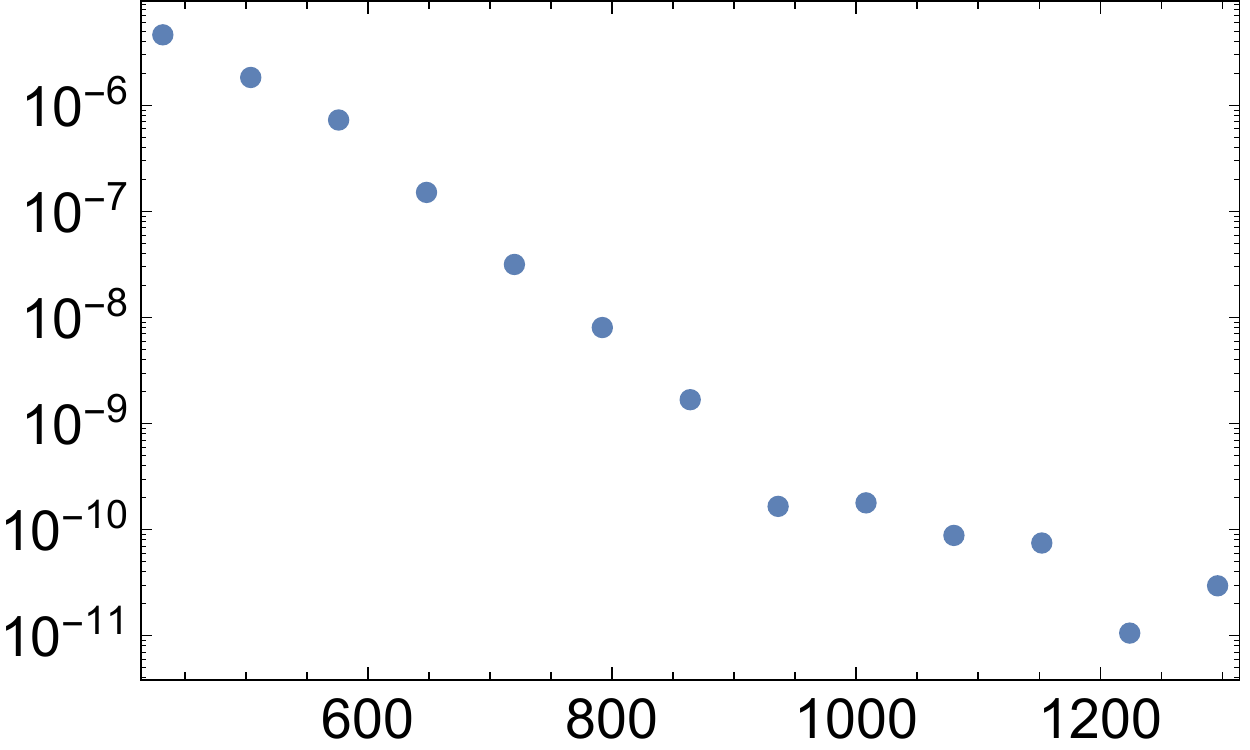}
    \put(45,-5){\large Collocation points}
    \put(-4,25){\large \rotatebox{90}{Error}}
  \end{overpic}
\vspace{.1in}

\caption{\label{F:Error} The error between a solution of the Toda lattice computed using time stepping and numerical inverse scattering.  The error shown is the maximum of the error computed at the points $(n,t) = (-51,30), (-31,30), (-11,30)$ plotted versus the number of collocation points used in the numerical solution of the associated {\RHP}. }
\end{figure}

\begin{figure}[htp]
  \subfigure[]{\includegraphics[width=.45\linewidth]{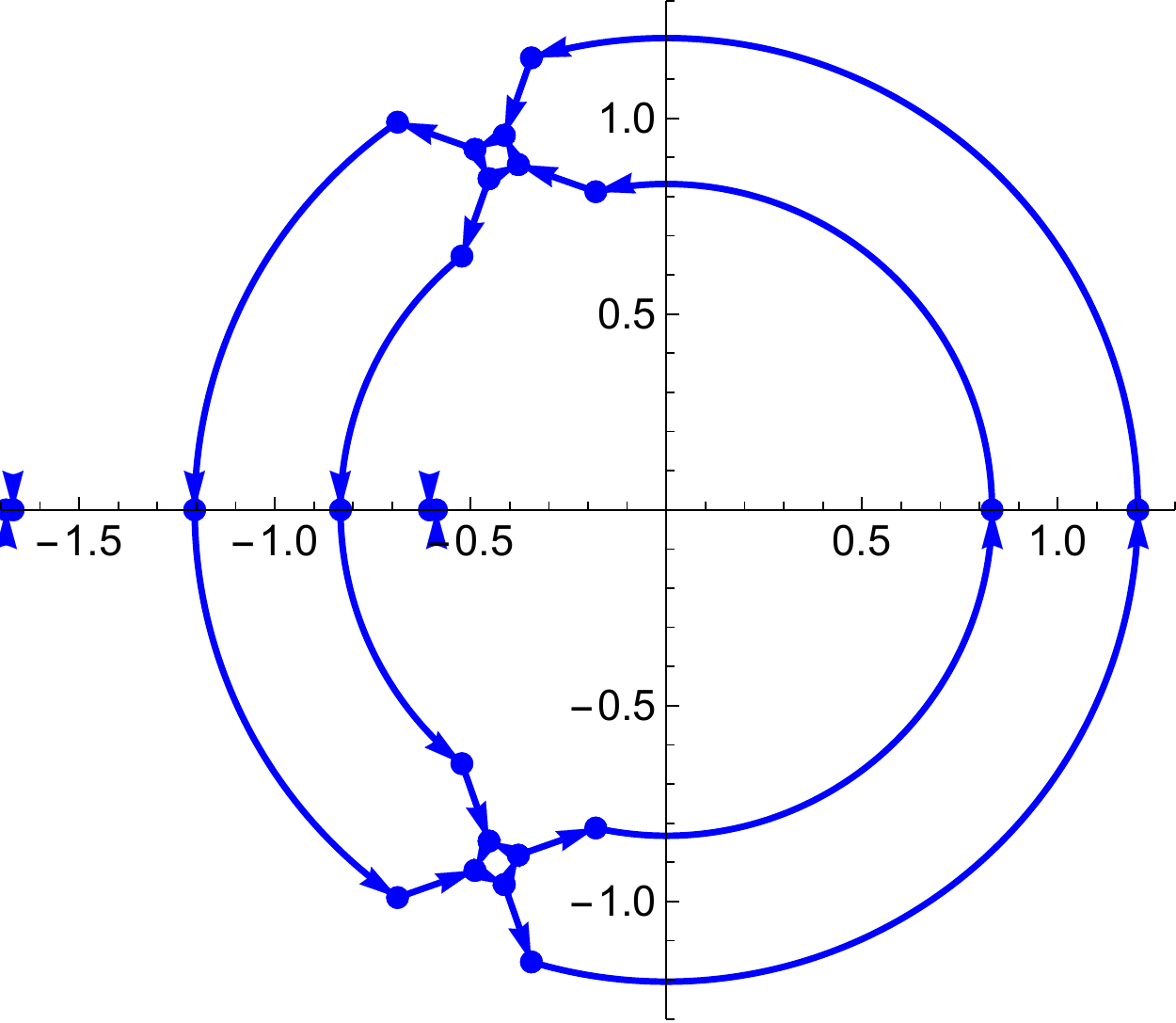}}
  \subfigure[]{\includegraphics[width=.45\linewidth]{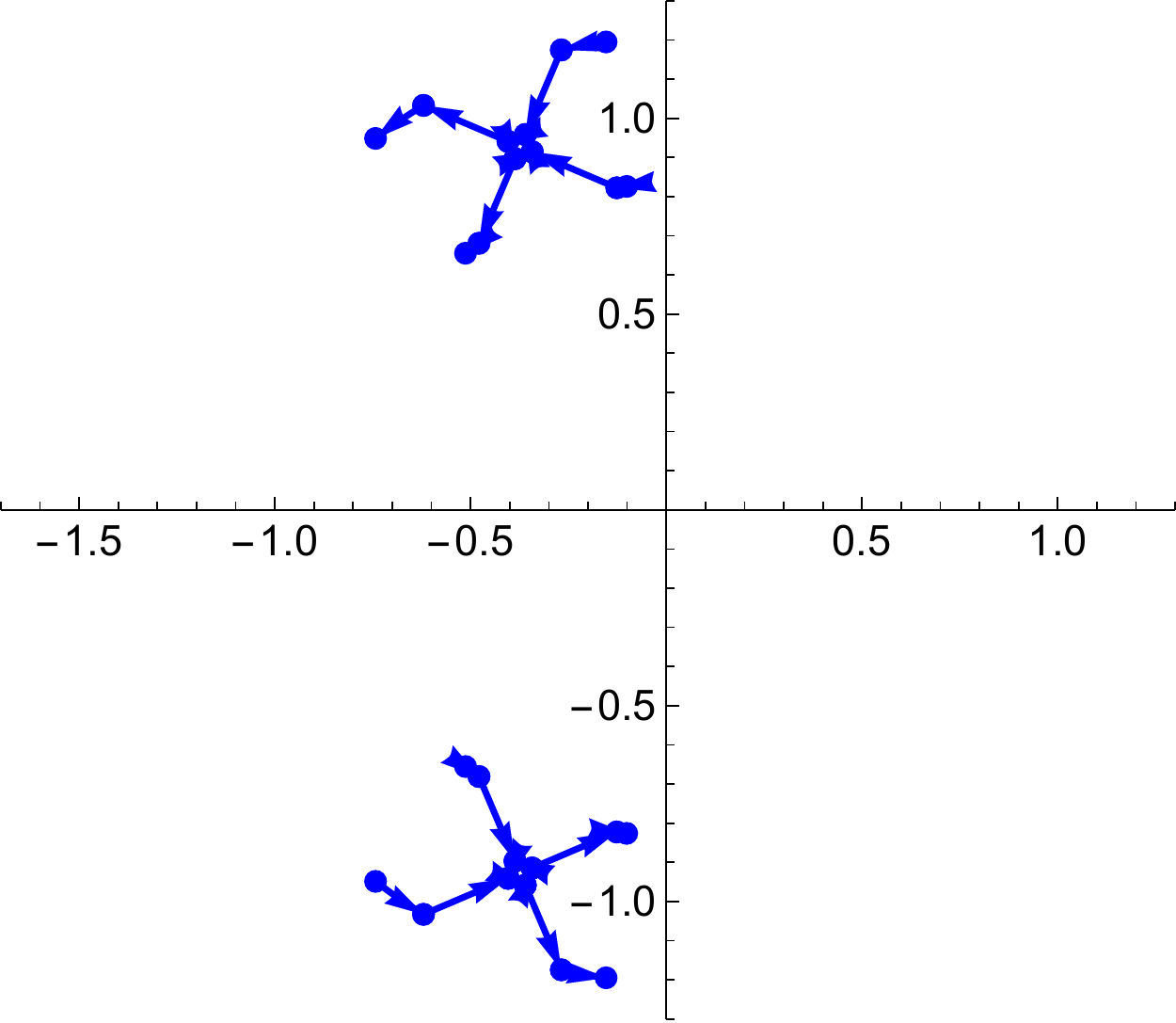}}
  \subfigure[]{\includegraphics[width=.45\linewidth]{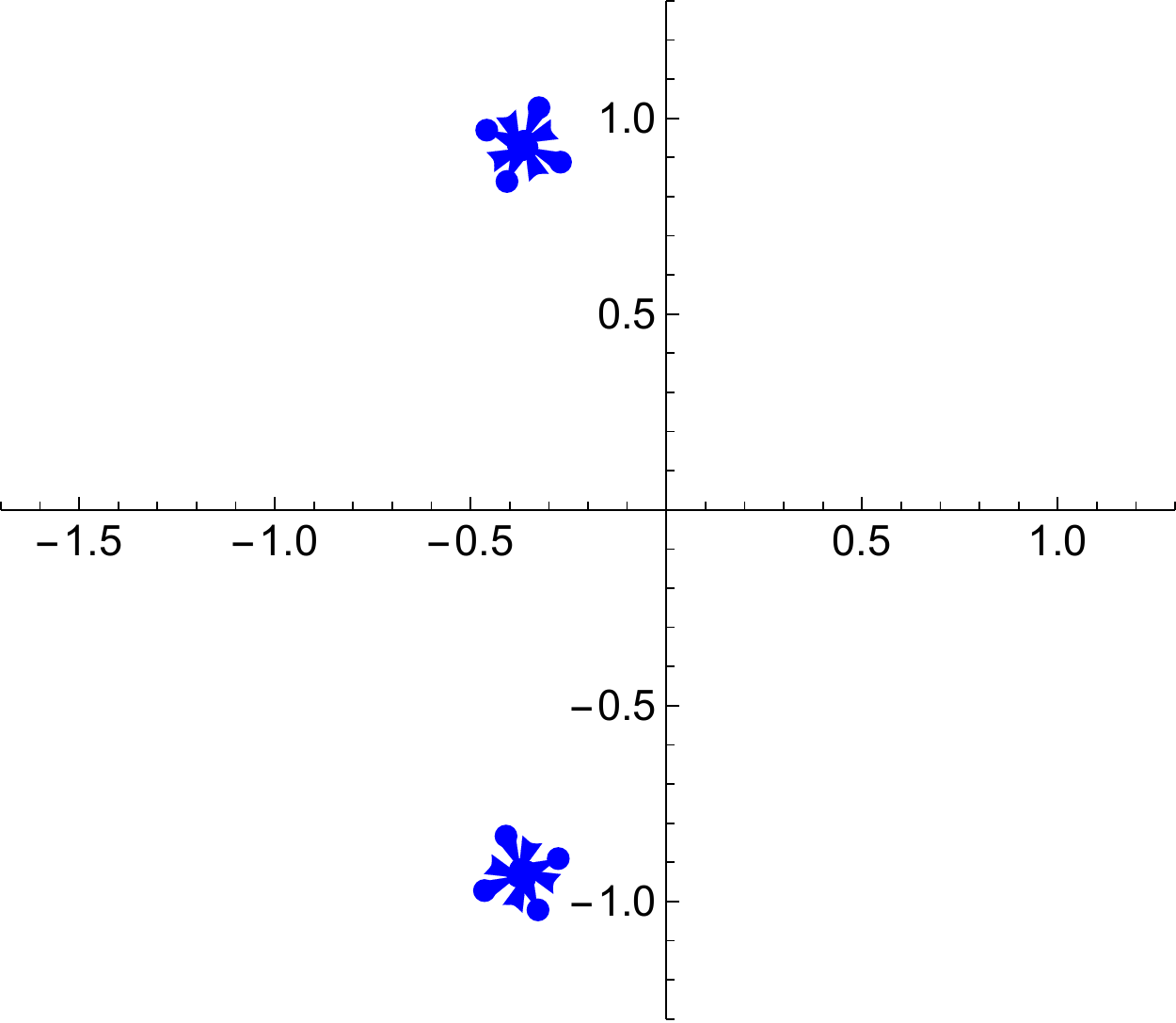}}
  \caption{Contours used in numerical inverse scattering for the Toda lattice in the dispersive region with TS initial data.  This figure shows how contour truncation localizes the contours as $t \goto \infty$.  (a) The contours for $(n,t) = (-11,30)$ with 28 total contours.  (b) The contours for $(n,t) = (-110,300)$ with 24 total contours (soliton contours have been truncated). (c) The contours for $(n,t) = (-1100,3000)$ with 14 total contours.}
 \label{F:truncate}
\end{figure}

\subsubsection{Numerical asymptotics}
We have seen that the accuracy of the numerical method is easily verified for short/moderate times by the comparison with time stepping methods. Thus in terms of accuracy, the method will also out-perform the asymptotic formulae (\cite{Tes, Tesch02, Kam}) for long-time. {Thus, with a numerical inverse scattering transform the numerical evaluation of asymptotic formulae (which is truly a non-trivial task) is no longer necessary in many cases.}  To demonstrate this, we show the long-time behavior of the solution with NS initial data in the dispersive, Painlev\'e and collisionless shock regions.  We call such computations \emph{numerical asymptotics}.
\begin{itemize}
\item {Dispersive region}.\newline
  To show the solution in the dispersive region we let $n$ depend on $t$ through $n = \lfloor 7t/10 \rfloor$ where $\lfloor \cdot \rfloor$ represents the integer part.  See Figure~\ref{F:NA-disp} for a plot of the solution into the dispersive region.
  \begin{figure}[h]
    \begin{overpic}[width=0.4\textwidth]{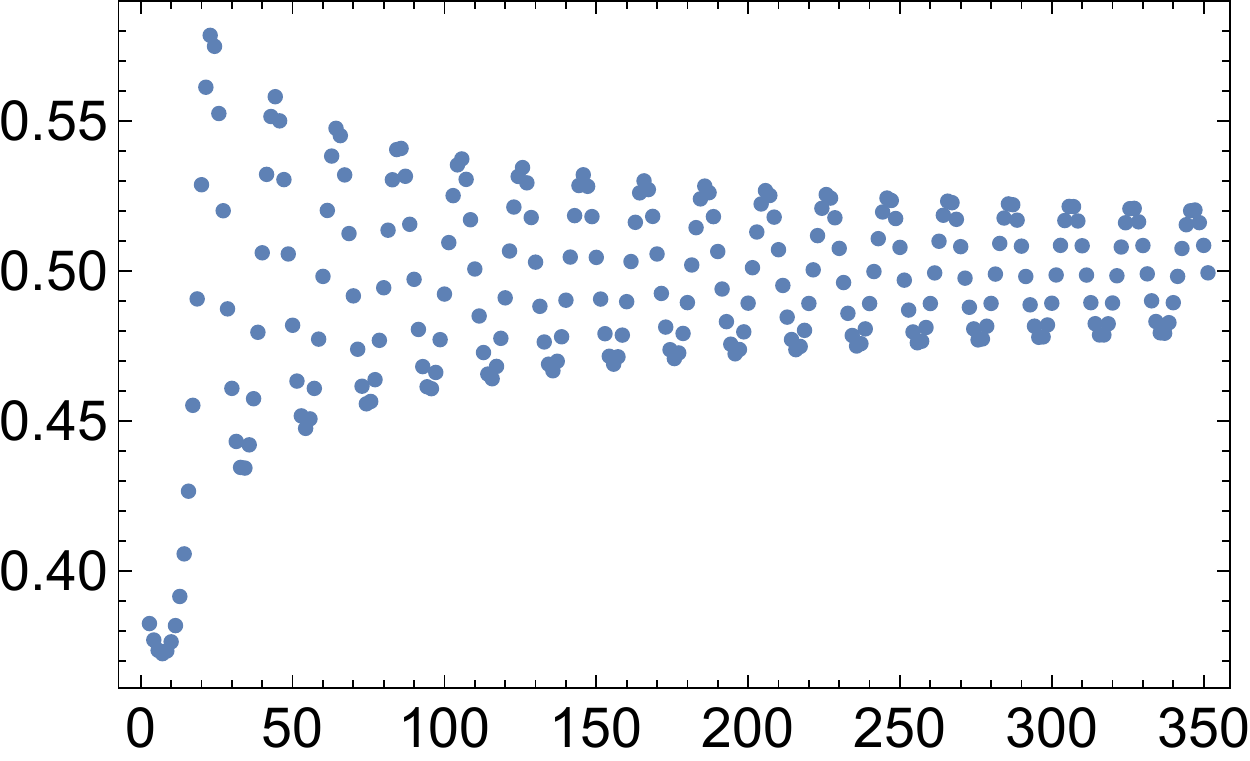}
    \put(52,-4){\large $t$}
    \put(-8,25){\rotatebox{90}{\large $a_n(t)$}}
    \end{overpic}
    \hspace{.12in}
    \begin{overpic}[width=0.4\textwidth]{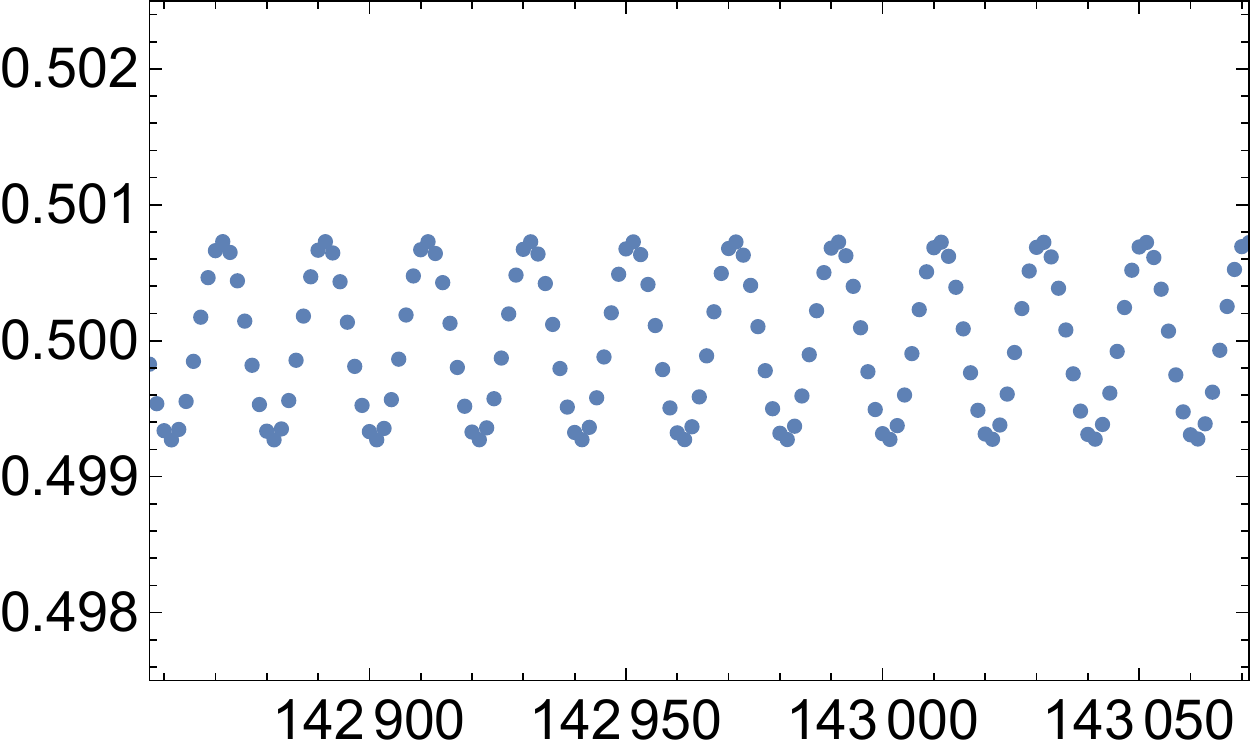}
    \put(52,-4){\large $t$}
    \put(-8,25){\rotatebox{90}{\large $a_n(t)$}}
    \end{overpic}
    \caption{\label{F:NA-disp}  Numerical asymptotics in the dispersive region: $n = \lfloor 7t/10 \rfloor$.  Such computations are accurate for arbitrarily large $t$.}
  \end{figure}

  \item {Painlev\'e region}.\newline
  To show the solution in the Painlev\'e region we let $n$ depend on $t$ through $n = \lfloor t-t^{1/3} \rfloor$.  See Figure~\ref{F:NA-PII} for a plot of the solution into the Painlev\'e region.
  \begin{figure}[h]
    \begin{overpic}[width=0.38\textwidth]{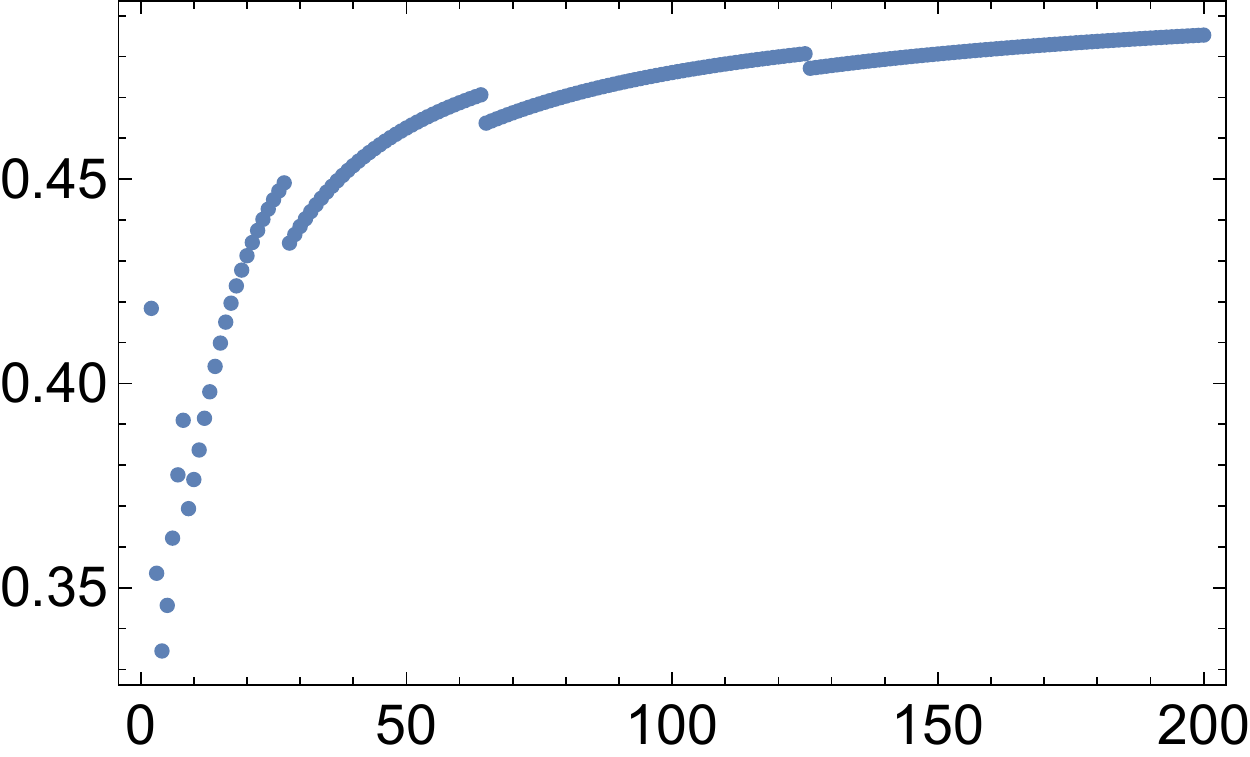}
    \put(52,-4){\large $t$}
    \put(-8,25){\rotatebox{90}{\large $a_n(t)$}}
    \end{overpic}
    \hspace{.12in}
    \begin{overpic}[width=0.42\textwidth]{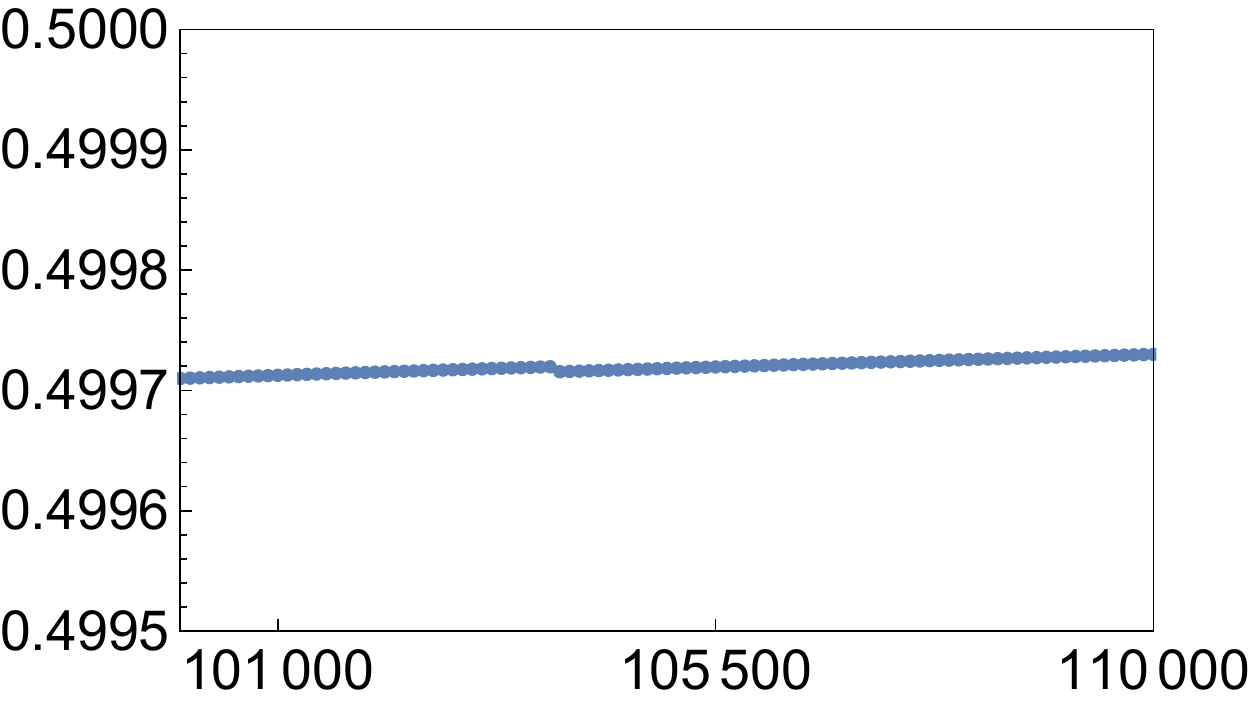}
    \put(52,-4){\large $t$}
    \put(-8,25){\rotatebox{90}{\large $a_n(t)$}}
    \end{overpic}
    \caption{\label{F:NA-PII}  Numerical asymptotics in the Painlev\'e region: $n = \lfloor t-t^{1/3} \rfloor$.  Such computations are accurate for arbitrarily large $t$}
  \end{figure}

  \item {Collisionless shock region}.\newline
  To show the solution in the collisionless shock region we let $n$ depend on $t$ through $n = \lfloor t-t^{1/3}(\log t)^{2/3} \rfloor$.  See Figure~\ref{F:NA-CS} for a plot of the solution into the collisionless shock region.
  \begin{figure}[h]
    \begin{overpic}[width=0.38\textwidth]{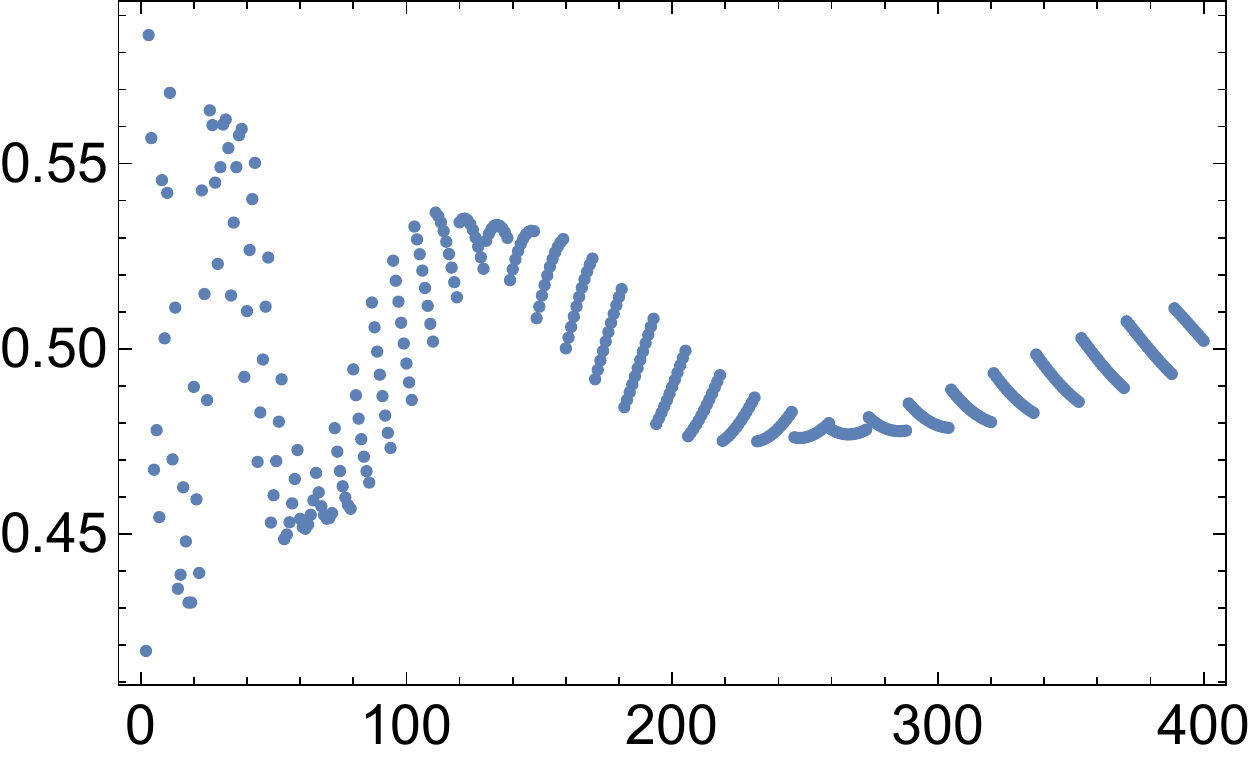}
    \put(52,-4){\large $t$}
    \put(-8,25){\rotatebox{90}{\large $a_n(t)$}}
    \end{overpic}
    \hspace{.12in}
    \begin{overpic}[width=0.42\textwidth]{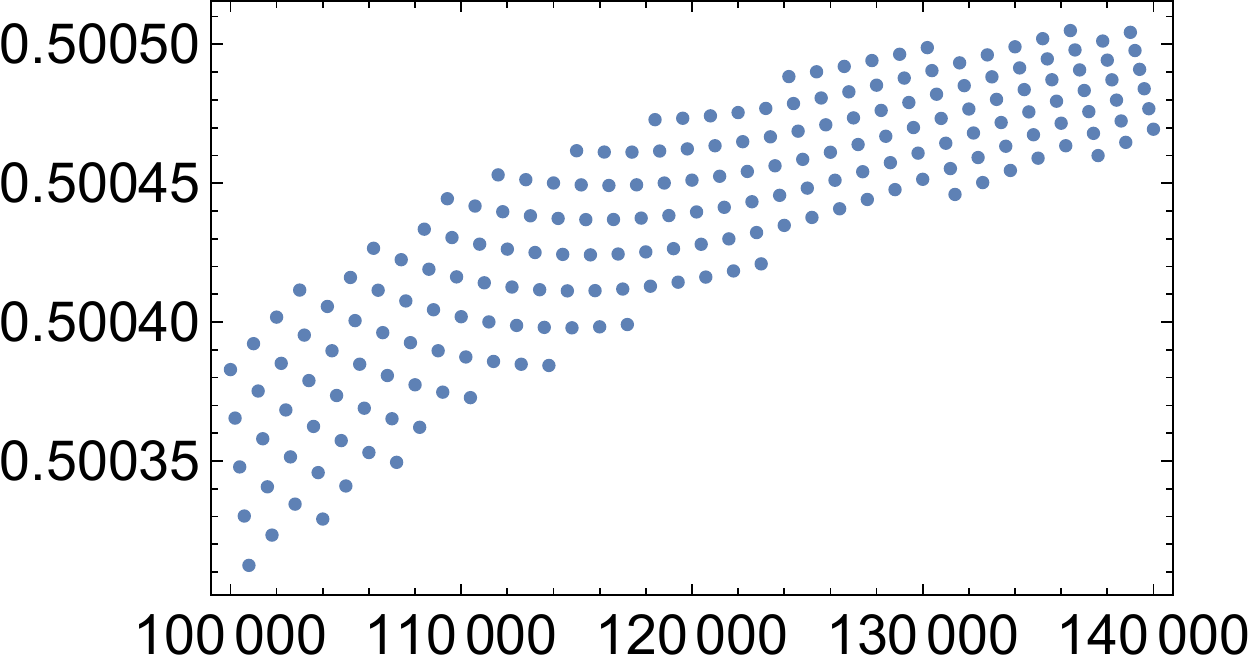}
    \put(52,-4){\large $t$}
    \put(-8,25){\rotatebox{90}{\large $a_n(t)$}}
    \end{overpic}
    \caption{\label{F:NA-CS}  Numerical asymptotics in the collisionless shock region: $n = \lfloor t-t^{1/3}(\log t)^{2/3} \rfloor$.  Such computations are accurate for arbitrarily large $t$.}
  \end{figure}

\end{itemize}

\appendix

\section{Solving the singular diagonal {\RHP}s}\label{A:sing}
There are two diagonal {\RHP}s that must be solved, and computed numerically, in our deformation procedures. The first is $\Delta(z)$ in \rhref{rhp:delta} and the second is $\psi(z)$ in~\rhref{rhp:psi}. Note that because $g^+(z)-g^-(z)$ is constant for $z \in \oarc{\alpha,\alpha^{-1}}$, these problems are both of the form:
\begin{rhp}\label{rhp:X}
\begin{align*}\begin{split}
\mathcal X^+(z) = \mathcal X^-(z) \begin{pmatrix} c \tau(z) & 0 \\ 0 & (c \tau(z))^{-1} \end{pmatrix},~~ z &\in \oarc{\gamma, \gamma^{-1}}, ~~ |\gamma| = 1, ~~\Im \gamma > 0,\\
\mathcal X(\infty) &= I,\\
\mathcal X(z)\diag\left(|z+1|^{-1},|z+1|\right) &= \mathcal O(1),~~ z \to -1, ~~ |z| < 1,\\
\mathcal X(z)\diag\left(|z+1|,|z+1|^{-1}\right) &= \mathcal O(1),~~ z \to -1, ~~ |z| > 1,
\end{split}\end{align*}
for a positive real constant $c$ such that $\mathcal X(z)$ is bounded for $z$ in a neighborhood of $\gamma$, $\gamma^{-1}$ and the boundary values $\mathcal X^\pm(z)$ are not continuous at $\gamma, \gamma^{-1}$ and $-1$.
\end{rhp}

It follows from classical theory that $\mathcal X_{12} \equiv \mathcal X_{21} \equiv 0$, $\mathcal X_{11}(z) = 1/\mathcal X_{12}(z)$. From this
\begin{align*}
\mathcal X_{11}(z) = \exp\left( \frac{1}{2 \pi i} \int_{\oarc{\gamma,\gamma^{-1}}} \frac{ \log(c\tau(s))}{s-z} ds \right).
\end{align*}
Furthermore, because $\log(c\tau(s))$ {is real valued} it follows that $\mathcal X_{11}(z)$ is bounded near $\gamma$ and $\gamma^{-1}$, see \cite{SIE}. The difficulty here is evaluating $\mathcal X_{11}(z)$ numerically because $\log \tau(s)$ is singular at $s = -1 \in \oarc{\gamma,\gamma^{-1}}$. We regularize the integrand by considering\footnote{As the amplitude of the initial data increases, $\tilde \tau(\pm 1)$ approaches zero. Round-off error that is amplified by forming the ratio may become significant and degrade the accuracy obtained in evaluating $\tilde \tau$.  In this case, once should look for a different method to compute $\mathcal X_{11}$.}
\begin{align*}
\tilde \tau(z) = -\frac{\tau(z)}{(z-z^{-1})^2} > 0, ~~ z \in \oarc{\gamma,\gamma^{-1}},
\end{align*}
so that the Cauchy integral of the smooth function $\log \tilde \tau(z)$ along $\oarc{\gamma,\gamma^{-1}}$ can be computed with the methods referenced in Section~\ref{S:numerics}. {Consider}
\begin{align*}
  \mathcal I(z;\gamma) &\defeq \frac{1}{2 \pi i} \int_{\oarc{\gamma,\gamma^{-1}}} \frac{\log(  -(s-s^{-1})^2)}{s-z} ds = \frac{1}{\pi i} \int_{\oarc{\gamma,\gamma^{-1}}} \frac{\log(s+1)}{s-z} ds+\frac{1}{\pi i} \int_{\oarc{\gamma,\gamma^{-1}}} \frac{\log (1-s^{-1})}{s-z} ds\\
  &- \frac{1}{2\pi i} \int_{\oarc{\gamma,-1}} \frac{i \pi}{s-z}ds + \frac{1}{2\pi i} \int_{\oarc{-1,\gamma^{-1}}} \frac{i \pi}{s-z}ds.
\end{align*}
    {To verify the second equality, we choose the branch cut of $\log (1-z^{-1})$ to be on $[0,1]$ with it being real-valued for $z < 0$.   We then choose the branch cut of $\log(1+z)$ to be on $(-\infty,-1]$ with it being real valued for $z > -1$.  Note that these are just the composition of the principal logarithm with $z \mapsto 1-z^{-1}$ and $z \mapsto 1+z$, respectively.     With these definitions, we must show that
    \begin{align*}
      2 \log (z+1) + 2 \log(1-z^{-1}) - \pi i \sign \Im z = \log(-(z-z^{-1})^2), \quad |z| = 1.
    \end{align*}
     For $\Im z > 0$, $\log (z+1) + \log(1-z^{-1}) = \log(z-z^{-1})$ and  $z-z^{-1}$ is purely imaginary with $\Im (z-z^{-1}) > 0$ so that $2 \log (z+1) + 2 \log(1-z^{-1}) - \pi i \sign \Im z = \log(-(z-z^{-1})^2)$.  A similar calculation follows for $\Im z < 0$.

Furthermore, $\sign \Im z$ is piecewise-smooth on $(\gamma,\gamma^{-1})$ and $\log(1-z^{-1})$ is smooth so that their Cauchy integrals can be, again, computed with the methods referenced in Section~\ref{S:numerics}. } Thus it remains to calculate by Cauchy's integral formula
\begin{align}\label{e:CI}
\frac{1}{2 \pi i} \int_{\oarc{\gamma,\gamma^{-1}}} \frac{ 2\log(s+1)}{s-z} ds = \frac{1}{2 \pi i} \int_{I_\gamma} \frac{ 2\log(s+1)}{s-z} ds + \begin{cases} 0, &\text{if } z \in D_{\gamma}^c,\\
2\log(z+1), &\text{if } z \in D_{\gamma},\end{cases}
\end{align}
where $I_\gamma$ is the vertical line connecting $\gamma$ and $\gamma^{-1}$ with downward orientation and $D_\gamma$ is the region enclosed by $I_\gamma$ and $\oarc{\gamma,\gamma^{-1}}$. It turns out that
\begin{align*}
\frac{1}{2 \pi i} \int_{I_\gamma} \frac{ 2\log(s+1)}{s-z} ds &= \frac{1}{2 \pi i} \left(\log (1 + \gamma^{-1}) \log \left( \frac{z-\gamma^{-1}}{1+z} \right)-\log (1 + \gamma) \log \left( \frac{z-\gamma}{1+z} \right)\right.\\
& \left.+ \mathrm{Li}_2\left(\frac{1+\gamma^{-1}}{1+z} \right) - \mathrm{Li}_2\left(\frac{1+\gamma}{1+z} \right)\right)
\end{align*}
where $\mathrm{Li}_2$ is the dilogarithm function, see \cite[Section~25.12]{DLMF}, with appropriately chosen branch cuts. Stock special function routines in {\tt Mathematica} allow this function to be computed accurately. Then
\begin{align*}
\mathcal X_{11}(z) = \exp \left( \frac{1}{2 \pi i} \int_{\oarc{\gamma,\gamma^{-1}}} \frac{ \log(c\tilde \tau(s))}{s-z} ds + \mathcal I(z;\gamma)\right).
\end{align*}
{Finally, we examine the singularities of $\mathcal X_{11}(z)$.  First, we note that if $f(s)$ is a continuously differentiable, real-valued function on a contour $C$ oriented from $a$ to $b$ then
\begin{align*}
  \frac{1}{2\pi i} \int_{C} \frac{f(s)}{s-z} ds =\left. \frac{1}{2\pi i} f(s) \log(s-z)\right|_{a}^b - \frac{1}{2\pi i} \int_C f'(s) \log(s-z)ds.
\end{align*}
As $z \to b$ the only unbounded term is $(2 \pi i)^{-1} f(b) \log(b-z)$ and when exponentiated,
\begin{align*}
  \exp\left[(2 \pi i)^{-1} f(b) \log(b-z)\right] = (b-z)^{f(b)/2\pi i},
\end{align*}
is bounded as $z \to b$ because $f(b)$ is real.  Therefore $\mathcal X_{11}(z)$ is bounded near $z_0$, $z_0^{-1}$.  We now must investigate the singularity at $z = -1$.   The only singularity can come from $\mathcal I(z;\gamma)$.
When exponentiated, \eqref{e:CI} contributes a second-order zero as $z \to -1$ from inside the unit circle. When exponentiated, the quantity
\begin{align*}
- \frac{1}{2\pi i} \int_{\oarc{\gamma,-1}} \frac{i \pi}{s-z}ds + \frac{1}{2\pi i} \int_{\oarc{-1,\gamma^{-1}}} \frac{i \pi}{s-z}ds
\end{align*}
produces a simple pole at $z = -1$.  Hence, we find that $\mathcal X_{11}(z)$ has a simple pole as $z \to -1$ from outside the unit circle and a simple zero as $z \to -1$ from inside the unit circle.  It is then clear that $\mathcal X(z)$ does indeed satisfy the conditions set forth in RH Problem~\ref{rhp:X}.  Now assume $\mathcal Y(z)$ is another solution. It follows that $\mathcal Y(z) \mathcal X^{-1}(z)$ has (possibly) isolated singularities at $z = -1,z_0,z_0^{-1}$ but, the singularities must be bounded and hence this product is entire.  By the asymptotic condition, $\mathcal Y(z) = \mathcal X(z)$ and $\mathcal X(z)$ is the unique solution.}

\section{On computing eigenvalues of $L$}\label{A:1}
In this section we present a case where we fail to capture all of the eigenvalues of $L$ numerically by using conventional eigenvalue algorithms on finite, $K\times K$ truncations $L_K$ of the doubly-infinite Jacobi matrix $L$. For
\begin{equation*}
\lambda^-_{1}(L) < \lambda^-_{2}(L)<\cdots <\lambda^{-}_{M^-}(L) < -1 < 1 < \lambda^{+}_{M^+}(L) < \cdots < \lambda^+_{2}(L) < \lambda^+_{1}(L)
\end{equation*}
denote the eigenvalues of $L$. Note that the integers $M^\pm$ are finite for the Jacobi matrices that appear in this text, but they can be zero. Similarly, let
\begin{equation*}
\lambda^{-}_1(L_K) < \lambda^{-}_2(L_K) < \cdots < \lambda^{-}_K(L_K)\quad\text{and}\quad \lambda^{+}_1(L_K) > \lambda^{+}_2(L_K) > \cdots > \lambda^{+}_K(L_K)
\end{equation*}
denote the real simple eigenvalues of any $K \times K$ truncation $L_K$, labeled in increasing ($\lambda^{-}_{j}$) and decreasing order ($\lambda^{+}_{j}$), respectively. Then
\begin{equation*}
\begin{aligned}
\lambda^{-}_j (L) \leq \lambda^{-}_j (L_K)\quad&\text{for } j = 1, 2, \dots, \min\{M^{-},K\},\\
\lambda^{+}_j (L) \geq \lambda^{+}_j (L_K)\quad&\text{for } j = 1, 2, \dots, \min\{M^{+},K\},
\end{aligned}
\end{equation*}
and
\begin{equation*}
\begin{aligned}
-1 \leq \lambda^{-}_j (L_K)\quad&\text{for }\min\{M^{-},K\}< j \leq K,\\
1 \geq \lambda^{+}_j (L_K)\quad&\text{for }\min\{M^{+},K\}< j \leq K \,.
\end{aligned}
\end{equation*}
In other words, pure point spectrum of $L$ shrinks around the a.c.-spectrum under truncations. In particular, an eigenvalue of $L$ that is close to the a.c.-spectrum might go inside the interval $\sigma_{\text{ac}}(L)=[-1,1]$ after applying a truncation, in which case it fails to be captured by conventional eigenvalue algorithms that are run on the finite truncations $L_K$. The example to be discussed below illustrates such a case. For a more detailed account on the spectra of finite truncations of doubly-infinite Jacobi matrices, we refer the reader to Section~4 in \cite{BN} and .

Consider the data $\{a_n^0, b_n^0\}$:
\begin{align}\label{E:small-sol}
 \begin{split}
a^0_n &= 1 - \half \frac{\sqrt{\ell_{n-1} \ell_{n+1}}}{\ell_n},\\
b^0_n & = \half (1-10^{-4}) (z^* - 1/z^*) \left( \frac{\ell_n -1}{\ell_n} - \frac{\ell_{n-1} -1}{\ell_{n-1}} \right),\\
\ell_n &= 1 + e^{- 4n/5}, \quad z^* = e^{-2/5}.
\end{split}
\end{align}
which is created by inverting pure soliton initial data.  Define the $(2K+1) \times (2K+1)$ truncation
\begin{equation*}
L_{2K+1}=
\begin{pmatrix}
b_{-K} & a_{-K}  &  	0	  		 &       & \\
a_{-K} & b_{-K+1} & a_{-K +1} 		 & \ddots		 & \\
0     & a_{-K +1} &\ddots & \ddots     & \ddots& \\
     &	\ddots	  & \ddots       & \ddots & a_{K-2}  & 0   	 \\
      &		  &		 \ddots     & a_{K-2} & b_{K-1}  & a_{K-1}    \\
     &		  &		  		 & 	0	 &a_{K-1}		 & b_K 	   \\
\end{pmatrix},
\end{equation*}
for $K= 2400$. Using standard eigenvalue algorithms for tridiagonal symmetric matrices on $L_{2K+1}$ yields no eigenvalues outside $[-1,1]$. However, the transmission coefficient has a pole outside $[-1,1]$, which can be captured by Newton iteration to find zeros of $1/R(z)$ and we find $R(z^*) = 0$ for $z^* \approx 0.99982297716$ which corresponds to an eigenvalue $\lambda \approx 1.00000001567$ of $L$.  Naturally, such an eigenvalue is difficult to capture via truncations because it does not emerge until $K$ is very large. See Figure~\ref{F:zero} for a illustration of the zero.

\begin{figure}[h]
  \begin{overpic}[width=0.4\textwidth]{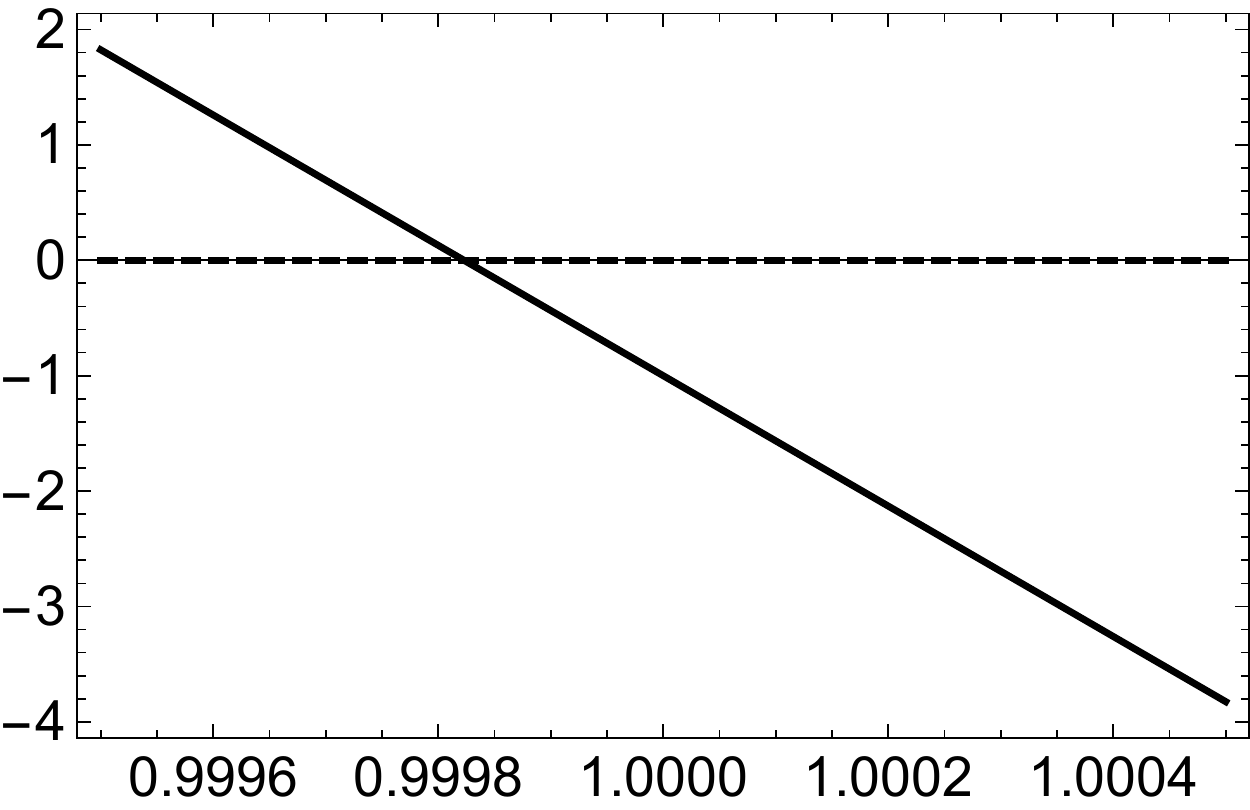}
    \put(50,-5){$z$}
    \put(-6,25){\rotatebox{90}{$1/R(z)$}}
  \end{overpic}\label{F:zero}
\caption{The reflection coefficient $R(z)$  for the initial data \eqref{E:small-sol}.  A zero for $1/R(z)$ is observed near $z = 1$ corresponding to an eigenvalue of $L$.  The zero is very close to the unit circle and therefore the corresponding eigenvalue is close to the right edge of the continuous spectrum of $L$.  This eigenvalue is difficult to detect with traditional eigenvalue techniques but it is easily captured using the reflection coefficient.}
\end{figure}

\section{Toda $g$-function}\label{A:g}
In this section we explicitly define the $g$-function computed in Section~\ref{S:col_shock}. We use this form to determine the exact asymptotic form of the collisionless shock region and then discuss its computation.

\subsection{Construction of the Toda $g$-function}\label{A:construct_g}

Let $\lambda_0 < 0$ and $\rho_0 > 0$ denote the real and imaginary parts of the stationary point, $z_0$, respectively. For $\alpha$, $\beta \in \mathbb{T}$ in the upper half-plane with $-1\leq \Re \alpha \leq \lambda_0 \leq \Re \beta \leq 1$, define
\begin{equation*}
F(\alpha,\beta) = \int_{\beta}^{\alpha}\frac{1}{p^2}\sqrt{\left(p-\alpha\right)\left(p-\alpha^{-1}\right)\left(p-\beta\right)\left(p-\beta^{-1}\right)}^+\, dp.
\end{equation*}
Here the square root is defined with branch cuts on $\carc{\beta,\alpha}$ and $\carc{\alpha^{-1},\beta^{-1}}$ and asymptotics
\begin{align*}
\sqrt{\left(p-\alpha\right)\left(p-\alpha^{-1}\right)\left(p-\beta\right)\left(p-\beta^{-1}\right)} \sim p^2, \quad p \goto \infty.
\end{align*}
\begin{lemma}
Under the additional restriction $\Re \alpha + \Re \beta = 2 \lambda_0$, $F(\alpha,\beta)\in[0,4]$. Moreover, for fixed $z_0$, $F$ is a monotone decreasing function of $\Re \alpha$ as $\Re \alpha$ increases from $-1$ to $\lambda_0$.
\end{lemma}
\begin{proof}
For notational simplicity set $\lambda_\alpha = \Re \alpha$, $\omega_\alpha= \arg \alpha$, and $\lambda_\beta = \Re \beta$, $\omega_\beta = \arg \beta$. Also, set
\begin{equation*}
X(p)=\sqrt{\left(p-\alpha\right)\left(p-\alpha^{-1}\right)\left(p-\beta\right)\left(p-\beta^{-1}\right)}\,.
\end{equation*}
Then
\begin{equation*}
F(\alpha,\beta) = \int_{\beta}^{\alpha} \frac{X^+(p)}{p} \frac{dp}{p},
\end{equation*}
First, note that
\begin{equation}\label{E:Xp}
\left(\frac{X^+ (p)}{p} \right)^2 = 4 (z - \lambda_\alpha)(z-\lambda_\beta),
\end{equation}
for $z= (p + p^{-1})/2$. Now, for $p$ on the unit circle, $z = \Re p $ and the right hand side of \eqref{E:Xp} is negative in the domain of integration, vanishing only at $z=\lambda_\alpha$ and $z=\lambda_\beta$. This implies that $X^+(p)/p$ is purely imaginary.

Second, by definition $X(p)/p \sim p$ as $p\to\infty$ so that $\Im p >0$ in the upper half-plane for $|p|$ large enough. We claim that $\Im(X(p)/p)$ cannot change sign outside the unit circle off the real axis. Suppose that it did. Then $X(p_*) / p_* = c$ would be real valued for some $p_*$ with $|p_*|>1$. For $z_*  = (p_* + p_*^{-1})/2$
\begin{equation}\label{E:quadz}
4 (z_* - \lambda_\alpha)(z_* - \lambda_\beta) = c^2 \in \mathbb{R}.
\end{equation}
Solving \eqref{E:quadz} for $z_*$, we obtain
\begin{equation*}
z_*=\frac{\lambda_\alpha + \lambda_\beta \pm  \sqrt{(\lambda_\alpha - \lambda_\beta)^2 + c^2}}{2},
\end{equation*}
which is clearly real valued. Thus $p_*$ is on the real axis or the unit circle. This implies that $X^{-}(p)/p$ has a positive imaginary part on the arc $\oarc{\beta, \alpha}$, and hence $X^{+}(p)/p$ is purely imaginary with a negative imaginary part on $\oarc{\beta, \alpha}$. This, together with the fact that $dp/p = i d\omega$ in polar coordinates, $p=e^{i\omega}$, implies that
\begin{equation}\label{E:Fpolar}
F(\alpha,\beta) = 2 \int_{\omega_\beta}^{\omega_\alpha}\sqrt{-(\cos\omega -\lambda_0)^2 + (\lambda_\alpha - \lambda_0)^2 }\,d\omega \geq 0.
\end{equation}
Setting $y=\cos\omega$ in \eqref{E:Fpolar} gives us
\begin{equation}
F(\alpha,\beta) = 2 \int_{2\lambda_0 - \lambda_\alpha}^{\lambda_\alpha} \frac{-1}{\sqrt{1-y^2}}\sqrt{-(y -\lambda_0)^2 + (\lambda_\alpha - \lambda_0)^2}\, dy\,.
\end{equation}
Now, for fixed $\lambda_0$ with $\lambda_\alpha \leq \lambda_0 \leq 0$, we differentiate $F$ with respect to $\lambda_\alpha$:
\begin{equation*}
\frac{\partial F}{\partial \lambda_\alpha} = 2\left(\lambda_\alpha - \lambda_0\right) \int^{2\lambda_0 - \lambda_\alpha}_{\lambda_\alpha}\frac{dy}{\sqrt{1-y^2}\sqrt{-(y -\lambda_0)^2 + (\lambda_\alpha - \lambda_0)^2}}\,.
\end{equation*}
Here the contributions from the endpoints vanish. Since $1-y^2 > 0$ and $-(y -\lambda_0)^2 + (\lambda_\alpha - \lambda_0)^2 > 0$ for $y\in(\lambda_\alpha, 2\lambda_0 - \lambda_\alpha)$, $\frac{\partial F}{\partial \lambda_\alpha} < 0$ for fixed $\lambda_0$. This implies that $F$ is a monotone decreasing function of $\lambda_\alpha$ as $\lambda_\alpha$ increases from $-1$ to $\lambda_0$. Thus, for any given $-1\leq \lambda_0\leq 0$, $F$ attains its maximum value $F^{*}(\lambda_0)$ when $\lambda_\alpha = -1$. Moreover,
\begin{equation*}
\frac{\partial F^*}{\partial \lambda_0} = 2 \int^{2\lambda_0 +1}_{-1}\frac{1+y}{\sqrt{1-y^2}\sqrt{-(y -\lambda_0)^2 + (-1 - \lambda_0)^2}}\, dy > 0,
\end{equation*}
which implies that $F^*$ is a monotone increasing function of $\lambda_0$ as $\lambda_0$ increases from $-1$ to $0$. Setting $z_0 = i$, $\alpha=-1$, and $\beta = 1$ gives
\begin{equation*}
F(-1,1) = \int_{0}^{\pi}2 \sqrt{1 -(\cos\omega)^2 }\,\,d\omega=2 \int_{0}^{\pi} \sin\omega \,d\omega = 4,
\end{equation*}
which is the maximum value of $F$ subject to the constraint $\lambda_\alpha+\lambda_\beta = 2\lambda_0$. Clearly, the choice $\alpha=\beta=z_0$ is admissible and it minimizes $F$ at $F=0$, whence we conclude $F(\alpha,\beta)\in[0,4]$.
\end{proof}

Similar to the case for the KdV (see \cite{DVZ,TOD_KdV}), introduce the variable
\begin{equation*}
s= - \frac{\log{\rho_0^2}}{t}.
\end{equation*}
Restricting to $\Re \alpha + \Re \beta = 2 \lambda_0$, with $\Im \alpha,\,\Im \beta \geq 0$, the expression
\begin{equation*}
\int_{\beta}^{\alpha}\frac{1}{p^2}\sqrt{\left(p-\alpha\right)\left(p-\alpha^{-1}\right)\left(p-\beta\right)\left(p-\beta^{-1}\right)}^{+}\, dp = s
\end{equation*}
defines both $\alpha(s)$ and $\beta(s)$ since $s$ is a monotone function of $\Re \alpha$ (or, equivalently, of $\arg \alpha$). We define the $g$-function to be
\begin{equation}\label{E:def_g}
g(z) = \frac{1}{2} \int_{\beta(s)}^{z}\frac{\sqrt{\left(p-\alpha\right)\left(p-\alpha^{-1}\right)\left(p-\beta\right)\left(p-\beta^{-1}\right)}}{p^2}\, dp + \frac{1}{2} \int_{-1}^{\alpha(s)}\frac{\sqrt{\left(p-\alpha\right)\left(p-\alpha^{-1}\right)\left(p-\beta\right)\left(p-\beta^{-1}\right)}}{p^2}\, dp,
\end{equation}
We choose the branch cut for $\sqrt{(z-\alpha)(z-\beta)\left(z-\alpha^{-1}\right)\left( z-\beta^{-1} \right)}$ to contain the circular arcs on the unit circle from $\beta$ to $\alpha$ and $\alpha^{-1}$ to $\beta$. In order for $g(z)$ to be single-valued, it is necessary to add a branch cut on the arc connecting $\alpha$ and $\alpha^{-1}$.

\begin{lemma}
The $g$-function given by \eqref{E:def_g} satisfies:
\begin{enumerate}
\item
\begin{equation*}
g^{+}(z) - g^{-}(z) =
\begin{cases}
\displaystyle\int_{\beta(s)}^{z}\frac{1}{p^2}\sqrt{\left(p-\alpha\right)\left(p-\alpha^{-1}\right)\left(p-\beta\right)\left(p-\beta^{-1}\right)}^{+}\, dp, &\quad z\in\oarc{\beta(s), \alpha(s) }\\
 \displaystyle\int_{\beta(s)}^{\alpha(s)}\frac{1}{p^2}\sqrt{\left(p-\alpha\right)\left(p-\alpha^{-1}\right)\left(p-\beta\right)\left(p-\beta^{-1}\right)}^{+}\, dp, &\quad z\in\oarc{\alpha(s), \alpha(s)^{-1}}\\
\displaystyle\int_{\beta(s)^{-1}}^{z}\frac{1}{p^2}\sqrt{\left(p-\alpha\right)\left(p-\alpha^{-1}\right)\left(p-\beta\right)\left(p-\beta^{-1}\right)}^{+}\, dp, &\quad z\in\oarc{\alpha(s)^{-1}, \beta^{-1}(s) }\\
0,&\quad \text{otherwise.}
\end{cases}
\end{equation*}
\item
\begin{equation}\label{E:g-plus}
g^{+}(z) +g^{-}(z) =
\begin{cases}
\displaystyle\int_{-1}^{\alpha(s)}\frac{1}{p^2}\sqrt{\left(p-\alpha\right)\left(p-\alpha^{-1}\right)\left(p-\beta\right)\left(p-\beta^{-1}\right)}\, dp, &\quad z\in\oarc{\beta(s), \alpha(s)}\\
 \displaystyle\int_{-1}^{z}\frac{1}{p^2}\sqrt{\left(p-\alpha\right)\left(p-\alpha^{-1}\right)\left(p-\beta\right)\left(p-\beta^{-1}\right)}\, dp, &\quad z\in\oarc{\alpha(s), \alpha(s)^{-1} }\\
\displaystyle\int_{-1}^{\alpha(s)^{-1}}\frac{1}{p^2}\sqrt{\left(p-\alpha\right)\left(p-\alpha^{-1}\right)\left(p-\beta\right)\left(p-\beta^{-1}\right)}\, dp, &\quad z\in\oarc{\alpha(s)^{-1}, \beta^{-1}(s) }\\
2 g(z),&\quad \text{otherwise.}
\end{cases}
\end{equation}
\item
\begin{equation*}
g(z) = \frac{1}{2}z - \lambda_0 \log z + \mathcal{O}\left(z^{-1}\right),\text{ as } z\to \infty.
\end{equation*}
\item
$g(z) - \frac{1}{2t}\theta(z)$ is bounded.
\item
\begin{equation*}
\rho_0^2 e^{t\left(g^+ (z) - g^- (z) \right)} = 1, \text{ for } z\in \oarc{\alpha(s), \alpha(s)^{-1}}.
\end{equation*}
\end{enumerate}
\end{lemma}

\begin{proof}
As before, set $X(p) = \sqrt{\left(p-\alpha\right)\left(p-\alpha^{-1}\right)\left(p-\beta\right)\left(p-\beta^{-1}\right)}$, and let $\lambda_\alpha$ and $\lambda_\beta$ denote the real parts of $\alpha$ and $\beta$, respectively. The properties (1) and (2) follow from contour integration and the fact that $X^+\left(p^{-1}\right)=\frac{X^+(p)}{p^2}$. To prove (3), we use
\begin{equation*}
\sqrt{1-y} = 1 - \frac{1}{2} y - \frac{1}{8}y^2 + \mathcal{O}\left(y^3\right), \text{ as } y\to 0,
\end{equation*}
and the condition that $\lambda_\alpha + \lambda_\beta = 2 \lambda_0 $, to obtain
\begin{equation*}
\frac{X(p)}{p^2} = p^{-2} - 2 \lambda_0 p^{-1} + \left( 1 - \tfrac{1}{2}(\lambda_\alpha - \lambda_\beta)^2 \right) + 4\lambda_0 \lambda_\alpha \lambda_\beta p + \mathcal{O}\left(p^2\right), \text{ as } p\to 0\,.
\end{equation*}
Also, for large $p$, we have
\begin{equation*}
\frac{X(p)}{p^2} = 1 - 2 \lambda_0 p^{-1} + \left( 1 - \tfrac{1}{2}(\lambda_\alpha - \lambda_\beta)^2 \right)p^{-2} + 4\lambda_0 \lambda_\alpha \lambda_\beta p^{-3} + \mathcal{O}\left(p^{-4}\right), \text{ as }p\to \infty.\\
\end{equation*}
Setting $H(z) = \displaystyle\int_{\alpha(s)}^{z}\frac{X(p)}{p^2}\, dp $, integration gives us
\begin{equation}\label{E:g_asymp}
\begin{aligned}
H(z) &= -z^{-1} + C - 2\lambda_0\log z + \left( 1 - \tfrac{1}{2}(\lambda_\alpha - \lambda_\beta)^2 \right)z + \mathcal{O}\left(z^2\right), \text{ as } z \to 0,\\
H(z) &= z - 2 \lambda_0 \log z + C - \left( 1 - \tfrac{1}{2}(\lambda_\alpha - \lambda_\beta)^2 \right)z^{-1} + \mathcal{O}\left(z^{-2}\right), \text{ as } z \to \infty,
\end{aligned}
\end{equation}
for some complex constant $C$. Then $\displaystyle\lim_{z\to\infty} H(z) + H\left(z^{-1}\right) = 2C$, which implies, for large $|z|$, that
\begin{equation}
\begin{aligned}
2 C &= \int_{\alpha(s)}^{z}\frac{X(p)}{p^2}\, dp + \int_{\alpha(s)}^{z^{-1}}\frac{X(p)}{p^2}\, dp = \int_{\alpha(s)}^{z}\frac{X(p)}{p^2}\, dp + \int_{z}^{\alpha(s)^{-1}}\frac{X(p)}{p^2}\, dp\\
& = - \int_{\alpha(s)^{-1}}^{\alpha(s)}\frac{X(p)}{p^2}\, dp = - 2 \int_{-1}^{\alpha(s)}\frac{X(p)}{p^2}\, dp.
\end{aligned}
\end{equation}
Therefore, $g(z) = \frac{1}{2} z -\lambda_0 \log z + \mathcal{O}\left(z^{-1}\right)$ as $z\to\infty$. (4) follows from the integral representation and the asymptotic expansion around $z=0$ in \eqref{E:g_asymp}. Finally, to prove (5), assume $z\in \oarc{\alpha(s),\alpha(s)^{-1}}$,
\begin{equation*}
\begin{aligned}
-\frac{\log{\rho_0^2} }{t} &= \int_{\beta(s)}^{\alpha(s)}\frac{X^+(p)}{p^2}\, dp\\
-\log{\rho_0^2} &= t \left(g^{+}(z) - g^{-}(z)\right)\\
e^{t \left(g^{+}(z) - g^{-}(z)\right)}\rho_0^2 &= 1.
\end{aligned}
\end{equation*}
\end{proof}

\subsection{Derivation of the collisionless shock scaling}\label{A:scaling}
We have
\begin{equation*}
\rho_0 = \sqrt{1 - \left(\frac{n}{t} \right)^2} = \sqrt{2}\sqrt{1 - \frac{n}{t}}\left(1 + \mathcal{O}\left(1 - \frac{n}{t} \right) \right),~\text{ for }~\frac{n}{t}\sim 1.
\end{equation*}
As described in Appendix~\ref{A:construct_g}, we choose $A = K(\alpha)$ and $B = K(\beta)$ (see \eqref{E:cov} for the transformation $K$) by
\begin{equation}\label{E:A-B-choice}
-\frac{\log \rho_0^{2}}{t\rho_0^3} = - i \int_{B}^{A}\frac{\sqrt{(q-A)(q-B)\left(q+\bar{A}\right)\left(q+\bar{B}\right)}}{(\lambda_0 - i\rho_0 q)^2}^+ \; dq.
\end{equation}
Note that as $\lambda_0 \to -1$, $\rho_0 \to 0$ (i.e.\ $z_0 \to -1$), the integral on the right hand side converges provided that $A$ and $B$ converge to finite values. To ensure this, we enforce that
\begin{equation*}
-\frac{\log\rho_0^2}{\rho_0^3 t} \to C,~\text{ as } t\to\infty,
\end{equation*}
for some constant $C\in\mathbb{R}$. Let $n = t - \frac{\mu(t)}{2}$, with $\frac{\mu(t)}{t}\to 0$ as $t\to\infty$, so that
\begin{equation}
\rho_0 = \sqrt{\mu(t)}\left(1 + \mathcal{O}\left(\tfrac{\mu(t)}{t}\right) \right).
\end{equation}
Then
\begin{equation*}
-\frac{\log\rho_0^2}{\rho_0^3 t} \sim \frac{-\log\left(\dfrac{\mu(t)}{t}\right)}{t\left(\frac{\mu(t)}{t}\right)^{3/2}} \to C\,,~\text{ as }~t\to\infty\,.
\end{equation*}
For this limit to exist, set $\mu(t)=C_1 t^{1/3}(\log t)^{2/3}$, $C_1>0$. Doing so yields
\begin{equation*}
-\frac{-\log\left(\frac{\mu(t)}{t}\right)}{t\left(\frac{\mu(t)}{t}\right)^{3/2}} =
\frac{\frac{2}{3}\log t - \frac{2}{3}\log(\log t) - \log C_1}{C_1^{3/2}\log t}\to \frac{2}{3C_1^{3/2}},~\text{ as }~t\to\infty\,,
\end{equation*}
as desired: In this limit, $A$ and $B$ in \eqref{E:A-B-choice}, tend to finite values. Therefore the scaling for the collisionless shock region is given by
\begin{equation*}
n = t - C_1 t^{1/3}(\log t)^{2/3}.
\end{equation*}

\subsection{Computing the Toda $g$-function}\label{g-func-comp}

To compute $g(z)$ we first compute $\mathfrak g$ and use the relation $t g(z) = t \mathfrak g(z) + \half \theta(z)$. It follows that $\mathfrak g'(z)$ solves the following {\RHP}
\begin{align}\label{frak-jump}
\mathfrak g'(z) &+ \mathfrak g'(z) = t^{-1} \theta'(z), ~~ z \in \Sigma_u \cup \Sigma_l,\\
\mathfrak g'(z) &= \mathcal O(z^{-2}), ~\text{as}~~ z \goto \infty.\label{frak-decay}
\end{align}
Furthermore, $\mathfrak g'(z)$ is a bounded function on $\mathbb C \setminus (\Sigma_u \cup \Sigma_l)$. We remark that imposing that \eqref{frak-jump} and that $\mathfrak g'(z)$ is a bounded function in the finite plane uniquely determines $\mathfrak g'(z)$ and \eqref{frak-decay} is a consequence of our choice of $\alpha$ and $\beta$.

We consider the function $\mathfrak G(k) = \mathfrak g(M(k))$, bounded in the finite plane, where $M(k) = \frac{k+i}{k-i}$ maps the real axis to the unit circle. Then $\mathfrak G'(k)$ solves
\begin{align}\label{Frak-jump}
\mathfrak G'(k) &+ \mathfrak G'(k) = t^{-1} \theta'(M(k))M'(k), ~~ z \in (-B,-A) \cup (A,B),\\
\mathfrak G'(k) &= \mathcal O((k-i)^2), ~\text{as}~~ k \goto i.\label{Frak-decay}
\end{align}
Here $A = M^{-1}(\alpha)$ and $B = M^{-1}(\beta)$. As in the case of $\mathfrak g'(z)$ the behavior at $k = i$ does not need to be imposed --- the boundedness of $\mathfrak G'$ along with the jump condition \eqref{Frak-jump} are enough to uniquely determine the function. Thus the numerical methodology in \cite{SOEquilibriumMeasure} applies directly to this situation and allows us to compute both $\mathfrak G'$ and $\mathfrak G$ to within machine precision, uniformly in the complex plane. Hence $\mathfrak g(z) = \mathfrak G\left(i\frac{z+1}{z-1} \right)$ can be computed accurately.
\section{The vanishing lemma and the unique solvability of {\RHP}s}\label{A:solve}
Let $G^*(z)$ denote the Hermitian transpose of the matrix $G(z)$.  In the following lemma, we allow a solution of a {\RHP} to fail to be continuous up the boundary but it must be uniformly bounded and satisfy the jump condition almost everywhere (a.e.).

\begin{lemma}[Vanishing lemma]\label{L:vanish}
Consider the {\RHP}
\begin{align*}
\Phi_+(z) = \Phi_-(z) G(z), \quad a.e. ~~z \in \mathbb T, \quad \Phi(\infty) = 0,
\end{align*}
where $\Phi$ is uniformly bounded on $\mathbb C\setminus \mathbb T$.  Assume $G \in L^\infty(\mathbb T)$ and $G^*(z) + G(z)$ is positive semi-definite a.e. and strictly positive definite on a set of positive measure. Then $\Phi \equiv 0$.
\end{lemma}
\begin{proof}
Because $\Phi(z)$ vanishes at infinity, $\Phi^*(1/\bar z)z^{-1}$ is a bounded analytic function on the unit disc. Then for $|z| =1$, $1/\bar z = z$ so that
\begin{align*}
0 = - i\int_{\mathbb T} \Phi_+(z) \Phi_-^*(1/\bar z)z^{-1} dz = \int_{\mathbb T} \Phi_+(z) \Phi_-^*(z) |d z| = \int_{\mathbb T} \Phi_+(z) G^*(z) \Phi_+^*(z) |dz|.
\end{align*}
Adding this equation to its Hermitian conjugate we find
\begin{align*}
\int_{\mathbb T} \Phi_+(z) \left[G^*(z) + G(z)\right] \Phi_+^*(z) |dz| = 0.
\end{align*}
And if $G^*(z)+G(z)$ is strictly positive definite on a set $S$ then $\Phi_+ = 0$ a.e. on $S$.   As $S$ has positive measure, it follows from classical results (see, for example, \cite{Duren}) that $\Phi_+ \equiv 0$ and therefore $\Phi \equiv 0$.
\end{proof}

\subsection{Unique solvability without solitons}
\begin{prop}
The associated matrix {\RHP} to \rhref{rhp:m} (\emph{i.e} dropping the symmetry condition and normalize to $I$ at $\infty$) is uniquely solvable in the absence of residue condition ( \emph{i.e.} if $N = 0$).
\end{prop}
\begin{proof}
This associated matrix {\RHP} is uniquely solvable if the operator defined by $u \mapsto u - \mathcal C_{\mathbb T}^- u \cdot (J-I)$ is invertible on $H^1(\mathbb T)$ \cite{mybook}. From classical results, this operator is Fredholm  and from $\det J(z) = 1$ the index is zero.  Then Lemma~\ref{L:vanish} demonstrates that the kernel must be trivial and the operator is invertible because $R(z^{-1}) = \overline{R(z)}$ and hence $J(z) + J^*(z)$ is diagonal with non-negative diagonal entries.
\end{proof}

We now introduce a function $\Delta_\mathrm{s}(z) = \Delta_\mathrm{s}(z;z_0)$ that satisfies \rhref{rhp:delta}, with the symmetry condition
\begin{align}\label{E:deltas-sym}
\Delta_\mathrm{s}(z) = \begin{pmatrix}
0 & 1 \\ 1 & 0
\end{pmatrix} \Delta_\mathrm{s}\big(z^{-1}\big) \begin{pmatrix}
0 & 1 \\ 1 & 0
\end{pmatrix}, \quad |z| > 1,
\end{align}
with $\Delta_\mathrm{s}(\infty) =  C$ and generally $C \neq I$.  $\Delta_\mathrm{s}(z)$ is called the partial transmission coefficient in \cite{KT_rev} (when $N = 0$) and is given by
\begin{align*}
\Delta_\mathrm{s}(z) &= \diag\left(\delta_\mathrm{s}(z), \delta_\mathrm{s}^{-1}(z)\right),\\
\delta_\mathrm{s}(z) &= \exp \left( \frac{1}{2 \pi i} \int_{\oarc{z_0,z_0^{-1}}} \log \tau(s) \frac{s+z}{s-z} \frac{ds}{2s} \right), ~~~ \tau(z) = 1 - R(z) R(z^{-1}).
\end{align*}
We now verify that $\Delta_\mathrm{s}$ satisfies the jump condition and singularity conditions of \rhref{rhp:delta} along with the symmetry condition \eqref{E:deltas-sym}.  We begin with the jump condition.  For $z \in \oarc{z_0,z_0^{-1}}$, $z \neq -1$ using the Sokhotski--Plemelj lemma \cite[p.~42]{SIE}
\begin{align*}
\log(\delta_\mathrm{s}^+(z)/\delta_\mathrm{s}^-(z)) =  \log \tau(z).
\end{align*}
For $|z| > 1$,
\begin{align*}
\delta_\mathrm{s}(1/z) = \exp \left( \frac{1}{2 \pi i} \int_{\oarc{z_0,z_0^{-1}}} \log \tau(s) \frac{s+z^{-1}}{s-z^{-1}} \frac{ds}{2s} \right).
\end{align*}
Then sending $s \mapsto s^{-1}$, using that $\tau(s^{-1}) = \tau(s)$ we have
\begin{align*}
\delta_\mathrm{s}(1/z) = 1/\delta_\mathrm{s}(z),
\end{align*}
showing that $\Delta_\mathrm{s}$ satisfies \eqref{E:deltas-sym}. We then find
\begin{align*}
	\delta_\mathrm{s}(z)/\delta(z) &= \exp\left( \frac{1}{2 \pi i} \int_{\oarc{z_0,z_0^{-1}}} \log \tau(s) \left[ \frac{s+z}{2s} - 1 \right]  \frac{ds}{s-z} \right)\\
    &=\exp\left( \frac{1}{2 \pi i} \int_{\oarc{z_0,z_0^{-1}}} \log \tau(s) \frac{ds}{2s} \right).
\end{align*}
Here $\delta(z) = \Delta_{11}(z)$ where $\Delta$ is the solution of \rhref{rhp:delta}. So $\delta_\mathrm{s}(z) = c \delta(z)$ for some $c \in \mathbb C$, $c \neq 0$.

Consider a new {\RHP} constructed using the same jump matrix $J$ from \rhref{rhp:hm}.  Assume $N =0$ (without solitons) and let $C_{z_0} = \{z : |z-z_0| = \epsilon\}$ with counter-clockwise orientation.  Let $w_0$ and $v_0$ be the intersection points of $\oarc{z_0,z_0^{-1}}$ and $C_{z_0}$ and $\mathbb T \setminus \oarc{z_0,z_0^{-1}}$ and $C_{z_0}$, respectively. Define the jump matrix $\hat J: \hat \Gamma \to \mathbb C^{2\times 2}$, $\hat\Gamma = \oarc{w_0,w_0^{-1}} \cup C_{z_0} \cup C_{z_0^{-1}} \cup (\mathbb T\setminus \oarc{v_0,v_0^{-1}})$:
\begin{align*}
\hat J(z) = \begin{cases}
\Delta_\mathrm{s}^-(z) J(z) \Delta_\mathrm{s}^+(z)^{-1}, & z \in \oarc{w_0,w_0^{-1}},\\
\Delta_\mathrm{s}(z) J(z), & z \in C_{z_0} \cup C_{z_0^{-1}} \text{ and } |z| > 1,\\
\Delta_\mathrm{s}(z), & z \in C_{z_0} \cup C_{z_0^{-1}} \text{ and } |z| < 1,\\
\Delta_\mathrm{s}(z) J(z) \Delta_\mathrm{s}(z)^{-1}, & z \in \mathbb T\setminus \oarc{v_0,v_0^{-1}}.
\end{cases}
\end{align*}
It follows that $\hat J$ satisfies the product condition \cite[Def.~2.55]{mybook} and therefore $u \mapsto u - C_{\hat \Gamma}^- u \cdot(\hat J-I)$ is Fredholm on $H_z^1(\hat \Gamma)$ (see \cite[Def.~2.48]{mybook}) and $\det \hat J = 1$ implies the index is zero.  Assume $u$ is in the kernel of this operator, so that $\Psi(z) := \mathcal C_{\hat \Gamma}u$ is a solution that is continuous up to $\hat \Gamma$ that satisfies $\Psi(\infty) = 0$.  Define for $z \in \mathbb C\setminus(\mathbb T \cup \hat \Gamma)$
\begin{align*}
\hat \Psi(z) = \Psi(z) \begin{cases}
\Delta_\mathrm{s}^{-1}(z), & |z-z_0| < \epsilon \text{ and } |z| < 1,\\
\Delta_\mathrm{s}^{-1}(z), & |z-z^{-1}_0| < \epsilon \text{ and } |z| < 1,\\
J^{-1}(z)\Delta_\mathrm{s}^{-1}(z), & |z-z_0| < \epsilon \text{ and } |z| > 1,\\
J^{-1}(z)\Delta_\mathrm{s}^{-1}(z), & |z-z^{-1}_0| < \epsilon \text{ and } |z| > 1,\\
I, & \text{otherwise}.
\end{cases}
\end{align*}
It follows that $\hat \Psi$ has a continuation that is analytic in $\mathbb C \setminus \mathbb T$ with the jump
\begin{align*}
\hat \Psi^+(z) &= \hat \Psi^-(z) G(z),\\
G(z) := \Delta^-_s(z)J(z) \Delta^+_s(z)^{-1} &= \begin{pmatrix}
[1-R(z)R(z^{-1})]\delta_\mathrm{s}^-(z)/\delta_\mathrm{s}^+(z) & - \overline{R(z)}\delta_\mathrm{s}^-(z) \delta_\mathrm{s}^+(z) e^{-\theta(z;n,t)}\\
\frac{R(z)}{\delta_\mathrm{s}^-(z) \delta_\mathrm{s}^+(z)} e^{\theta(z;n,t)} & \delta_\mathrm{s}^+(z)/\delta_\mathrm{s}^-(z)
\end{pmatrix}.
\end{align*}
This follows because $\hat \Psi^+(z) = \hat \Psi^-(z)$ on $C_{z_0}$ and $C_{z_0^{-1}}$.  Note that by the singularity conditions in \rhref{rhp:delta}, $G(z)$ is continuous near $z = -1$.  It follows that $\delta_\mathrm{s}(z)$ also satisfies $\overline{\delta^+_s(z)} = 1/\delta^-_s(z)$ so that
\begin{align*}
G^*(z) = \begin{pmatrix}
\big[1-R(z)R(z^{-1})\big]|\delta_\mathrm{s}^-(z)|^2&  \overline{R(z)}\delta_\mathrm{s}^-(z) \delta_\mathrm{s}^+(z) e^{-\theta(z;n,t)}\\
-\frac{R(z)}{\delta_\mathrm{s}^-(z) \delta_\mathrm{s}^+(z)} e^{\theta(z;n,t)} & |\delta_\mathrm{s}^+(z)|^2
\end{pmatrix}
\end{align*}
and
\begin{align*}
G^*(z) + G(z)  = \begin{pmatrix}
2\big[1-R(z)R(z^{-1})\big]|\delta_\mathrm{s}^-(z)|^2&  0\\
0  & 2|\delta_\mathrm{s}^+(z)|^2
\end{pmatrix}.
\end{align*}
Applying Lemma~\ref{L:vanish}, $\hat \Psi \equiv 0$ and hence $\Psi \equiv 0$.  This gives the following proposition:

\begin{prop}
\label{P:solve}
The matrix {\RHP} with jump matrix $\hat J$ is uniquely solvable.
\end{prop}

\subsection{The addition of solitons} \label{S:addsolitons}

Our main approach to adding solitons is to just include the jumps on the contours $D_j^\pm$ in \rhref{rhp:disp}.  Because the matrix {\RHP} without these contours is uniquely solvable it is reasonable to expect that the addition of these jumps will not completely destroy the unique solvability of the problem.  As stated in Remark~\ref{r:singular}, for each fixed $n$, the problem is either solvable for no $t$ or solvable on the compliment of a discrete set of $t$ values.  One should expect the latter and this is indeed the case: The matrix {\RHP}s we consider are uniquely solvable for sufficiently large $t$.

We now present an approach that incorporate the symmetry condition (see \eqref{E:symm}) for the vector {\RHP}s to give equations that are uniquely solvable for every $n$ and $t$ value.  We emphasize that this approach is not necessary to compute the solution and the fact that the numerical method presented here is robust despite ignoring symmetry is important, as evidenced in Figure~\ref{f:singular}.  Ignoring symmetry is also often more convenient for implementation as well.  Nonetheless, inspired by \cite[(4.7)]{KT_rev} and by the deformation for $\Phi_{1,\mathrm{d}}$ in \rhref{rhp:disp}, we define
\begin{align*}
Q_\mathrm{s}(z) = \diag( q_\mathrm{s}(z), 1/q_\mathrm{s}(z) ), \quad q_\mathrm{s}(z) = \prod_{j \in K_{n,t}} \frac{z - \zeta_j}{z \zeta_j -1}.
\end{align*}
We use this here because
\begin{align*}
Q_\mathrm{s}(z) = \begin{pmatrix}
0 & 1 \\ 1 & 0
\end{pmatrix} Q_\mathrm{s}(z^{-1}) \begin{pmatrix}
0 & 1 \\ 1 & 0
\end{pmatrix}
\end{align*}
and if $m(z)$ (from \rhref{rhp:m}) satisfies the symmetry condition \eqref{E:symm} then so does $m^\star(z):=m(z)Q_\mathrm{s}(z)$.  Define
\begin{align*}
h^+_j(n,t): = h_j^+ = \begin{cases}
\lim_{z \to \zeta_j} (z-\zeta_j)^2 q_\mathrm{s}(z)^{-2}, & j \in K_{n,t},\\
q_\mathrm{s}(\zeta_j)^{-2}, & j \not\in K_{n,t},
\end{cases} \quad
h^-_j(n,t): = h_j^- = \begin{cases}
\lim_{z \to \zeta_j^{-1}} (z-\zeta^{-1}_j)^2 q_\mathrm{s}(z)^2, & j \in K_{n,t},\\
q_\mathrm{s}(\zeta_j)^{2}, & j \not\in K_{n,t}.
\end{cases}
\end{align*}
It follows that $m^\star(z)$ satisfies the following conditions.
\begin{rhp}
Find the function $m^\star \colon \mathbb C \setminus \mathbb{T} \to \mathbb{C}^{1\times 2}$ that is sectionally meromorphic, continuous up to $\mathbb T$, with simple poles at $\zeta_j^{\pm 1}$, $j=1\dots,N$, and satisfies:
\begin{itemize}
\item \emph{the jump condition:}
\begin{equation*}
m^{\star}_+(z;n,t) = m^{\star}_-(z;n,t) Q_\mathrm{s}^{-1}(z)J(z;n,t)Q_\mathrm{s}(z),\phantom{x} z\in\mathbb{T},
\end{equation*}
\item \emph{the residue conditions:} For $j = 1,2,\ldots,N$,
\begin{equation*}
\begin{aligned}
\underset{z=\zeta_j}{\text{Res}}\,m^\star(z;n,t) &= \lim_{z\to\zeta_j} m^\star(z;n,t)\begin{pmatrix}
0 & h_j^+\zeta^{-1}_j {\gamma_j^{-1} e^{-\theta(\zeta_j;n,t)}} \\
0 & 0
\end{pmatrix},\quad j\in K_{n,t},\\
\underset{z=\zeta_j^{-1}}{\text{Res}}\,m^\star(z;n,t) &= \lim_{z\to\zeta_j^{-1}} m^\star(z;n,t)\begin{pmatrix}
0 & 0 \\
-\zeta_j  h_j^-\gamma_j^{-1} e^{-\theta(\zeta_j;n,t)}  & 0
\end{pmatrix}, \quad j\in K_{n,t},\\
\underset{z=\zeta_j}{\text{Res}}\,m^\star(z;n,t) &= \lim_{z\to\zeta_j} m^\star(z;n,t)\begin{pmatrix}
0 & 0\\
-\zeta_j {\gamma_j e^{\theta(\zeta_j;n,t)}}/h_j^+  & 0
\end{pmatrix},\quad j\not\in K_{n,t},\\
\underset{z=\zeta_j^{-1}}{\text{Res}}\,m^\star(z;n,t) &= \lim_{z\to\zeta_j^{-1}} m^\star(z;n,t)\begin{pmatrix}
0 & \zeta^{-1}_j \gamma_j e^{\theta(\zeta_j;n,t)}/h_j^- \\
0 & 0
\end{pmatrix}, \quad j\in K_{n,t},
\end{aligned}
\end{equation*}
\item \emph{the symmetry condition:}
\begin{equation*}
m^\star\left(z^{-1}; n, t\right) = m^\star(z;n,t)\begin{pmatrix} 0& 1 \\ 1 & 0 \end{pmatrix},
\end{equation*}
\item \emph{the normalization condition:}
\begin{equation*}
\lim_{z\to 0} m^\star(z;n,t) = \begin{pmatrix} m^\star_1 & m^\star_2 \end{pmatrix},\text{ with } m^\star_1 \cdot m^\star_2 = 1 \text{ and } \left(\prod_{j\in K_{n,t}}  \zeta_j\right)m^\star_1 > 0\,.
\end{equation*}
\end{itemize}
\end{rhp}
Now, let $K(z) = K(z;n,t)$ be the solution of the following {\RHP}:
\begin{rhp}\label{rhp:K}
Find the function $K \colon \mathbb C \setminus \mathbb{T} \to \mathbb{C}^{2\times 2}$ that is sectionally analytic, continuous up to $\mathbb T$, and satisfies
\begin{itemize}
\item \emph{the jump condition:}
\begin{equation*}
K_+(z) = K_-(z) Q_\mathrm{s}^{-1}(z)J(z;n,t)Q_\mathrm{s}(z),\phantom{x} z\in\mathbb{T},
\end{equation*}
\item \emph{the symmetry condition:}
\begin{equation}\label{E:symm-K}
K\left(z^{-1}\right) = \begin{pmatrix} 0& 1 \\ 1 & 0 \end{pmatrix}K(z)\begin{pmatrix} 0& 1 \\ 1 & 0 \end{pmatrix},
\end{equation}
\item \emph{the normalization condition:}
\begin{equation*}
K(\infty) = I.
\end{equation*}
\end{itemize}
\end{rhp}
\rhref{rhp:K} has a unique solution.  Indeed, because $\overline{q_\mathrm{s}(z)} = 1/q_\mathrm{s}(z)$ for $z \in \mathbb T$, the jump matrix and contour satisfy the hypotheses of the vanishing lemma, Lemma~\ref{L:vanish}.  This implies that the solution without the symmetry condition exists and is unique.  Then from the symmetries of the jump matrix
\begin{align*}
\hat K(z) := \begin{pmatrix} 0& 1 \\ 1 & 0 \end{pmatrix}K(z^{-1})\begin{pmatrix} 0& 1 \\ 1 & 0 \end{pmatrix},
\end{align*}
is also a solution and therefore \eqref{E:symm-K} follows.  Furthermore, $\det K(z) = 1$. Define $E_{j,\pm}$ by
\begin{equation*}
\begin{aligned}
E_{j,+}&:=\begin{cases}
K(\zeta_j)\begin{pmatrix}
0 & h_j^+\zeta^{-1}_j {\gamma_j^{-1} e^{-\theta(\zeta_j;n,t)}} \\
0 & 0
\end{pmatrix}K^{-1}(\zeta_j),\quad j\in K_{n,t},\\
K(\zeta_j) \begin{pmatrix}
0 & 0\\
-\zeta_j {\gamma_j e^{\theta(\zeta_j;n,t)}}/h_j^+  & 0
\end{pmatrix} K^{-1}(\zeta_j),\quad j\not\in K_{n,t},
\end{cases}\\
E_{j,-}&:=\begin{cases}
K(\zeta_{j}^{-1}) \begin{pmatrix}
0 & 0 \\
-\zeta_j  h_j^-\gamma_j^{-1} e^{-\theta(\zeta_j;n,t)}  & 0
\end{pmatrix}K^{-1}(\zeta_{j}^{-1}), \quad j\in K_{n,t},\\
K(\zeta_j^{-1})\begin{pmatrix}
0 & \zeta^{-1}_j \gamma_j e^{\theta(\zeta_j;n,t)}/h_j^- \\
0 & 0
\end{pmatrix}K^{-1}(\zeta_j^{-1}), \quad j\not\in K_{n,t}.
\end{cases}\\
\end{aligned}
\end{equation*}
Now consider $v^\star(z) : = m^\star(z) K^{-1}(z)$ which has an analytic continuation across $\mathbb T$.  Thus $v^\star(z)$ satisfies the following discrete {\RHP}:
\begin{rhp}
Find the function $v^\star \colon \mathbb C \to \mathbb{C}^{1\times 2}$ that is sectionally meromorphic with simple poles at $\zeta_j^{\pm 1}$, $j=1\dots,N$, and satisfies:
\begin{itemize}
\item \emph{the residue conditions:} For $j = 1,2,\ldots,N$,
\begin{equation*}
\begin{aligned}
\underset{z=\zeta_j}{\text{Res}}\,v^\star(z) &= \lim_{z\to\zeta_j} v^\star(z)E_{j,+},\quad j\in K_{n,t},\\
\underset{z=\zeta_j^{-1}}{\text{Res}}\,v^\star(z) &= \lim_{z\to\zeta_j^{-1}} v^\star(z) E_{j,-}, \quad j\in K_{n,t},\\
\underset{z=\zeta_j}{\text{Res}}\,v^\star(z) &= \lim_{z\to\zeta_j} v^\star(z) E_{j,+},\quad j\not\in K_{n,t},\\
\underset{z=\zeta_j^{-1}}{\text{Res}}\,v^\star(z) &= \lim_{z\to\zeta_j^{-1}} v^\star(z)E_{j,-}, \quad j\not\in K_{n,t},
\end{aligned}
\end{equation*}
\item \emph{the symmetry condition:}
\begin{equation}\label{E:symm-v-star}
v^\star\left(z^{-1}\right) = v^\star(z)\begin{pmatrix} 0& 1 \\ 1 & 0 \end{pmatrix},
\end{equation}
\item \emph{the normalization condition:}
\begin{equation}\label{E:norm-v-star}
\lim_{z\to 0} v^\star(z) = \begin{pmatrix} v^\star_1 & v^\star_2 \end{pmatrix},\text{ with } v^\star_1 \cdot v^\star_2 = 1 \text{ and } \left(\prod_{j\in K_{n,t}}  \zeta_j\right)v^\star_1 > 0\,.
\end{equation}
\end{itemize}
\end{rhp}

We assume we can solve for $K(z)$ and therefore compute each of $E_{j,\pm}$, $j=1,2,\ldots,N$.  It follows that
\begin{align*}
v^\star(z) = \begin{pmatrix}
c + \sum_{j=1}^N a_{j,+} \frac{z \zeta_j -1}{z-\zeta_j} + \sum_{j=1}^N a_{j,-} \frac{z - \zeta_j }{z\zeta_j-1} & d + \sum_{j=1}^N b_{j,+} \frac{z \zeta_j -1}{z-\zeta_j} + \sum_{j=1}^N b_{j,-} \frac{z - \zeta_j }{z\zeta_j-1}
\end{pmatrix},
\end{align*}
for some choice of constants $a_{j,\pm}$, $b_{j,\pm}$.  From the symmetry condition \eqref{E:symm-v-star} it follows that $a_{j,\pm} = b_{j,\mp}$, $c = d$ so that
\begin{align*}
v^\star(z) = \begin{pmatrix}
c + \sum_{j=1}^N a_{j,+} \frac{z \zeta_j -1}{z-\zeta_j} + \sum_{j=1}^N a_{j,-} \frac{z - \zeta_j }{z\zeta_j-1} & c + \sum_{j=1}^N a_{j,-} \frac{z \zeta_j -1}{z-\zeta_j} + \sum_{j=1}^N a_{j,+} \frac{z - \zeta_j }{z\zeta_j-1}
\end{pmatrix}.
\end{align*}
We now obtain a linear system for these constants under the assumption that $c$ is known.  For $k = 1,2,\ldots, N$
\begin{align}
&\underset{{z = \zeta_k}}{\text{Res}}\, v^\star(z) =\begin{pmatrix}
a_{k,+}(\zeta_k^2 -1) & a_{k,-}(\zeta_k^2-1)
\end{pmatrix}\notag \\\label{E:sys1}
&= \begin{pmatrix}
c + \sum_{j \neq k} a_{j,+} \frac{\zeta_k \zeta_j -1}{\zeta_k-\zeta_j} + \sum_{j=1}^N a_{j,-} \frac{\zeta_k - \zeta_j }{\zeta_k\zeta_j-1} & c + \sum_{j \neq k} a_{j,-} \frac{\zeta_k \zeta_j -1}{\zeta_k-\zeta_j} + \sum_{j=1}^N a_{j,+} \frac{\zeta_k - \zeta_j }{\zeta_k\zeta_j-1}
\end{pmatrix} E_{k,+},\\
&\underset{{z = \zeta_k^{-1}}}{\text{Res}} v^\star(z) =\begin{pmatrix}
a_{k,-}(\zeta_k^{-2} -1) & a_{k,+}(\zeta_k^{-2}-1)
\end{pmatrix}\notag \\
&= \begin{pmatrix}
c + \sum_{j =1}^N a_{j,+} \frac{\zeta^{-1}_k \zeta_j -1}{\zeta^{-1}_k-\zeta_j} + \sum_{j\neq k} a_{j,-} \frac{\zeta^{-1}_k - \zeta_j }{\zeta^{-1}_k\zeta_j-1} & c + \sum_{j =1}^N a_{j,-} \frac{\zeta^{-1}_k \zeta_j -1}{\zeta^{-1}_k-\zeta_j} + \sum_{j\neq k} a_{j,+} \frac{\zeta^{-1}_k - \zeta_j }{\zeta^{-1}_k\zeta_j-1}
\end{pmatrix} E_{k,-}.\label{E:sys2}
\end{align}
This is a system of $4N$ equations for $2N$ unknowns so it must have redundancy.  We compute the symmetries of $E_{j,\pm}$ from the symmetry of $v^\star(z)$:
\begin{align*}
\underset{{z = \zeta_j}}{\text{Res}}\, v^\star(z) &= \lim_{z \to \zeta_j} v^\star(z) E_{j,+},\\
\lim_{z \to \zeta_j} (z-\zeta_j) v^\star(z) &= \lim_{z \to \zeta_j} v^\star(z) E_{j,+},\\
\lim_{z \to \zeta_j} (z-\zeta_j) v^\star(z^{-1}) &= \lim_{z \to \zeta_j} v^\star(z^{-1}) \begin{pmatrix}
0 & 1 \\ 1 & 0
\end{pmatrix} E_{j,+}\begin{pmatrix}
0 & 1 \\ 1 & 0
\end{pmatrix} = \lim_{z \to \zeta_j^{-1}} v^\star(z) \begin{pmatrix}
0 & 1 \\ 1 & 0
\end{pmatrix} E_{j,+}\begin{pmatrix}
0 & 1 \\ 1 & 0
\end{pmatrix},\\
\lim_{z \to \zeta^{-1}_j} (z^{-1}-\zeta_j) v^\star(z) &= \lim_{z \to \zeta_j^{-1}} v^\star(z) \begin{pmatrix}
0 & 1 \\ 1 & 0
\end{pmatrix} E_{j,+}\begin{pmatrix}
0 & 1 \\ 1 & 0
\end{pmatrix},\\
-\lim_{z \to \zeta^{-1}_j} z^{-1}\zeta_j(z-\zeta_j^{-1}) v^\star(z) &= \lim_{z \to \zeta_j^{-1}} v^\star(z) \begin{pmatrix}
0 & 1 \\ 1 & 0
\end{pmatrix} E_{j,+}\begin{pmatrix}
0 & 1 \\ 1 & 0
\end{pmatrix},\\
\underset{{z = \zeta^{-1}_j}}{\text{Res}}\,v^\star(z) &= - \zeta_j^{-2}\lim_{z \to \zeta_j^{-1}} v^\star(z) \begin{pmatrix}
0 & 1 \\ 1 & 0
\end{pmatrix} E_{j,+}\begin{pmatrix}
0 & 1 \\ 1 & 0
\end{pmatrix}.
\end{align*}
This leads to the conclusion that \eqref{E:sys1} is equivalent to \eqref{E:sys2}.  And since a unique solution of \eqref{E:sys1} exists for the correct $c$, a solution must exist for every $c$.  Due to linearity, the coefficients $a_{j,\pm} = a_{j,\pm}(c)$ have simple dependence on $c$.  Namely, $a_{j,\pm}(c) = a_{j,\pm}(1)\cdot c$.  To find $c$, we set $c = 1$ and solve for $a_{j,\pm}(1)$.  Then for $v^\star(z) = v^\star(z;c)$ we have
\begin{align*}
v^\star(0;c) = c v^\star(0;1).
\end{align*}
And $c$ is then found by enforcing \eqref{E:norm-v-star}.
\section{A proof that generically $R(-1) = -1$}\label{A:genericity}
{In this appendix we prove a theorem to establish the genericity of $R(\pm 1) = -1$ for the Jacobi matrices used in this work. Consider the weighted $\ell^1$-space of doubly infinite sequences, $\ell^{1}_w(\mathbb{Z})$, with the weight function given by $n\mapsto 1+|n|$, and define the Banach space $\mathcal{L}^1_{w}(\mathbb{Z}) = \ell^1_w(\mathbb{Z}) \oplus \ell^1_w(\mathbb{Z})$ equipped with the norm
\begin{equation*}
\| (x,y) \|_{\mathcal{L}^{1}_w} = \sum_{n\in\mathbb{Z}}(1+|n|)\big(|x_n| + |y_n| \big).
\end{equation*}

We let $\mathcal{M}$ denote the \emph{Marchenko} class of doubly infinite Jacobi matrices $J(a,b)$ whose coefficients have the property that
\begin{equation*}
\left( \left\lbrace a_n -\tfrac{1}{2} \right\rbrace_{n\in\mathbb{Z}}, \lbrace b_n \rbrace_{n\in\mathbb{Z}} \right) \in {\mathcal{L}^{1}_w}(\mathbb{Z}).
\end{equation*}
For ease of notation, we set $\alpha_n = a_n - \tfrac{1}{2}$ for each $n\in\mathbb{Z}$ and define $J_o(\cdot,\cdot)$ by  $J_o(a-1/2,b) \defeq J(a,b)$ for all of the Jacobi matrices mentioned throughout this appendix.  Note that $J_o(\alpha,b) = J(a,b)$.
\begin{theorem}
The set of doubly infinite Jacobi matrices $J_o(\alpha, b)$ in $\mathcal{M}$ with the associated reflection coefficient satisfying $R(\pm 1) = -1$ is an open and dense subset of $\mathcal{M}$ in the topology induced by the norm
\begin{equation*}
\| J_o(\alpha,b) \|_{\mathcal{M}} = \sum_{n\in\mathbb{Z}}{(1+|n|)\left(\left|\alpha_n\right | + |b_n| \right)}.
\end{equation*}
\end{theorem}

\begin{proof}
We first show that the subset of Jacobi matrices in $\mathcal{M}$ with $R(\pm 1) = -1$ is dense in $\mathcal{M}$. Suppose that $L = J_o(\alpha,b)$ is in $\mathcal{M}$ with the reflection coefficient satisfying $R_L(\pm 1) \neq -1$, and let $\varepsilon>0$ be given. Also, let $L_0$ denote the free Jacobi matrix with coefficients $a_n=\tfrac{1}{2}$ (\emph{i.e.}$~\alpha_n = 0$) and $b_n=0$ for all $n\in\mathbb{Z}$ -- namely, the discrete Schr\"{o}dinger operator with the zero potential. As $J_o(\alpha,b) \in \mathcal{M}$, there exists some $N=N(\varepsilon)\in\mathbb{N}$ such that
\begin{equation}
\sum_{|n|\geq N} (1+|n|)(|\alpha_n| + |b_n|) < \frac{\varepsilon}{3}.
\label{eq:epsilon1}
\end{equation}
Now consider $\tilde L = J_o\left(\tilde{\alpha},\tilde{b}\right)$ which is the Jacobi matrix that is a finite-rank perturbation of $L_0$, with the property that
\begin{equation*}
\tilde{\alpha}_n=\begin{cases}\alpha_n,\quad&\text{if~} -N\leq n \leq N-1,\\ 0,\quad&\text{otherwise,}\end{cases}~\text{and}\quad
\tilde{b}_n=\begin{cases}b_n,\quad&\text{if~} -N\leq n \leq N,\\ 0,\quad&\text{otherwise.}\end{cases}
\end{equation*}
Define the \emph{Jost} solutions $\varphi_{\pm}^\mathrm{r}$ and $\varphi_{\pm}^\mathrm{l}$ of the problem
\begin{equation}
(\tilde L \varphi)(n) = \frac{z+z^{-1}}{2} \varphi(n),\quad\text{for all~}n\in\mathbb{Z},
\label{eq:Jproblem}
\end{equation}
by their asymptotic behaviors:
\begin{equation}
\begin{aligned}
\lim_{n\to +\infty} \varphi_{+}^\mathrm{r}(z;n) z^{-n} =1\quad&\text{and}\quad\lim_{n\to +\infty} \varphi_{-}^\mathrm{r}(z;n) z^{n} =1,\\
\lim_{n\to -\infty} \varphi_{+}^\mathrm{l}(z;n) z^{-n} =1\quad&\text{and}\quad\lim_{n\to -\infty} \varphi_{-}^\mathrm{l}(z;n) z^{n} =1.
\end{aligned}
\label{eq:Jost}
\end{equation}
First recall that the Wronskian of any two solutions $\varphi$ and $\psi$ to the problem \eqref{eq:Jproblem} is given by
\begin{equation*}
W\left(\varphi(z;\cdot),\psi(z;\cdot)\right) = \tilde{a}_n \left[\varphi(z;n)\psi(z;n+1) - \varphi(z;n+1)\psi(z;n) \right] = - \det \begin{bmatrix}\varphi(z;n+1) & \psi(z;n+1) \\ a_n \varphi(z;n) & a_n \psi(z;n) \end{bmatrix},
\end{equation*}
and that it is independent of $n$ as long as $f$ and $g$ solve \eqref{eq:Jproblem} with the same value of $\tfrac{1}{2}\left(z+z^{-1}\right)$. Recall also that the reflection coefficient $R_{\tilde{L}}(z)$ associated with $\tilde{L}$ (or for any Jacobi matrix $L$ in $\mathcal{M}$) can be obtained by ratio of two Wronskians of certain Jost solutions that correspond to the same value of $\tfrac{1}{2}\left(z+z^{-1}\right)$:
\begin{equation}
R_{\tilde{L}}(z) = \frac{W\left(\varphi_{-}^\mathrm{l}(z;\cdot), \varphi_{-}^\mathrm{r}(z;\cdot)\right)}{W\left(\varphi_+^\mathrm{r}(z;\cdot), \varphi_{-}^\mathrm{l}(z;\cdot)\right)}.
\label{eq:Refco}
\end{equation}
For any solution $\varphi$ of \eqref{eq:Jproblem}, observe that
\begin{equation*}
\begin{bmatrix}
\varphi(z;n+1)\\
a_n \varphi(z;n)
\end{bmatrix}
=
\Lambda_n(z)
\begin{bmatrix}
\varphi(z;n)\\a_{n-1}\varphi(z;n-1)
\end{bmatrix}
\quad\text{and}\quad
\begin{bmatrix}
\varphi(z;n)\\
a_{n-1} \varphi(z;n-1)
\end{bmatrix}
=
\Lambda_n(z)^{-1}
\begin{bmatrix}
\varphi(z;n+1)\\a_{n}\varphi(z;n)
\end{bmatrix},
\end{equation*}
where $\Lambda_n(z)^{\pm 1}$ are the transfer matrices defined by
\begin{equation*}
\Lambda_n(z) =\frac{1}{a_n}\begin{bmatrix} \frac{z+z^{-1}}{2} - b_n & -1\\ a_n^2 & 0 \end{bmatrix}
\quad\text{and}\quad
\Lambda_n(z)^{-1} =\frac{1}{a_n}\begin{bmatrix} 0 & 1\\ -a_n^2 & \frac{z+z^{-1}}{2} - b_n \end{bmatrix}.
\end{equation*}
Now, since $n\mapsto z^n$ and $n\mapsto z^{-n}$ solve \eqref{eq:Jproblem} with $\tilde{L}$ replaced with the free matrix $L_0$, as long as the Jost solutions of \eqref{eq:Jproblem} are outside the support of $L-L_0$, they coincide with the solutions of the problem for the free matrix $L_0$. In particular:
\begin{equation}
\begin{aligned}
\varphi^\mathrm{r}_{+}(z;N+1)&= z^{N+1},\quad
\varphi^\mathrm{r}_{+}(z;N)= z^{N},\quad
\varphi^\mathrm{r}_{-}(z;N+1)= z^{-N-1},\quad
\varphi^\mathrm{r}_{-}(z;N)= z^{-N},\\
\varphi^\mathrm{l}_{+}(z;-N-1)&= z^{-N-1},\quad
\varphi^\mathrm{l}_{+}(z;-N)= z^{-N},\quad
\varphi^\mathrm{l}_{-}(z;-N-1)= z^{N+1},\quad
\varphi^\mathrm{l}_{-}(z;-N)= z^{N}.
\end{aligned}
\label{eq:f-edge}
\end{equation}
To compute the reflection coefficient, we evaluate the Wronskians in \eqref{eq:Refco} at $n=0$. Then we have
\begin{equation}
R_{\tilde{L}}(z) = \frac{W\left(\varphi_{-}^\mathrm{l}(z;\cdot), \varphi_{-}^\mathrm{r}(z;\cdot)\right)}{W\left(\varphi_+^\mathrm{r}(z;\cdot), \varphi_{-}^\mathrm{l}(z;\cdot)\right)}=\frac{
\det\begin{bmatrix}
\varphi^\mathrm{l}_{-}(z;1) & \varphi^\mathrm{r}_{-}(z;1)\\  a_0 \varphi^\mathrm{l}_{-}(z;0) & a_0 \varphi^\mathrm{r}_{-}(z;0)
\end{bmatrix}
}{\det\begin{bmatrix}
\varphi^\mathrm{r}_{+}(z;1) &  \varphi^\mathrm{l}_{-}(z;1) \\  a_0 \varphi^\mathrm{r}_{+}(z;0)& a_0 \varphi^\mathrm{l}_{-}(z;0)
\end{bmatrix}
}.
\label{eq:Refco2}
\end{equation}
Using the transfer matrices, this could be expressed as:
\begin{equation*}
R_{\tilde{L}}(z)=\frac{
\det\begin{bmatrix}
\Lambda_0(z) \Lambda_{-1}(z)\cdots\Lambda_{-N}(z)
\begin{bmatrix}
z^{N}\\ \tfrac{1}{2} z^{N+1}
\end{bmatrix};
\Lambda_1(z)^{-1} \Lambda_{2}(z)^{-1}\cdots\Lambda_{N}(z)^{-1}
\begin{bmatrix}
z^{-(N+1)}\\ \tfrac{1}{2} z^{-N}
\end{bmatrix}
\end{bmatrix}
}{
\det\begin{bmatrix}
\Lambda_1(z)^{-1} \Lambda_{2}(z)^{-1}\cdots\Lambda_{N}(z)^{-1}
\begin{bmatrix}
z^{N+1}\\ \tfrac{1}{2} z^{N}
\end{bmatrix};
\Lambda_0(z) \Lambda_{-1}(z)\cdots\Lambda_{-N}(z)
\begin{bmatrix}
z^{N}\\ \tfrac{1}{2} z^{N+1}
\end{bmatrix}
\end{bmatrix}
}.
\end{equation*}
For ease of notation, label the products of transfer matrices that appear above:
\begin{equation*}
\begin{aligned}
\mathbf{A}_N(z)&:=\Lambda_0(z) \Lambda_{-1}(z)\cdots\Lambda_{-N}(z),\\
\mathbf{B}_N(z)&:=\Lambda_1(z)^{-1} \Lambda_{2}(z)^{-1}\cdots\Lambda_{N}(z)^{-1},
\end{aligned}
\end{equation*}
and rewrite
\begin{equation*}
R_{\tilde{L}}(z)=\frac{
\det\begin{bmatrix}
\mathbf{A}_N(z)
\begin{bmatrix}
z^{N}\\ \tfrac{1}{2} z^{N+1}
\end{bmatrix};
\mathbf{B}_N(z)
\begin{bmatrix}
z^{-(N+1)}\\ \tfrac{1}{2} z^{-N}
\end{bmatrix}
\end{bmatrix}
}{
\det\begin{bmatrix}
\mathbf{B}_N(z)
\begin{bmatrix}
z^{N+1}\\ \tfrac{1}{2} z^{N}
\end{bmatrix};
\mathbf{A}_N(z)
\begin{bmatrix}
z^{N}\\ \tfrac{1}{2} z^{N+1}
\end{bmatrix}
\end{bmatrix}
}.
\end{equation*}
A direct calculation shows that
\begin{equation*}
R_{\tilde{L}}(\pm 1) = -1
\end{equation*}
holds unless
\begin{equation*}
\det\begin{bmatrix}
\mathbf{B}_N(\pm 1)
\begin{bmatrix}
1\\ \tfrac{1}{2}
\end{bmatrix};
\mathbf{A}_N(\pm 1)
\begin{bmatrix}
1\\ \tfrac{1}{2}
\end{bmatrix}
\end{bmatrix} = 0.
\label{eq:determinant}
\end{equation*}
In case the determinant in \eqref{eq:determinant} is nonzero, since \eqref{eq:epsilon1} implies that $\| L - \tilde{L} \|_{\mathcal{M}} < \tfrac{\varepsilon}{3}$, we have proven that the subset of Jacobi matrices in $\mathcal{M}$ whose reflection coefficients satisfies $R(\pm 1)=-1$ is dense in $\mathcal{M}$. The case where the determinant in \eqref{eq:determinant} vanishes requires more work and will be treated below. For notational brevity, we present the argument for $z=-1$ only. The argument for the case $z=1$ is identical.

Suppose that \eqref{eq:determinant} holds, and note that the determinant on the left hand side of \eqref{eq:determinant} is a rational function in $2N+1$ real variables: ${a}_{-N},\dots,{a}_{N-1}, {b}_{-N},\dots,{b}_N$. More precisely, it is of the form:
\begin{equation*}
 \det\begin{bmatrix}
\mathbf{B}_N(-1)
\begin{bmatrix}
1\\ \tfrac{1}{2}
\end{bmatrix};
\mathbf{A}_N(-1)
\begin{bmatrix}
1\\ \tfrac{1}{2}
\end{bmatrix}
\end{bmatrix}=
\frac{y\left({a}_{-N},\dots,{a}_{N-1}, {b}_{-N},\dots,{b}_N\right)}{a_{-N} \cdots a_0 \cdots a_{N-1} },
\end{equation*}
where the denominator is nonzero, and $y$ is a polynomial in $2N+1$ real variables. Therefore determinant vanishes at a point in $\mathbb{R}^{2N+1}$ if and only if $y$ vanishes at that point. Note that $y$ is not the zero polynomial since $y(0,0,\dots,0) =  \tfrac{1}{4}$.
Now suppose that $y$ vanishes at a point $(\tilde{a}_{-N},\dots,\tilde{a}_{N-1}, \tilde{b}_{-N},\dots,\tilde{b}_N)$. Any polynomial that vanishes on an open set is identically zero, and since $y$ is not identically zero, there must be a point where $y$ is non-zero in every neighborhood of where it vanishes. In particular, there is a point in any neighborhood of $(\tilde{a}_{-N},\dots,\tilde{a}_{N-1}, \tilde{b}_{-N},\dots,\tilde{b}_N)$ where $y$ does not vanish.
Take the open ball $D_{\varepsilon,N}$ (in $\mathbb{R}^{2N+1}$) centered at $(\tilde{a}_{-N},\dots,\tilde{a}_{N-1}, \tilde{b}_{-N},\dots,\tilde{b}_N)$ with radius $\frac{\varepsilon}{6(1+N)(1+2N)}$.
There exists a point $(a'_{-N},\dots,a'_{N-1}, b'_{-N},\dots, b'_N)\in D_{\varepsilon,N}$ where
\begin{equation*}
y\left(a'_{-N},\dots,a'_{N-1}, b'_{-N},\dots, b'_N\right)\neq 0.
\end{equation*}
This immediately implies that the reflection coefficient $R_{L^*}$ for the Jacobi matrix $L^* = J_o\left(\alpha^*,b^*\right)$ with the coefficients
\begin{equation*}
\alpha^*_n=\begin{cases}\alpha'_n,\quad&\text{if~} -N\leq n \leq N-1,\\ 0,\quad&\text{otherwise,}\end{cases}~\text{and}\quad
b^*_n=\begin{cases}b'_n,\quad&\text{if~} -N\leq n \leq N,\\ 0,\quad&\text{otherwise.}\end{cases}
\end{equation*}
satisfies
\begin{equation*}
R_{L^{*}}(- 1) = -1.
\end{equation*}
Moreover,
\begin{equation*}
\left\|L^{*}- \tilde{L} \right\|_{\mathcal{M}} = \sum_{n\in\mathbb{Z}}(1+|n|)\left(\left|\tilde{\alpha}_n - \alpha^*_n\right| + \left|\tilde{b}_n - b^*_n\right|\right) \leq \frac{\varepsilon}{6(1+N)(2N+1)} 2(1+N)(2N+1) =\frac{\varepsilon}{3}.
\label{eq:epsilon2}
\end{equation*}
It follows from \eqref{eq:epsilon1} and $\eqref{eq:epsilon2}$ that $\left \| L - L^* \right\|_{\mathcal{M}} <\varepsilon$. This shows that the subset of Jacobi matrices in $\mathcal{M}$ whose reflection coefficients have the property that $R(\pm 1)=-1$ is dense in $\mathcal{M}$.

We now prove that such matrices form an open set in $\mathcal{M}$.  When $z = \pm 1$ we note that $\varphi_\pm^{\mathrm{l}} = \varphi_\pm^{\mathrm{r}}$ and when considering \eqref{eq:Refco} we see that $R(\pm 1) \neq -1$ only when $W(\varphi_-^{\mathrm{l}}(\pm 1,\cdot), \varphi_{-}^{\mathrm{r}}(\pm 1,\cdot))  = 0$.  So, we take $L = J_o(\alpha, b)$ such that $W(\varphi_{-}^{\mathrm{l}}(\pm 1,\cdot), \varphi_{-}^{\mathrm{r}}(\pm 1,\cdot)) \neq 0$ and it is enough to show that this Wronskian is continuous with respect to the topology of $\mathcal M$. We will show that the following mappings

\begin{align*}
  T_{\mathrm{r}}: \mathcal M &\to \ell^\infty(\mathbb Z^+), ~~\quad\quad\quad T_{\mathrm{r}}(L) = \varphi_{-}^{\mathrm{r}}(\pm 1,\cdot),\\
  T_{\mathrm{l}}: \mathcal M &\to \ell^\infty(\mathbb Z^- \cup \{1\}), \quad T_{\mathrm{l}}(L) = \varphi_{-}^{\mathrm{l}}(\pm 1,\cdot),
\end{align*}
are well-defined and continuous which imply the continuity of the Wronskian at $n = 0$. It follows that
\begin{align*}
  \varphi_{-}^{\mathrm{r}}(\pm 1,n) &= (\pm 1)^n -  \sum_{m=n+1}^\infty (\pm 1)^{m-n}(m-n) h^\mathrm{r}_m,\\
  \varphi_{-}^{\mathrm{l}}(\pm 1,n) &= (\pm 1)^n -  \sum_{m=-\infty}^{n-1} (\pm 1)^{m-n}(m-n) h^\mathrm{l}_m,\\
  h^\mathrm{r}_m &= \alpha_{m-1}\varphi_{-}^{\mathrm{r}}(\pm 1,m-1) + b_m \varphi_{-}^{\mathrm{r}}(\pm 1,m) + \alpha_m\varphi_{-}^{\mathrm{r}}(\pm 1,m+1),\\
  h^\mathrm{l}_m &= \alpha_{m-1}\varphi_{-}^{\mathrm{l}}(\pm 1,m-1) + b_m \varphi_{-}^{\mathrm{l}}(\pm 1,m) + \alpha_m\varphi_{-}^{\mathrm{l}}(\pm 1,m+1).
\end{align*}
We then write
\begin{align}
  \varphi_{-}^{\mathrm{r}}(\pm 1,n) + \sum_{m=n+1}^\infty K^\mathrm{r}_{n,m}(\alpha,b)\varphi_{-}^{\mathrm{r}}(\pm 1,m) &= (\pm 1)^n,\label{e:volt}\\
  \varphi_{-}^{\mathrm{l}}(\pm 1,n) + \sum_{m=-\infty}^{n-1} K^\mathrm{l}_{n,m}(\alpha,b)\varphi_{-}^{\mathrm{l}}(\pm 1,m) &= (\pm 1)^n.\notag
\end{align}
We do not construct the kernels in the sums explicitly, but rather, obtain bounds on them.  In both cases $|n| \leq |m|$ so that
\begin{align*}
  \sup_{n \geq 0} \sum_{m=n+1}^\infty K^{\mathrm{r}}_{n,m}(\alpha,b) &\leq 8 \|J_o(\alpha, b)\|_{\mathcal M},\\
  \sup_{n \leq 1} \sum_{m=-\infty}^{n-1} K^{\mathrm{l}}_{n,m}(\alpha,b) &\leq 8 \|J_o(\alpha, b)\|_{\mathcal M}.
\end{align*}
Then it follows from \cite[Lemma 7.8]{Tes} that $\varphi^\mathrm{r}_-$ and $\varphi^\mathrm{l}_-$ are the unique solutions of these Volterra summation equations.  Now, let $f$ be the solution of \eqref{e:volt} with a different Jacobi matrix $J_o(\check \alpha, \check b)$.  Then $u := \varphi_-^\mathrm{r} - f$ solves
\begin{align*}
  u(n) + \sum_{m=n+1}^\infty K^\mathrm{r}_{n,m}(\check\alpha,\check b)u(m) &= \sum_{m=n+1}^\infty (K^\mathrm{r}_{n,m}(\alpha - \check\alpha, b- \check b)) \varphi_-^\mathrm{r}(\pm 1,m).
\end{align*}
Then by \cite[Lemma 7.8]{Tes}
\begin{align*}
\|u\|_{\ell^\infty(\mathbb Z^+)} \leq  8\|\varphi_-^\mathrm{r}\|_{\ell^\infty(\mathbb Z^+)} \|J_o(\alpha-\check \alpha,b -\check b)\|_{\mathcal M} \exp\left(8 \|J_o(\check \alpha,\check b)\|_{\mathcal M} \right)
\end{align*}
and hence if $J_o(\check \alpha, \check b) \to J_o(\alpha, b)$ in the topology of $\mathcal M$ then $f \to \varphi_-^\mathrm{r}$ uniformly.  Similar arguments follow for the continuity of $\varphi_-^\mathrm{l}$ in this topology.  This completes the proof.
\end{proof}

\subsection{A simple example.}
Suppose that $b_n =0$ for all $n\neq 0$, $b_0 = \beta$, and we have $a_n = \tfrac{1}{2}$ for all $n\in\mathbb{Z}$, then
\begin{equation}
\begin{aligned}
\varphi^\mathrm{r}_{+}(z;1)&= z,\quad
\varphi^\mathrm{r}_{+}(z;0)= 1,\quad
\varphi^\mathrm{r}_{-}(z;1)= z^{-1},\quad
\varphi^\mathrm{r}_{-}(z;0)= 1\\
\varphi^\mathrm{l}_{+}(z;1&)= z-2\beta,\quad
\varphi^\mathrm{l}_{+}(z;0)= 1,\quad
\varphi^\mathrm{l}_{-}(z;1)= z^{-1}-2\beta,\quad
\varphi^\mathrm{l}_{-}(z;0)= 1.
\end{aligned}
\end{equation}
Using these gives
\begin{equation}
R(z) = \frac{
\det\begin{bmatrix}
\varphi^\mathrm{l}_{-}(z;1) & \varphi^\mathrm{r}_{-}(z;1)\\  \frac{1}{2} \varphi^\mathrm{l}_{-}(z;0) & \frac{1}{2} \varphi^\mathrm{r}_{-}(z;0)
\end{bmatrix}
}{\det\begin{bmatrix}
\varphi^\mathrm{r}_{+}(z;1) &  \varphi^\mathrm{l}_{-}(z;1) \\  \frac{1}{2} \varphi^\mathrm{r}_{+}(z;0)& \frac{1}{2} \varphi^\mathrm{l}_{-}(z;0)
\end{bmatrix}}=\frac{2\beta}{z^{-1} - z - 2\beta}.
\end{equation}
If $\beta \neq 0$, then $R(\pm 1)=-1$. If $\beta=0$, then we have the free matrix $L_0$, which has $R(z)=0$ for all $z$ on the unit circle.
}


\end{document}